%% file: Truly_Concurrent_Process_Algebra_with_Localities.tex
\documentclass[11pt]{article} 

\usepackage{textcomp, amssymb}  
\usepackage{anysize}            
\usepackage{amsfonts}
\usepackage{amsmath}
\usepackage{dsfont}
\usepackage{MnSymbol}
\usepackage{mathtools}
\usepackage{graphicx}
\usepackage{epsfig}
\usepackage{epstopdf}
\usepackage{lmerge}
\usepackage{tikz}
\usepackage{amsthm}
\usepackage{ifpdf}

\usetikzlibrary{calc,decorations.pathmorphing,shapes,graphs}

\newcounter{sarrow}

\newcounter{sarrow1}
\newcommand\xnrsquigarrow[1]{%
\stepcounter{sarrow1}%
\mathrel{\begin{tikzpicture}[baseline= {( $ (current bounding box.south) + (0,-0.5ex) $ )}]
\node[inner sep=.5ex] (\thesarrow) {$\scriptstyle #1$};
\path[draw,<-,decorate,
  decoration={zigzag,amplitude=0.7pt,segment length=1.2mm,pre=lineto,pre length=4pt}]
    (\thesarrow1.south east) -- (\thesarrow1.south west);
    $\slashedarrowfill@\relbar\relbar/$
    \end{tikzpicture}}%
}

\makeatletter
\def\slashedarrowfill@#1#2#3#4#5{%
  $\m@th\thickmuskip0mu\medmuskip\thickmuskip\thinmuskip\thickmuskip
   \relax#5#1\mkern-7mu%
   \cleaders\hbox{$#5\mkern-2mu#2\mkern-2mu$}\hfill
   \mathclap{#3}\mathclap{#2}%
   \cleaders\hbox{$#5\mkern-2mu#2\mkern-2mu$}\hfill
   \mkern-7mu#4$%
}
\def\rightslashedarrowfillb@{%
  \slashedarrowfill@\relbar\relbar/\rightarrow}
\newcommand\xnrightarrow[2][]{%
  \ext@arrow 0055{\rightslashedarrowfillb@}{#1}{#2}}

\def\rightslashedarrowfille@{%
  \slashedarrowfill@\relbar\relbar/\twoheadrightarrow}
\newcommand\xntworightarrow[2][]{%
  \ext@arrow 0055{\rightslashedarrowfille@}{#1}{#2}}

\def\rightslashedarrowfillg@{%
  \slashedarrowfill@\relbar\relbar{\raisebox{.12em}{}}\twoheadrightarrow}
\newcommand\xtworightarrow[2][]{%
  \ext@arrow 0055{\rightslashedarrowfillg@}{#1}{#2}}

\def\rightslashedarrowfillx@{%
  \slashedarrowfill@\Relbar\Relbar/\rightrightarrows}
\newcommand\xnTworightarrow[2][]{%
  \ext@arrow 0055{\rightslashedarrowfillx@}{#1}{#2}}

\def\rightslashedarrowfilly@{%
  \slashedarrowfill@\Relbar\Relbar{\raisebox{.12em}{}}\rightrightarrows}
\newcommand\xTworightarrow[2][]{%
  \ext@arrow 0055{\rightslashedarrowfilly@}{#1}{#2}}

\pgfdeclareshape{slash underlined}
{
  \inheritsavedanchors[from=rectangle] 
  \inheritanchorborder[from=rectangle]
  \inheritanchor[from=rectangle]{north}
  \inheritanchor[from=rectangle]{north west}
  \inheritanchor[from=rectangle]{north east}
  \inheritanchor[from=rectangle]{center}
  \inheritanchor[from=rectangle]{west}
  \inheritanchor[from=rectangle]{east}
  \inheritanchor[from=rectangle]{mid}
  \inheritanchor[from=rectangle]{mid west}
  \inheritanchor[from=rectangle]{mid east}
  \inheritanchor[from=rectangle]{base}
  \inheritanchor[from=rectangle]{base west}
  \inheritanchor[from=rectangle]{base east}
  \inheritanchor[from=rectangle]{south}
  \inheritanchor[from=rectangle]{south west}
  \inheritanchor[from=rectangle]{south east}
  \inheritanchorborder[from=rectangle]
  \foregroundpath{
    \southwest \pgf@xa=\pgf@x \pgf@ya=\pgf@y
    \northeast \pgf@xb=\pgf@x \pgf@yb=\pgf@y
    \pgf@xc=\pgf@xa
    \advance\pgf@xc by .5\pgf@xb
    \pgf@yc=\pgf@ya
    \advance\pgf@xc by -1.3pt
    \advance\pgf@yc by -1.8pt
    \pgfpathmoveto{\pgfqpoint{\pgf@xc}{\pgf@yc}}
    \advance\pgf@xc by  2.6pt
    \advance\pgf@yc by  3.6pt
    \pgfpathlineto{\pgfqpoint{\pgf@xc}{\pgf@yc}}
    \pgfpathmoveto{\pgfqpoint{\pgf@xa}{\pgf@ya}}
    \pgfpathlineto{\pgfqpoint{\pgf@xb}{\pgf@ya}}
 }
}
\tikzset{nomorepostaction/.code=\let\tikz@postactions\pgfutil@empty}

\newtheorem{theorem}{Theorem}[section]
\newtheorem{definition}[theorem]{Definition}
\newtheorem{proposition}[theorem]{Proposition}
\newtheorem{lemma}[theorem]{Lemma}

\begin{document}

\begin{titlepage}
\thispagestyle{empty}

\hrule
\begin{center}
{\bf\LARGE Truly Concurrent Process Algebra with Localities}

\vspace{0.7cm}
--- Yong Wang ---

\vspace{2cm}
\begin{figure}[!htbp]
 \centering
 \includegraphics[width=1.0\textwidth]{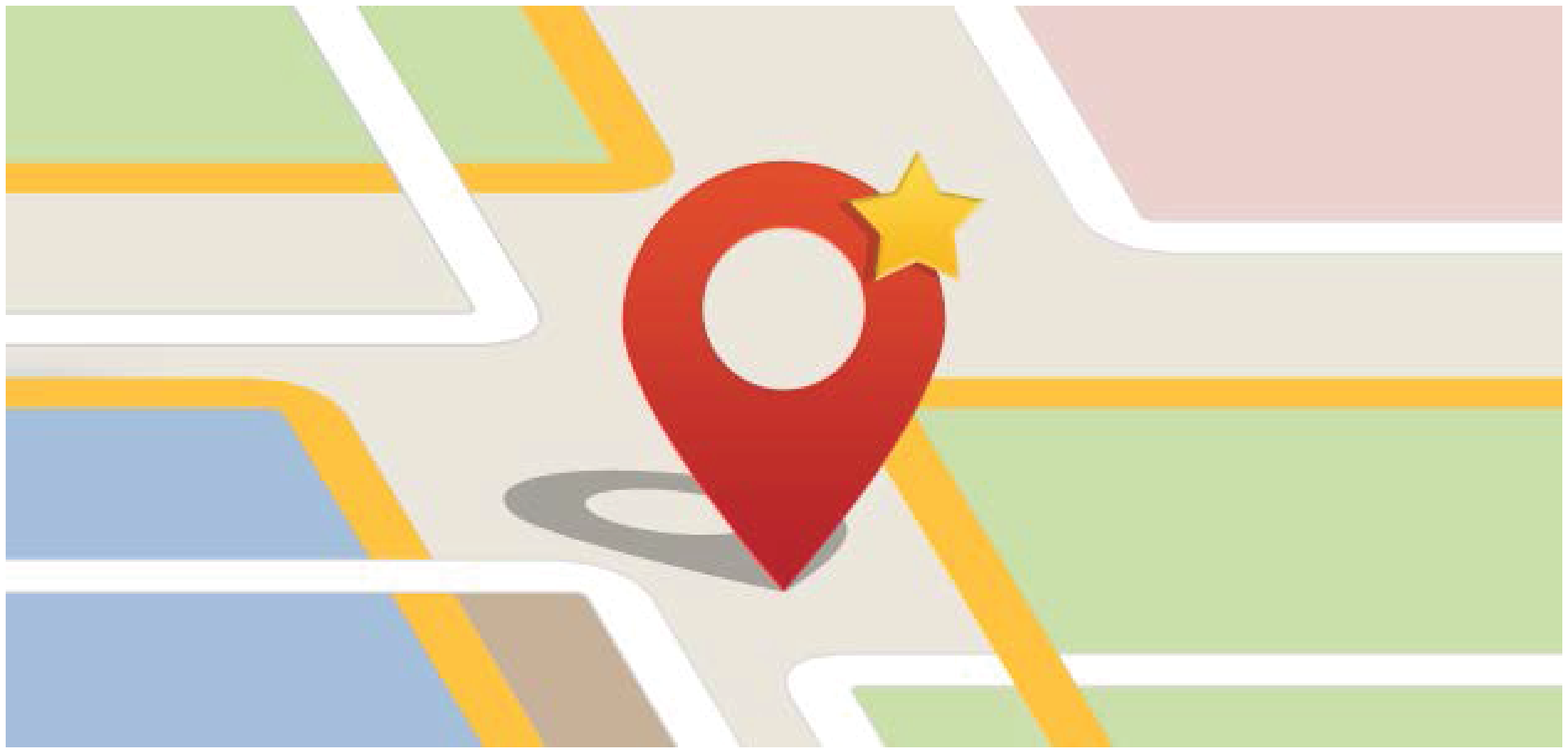}
\end{figure}

\end{center}
\end{titlepage}

\newpage 

\setcounter{page}{1}\pagenumbering{roman}

\tableofcontents

\newpage

\setcounter{page}{1}\pagenumbering{arabic}

        \input{section1.tex}
\newpage\input{section2.tex}

\newpage\input{section3.tex}
\newpage\input{section4.tex}

\newpage\input{section5.tex}

\newpage\input{section6.tex}
\newpage

\input{refs.tex}
\end{document}

%% file: section1.tex
\section{Introduction}

The well-known process algebras, such as CCS \cite{CC} \cite{CCS}, ACP \cite{ACP} and $\pi$-calculus \cite{PI1} \cite{PI2}, capture the interleaving concurrency based on bisimilarity semantics.
We did some work on truly concurrent process algebras, such as CTC \cite{CTC}, APTC \cite{ATC} and $\pi_{tc}$ \cite{PITC}, capture the true concurrency based on truly concurrent bisimilarities, such as
pomset bisimilarity, step bisimilarity, history-preserving (hp-) bisimilarity and hereditary history-preserving (hhp-) bisimilarity. Truly concurrent process algebras are generalizations
of the corresponding traditional process algebras.

In this book, we introduce localities into truly concurrent process algebras, based on the work on process algebra with localities \cite{LOC1}. We introduce the
preliminaries in chapter \ref{bg}. Truly concurrent bisimulations with localities in chapter \ref{osl}, CTC with localities in chapter \ref{ctcl}, APTC with localities in chapter \ref{aptcl}, 
$\pi_{tc}$ with localities in chapter \ref{pitcl}.

%% file: section2.tex
\section{Backgrounds}\label{bg}

To make this book self-satisfied, we introduce some preliminaries in this chapter, including some introductions on operational semantics, proof techniques, truly concurrent process algebra \cite{ATC} \cite{CTC} \cite{PITC}, which is based on truly
concurrent operational semantics.

\subsection{Operational Semantics}

The semantics of $ACP$ is based on bisimulation/rooted branching bisimulation equivalences, and the modularity of $ACP$ relies on the concept of conservative extension, for the
conveniences, we introduce some concepts and conclusions on them.

\begin{definition}[Bisimulation]
A bisimulation relation $R$ is a binary relation on processes such that: (1) if $p R q$ and $p\xrightarrow{a}p'$ then $q\xrightarrow{a}q'$ with $p' R q'$; (2) if $p R q$ and
$q\xrightarrow{a}q'$ then $p\xrightarrow{a}p'$ with $p' R q'$; (3) if $p R q$ and $pP$, then $qP$; (4) if $p R q$ and $qP$, then $pP$. Two processes $p$ and $q$ are bisimilar,
denoted by $p\sim_{HM} q$, if there is a bisimulation relation $R$ such that $p R q$.
\end{definition}

\begin{definition}[Congruence]
Let $\Sigma$ be a signature. An equivalence relation $R$ on $\mathcal{T}(\Sigma)$ is a congruence if for each $f\in\Sigma$, if $s_i R t_i$ for $i\in\{1,\cdots,ar(f)\}$, then
$f(s_1,\cdots,s_{ar(f)}) R f(t_1,\cdots,t_{ar(f)})$.
\end{definition}

\begin{definition}[Branching bisimulation]
A branching bisimulation relation $R$ is a binary relation on the collection of processes such that: (1) if $p R q$ and $p\xrightarrow{a}p'$ then either $a\equiv \tau$ and $p' R q$ or there is a sequence of (zero or more) $\tau$-transitions $q\xrightarrow{\tau}\cdots\xrightarrow{\tau}q_0$ such that $p R q_0$ and $q_0\xrightarrow{a}q'$ with $p' R q'$; (2) if $p R q$ and $q\xrightarrow{a}q'$ then either $a\equiv \tau$ and $p R q'$ or there is a sequence of (zero or more) $\tau$-transitions $p\xrightarrow{\tau}\cdots\xrightarrow{\tau}p_0$ such that $p_0 R q$ and $p_0\xrightarrow{a}p'$ with $p' R q'$; (3) if $p R q$ and $pP$, then there is a sequence of (zero or more) $\tau$-transitions $q\xrightarrow{\tau}\cdots\xrightarrow{\tau}q_0$ such that $p R q_0$ and $q_0P$; (4) if $p R q$ and $qP$, then there is a sequence of (zero or more) $\tau$-transitions $p\xrightarrow{\tau}\cdots\xrightarrow{\tau}p_0$ such that $p_0 R q$ and $p_0P$. Two processes $p$ and $q$ are branching bisimilar, denoted by $p\approx_{bHM} q$, if there is a branching bisimulation relation $R$ such that $p R q$.
\end{definition}

\begin{definition}[Rooted branching bisimulation]
A rooted branching bisimulation relation $R$ is a binary relation on processes such that: (1) if $p R q$ and $p\xrightarrow{a}p'$ then $q\xrightarrow{a}q'$ with $p'\approx_{bHM} q'$;
(2) if $p R q$ and $q\xrightarrow{a}q'$ then $p\xrightarrow{a}p'$ with $p'\approx_{bHM} q'$; (3) if $p R q$ and $pP$, then $qP$; (4) if $p R q$ and $qP$, then $pP$. Two processes $p$ and $q$ are rooted branching bisimilar, denoted by $p\approx_{rbHM} q$, if there is a rooted branching bisimulation relation $R$ such that $p R q$.
\end{definition}

\begin{definition}[Conservative extension]
Let $T_0$ and $T_1$ be TSSs (transition system specifications) over signatures $\Sigma_0$ and $\Sigma_1$, respectively. The TSS $T_0\oplus T_1$ is a conservative extension of $T_0$ if
the LTSs (labeled transition systems) generated by $T_0$ and $T_0\oplus T_1$ contain exactly the same transitions $t\xrightarrow{a}t'$ and $tP$ with $t\in \mathcal{T}(\Sigma_0)$.
\end{definition}

\begin{definition}[Source-dependency]
The source-dependent variables in a transition rule of $\rho$ are defined inductively as follows: (1) all variables in the source of $\rho$ are source-dependent; (2) if
$t\xrightarrow{a}t'$ is a premise of $\rho$ and all variables in $t$ are source-dependent, then all variables in $t'$ are source-dependent. A transition rule is source-dependent if
all its variables are. A TSS is source-dependent if all its rules are.
\end{definition}

\begin{definition}[Freshness]
Let $T_0$ and $T_1$ be TSSs over signatures $\Sigma_0$ and $\Sigma_1$, respectively. A term in $\mathbb{T}(T_0\oplus T_1)$ is said to be fresh if it contains a function symbol from
$\Sigma_1\setminus\Sigma_0$. Similarly, a transition label or predicate symbol in $T_1$ is fresh if it does not occur in $T_0$.
\end{definition}

\begin{theorem}[Conservative extension]
Let $T_0$ and $T_1$ be TSSs over signatures $\Sigma_0$ and $\Sigma_1$, respectively, where $T_0$ and $T_0\oplus T_1$ are positive after reduction. Under the following conditions,
$T_0\oplus T_1$ is a conservative extension of $T_0$. (1) $T_0$ is source-dependent. (2) For each $\rho\in T_1$, either the source of $\rho$ is fresh, or $\rho$ has a premise of the
form $t\xrightarrow{a}t'$ or $tP$, where $t\in \mathbb{T}(\Sigma_0)$, all variables in $t$ occur in the source of $\rho$ and $t'$, $a$ or $P$ is fresh.
\end{theorem}

\begin{definition}[Prime event structure with silent event]
Let $\Lambda$ be a fixed set of labels, ranged over $a,b,c,\cdots$ and $\tau$. A ($\Lambda$-labelled) prime event structure with silent event $\tau$ is a tuple
$\mathcal{E}=\langle \mathbb{E}, \leq, \sharp, \lambda\rangle$, where $\mathbb{E}$ is a denumerable set of events, including the silent event $\tau$. Let
$\hat{\mathbb{E}}=\mathbb{E}\backslash\{\tau\}$, exactly excluding $\tau$, it is obvious that $\hat{\tau^*}=\epsilon$, where $\epsilon$ is the empty event.
Let $\lambda:\mathbb{E}\rightarrow\Lambda$ be a labelling function and let $\lambda(\tau)=\tau$. And $\leq$, $\sharp$ are binary relations on $\mathbb{E}$,
called causality and conflict respectively, such that:

\begin{enumerate}
  \item $\leq$ is a partial order and $\lceil e \rceil = \{e'\in \mathbb{E}|e'\leq e\}$ is finite for all $e\in \mathbb{E}$. It is easy to see that
  $e\leq\tau^*\leq e'=e\leq\tau\leq\cdots\leq\tau\leq e'$, then $e\leq e'$.
  \item $\sharp$ is irreflexive, symmetric and hereditary with respect to $\leq$, that is, for all $e,e',e''\in \mathbb{E}$, if $e\sharp e'\leq e''$, then $e\sharp e''$.
\end{enumerate}

Then, the concepts of consistency and concurrency can be drawn from the above definition:

\begin{enumerate}
  \item $e,e'\in \mathbb{E}$ are consistent, denoted as $e\frown e'$, if $\neg(e\sharp e')$. A subset $X\subseteq \mathbb{E}$ is called consistent, if $e\frown e'$ for all
  $e,e'\in X$.
  \item $e,e'\in \mathbb{E}$ are concurrent, denoted as $e\parallel e'$, if $\neg(e\leq e')$, $\neg(e'\leq e)$, and $\neg(e\sharp e')$.
\end{enumerate}
\end{definition}

\begin{definition}[Configuration]
Let $\mathcal{E}$ be a PES. A (finite) configuration in $\mathcal{E}$ is a (finite) consistent subset of events $C\subseteq \mathcal{E}$, closed with respect to causality
(i.e. $\lceil C\rceil=C$). The set of finite configurations of $\mathcal{E}$ is denoted by $\mathcal{C}(\mathcal{E})$. We let $\hat{C}=C\backslash\{\tau\}$.
\end{definition}

A consistent subset of $X\subseteq \mathbb{E}$ of events can be seen as a pomset. Given $X, Y\subseteq \mathbb{E}$, $\hat{X}\sim \hat{Y}$ if $\hat{X}$ and $\hat{Y}$ are
isomorphic as pomsets. In the following of the paper, we say $C_1\sim C_2$, we mean $\hat{C_1}\sim\hat{C_2}$.

\begin{definition}[Pomset transitions and step]
Let $\mathcal{E}$ be a PES and let $C\in\mathcal{C}(\mathcal{E})$, and $\emptyset\neq X\subseteq \mathbb{E}$, if $C\cap X=\emptyset$ and $C'=C\cup X\in\mathcal{C}(\mathcal{E})$,
then $C\xrightarrow{X} C'$ is called a pomset transition from $C$ to $C'$. When the events in $X$ are pairwise concurrent, we say that $C\xrightarrow{X}C'$ is a step.
\end{definition}

\begin{definition}[Pomset, step bisimulation]\label{PSB}
Let $\mathcal{E}_1$, $\mathcal{E}_2$ be PESs. A pomset bisimulation is a relation $R\subseteq\mathcal{C}(\mathcal{E}_1)\times\mathcal{C}(\mathcal{E}_2)$, such that if
$(C_1,C_2)\in R$, and $C_1\xrightarrow{X_1}C_1'$ then $C_2\xrightarrow{X_2}C_2'$, with $X_1\subseteq \mathbb{E}_1$, $X_2\subseteq \mathbb{E}_2$, $X_1\sim X_2$ and $(C_1',C_2')\in R$,
and vice-versa. We say that $\mathcal{E}_1$, $\mathcal{E}_2$ are pomset bisimilar, written $\mathcal{E}_1\sim_p\mathcal{E}_2$, if there exists a pomset bisimulation $R$, such that
$(\emptyset,\emptyset)\in R$. By replacing pomset transitions with steps, we can get the definition of step bisimulation. When PESs $\mathcal{E}_1$ and $\mathcal{E}_2$ are step
bisimilar, we write $\mathcal{E}_1\sim_s\mathcal{E}_2$.
\end{definition}

\begin{definition}[Posetal product]
Given two PESs $\mathcal{E}_1$, $\mathcal{E}_2$, the posetal product of their configurations, denoted $\mathcal{C}(\mathcal{E}_1)\overline{\times}\mathcal{C}(\mathcal{E}_2)$,
is defined as

$$\{(C_1,f,C_2)|C_1\in\mathcal{C}(\mathcal{E}_1),C_2\in\mathcal{C}(\mathcal{E}_2),f:C_1\rightarrow C_2 \textrm{ isomorphism}\}.$$

A subset $R\subseteq\mathcal{C}(\mathcal{E}_1)\overline{\times}\mathcal{C}(\mathcal{E}_2)$ is called a posetal relation. We say that $R$ is downward closed when for any
$(C_1,f,C_2),(C_1',f',C_2')\in \mathcal{C}(\mathcal{E}_1)\overline{\times}\mathcal{C}(\mathcal{E}_2)$, if $(C_1,f,C_2)\subseteq (C_1',f',C_2')$ pointwise and $(C_1',f',C_2')\in R$,
then $(C_1,f,C_2)\in R$.

For $f:X_1\rightarrow X_2$, we define $f[x_1\mapsto x_2]:X_1\cup\{x_1\}\rightarrow X_2\cup\{x_2\}$, $z\in X_1\cup\{x_1\}$,(1)$f[x_1\mapsto x_2](z)=
x_2$,if $z=x_1$;(2)$f[x_1\mapsto x_2](z)=f(z)$, otherwise. Where $X_1\subseteq \mathbb{E}_1$, $X_2\subseteq \mathbb{E}_2$, $x_1\in \mathbb{E}_1$, $x_2\in \mathbb{E}_2$.
\end{definition}

\begin{definition}[(Hereditary) history-preserving bisimulation]\label{HHPB}
A history-preserving (hp-) bisimulation is a posetal relation $R\subseteq\mathcal{C}(\mathcal{E}_1)\overline{\times}\mathcal{C}(\mathcal{E}_2)$ such that if $(C_1,f,C_2)\in R$,
and $C_1\xrightarrow{e_1} C_1'$, then $C_2\xrightarrow{e_2} C_2'$, with $(C_1',f[e_1\mapsto e_2],C_2')\in R$, and vice-versa. $\mathcal{E}_1,\mathcal{E}_2$ are history-preserving
(hp-)bisimilar and are written $\mathcal{E}_1\sim_{hp}\mathcal{E}_2$ if there exists a hp-bisimulation $R$ such that $(\emptyset,\emptyset,\emptyset)\in R$.

A hereditary history-preserving (hhp-)bisimulation is a downward closed hp-bisimulation. $\mathcal{E}_1,\mathcal{E}_2$ are hereditary history-preserving (hhp-)bisimilar and are
written $\mathcal{E}_1\sim_{hhp}\mathcal{E}_2$.
\end{definition}

\begin{definition}[Weak pomset transitions and weak step]
Let $\mathcal{E}$ be a PES and let $C\in\mathcal{C}(\mathcal{E})$, and $\emptyset\neq X\subseteq \hat{\mathbb{E}}$, if $C\cap X=\emptyset$ and
$\hat{C'}=\hat{C}\cup X\in\mathcal{C}(\mathcal{E})$, then $C\xRightarrow{X} C'$ is called a weak pomset transition from $C$ to $C'$, where we define
$\xRightarrow{e}\triangleq\xrightarrow{\tau^*}\xrightarrow{e}\xrightarrow{\tau^*}$. And $\xRightarrow{X}\triangleq\xrightarrow{\tau^*}\xrightarrow{e}\xrightarrow{\tau^*}$,
for every $e\in X$. When the events in $X$ are pairwise concurrent, we say that $C\xRightarrow{X}C'$ is a weak step.
\end{definition}

\begin{definition}[Weak pomset, step bisimulation]
Let $\mathcal{E}_1$, $\mathcal{E}_2$ be PESs. A weak pomset bisimulation is a relation $R\subseteq\mathcal{C}(\mathcal{E}_1)\times\mathcal{C}(\mathcal{E}_2)$, such that if
$(C_1,C_2)\in R$, and $C_1\xRightarrow{X_1}C_1'$ then $C_2\xRightarrow{X_2}C_2'$, with $X_1\subseteq \hat{\mathbb{E}_1}$, $X_2\subseteq \hat{\mathbb{E}_2}$, $X_1\sim X_2$ and
$(C_1',C_2')\in R$, and vice-versa. We say that $\mathcal{E}_1$, $\mathcal{E}_2$ are weak pomset bisimilar, written $\mathcal{E}_1\approx_p\mathcal{E}_2$, if there exists a weak pomset
bisimulation $R$, such that $(\emptyset,\emptyset)\in R$. By replacing weak pomset transitions with weak steps, we can get the definition of weak step bisimulation. When PESs
$\mathcal{E}_1$ and $\mathcal{E}_2$ are weak step bisimilar, we write $\mathcal{E}_1\approx_s\mathcal{E}_2$.
\end{definition}

\begin{definition}[Weakly posetal product]
Given two PESs $\mathcal{E}_1$, $\mathcal{E}_2$, the weakly posetal product of their configurations, denoted $\mathcal{C}(\mathcal{E}_1)\overline{\times}\mathcal{C}(\mathcal{E}_2)$, is
defined as

$$\{(C_1,f,C_2)|C_1\in\mathcal{C}(\mathcal{E}_1),C_2\in\mathcal{C}(\mathcal{E}_2),f:\hat{C_1}\rightarrow \hat{C_2} \textrm{ isomorphism}\}.$$

A subset $R\subseteq\mathcal{C}(\mathcal{E}_1)\overline{\times}\mathcal{C}(\mathcal{E}_2)$ is called a weakly posetal relation. We say that $R$ is downward closed when for any
$(C_1,f,C_2),(C_1',f,C_2')\in \mathcal{C}(\mathcal{E}_1)\overline{\times}\mathcal{C}(\mathcal{E}_2)$, if $(C_1,f,C_2)\subseteq (C_1',f',C_2')$ pointwise and $(C_1',f',C_2')\in R$, then
$(C_1,f,C_2)\in R$.

For $f:X_1\rightarrow X_2$, we define $f[x_1\mapsto x_2]:X_1\cup\{x_1\}\rightarrow X_2\cup\{x_2\}$, $z\in X_1\cup\{x_1\}$,(1)$f[x_1\mapsto x_2](z)=
x_2$,if $z=x_1$;(2)$f[x_1\mapsto x_2](z)=f(z)$, otherwise. Where $X_1\subseteq \hat{\mathbb{E}_1}$, $X_2\subseteq \hat{\mathbb{E}_2}$, $x_1\in \hat{\mathbb{E}}_1$,
$x_2\in \hat{\mathbb{E}}_2$. Also, we define $f(\tau^*)=f(\tau^*)$.
\end{definition}

\begin{definition}[Weak (hereditary) history-preserving bisimulation]
A weak history-preserving (hp-) bisimulation is a weakly posetal relation $R\subseteq\mathcal{C}(\mathcal{E}_1)\overline{\times}\mathcal{C}(\mathcal{E}_2)$ such that if
$(C_1,f,C_2)\in R$, and $C_1\xRightarrow{e_1} C_1'$, then $C_2\xRightarrow{e_2} C_2'$, with $(C_1',f[e_1\mapsto e_2],C_2')\in R$, and vice-versa. $\mathcal{E}_1,\mathcal{E}_2$ are weak
history-preserving (hp-)bisimilar and are written $\mathcal{E}_1\approx_{hp}\mathcal{E}_2$ if there exists a weak hp-bisimulation $R$ such that $(\emptyset,\emptyset,\emptyset)\in R$.

A weakly hereditary history-preserving (hhp-)bisimulation is a downward closed weak hp-bisimulation. $\mathcal{E}_1,\mathcal{E}_2$ are weakly hereditary history-preserving
(hhp-)bisimilar and are written $\mathcal{E}_1\approx_{hhp}\mathcal{E}_2$.
\end{definition}

\begin{definition}[Branching pomset, step bisimulation]
Assume a special termination predicate $\downarrow$, and let $\surd$ represent a state with $\surd\downarrow$. Let $\mathcal{E}_1$, $\mathcal{E}_2$ be PESs. A branching pomset
bisimulation is a relation $R\subseteq\mathcal{C}(\mathcal{E}_1)\times\mathcal{C}(\mathcal{E}_2)$, such that:
 \begin{enumerate}
   \item if $(C_1,C_2)\in R$, and $C_1\xrightarrow{X}C_1'$ then
   \begin{itemize}
     \item either $X\equiv \tau^*$, and $(C_1',C_2)\in R$;
     \item or there is a sequence of (zero or more) $\tau$-transitions $C_2\xrightarrow{\tau^*} C_2^0$, such that $(C_1,C_2^0)\in R$ and $C_2^0\xRightarrow{X}C_2'$ with
     $(C_1',C_2')\in R$;
   \end{itemize}
   \item if $(C_1,C_2)\in R$, and $C_2\xrightarrow{X}C_2'$ then
   \begin{itemize}
     \item either $X\equiv \tau^*$, and $(C_1,C_2')\in R$;
     \item or there is a sequence of (zero or more) $\tau$-transitions $C_1\xrightarrow{\tau^*} C_1^0$, such that $(C_1^0,C_2)\in R$ and $C_1^0\xRightarrow{X}C_1'$ with
     $(C_1',C_2')\in R$;
   \end{itemize}
   \item if $(C_1,C_2)\in R$ and $C_1\downarrow$, then there is a sequence of (zero or more) $\tau$-transitions $C_2\xrightarrow{\tau^*}C_2^0$ such that $(C_1,C_2^0)\in R$
   and $C_2^0\downarrow$;
   \item if $(C_1,C_2)\in R$ and $C_2\downarrow$, then there is a sequence of (zero or more) $\tau$-transitions $C_1\xrightarrow{\tau^*}C_1^0$ such that $(C_1^0,C_2)\in R$
   and $C_1^0\downarrow$.
 \end{enumerate}

We say that $\mathcal{E}_1$, $\mathcal{E}_2$ are branching pomset bisimilar, written $\mathcal{E}_1\approx_{bp}\mathcal{E}_2$, if there exists a branching pomset bisimulation $R$,
such that $(\emptyset,\emptyset)\in R$.

By replacing pomset transitions with steps, we can get the definition of branching step bisimulation. When PESs $\mathcal{E}_1$ and $\mathcal{E}_2$ are branching step bisimilar,
we write $\mathcal{E}_1\approx_{bs}\mathcal{E}_2$.
\end{definition}

\begin{definition}[Rooted branching pomset, step bisimulation]
Assume a special termination predicate $\downarrow$, and let $\surd$ represent a state with $\surd\downarrow$. Let $\mathcal{E}_1$, $\mathcal{E}_2$ be PESs. A rooted branching pomset
bisimulation is a relation $R\subseteq\mathcal{C}(\mathcal{E}_1)\times\mathcal{C}(\mathcal{E}_2)$, such that:
 \begin{enumerate}
   \item if $(C_1,C_2)\in R$, and $C_1\xrightarrow{X}C_1'$ then $C_2\xrightarrow{X}C_2'$ with $C_1'\approx_{bp}C_2'$;
   \item if $(C_1,C_2)\in R$, and $C_2\xrightarrow{X}C_2'$ then $C_1\xrightarrow{X}C_1'$ with $C_1'\approx_{bp}C_2'$;
   \item if $(C_1,C_2)\in R$ and $C_1\downarrow$, then $C_2\downarrow$;
   \item if $(C_1,C_2)\in R$ and $C_2\downarrow$, then $C_1\downarrow$.
 \end{enumerate}

We say that $\mathcal{E}_1$, $\mathcal{E}_2$ are rooted branching pomset bisimilar, written $\mathcal{E}_1\approx_{rbp}\mathcal{E}_2$, if there exists a rooted branching pomset
bisimulation $R$, such that $(\emptyset,\emptyset)\in R$.

By replacing pomset transitions with steps, we can get the definition of rooted branching step bisimulation. When PESs $\mathcal{E}_1$ and $\mathcal{E}_2$ are rooted branching step
bisimilar, we write $\mathcal{E}_1\approx_{rbs}\mathcal{E}_2$.
\end{definition}

\begin{definition}[Branching (hereditary) history-preserving bisimulation]
Assume a special termination predicate $\downarrow$, and let $\surd$ represent a state with $\surd\downarrow$. A branching history-preserving (hp-) bisimulation is a posetal
relation $R\subseteq\mathcal{C}(\mathcal{E}_1)\overline{\times}\mathcal{C}(\mathcal{E}_2)$ such that:

 \begin{enumerate}
   \item if $(C_1,f,C_2)\in R$, and $C_1\xrightarrow{e_1}C_1'$ then
   \begin{itemize}
     \item either $e_1\equiv \tau$, and $(C_1',f[e_1\mapsto \tau],C_2)\in R$;
     \item or there is a sequence of (zero or more) $\tau$-transitions $C_2\xrightarrow{\tau^*} C_2^0$, such that $(C_1,f,C_2^0)\in R$ and $C_2^0\xrightarrow{e_2}C_2'$ with
     $(C_1',f[e_1\mapsto e_2],C_2')\in R$;
   \end{itemize}
   \item if $(C_1,f,C_2)\in R$, and $C_2\xrightarrow{e_2}C_2'$ then
   \begin{itemize}
     \item either $X\equiv \tau$, and $(C_1,f[e_2\mapsto \tau],C_2')\in R$;
     \item or there is a sequence of (zero or more) $\tau$-transitions $C_1\xrightarrow{\tau^*} C_1^0$, such that $(C_1^0,f,C_2)\in R$ and $C_1^0\xrightarrow{e_1}C_1'$ with
     $(C_1',f[e_2\mapsto e_1],C_2')\in R$;
   \end{itemize}
   \item if $(C_1,f,C_2)\in R$ and $C_1\downarrow$, then there is a sequence of (zero or more) $\tau$-transitions $C_2\xrightarrow{\tau^*}C_2^0$ such that $(C_1,f,C_2^0)\in R$
   and $C_2^0\downarrow$;
   \item if $(C_1,f,C_2)\in R$ and $C_2\downarrow$, then there is a sequence of (zero or more) $\tau$-transitions $C_1\xrightarrow{\tau^*}C_1^0$ such that $(C_1^0,f,C_2)\in R$
   and $C_1^0\downarrow$.
 \end{enumerate}

$\mathcal{E}_1,\mathcal{E}_2$ are branching history-preserving (hp-)bisimilar and are written $\mathcal{E}_1\approx_{bhp}\mathcal{E}_2$ if there exists a branching hp-bisimulation
$R$ such that $(\emptyset,\emptyset,\emptyset)\in R$.

A branching hereditary history-preserving (hhp-)bisimulation is a downward closed branching hp-bisimulation. $\mathcal{E}_1,\mathcal{E}_2$ are branching hereditary history-preserving
(hhp-)bisimilar and are written $\mathcal{E}_1\approx_{bhhp}\mathcal{E}_2$.
\end{definition}

\begin{definition}[Rooted branching (hereditary) history-preserving bisimulation]
Assume a special termination predicate $\downarrow$, and let $\surd$ represent a state with $\surd\downarrow$. A rooted branching history-preserving (hp-) bisimulation is a posetal relation $R\subseteq\mathcal{C}(\mathcal{E}_1)\overline{\times}\mathcal{C}(\mathcal{E}_2)$ such that:

 \begin{enumerate}
   \item if $(C_1,f,C_2)\in R$, and $C_1\xrightarrow{e_1}C_1'$, then $C_2\xrightarrow{e_2}C_2'$ with $C_1'\approx_{bhp}C_2'$;
   \item if $(C_1,f,C_2)\in R$, and $C_2\xrightarrow{e_2}C_2'$, then $C_1\xrightarrow{e_1}C_1'$ with $C_1'\approx_{bhp}C_2'$;
   \item if $(C_1,f,C_2)\in R$ and $C_1\downarrow$, then $C_2\downarrow$;
   \item if $(C_1,f,C_2)\in R$ and $C_2\downarrow$, then $C_1\downarrow$.
 \end{enumerate}

$\mathcal{E}_1,\mathcal{E}_2$ are rooted branching history-preserving (hp-)bisimilar and are written $\mathcal{E}_1\approx_{rbhp}\mathcal{E}_2$ if there exists a rooted branching
hp-bisimulation $R$ such that $(\emptyset,\emptyset,\emptyset)\in R$.

A rooted branching hereditary history-preserving (hhp-)bisimulation is a downward closed rooted branching hp-bisimulation. $\mathcal{E}_1,\mathcal{E}_2$ are rooted branching
hereditary history-preserving (hhp-)bisimilar and are written $\mathcal{E}_1\approx_{rbhhp}\mathcal{E}_2$.
\end{definition}

\subsection{Proof Techniques}\label{PT}

In this subsection, we introduce the concepts and conclusions about elimination, which is very important in the proof of completeness theorem.

\begin{definition}[Elimination property]
Let a process algebra with a defined set of basic terms as a subset of the set of closed terms over the process algebra. Then the process algebra has the elimination to basic terms
property if for every closed term $s$ of the algebra, there exists a basic term $t$ of the algebra such that the algebra$\vdash s=t$.
\end{definition}

\begin{definition}[Strongly normalizing]
A term $s_0$ is called strongly normalizing if does not an infinite series of reductions beginning in $s_0$.
\end{definition}

\begin{definition}
We write $s>_{lpo} t$ if $s\rightarrow^+ t$ where $\rightarrow^+$ is the transitive closure of the reduction relation defined by the transition rules of an algebra.
\end{definition}

\begin{theorem}[Strong normalization]\label{SN}
Let a term rewriting system (TRS) with finitely many rewriting rules and let $>$ be a well-founded ordering on the signature of the corresponding algebra. If $s>_{lpo} t$ for each
rewriting rule $s\rightarrow t$ in the TRS, then the term rewriting system is strongly normalizing.
\end{theorem}

\subsection{CTC}

CTC \cite{CTC} is a calculus of truly concurrent systems. It includes syntax and semantics:

\begin{enumerate}
  \item Its syntax includes actions, process constant, and operators acting between actions, like Prefix, Summation, Composition, Restriction, Relabelling.
  \item Its semantics is based on labeled transition systems, Prefix, Summation, Composition, Restriction, Relabelling have their transition rules. CTC has good semantic properties
  based on the truly concurrent bisimulations. These properties include monoid laws, static laws, new expansion law for strongly truly concurrent bisimulations, $\tau$ laws for weakly
  truly concurrent bisimulations, and full congruences for strongly and weakly truly concurrent bisimulations, and also unique solution for recursion.
\end{enumerate}

\subsection{APTC}

$APTC$ \cite{ATC} captures several computational properties in the form of algebraic laws, and proves the soundness and completeness modulo truly concurrent bisimulation/rooted branching truly concurrent bisimulation equivalence. These computational properties are organized in a modular way by use of the concept of conservational extension, which include the following modules, note that, every algebra are composed of constants and operators, the constants are the computational objects, while operators capture the computational properties.

\begin{enumerate}
  \item \textbf{$BATC$ (Basic Algebras for True Concurrency)}. $BATC$ has sequential composition $\cdot$ and alternative composition $+$ to capture causality computation and conflict.
  The constants are ranged over $\mathbb{E}$, the set of atomic events. The algebraic laws on $\cdot$ and $+$ are sound and complete modulo truly concurrent bisimulation equivalences,
  such as pomset bisimulation $\sim_p$, step bisimulation $\sim_s$, history-preserving (hp-) bisimulation $\sim_{hp}$ and hereditary history-preserving (hhp-) bisimulation $\sim_{hhp}$.
  \item \textbf{$APTC$ (Algebra for Parallelism for True Concurrency)}. $APTC$ uses the whole parallel operator $\between$, the parallel operator $\parallel$ to model parallelism, and
  the communication merge $\mid$ to model causality (communication) among different parallel branches. Since a communication may be blocked, a new constant called deadlock $\delta$ is
  extended to $\mathbb{E}$, and also a new unary encapsulation operator $\partial_H$ is introduced to eliminate $\delta$, which may exist in the processes. And also a conflict
  elimination operator $\Theta$ to eliminate conflicts existing in different parallel branches. The algebraic laws on these operators are also sound and complete modulo truly concurrent
  bisimulation equivalences, such as pomset bisimulation $\sim_p$, step bisimulation $\sim_s$, history-preserving (hp-) bisimulation $\sim_{hp}$. Note that, these operators in a process
  except the parallel operator $\parallel$ can be eliminated by deductions on the process using axioms of $APTC$, and eventually be steadied by $\cdot$, $+$ and $\parallel$, this is
  also why bisimulations are called an \emph{truly concurrent} semantics.
  \item \textbf{Recursion}. To model infinite computation, recursion is introduced into $APTC$. In order to obtain a sound and complete theory, guarded recursion and linear recursion
  are needed. The corresponding axioms are $RSP$ (Recursive Specification Principle) and $RDP$ (Recursive Definition Principle), $RDP$ says the solutions of a recursive specification
  can represent the behaviors of the specification, while $RSP$ says that a guarded recursive specification has only one solution, they are sound with respect to $APTC$ with guarded
  recursion modulo truly concurrent bisimulation equivalences, such as pomset bisimulation $\sim_p$, step bisimulation $\sim_s$, history-preserving (hp-) bisimulation $\sim_{hp}$, and
  they are complete with respect to $APTC$ with linear recursion modulo truly concurrent bisimulation equivalence, such as pomset bisimulation $\sim_p$, step bisimulation $\sim_s$,
  history-preserving (hp-) bisimulation $\sim_{hp}$.
  \item \textbf{Abstraction}. To abstract away internal implementations from the external behaviors, a new constant $\tau$ called silent step is added to $\mathbb{E}$, and also a new
  unary abstraction operator $\tau_I$ is used to rename actions in $I$ into $\tau$ (the resulted $APTC$ with silent step and abstraction operator is called $APTC_{\tau}$). The recursive
  specification is adapted to guarded linear recursion to prevent infinite $\tau$-loops specifically. The axioms for $\tau$ and $\tau_I$ are sound modulo rooted branching truly
  concurrent bisimulation equivalences (a kind of weak truly concurrent bisimulation equivalence), such as rooted branching pomset bisimulation $\approx_p$, rooted branching step
  bisimulation $\approx_s$, rooted branching history-preserving (hp-) bisimulation $\approx_{hp}$. To eliminate infinite $\tau$-loops caused by $\tau_I$ and obtain the completeness,
  $CFAR$ (Cluster Fair Abstraction Rule) is used to prevent infinite $\tau$-loops in a constructible way.
\end{enumerate}

\subsection{$\pi_{tc}$}

$\pi_{tc}$ \cite{PITC} is a calculus of truly concurrent mobile processes. It includes syntax and semantics:

\begin{enumerate}
  \item Its syntax includes actions, process constant, and operators acting between actions, like Prefix, Summation, Composition, Restriction, Input and Output.
  \item Its semantics is based on labeled transition systems, Prefix, Summation, Composition, Restriction, Input and Output have their transition rules. $\pi_{tc}$ has good semantic properties
  based on the truly concurrent bisimulations. These properties include summation laws, identity laws, restriction laws, parallel laws, expansion laws, congruences, and also unique solution for recursion.
\end{enumerate}

%% file: section3.tex
\section{Location-related Bisimulation Semantics for True Concurrency}\label{osl}

In this chapter, we give the location-related bisimulation semantics for true concurrency. This chapter is organized as follows. We introduce the static location bisimulations in section
\ref{slb}, the dynamic location bisimulations in section \ref{dlb}.

\input{section3/section3.1.tex}

\input{section3/section3.2.tex}

%% file: section3/section3.1.tex
\subsection{Static Location Bisimulations}\label{slb}

Let $Loc$ be the set of locations, and $u,v\in Loc^*$. Let $\ll$ be the sequential ordering on $Loc^*$, we call $v$ is an extension or a sublocation of $u$ in $u\ll v$; and if $u\nll v$
$v\nll u$, then $u$ and $v$ are independent and denoted $u\diamond v$.

\begin{definition}[Consistent location association]
A relation $\varphi\subseteq (Loc^*\times Loc^*)$ is a consistent location association (cla), if $(u,v)\in \varphi \&(u',v')\in\varphi$, then $u\diamond u'\Leftrightarrow v\diamond v'$.
\end{definition}

\begin{definition}[Static location pomset, step bisimulation]
Let $\mathcal{E}_1$, $\mathcal{E}_2$ be PESs. A static location pomset bisimulation is a relation $R_{\varphi}\subseteq\mathcal{C}(\mathcal{E}_1)\times\mathcal{C}(\mathcal{E}_2)$, such that if
$(C_1,C_2)\in R_{\varphi}$, and $C_1\xrightarrow[u]{X_1}C_1'$ then $C_2\xrightarrow[v]{X_2}C_2'$, with $X_1\subseteq \mathbb{E}_1$, $X_2\subseteq \mathbb{E}_2$, $X_1\sim X_2$ and
$(C_1',C_2')\in R_{\varphi\cup\{(u,v)\}}$,
and vice-versa. We say that $\mathcal{E}_1$, $\mathcal{E}_2$ are static location pomset bisimilar, written $\mathcal{E}_1\sim_p^{sl}\mathcal{E}_2$, if there exists a static location pomset bisimulation $R_{\varphi}$, such that
$(\emptyset,\emptyset)\in R_{\varphi}$. By replacing pomset transitions with steps, we can get the definition of static location step bisimulation. When PESs $\mathcal{E}_1$ and $\mathcal{E}_2$ are static location step
bisimilar, we write $\mathcal{E}_1\sim_s^{sl}\mathcal{E}_2$.
\end{definition}

\begin{definition}[Static location (hereditary) history-preserving bisimulation]
A static location history-preserving (hp-) bisimulation is a posetal relation $R_{\varphi}\subseteq\mathcal{C}(\mathcal{E}_1)\overline{\times}\mathcal{C}(\mathcal{E}_2)$ such that if $(C_1,f,C_2)\in R_{\varphi}$,
and $C_1\xrightarrow[u]{e_1} C_1'$, then $C_2\xrightarrow[v]{e_2} C_2'$, with $(C_1',f[e_1\mapsto e_2],C_2')\in R_{\varphi\cup\{(u,v)\}}$, and vice-versa. $\mathcal{E}_1,\mathcal{E}_2$ are static location history-preserving
(hp-)bisimilar and are written $\mathcal{E}_1\sim_{hp}^{sl}\mathcal{E}_2$ if there exists a static location hp-bisimulation $R_{\varphi}$ such that $(\emptyset,\emptyset,\emptyset)\in R_{\varphi}$.

A static location hereditary history-preserving (hhp-)bisimulation is a downward closed static location hp-bisimulation. $\mathcal{E}_1,\mathcal{E}_2$ are static location hereditary history-preserving (hhp-)bisimilar and are
written $\mathcal{E}_1\sim_{hhp}^{sl}\mathcal{E}_2$.
\end{definition}

\begin{definition}[Weak static location pomset, step bisimulation]
Let $\mathcal{E}_1$, $\mathcal{E}_2$ be PESs. A weak static location pomset bisimulation is a relation $R_{\varphi}\subseteq\mathcal{C}(\mathcal{E}_1)\times\mathcal{C}(\mathcal{E}_2)$, such that if
$(C_1,C_2)\in R_{\varphi}$, and $C_1\xRightarrow[u]{X_1}C_1'$ then $C_2\xRightarrow[v]{X_2}C_2'$, with $X_1\subseteq \hat{\mathbb{E}_1}$, $X_2\subseteq \hat{\mathbb{E}_2}$, $X_1\sim X_2$ and
$(C_1',C_2')\in R_{\varphi\cup\{(u,v)\}}$, and vice-versa. We say that $\mathcal{E}_1$, $\mathcal{E}_2$ are weak static location pomset bisimilar, written $\mathcal{E}_1\approx_p^{sl}\mathcal{E}_2$, if there exists a weak static location pomset
bisimulation $R_{\varphi}$, such that $(\emptyset,\emptyset)\in R_{\varphi}$. By replacing weak pomset transitions with weak steps, we can get the definition of weak static location step bisimulation. When PESs
$\mathcal{E}_1$ and $\mathcal{E}_2$ are weak static location step bisimilar, we write $\mathcal{E}_1\approx_s^{sl}\mathcal{E}_2$.
\end{definition}

\begin{definition}[Weak static location (hereditary) history-preserving bisimulation]
A weak static location history-preserving (hp-) bisimulation is a weakly posetal relation $R_{\varphi}\subseteq\mathcal{C}(\mathcal{E}_1)\overline{\times}\mathcal{C}(\mathcal{E}_2)$ such that if
$(C_1,f,C_2)\in R_{\varphi}$, and $C_1\xRightarrow[u]{e_1} C_1'$, then $C_2\xRightarrow[v]{e_2} C_2'$, with $(C_1',f[e_1\mapsto e_2],C_2')\in R_{\varphi\cup\{(u,v)\}}$, and vice-versa. $\mathcal{E}_1,\mathcal{E}_2$ are weak
static location history-preserving (hp-)bisimilar and are written $\mathcal{E}_1\approx_{hp}^{sl}\mathcal{E}_2$ if there exists a weak static location hp-bisimulation $R_{\varphi}$ such that $(\emptyset,\emptyset,\emptyset)\in R_{\varphi}$.

A weak static location hereditary history-preserving (hhp-)bisimulation is a downward closed weak static location hp-bisimulation. $\mathcal{E}_1,\mathcal{E}_2$ are weak static location hereditary history-preserving
(hhp-)bisimilar and are written $\mathcal{E}_1\approx_{hhp}^{sl}\mathcal{E}_2$.
\end{definition}

\begin{definition}[Branching static location pomset, step bisimulation]
Assume a special termination predicate $\downarrow$, and let $\surd$ represent a state with $\surd\downarrow$. Let $\mathcal{E}_1$, $\mathcal{E}_2$ be PESs. A branching static location pomset
bisimulation is a relation $R_{\varphi}\subseteq\mathcal{C}(\mathcal{E}_1)\times\mathcal{C}(\mathcal{E}_2)$, such that:
 \begin{enumerate}
   \item if $(C_1,C_2)\in R_{\varphi}$, and $C_1\xrightarrow[u]{X}C_1'$ then
   \begin{itemize}
     \item either $X\equiv \tau^*$, and $(C_1',C_2)\in R_{\varphi}$;
     \item or there is a sequence of (zero or more) $\tau$-transitions $C_2\xrightarrow{\tau^*} C_2^0$, such that $(C_1,C_2^0)\in R_{\varphi}$ and $C_2^0\xRightarrow[v]{X}C_2'$ with
     $(C_1',C_2')\in R_{\varphi\cup\{(u,v)\}}$;
   \end{itemize}
   \item if $(C_1,C_2)\in R_{\varphi}$, and $C_2\xrightarrow[v]{X}C_2'$ then
   \begin{itemize}
     \item either $X\equiv \tau^*$, and $(C_1,C_2')\in R_{\varphi}$;
     \item or there is a sequence of (zero or more) $\tau$-transitions $C_1\xrightarrow{\tau^*} C_1^0$, such that $(C_1^0,C_2)\in R_{\varphi}$ and $C_1^0\xRightarrow[u]{X}C_1'$ with
     $(C_1',C_2')\in R_{\varphi\cup\{(u,v)\}}$;
   \end{itemize}
   \item if $(C_1,C_2)\in R_{\varphi}$ and $C_1\downarrow$, then there is a sequence of (zero or more) $\tau$-transitions $C_2\xrightarrow{\tau^*}C_2^0$ such that $(C_1,C_2^0)\in R_{\varphi}$
   and $C_2^0\downarrow$;
   \item if $(C_1,C_2)\in R_{\varphi}$ and $C_2\downarrow$, then there is a sequence of (zero or more) $\tau$-transitions $C_1\xrightarrow{\tau^*}C_1^0$ such that $(C_1^0,C_2)\in R_{\varphi}$
   and $C_1^0\downarrow$.
 \end{enumerate}

We say that $\mathcal{E}_1$, $\mathcal{E}_2$ are branching static location pomset bisimilar, written $\mathcal{E}_1\approx_{bp}^{sl}\mathcal{E}_2$, if there exists a branching static location pomset bisimulation $R_{\varphi}$,
such that $(\emptyset,\emptyset)\in R_{\varphi}$.

By replacing pomset transitions with steps, we can get the definition of branching static location step bisimulation. When PESs $\mathcal{E}_1$ and $\mathcal{E}_2$ are branching static location step bisimilar,
we write $\mathcal{E}_1\approx_{bs}^{sl}\mathcal{E}_2$.
\end{definition}

\begin{definition}[Rooted branching static location pomset, step bisimulation]
Assume a special termination predicate $\downarrow$, and let $\surd$ represent a state with $\surd\downarrow$. Let $\mathcal{E}_1$, $\mathcal{E}_2$ be PESs. A rooted branching static location pomset
bisimulation is a relation $R_{\varphi}\subseteq\mathcal{C}(\mathcal{E}_1)\times\mathcal{C}(\mathcal{E}_2)$, such that:
 \begin{enumerate}
   \item if $(C_1,C_2)\in R_{\varphi}$, and $C_1\xrightarrow[u]{X}C_1'$ then $C_2\xrightarrow[v]{X}C_2'$ with $C_1'\approx_{bp}^{sl}C_2'$;
   \item if $(C_1,C_2)\in R_{\varphi}$, and $C_2\xrightarrow[v]{X}C_2'$ then $C_1\xrightarrow[u]{X}C_1'$ with $C_1'\approx_{bp}^{sl}C_2'$;
   \item if $(C_1,C_2)\in R_{\varphi}$ and $C_1\downarrow$, then $C_2\downarrow$;
   \item if $(C_1,C_2)\in R_{\varphi}$ and $C_2\downarrow$, then $C_1\downarrow$.
 \end{enumerate}

We say that $\mathcal{E}_1$, $\mathcal{E}_2$ are rooted branching static location pomset bisimilar, written $\mathcal{E}_1\approx_{rbp}^{sl}\mathcal{E}_2$, if there exists a rooted branching static location pomset
bisimulation $R_{\varphi}$, such that $(\emptyset,\emptyset)\in R_{\varphi}$.

By replacing pomset transitions with steps, we can get the definition of rooted branching static location step bisimulation. When PESs $\mathcal{E}_1$ and $\mathcal{E}_2$ are rooted branching static location step
bisimilar, we write $\mathcal{E}_1\approx_{rbs}^{sl}\mathcal{E}_2$.
\end{definition}

\begin{definition}[Branching static location (hereditary) history-preserving bisimulation]
Assume a special termination predicate $\downarrow$, and let $\surd$ represent a state with $\surd\downarrow$. A branching static location history-preserving (hp-) bisimulation is a posetal
relation $R_{\varphi}\subseteq\mathcal{C}(\mathcal{E}_1)\overline{\times}\mathcal{C}(\mathcal{E}_2)$ such that:

 \begin{enumerate}
   \item if $(C_1,f,C_2)\in R$, and $C_1\xrightarrow[u]{e_1}C_1'$ then
   \begin{itemize}
     \item either $e_1\equiv \tau$, and $(C_1',f[e_1\mapsto \tau],C_2)\in R_{\varphi}$;
     \item or there is a sequence of (zero or more) $\tau$-transitions $C_2\xrightarrow{\tau^*} C_2^0$, such that $(C_1,f,C_2^0)\in R_{\varphi}$ and $C_2^0\xrightarrow[v]{e_2}C_2'$ with
     $(C_1',f[e_1\mapsto e_2],C_2')\in R_{\varphi\cup\{(u,v)\}}$;
   \end{itemize}
   \item if $(C_1,f,C_2)\in R_{\varphi}$, and $C_2\xrightarrow[v]{e_2}C_2'$ then
   \begin{itemize}
     \item either $X\equiv \tau$, and $(C_1,f[e_2\mapsto \tau],C_2')\in R_{\varphi}$;
     \item or there is a sequence of (zero or more) $\tau$-transitions $C_1\xrightarrow{\tau^*} C_1^0$, such that $(C_1^0,f,C_2)\in R_{\varphi}$ and $C_1^0\xrightarrow[u]{e_1}C_1'$ with
     $(C_1',f[e_2\mapsto e_1],C_2')\in R_{\varphi\cup\{(u,v)\}}$;
   \end{itemize}
   \item if $(C_1,f,C_2)\in R_{\varphi}$ and $C_1\downarrow$, then there is a sequence of (zero or more) $\tau$-transitions $C_2\xrightarrow{\tau^*}C_2^0$ such that $(C_1,f,C_2^0)\in R_{\varphi}$
   and $C_2^0\downarrow$;
   \item if $(C_1,f,C_2)\in R_{\varphi}$ and $C_2\downarrow$, then there is a sequence of (zero or more) $\tau$-transitions $C_1\xrightarrow{\tau^*}C_1^0$ such that $(C_1^0,f,C_2)\in R_{\varphi}$
   and $C_1^0\downarrow$.
 \end{enumerate}

$\mathcal{E}_1,\mathcal{E}_2$ are branching static location history-preserving (hp-)bisimilar and are written $\mathcal{E}_1\approx_{bhp}^{sl}\mathcal{E}_2$ if there exists a branching static location hp-bisimulation
$R_{\varphi}$ such that $(\emptyset,\emptyset,\emptyset)\in R_{\varphi}$.

A branching static location hereditary history-preserving (hhp-)bisimulation is a downward closed branching static location hhp-bisimulation. $\mathcal{E}_1,\mathcal{E}_2$ are branching static location hereditary history-preserving
(hhp-)bisimilar and are written $\mathcal{E}_1\approx_{bhhp}^{sl}\mathcal{E}_2$.
\end{definition}

\begin{definition}[Rooted branching static location (hereditary) history-preserving bisimulation]
Assume a special termination predicate $\downarrow$, and let $\surd$ represent a state with $\surd\downarrow$. A rooted branching static location history-preserving (hp-) bisimulation is a posetal relation $R_{\varphi}\subseteq\mathcal{C}(\mathcal{E}_1)\overline{\times}\mathcal{C}(\mathcal{E}_2)$ such that:

 \begin{enumerate}
   \item if $(C_1,f,C_2)\in R_{\varphi}$, and $C_1\xrightarrow[u]{e_1}C_1'$, then $C_2\xrightarrow[v]{e_2}C_2'$ with $C_1'\approx_{bhp}^{sl}C_2'$;
   \item if $(C_1,f,C_2)\in R_{\varphi}$, and $C_2\xrightarrow[v]{e_2}C_2'$, then $C_1\xrightarrow[u]{e_1}C_1'$ with $C_1'\approx_{bhp}^{sl}C_2'$;
   \item if $(C_1,f,C_2)\in R_{\varphi}$ and $C_1\downarrow$, then $C_2\downarrow$;
   \item if $(C_1,f,C_2)\in R_{\varphi}$ and $C_2\downarrow$, then $C_1\downarrow$.
 \end{enumerate}

$\mathcal{E}_1,\mathcal{E}_2$ are rooted branching static location history-preserving (hp-)bisimilar and are written $\mathcal{E}_1\approx_{rbhp}^{sl}\mathcal{E}_2$ if there exists a rooted branching
static location hp-bisimulation $R_{\varphi}$ such that $(\emptyset,\emptyset,\emptyset)\in R_{\varphi}$.

A rooted branching static location hereditary history-preserving (hhp-)bisimulation is a downward closed rooted branching static location hp-bisimulation. $\mathcal{E}_1,\mathcal{E}_2$ are rooted branching
static location hereditary history-preserving (hhp-)bisimilar and are written $\mathcal{E}_1\approx_{rbhhp}^{sl}\mathcal{E}_2$.
\end{definition}

%% file: section3/section3.2.tex
\subsection{Dynamic Location Bisimulations}\label{dlb}

We assume that $u,v\in Loc^+$.

\begin{definition}[Dynamic location pomset, step bisimulation]
Let $\mathcal{E}_1$, $\mathcal{E}_2$ be PESs. A dynamic location pomset bisimulation is a relation $R\subseteq\mathcal{C}(\mathcal{E}_1)\times\mathcal{C}(\mathcal{E}_2)$, such that if
$(C_1,C_2)\in R$, and $C_1\xrightarrow[u]{X_1}C_1'$ then $C_2\xrightarrow[u]{X_2}C_2'$, with $X_1\subseteq \mathbb{E}_1$, $X_2\subseteq \mathbb{E}_2$, $X_1\sim X_2$ and $(C_1',C_2')\in R$,
and vice-versa. We say that $\mathcal{E}_1$, $\mathcal{E}_2$ are dynamic location pomset bisimilar, written $\mathcal{E}_1\sim_p^{dl}\mathcal{E}_2$, if there exists a dynamic location pomset bisimulation $R$, such that
$(\emptyset,\emptyset)\in R$. By replacing pomset transitions with steps, we can get the definition of dynamic location step bisimulation. When PESs $\mathcal{E}_1$ and $\mathcal{E}_2$ are dynamic location step
bisimilar, we write $\mathcal{E}_1\sim_s^{dl}\mathcal{E}_2$.
\end{definition}

\begin{definition}[Dynamic location (hereditary) history-preserving bisimulation]
A dynamic location history-preserving (hp-) bisimulation is a posetal relation $R\subseteq\mathcal{C}(\mathcal{E}_1)\overline{\times}\mathcal{C}(\mathcal{E}_2)$ such that if $(C_1,f,C_2)\in R$,
and $C_1\xrightarrow[u]{e_1} C_1'$, then $C_2\xrightarrow[u]{e_2} C_2'$, with $(C_1',f[e_1\mapsto e_2],C_2')\in R$, and vice-versa. $\mathcal{E}_1,\mathcal{E}_2$ are dynamic location history-preserving
(hp-)bisimilar and are written $\mathcal{E}_1\sim_{hp}^{dl}\mathcal{E}_2$ if there exists a dynamic location hp-bisimulation $R$ such that $(\emptyset,\emptyset,\emptyset)\in R$.

A dynamic location hereditary history-preserving (hhp-)bisimulation is a downward closed dynamic location hp-bisimulation. $\mathcal{E}_1,\mathcal{E}_2$ are dynamic location hereditary history-preserving (hhp-)bisimilar and are
written $\mathcal{E}_1\sim_{hhp}^{dl}\mathcal{E}_2$.
\end{definition}

\begin{definition}[Weak dynamic location pomset, step bisimulation]
Let $\mathcal{E}_1$, $\mathcal{E}_2$ be PESs. A weak dynamic location pomset bisimulation is a relation $R\subseteq\mathcal{C}(\mathcal{E}_1)\times\mathcal{C}(\mathcal{E}_2)$, such that if
$(C_1,C_2)\in R$, and $C_1\xRightarrow[u]{X_1}C_1'$ then $C_2\xRightarrow[u]{X_2}C_2'$, with $X_1\subseteq \hat{\mathbb{E}_1}$, $X_2\subseteq \hat{\mathbb{E}_2}$, $X_1\sim X_2$ and
$(C_1',C_2')\in R$, and vice-versa. We say that $\mathcal{E}_1$, $\mathcal{E}_2$ are weak dynamic location pomset bisimilar, written $\mathcal{E}_1\approx_p^{dl}\mathcal{E}_2$, if there exists a weak dynamic location pomset
bisimulation $R$, such that $(\emptyset,\emptyset)\in R$. By replacing weak pomset transitions with weak steps, we can get the definition of weak dynamic location step bisimulation. When PESs
$\mathcal{E}_1$ and $\mathcal{E}_2$ are weak dynamic location step bisimilar, we write $\mathcal{E}_1\approx_s^{dl}\mathcal{E}_2$.
\end{definition}

\begin{definition}[Weak dynamic location (hereditary) history-preserving bisimulation]
A weak dynamic location history-preserving (hp-) bisimulation is a weakly posetal relation $R\subseteq\mathcal{C}(\mathcal{E}_1)\overline{\times}\mathcal{C}(\mathcal{E}_2)$ such that if
$(C_1,f,C_2)\in R$, and $C_1\xRightarrow[u]{e_1} C_1'$, then $C_2\xRightarrow[u]{e_2} C_2'$, with $(C_1',f[e_1\mapsto e_2],C_2')\in R$, and vice-versa. $\mathcal{E}_1,\mathcal{E}_2$ are weak
dynamic location history-preserving (hp-)bisimilar and are written $\mathcal{E}_1\approx_{hp}^{dl}\mathcal{E}_2$ if there exists a weak dynamic location hp-bisimulation $R$ such that $(\emptyset,\emptyset,\emptyset)\in R$.

A weak dynamic location hereditary history-preserving (hhp-)bisimulation is a downward closed weak dynamic location hp-bisimulation. $\mathcal{E}_1,\mathcal{E}_2$ are weak dynamic location hereditary history-preserving
(hhp-)bisimilar and are written $\mathcal{E}_1\approx_{hhp}^{dl}\mathcal{E}_2$.
\end{definition}

\begin{definition}[Branching dynamic location pomset, step bisimulation]
Assume a special termination predicate $\downarrow$, and let $\surd$ represent a state with $\surd\downarrow$. Let $\mathcal{E}_1$, $\mathcal{E}_2$ be PESs. A branching dynamic location pomset
bisimulation is a relation $R\subseteq\mathcal{C}(\mathcal{E}_1)\times\mathcal{C}(\mathcal{E}_2)$, such that:
 \begin{enumerate}
   \item if $(C_1,C_2)\in R$, and $C_1\xrightarrow[u]{X}C_1'$ then
   \begin{itemize}
     \item either $X\equiv \tau^*$, and $(C_1',C_2)\in R$;
     \item or there is a sequence of (zero or more) $\tau$-transitions $C_2\xrightarrow{\tau^*} C_2^0$, such that $(C_1,C_2^0)\in R$ and $C_2^0\xRightarrow[u]{X}C_2'$ with
     $(C_1',C_2')\in R$;
   \end{itemize}
   \item if $(C_1,C_2)\in R$, and $C_2\xrightarrow[u]{X}C_2'$ then
   \begin{itemize}
     \item either $X\equiv \tau^*$, and $(C_1,C_2')\in R$;
     \item or there is a sequence of (zero or more) $\tau$-transitions $C_1\xrightarrow{\tau^*} C_1^0$, such that $(C_1^0,C_2)\in R$ and $C_1^0\xRightarrow[u]{X}C_1'$ with
     $(C_1',C_2')\in R$;
   \end{itemize}
   \item if $(C_1,C_2)\in R$ and $C_1\downarrow$, then there is a sequence of (zero or more) $\tau$-transitions $C_2\xrightarrow{\tau^*}C_2^0$ such that $(C_1,C_2^0)\in R$
   and $C_2^0\downarrow$;
   \item if $(C_1,C_2)\in R$ and $C_2\downarrow$, then there is a sequence of (zero or more) $\tau$-transitions $C_1\xrightarrow{\tau^*}C_1^0$ such that $(C_1^0,C_2)\in R$
   and $C_1^0\downarrow$.
 \end{enumerate}

We say that $\mathcal{E}_1$, $\mathcal{E}_2$ are branching dynamic location pomset bisimilar, written $\mathcal{E}_1\approx_{bp}^{dl}\mathcal{E}_2$, if there exists a branching dynamic location pomset bisimulation $R$,
such that $(\emptyset,\emptyset)\in R$.

By replacing pomset transitions with steps, we can get the definition of branching dynamic location step bisimulation. When PESs $\mathcal{E}_1$ and $\mathcal{E}_2$ are branching dynamic location step bisimilar,
we write $\mathcal{E}_1\approx_{bs}^{dl}\mathcal{E}_2$.
\end{definition}

\begin{definition}[Rooted branching dynamic location pomset, step bisimulation]
Assume a special termination predicate $\downarrow$, and let $\surd$ represent a state with $\surd\downarrow$. Let $\mathcal{E}_1$, $\mathcal{E}_2$ be PESs. A rooted branching dynamic location pomset
bisimulation is a relation $R\subseteq\mathcal{C}(\mathcal{E}_1)\times\mathcal{C}(\mathcal{E}_2)$, such that:
 \begin{enumerate}
   \item if $(C_1,C_2)\in R$, and $C_1\xrightarrow[u]{X}C_1'$ then $C_2\xrightarrow[u]{X}C_2'$ with $C_1'\approx_{bp}^{dl}C_2'$;
   \item if $(C_1,C_2)\in R$, and $C_2\xrightarrow[u]{X}C_2'$ then $C_1\xrightarrow[u]{X}C_1'$ with $C_1'\approx_{bp}^{dl}C_2'$;
   \item if $(C_1,C_2)\in R$ and $C_1\downarrow$, then $C_2\downarrow$;
   \item if $(C_1,C_2)\in R$ and $C_2\downarrow$, then $C_1\downarrow$.
 \end{enumerate}

We say that $\mathcal{E}_1$, $\mathcal{E}_2$ are rooted branching dynamic location pomset bisimilar, written $\mathcal{E}_1\approx_{rbp}^{dl}\mathcal{E}_2$, if there exists a rooted branching dynamic location pomset
bisimulation $R$, such that $(\emptyset,\emptyset)\in R$.

By replacing pomset transitions with steps, we can get the definition of rooted branching dynamic location step bisimulation. When PESs $\mathcal{E}_1$ and $\mathcal{E}_2$ are rooted branching dynamic location step
bisimilar, we write $\mathcal{E}_1\approx_{rbs}^{dl}\mathcal{E}_2$.
\end{definition}

\begin{definition}[Branching dynamic location (hereditary) history-preserving bisimulation]
Assume a special termination predicate $\downarrow$, and let $\surd$ represent a state with $\surd\downarrow$. A branching dynamic location history-preserving (hp-) bisimulation is a posetal
relation $R\subseteq\mathcal{C}(\mathcal{E}_1)\overline{\times}\mathcal{C}(\mathcal{E}_2)$ such that:

 \begin{enumerate}
   \item if $(C_1,f,C_2)\in R$, and $C_1\xrightarrow[u]{e_1}C_1'$ then
   \begin{itemize}
     \item either $e_1\equiv \tau$, and $(C_1',f[e_1\mapsto \tau],C_2)\in R$;
     \item or there is a sequence of (zero or more) $\tau$-transitions $C_2\xrightarrow{\tau^*} C_2^0$, such that $(C_1,f,C_2^0)\in R$ and $C_2^0\xrightarrow[u]{e_2}C_2'$ with
     $(C_1',f[e_1\mapsto e_2],C_2')\in R$;
   \end{itemize}
   \item if $(C_1,f,C_2)\in R$, and $C_2\xrightarrow[u]{e_2}C_2'$ then
   \begin{itemize}
     \item either $X\equiv \tau$, and $(C_1,f[e_2\mapsto \tau],C_2')\in R$;
     \item or there is a sequence of (zero or more) $\tau$-transitions $C_1\xrightarrow{\tau^*} C_1^0$, such that $(C_1^0,f,C_2)\in R$ and $C_1^0\xrightarrow[u]{e_1}C_1'$ with
     $(C_1',f[e_2\mapsto e_1],C_2')\in R$;
   \end{itemize}
   \item if $(C_1,f,C_2)\in R$ and $C_1\downarrow$, then there is a sequence of (zero or more) $\tau$-transitions $C_2\xrightarrow{\tau^*}C_2^0$ such that $(C_1,f,C_2^0)\in R$
   and $C_2^0\downarrow$;
   \item if $(C_1,f,C_2)\in R$ and $C_2\downarrow$, then there is a sequence of (zero or more) $\tau$-transitions $C_1\xrightarrow{\tau^*}C_1^0$ such that $(C_1^0,f,C_2)\in R$
   and $C_1^0\downarrow$.
 \end{enumerate}

$\mathcal{E}_1,\mathcal{E}_2$ are branching dynamic location history-preserving (hp-)bisimilar and are written $\mathcal{E}_1\approx_{bhp}^{dl}\mathcal{E}_2$ if there exists a branching dynamic location hp-bisimulation
$R$ such that $(\emptyset,\emptyset,\emptyset)\in R$.

A branching dynamic location hereditary history-preserving (hhp-)bisimulation is a downward closed branching dynamic location hp-bisimulation. $\mathcal{E}_1,\mathcal{E}_2$ are branching dynamic location hereditary history-preserving
(hhp-)bisimilar and are written $\mathcal{E}_1\approx_{bhhp}^{dl}\mathcal{E}_2$.
\end{definition}

\begin{definition}[Rooted branching dynamic location (hereditary) history-preserving bisimulation]
Assume a special termination predicate $\downarrow$, and let $\surd$ represent a state with $\surd\downarrow$. A rooted branching dynamic location history-preserving (hp-) bisimulation is a posetal relation $R\subseteq\mathcal{C}(\mathcal{E}_1)\overline{\times}\mathcal{C}(\mathcal{E}_2)$ such that:

 \begin{enumerate}
   \item if $(C_1,f,C_2)\in R$, and $C_1\xrightarrow[u]{e_1}C_1'$, then $C_2\xrightarrow[u]{e_2}C_2'$ with $C_1'\approx_{bhp}^{dl}C_2'$;
   \item if $(C_1,f,C_2)\in R$, and $C_2\xrightarrow[u]{e_2}C_2'$, then $C_1\xrightarrow[u]{e_1}C_1'$ with $C_1'\approx_{bhp}^{dl}C_2'$;
   \item if $(C_1,f,C_2)\in R$ and $C_1\downarrow$, then $C_2\downarrow$;
   \item if $(C_1,f,C_2)\in R$ and $C_2\downarrow$, then $C_1\downarrow$.
 \end{enumerate}

$\mathcal{E}_1,\mathcal{E}_2$ are rooted branching dynamic location history-preserving (hp-)bisimilar and are written $\mathcal{E}_1\approx_{rbhp}^{dl}\mathcal{E}_2$ if there exists a rooted branching
dynamic location hp-bisimulation $R$ such that $(\emptyset,\emptyset,\emptyset)\in R$.

A rooted branching dynamic location hereditary history-preserving (hhp-)bisimulation is a downward closed rooted branching dynamic location hp-bisimulation. $\mathcal{E}_1,\mathcal{E}_2$ are rooted branching
dynamic location hereditary history-preserving (hhp-)bisimilar and are written $\mathcal{E}_1\approx_{rbhhp}^{dl}\mathcal{E}_2$.
\end{definition}

%% file: section4.tex
\section{CTC with Localities}\label{ctcl}

In this chapter, we introduce CTC with localities, including CTC with static localities in section \ref{ctcsl}, CTC with dynamic localities in section \ref{ctcdl}.

\input{section4/section4.1.tex}

\input{section4/section4.2.tex}

%% file: section4/section4.1.tex
\subsection{CTC with Static Localities}{\label{ctcsl}}

\subsubsection{Syntax and Operational Semantics}

We assume an infinite set $\mathcal{N}$ of (action or event) names, and use $a,b,c,\cdots$ to range over $\mathcal{N}$. We denote by $\overline{\mathcal{N}}$ the set of co-names and
let $\overline{a},\overline{b},\overline{c},\cdots$ range over $\overline{\mathcal{N}}$. Then we set $\mathcal{L}=\mathcal{N}\cup\overline{\mathcal{N}}$ as the set of labels, and use
$l,\overline{l}$ to range over $\mathcal{L}$. We extend complementation to $\mathcal{L}$ such that $\overline{\overline{a}}=a$. Let $\tau$ denote the silent step (internal action or
event) and define $Act=\mathcal{L}\cup\{\tau\}$ to be the set of actions, $\alpha,\beta$ range over $Act$. And $K,L$ are used to stand for subsets of $\mathcal{L}$ and $\overline{L}$
is used for the set of complements of labels in $L$. A relabelling function $f$ is a function from $\mathcal{L}$ to $\mathcal{L}$ such that $f(\overline{l})=\overline{f(l)}$. By
defining $f(\tau)=\tau$, we extend $f$ to $Act$.

Further, we introduce a set $\mathcal{X}$ of process variables, and a set $\mathcal{K}$ of process constants, and let $X,Y,\cdots$ range over $\mathcal{X}$, and $A,B,\cdots$ range
over $\mathcal{K}$, $\widetilde{X}$ is a tuple of distinct process variables, and also $E,F,\cdots$ range over the recursive expressions. We write $\mathcal{P}$ for the set of
processes. Sometimes, we use $I,J$ to stand for an indexing set, and we write $E_i:i\in I$ for a family of expressions indexed by $I$. $Id_D$ is the identity function or relation over
set $D$.

For each process constant schema $A$, a defining equation of the form

$$A\overset{\text{def}}{=}P$$

is assumed, where $P$ is a process.


Let $Loc$ be the set of locations, and $loc\in Loc$, $u,v\in Loc^*$, $\epsilon$ is the empty location. A distribution allocates a location $u\in Loc*$ to an action $\alpha$ denoted
$u::\alpha$ or a process $P$ denoted $u::P$.

\begin{definition}[Syntax]\label{syntax3}
Truly concurrent processes with static location are defined inductively by the following formation rules:

\begin{enumerate}
  \item $A\in\mathcal{P}$;
  \item $\textbf{nil}\in\mathcal{P}$;
  \item if $P\in\mathcal{P}$ and $loc\in Loc$, the Location $loc::P\in\mathcal{P}$;
  \item if $P\in\mathcal{P}$, then the Prefix $\alpha.P\in\mathcal{P}$, for $\alpha\in Act$;
  \item if $P,Q\in\mathcal{P}$, then the Summation $P+Q\in\mathcal{P}$;
  \item if $P,Q\in\mathcal{P}$, then the Composition $P\parallel Q\in\mathcal{P}$;
  \item if $P\in\mathcal{P}$, then the Prefix $(\alpha_1\parallel\cdots\parallel\alpha_n).P\in\mathcal{P}\quad(n\in I)$, for $\alpha_,\cdots,\alpha_n\in Act$;
  \item if $P\in\mathcal{P}$, then the Restriction $P\setminus L\in\mathcal{P}$ with $L\in\mathcal{L}$;
  \item if $P\in\mathcal{P}$, then the Relabelling $P[f]\in\mathcal{P}$.
\end{enumerate}

The standard BNF grammar of syntax of CTC with static localities can be summarized as follows:

$$P::=A|\textbf{nil}|loc::P|\alpha.P| P+P | P\parallel P | (\alpha_1\parallel\cdots\parallel\alpha_n).P | P\setminus L | P[f]$$
\end{definition}


The operational semantics is defined by LTSs (labelled transition systems), and it is detailed by the following definition.

\begin{definition}[Semantics]\label{semantics3}
The operational semantics of CTC with static localities corresponding to the syntax in Definition \ref{syntax3} is defined by a series of transition rules, named $\textbf{Act}$, $\textbf{Loc}$,
$\textbf{Sum}$, $\textbf{Com}$, $\textbf{Res}$, $\textbf{Rel}$ and $\textbf{Con}$ indicate that the rules are associated respectively with Prefix, Summation, Composition, Restriction,
Relabelling and Constants in Definition \ref{syntax3}. They are shown in Table \ref{TRForCTC3}.

\begin{center}
    \begin{table}
        \[\textbf{Act}_1\quad \frac{}{\alpha.P\xrightarrow[\epsilon]{\alpha}P}\]

        \[\textbf{Loc}\quad \frac{P\xrightarrow[u]{\alpha}P'}{loc::P\xrightarrow[loc\ll u]{\alpha}loc::P'}\]

        \[\textbf{Sum}_1\quad \frac{P\xrightarrow[u]{\alpha}P'}{P+Q\xrightarrow[u]{\alpha}P'}\]

        \[\textbf{Com}_1\quad \frac{P\xrightarrow[u]{\alpha}P'\quad Q\nrightarrow}{P\parallel Q\xrightarrow[u]{\alpha}P'\parallel Q}\]

        \[\textbf{Com}_2\quad \frac{Q\xrightarrow[u]{\alpha}Q'\quad P\nrightarrow}{P\parallel Q\xrightarrow[u]{\alpha}P\parallel Q'}\]

        \[\textbf{Com}_3\quad \frac{P\xrightarrow[u]{\alpha}P'\quad Q\xrightarrow[v]{\beta}Q'}{P\parallel Q\xrightarrow[u\diamond v]{\{\alpha,\beta\}}P'\parallel Q'}\quad (\beta\neq\overline{\alpha})\]

        \[\textbf{Com}_4\quad \frac{P\xrightarrow[u]{l}P'\quad Q\xrightarrow[v]{\overline{l}}Q'}{P\parallel Q\xrightarrow[u\diamond v]{\tau}P'\parallel Q'}\]

        \[\textbf{Act}_2\quad \frac{}{(\alpha_1\parallel\cdots\parallel\alpha_n).P\xrightarrow[\epsilon]{\{\alpha_1,\cdots,\alpha_n\}}P}\quad (\alpha_i\neq\overline{\alpha_j}\quad i,j\in\{1,\cdots,n\})\]

        \[\textbf{Sum}_2\quad \frac{P\xrightarrow[u]{\{\alpha_1,\cdots,\alpha_n\}}P'}{P+Q\xrightarrow[u]{\{\alpha_1,\cdots,\alpha_n\}}P'}\]

%
%
%
%
%

        \caption{Transition rules of CTC with static localities}
        \label{TRForCTC3}
    \end{table}
\end{center}

\begin{center}
    \begin{table}
%
%
%
%
%
%
%
%
%
        \[\textbf{Res}_1\quad \frac{P\xrightarrow[u]{\alpha}P'}{P\setminus L\xrightarrow[u]{\alpha}P'\setminus L}\quad (\alpha,\overline{\alpha}\notin L)\]

        \[\textbf{Res}_2\quad \frac{P\xrightarrow[u]{\{\alpha_1,\cdots,\alpha_n\}}P'}{P\setminus L\xrightarrow[u]{\{\alpha_1,\cdots,\alpha_n\}}P'\setminus L}\quad (\alpha_1,\overline{\alpha_1},\cdots,\alpha_n,\overline{\alpha_n}\notin L)\]

        \[\textbf{Rel}_1\quad \frac{P\xrightarrow[u]{\alpha}P'}{P[f]\xrightarrow[u]{f(\alpha)}P'[f]}\]

        \[\textbf{Rel}_2\quad \frac{P\xrightarrow[u]{\{\alpha_1,\cdots,\alpha_n\}}P'}{P[f]\xrightarrow[u]{\{f(\alpha_1),\cdots,f(\alpha_n)\}}P'[f]}\]

        \[\textbf{Con}_1\quad\frac{P\xrightarrow[u]{\alpha}P'}{A\xrightarrow[u]{\alpha}P'}\quad (A\overset{\text{def}}{=}P)\]

        \[\textbf{Con}_2\quad\frac{P\xrightarrow[u]{\{\alpha_1,\cdots,\alpha_n\}}P'}{A\xrightarrow[u]{\{\alpha_1,\cdots,\alpha_n\}}P'}\quad (A\overset{\text{def}}{=}P)\]

        \caption{Transition rules of CTC with static localities (continuing)}
        \label{TRForCTC32}
    \end{table}
\end{center}
\end{definition}


\begin{definition}[Sorts]\label{sorts3}
Given the sorts $\mathcal{L}(A)$ and $\mathcal{L}(X)$ of constants and variables, we define $\mathcal{L}(P)$ inductively as follows.

\begin{enumerate}
  \item $\mathcal{L}(loc::P)=\mathcal{L}(P)$;
  \item $\mathcal{L}(l.P)=\{l\}\cup\mathcal{L}(P)$;
  \item $\mathcal{L}((l_1\parallel \cdots\parallel l_n).P)=\{l_1,\cdots,l_n\}\cup\mathcal{L}(P)$;
  \item $\mathcal{L}(\tau.P)=\mathcal{L}(P)$;
  \item $\mathcal{L}(P+Q)=\mathcal{L}(P)\cup\mathcal{L}(Q)$;
  \item $\mathcal{L}(P\parallel Q)=\mathcal{L}(P)\cup\mathcal{L}(Q)$;
  \item $\mathcal{L}(P\setminus L)=\mathcal{L}(P)-(L\cup\overline{L})$;
  \item $\mathcal{L}(P[f])=\{f(l):l\in\mathcal{L}(P)\}$;
  \item for $A\overset{\text{def}}{=}P$, $\mathcal{L}(P)\subseteq\mathcal{L}(A)$.
\end{enumerate}
\end{definition}

Now, we present some properties of the transition rules defined in Table \ref{TRForCTC3}.

\begin{proposition}
If $P\xrightarrow[u]{\alpha}P'$, then
\begin{enumerate}
  \item $\alpha\in\mathcal{L}(P)\cup\{\tau\}$;
  \item $\mathcal{L}(P')\subseteq\mathcal{L}(P)$.
\end{enumerate}

If $P\xrightarrow[u]{\{\alpha_1,\cdots,\alpha_n\}}P'$, then
\begin{enumerate}
  \item $\alpha_1,\cdots,\alpha_n\in\mathcal{L}(P)\cup\{\tau\}$;
  \item $\mathcal{L}(P')\subseteq\mathcal{L}(P)$.
\end{enumerate}
\end{proposition}

\begin{proof}
By induction on the inference of $P\xrightarrow[u]{\alpha}P'$ and $P\xrightarrow[u]{\{\alpha_1,\cdots,\alpha_n\}}P'$, there are several cases corresponding to the transition rules
in Table \ref{TRForCTC3}, we just prove the one case $\textbf{Act}_1$ and $\textbf{Act}_2$, and omit the others.

Case $\textbf{Act}_1$: by $\textbf{Act}_1$, with $P\equiv\alpha.P'$. Then by Definition \ref{sorts3}, we have (1) $\mathcal{L}(P)=\{\alpha\}\cup\mathcal{L}(P')$ if $\alpha\neq\tau$;
(2) $\mathcal{L}(P)=\mathcal{L}(P')$ if $\alpha=\tau$. So, $\alpha\in\mathcal{L}(P)\cup\{\tau\}$, and $\mathcal{L}(P')\subseteq\mathcal{L}(P)$, as desired.

Case $\textbf{Act}_2$: by $\textbf{Act}_2$, with $P\equiv(\alpha_1\parallel\cdots\parallel\alpha_n).P'$. Then by Definition \ref{sorts3}, we have (1)
$\mathcal{L}(P)=\{\alpha_1,\cdots,\alpha_n\}\cup\mathcal{L}(P')$ if $\alpha_i\neq\tau$ for $i\leq n$; (2) $\mathcal{L}(P)=\mathcal{L}(P')$ if
$\alpha_1,\cdots,\alpha_n=\tau$. So, $\alpha_1,\cdots,\alpha_n\in\mathcal{L}(P)\cup\{\tau\}$, and $\mathcal{L}(P')\subseteq\mathcal{L}(P)$, as desired.
\end{proof}

\subsubsection{Strong Bisimulations}


Based on the concepts of strong bisimulation equivalences, we get the following laws.

\begin{proposition}[Monoid laws for strong static location pomset bisimulation]
The monoid laws for strong static location pomset bisimulation are as follows.

\begin{enumerate}
  \item $P+Q\sim_p^{sl} Q+P$;
  \item $P+(Q+R)\sim_p^{sl} (P+Q)+R$;
  \item $P+P\sim_p^{sl} P$;
  \item $P+\textbf{nil}\sim_p^{sl} P$.
\end{enumerate}

\end{proposition}

\begin{proof}
\begin{enumerate}
  \item $P+Q\sim_p^{sl} Q+P$. It is sufficient to prove the relation $R=\{(P+Q, Q+P)\}\cup \textbf{Id}$ is a strong static location pomset bisimulation for some distributions. It can be proved similarly to the proof of
  Monoid laws for strong pomset bisimulation in CTC, we omit it;
  \item $P+(Q+R)\sim_p^{sl} (P+Q)+R$. It is sufficient to prove the relation $R=\{(P+(Q+R), (P+Q)+R)\}\cup \textbf{Id}$ is a strong static location pomset bisimulation for some distributions. It can be proved similarly to the proof of
  Monoid laws for strong pomset bisimulation in CTC, we omit it;
  \item $P+P\sim_p^{sl} P$. It is sufficient to prove the relation $R=\{(P+P, P)\}\cup \textbf{Id}$ is a strong static location pomset bisimulation for some distributions. It can be proved similarly to the proof of
  Monoid laws for strong pomset bisimulation in CTC, we omit it;
  \item $P+\textbf{nil}\sim_p^{sl} P$. It is sufficient to prove the relation $R=\{(P+\textbf{nil}, P)\}\cup \textbf{Id}$ is a strong static location pomset bisimulation for some distributions. It can be proved similarly to the proof of
  Monoid laws for strong pomset bisimulation in CTC, we omit it.
\end{enumerate}
\end{proof}

\begin{proposition}[Monoid laws for strong static location step bisimulation]
The monoid laws for strong static location step bisimulation are as follows.

\begin{enumerate}
  \item $P+Q\sim_s^{sl} Q+P$;
  \item $P+(Q+R)\sim_s^{sl} (P+Q)+R$;
  \item $P+P\sim_s^{sl} P$;
  \item $P+\textbf{nil}\sim_s^{sl} P$.
\end{enumerate}

\end{proposition}

\begin{proof}
\begin{enumerate}
  \item $P+Q\sim_s^{sl} Q+P$. It is sufficient to prove the relation $R=\{(P+Q, Q+P)\}\cup \textbf{Id}$ is a strong static location step bisimulation for some distributions. It can be proved similarly to the proof of
  Monoid laws for strong step bisimulation in CTC, we omit it;
  \item $P+(Q+R)\sim_s^{sl} (P+Q)+R$. It is sufficient to prove the relation $R=\{(P+(Q+R), (P+Q)+R)\}\cup \textbf{Id}$ is a strong static location step bisimulation for some distributions. It can be proved similarly to the proof of
  Monoid laws for strong step bisimulation in CTC, we omit it;
  \item $P+P\sim_s^{sl} P$. It is sufficient to prove the relation $R=\{(P+P, P)\}\cup \textbf{Id}$ is a strong static location step bisimulation for some distributions. It can be proved similarly to the proof of
  Monoid laws for strong step bisimulation in CTC, we omit it;
  \item $P+\textbf{nil}\sim_s^{sl} P$. It is sufficient to prove the relation $R=\{(P+\textbf{nil}, P)\}\cup \textbf{Id}$ is a strong static location step bisimulation for some distributions. It can be proved similarly to the proof of
  Monoid laws for strong step bisimulation in CTC, we omit it.
\end{enumerate}
\end{proof}

\begin{proposition}[Monoid laws for strong static location hp-bisimulation]
The monoid laws for strong static location hp-bisimulation are as follows.

\begin{enumerate}
  \item $P+Q\sim_{hp}^{sl} Q+P$;
  \item $P+(Q+R)\sim_{hp}^{sl} (P+Q)+R$;
  \item $P+P\sim_{hp}^{sl} P$;
  \item $P+\textbf{nil}\sim_{hp}^{sl} P$.
\end{enumerate}

\end{proposition}

\begin{proof}
\begin{enumerate}
  \item $P+Q\sim_{hp}^{sl} Q+P$. It is sufficient to prove the relation $R=\{(P+Q, Q+P)\}\cup \textbf{Id}$ is a strong static location hp-bisimulation for some distributions. It can be proved similarly to the proof of
  Monoid laws for strong hp-bisimulation in CTC, we omit it;
  \item $P+(Q+R)\sim_{hp}^{sl} (P+Q)+R$. It is sufficient to prove the relation $R=\{(P+(Q+R), (P+Q)+R)\}\cup \textbf{Id}$ is a strong static location hp-bisimulation for some distributions. It can be proved similarly to the proof of
  Monoid laws for strong hp-bisimulation in CTC, we omit it;
  \item $P+P\sim_{hp}^{sl} P$. It is sufficient to prove the relation $R=\{(P+P, P)\}\cup \textbf{Id}$ is a strong static location hp-bisimulation for some distributions. It can be proved similarly to the proof of
  Monoid laws for strong hp-bisimulation in CTC, we omit it;
  \item $P+\textbf{nil}\sim_{hp}^{sl} P$. It is sufficient to prove the relation $R=\{(P+\textbf{nil}, P)\}\cup \textbf{Id}$ is a strong static location hp-bisimulation for some distributions. It can be proved similarly to the proof of
  Monoid laws for strong hp-bisimulation in CTC, we omit it.
\end{enumerate}
\end{proof}

\begin{proposition}[Monoid laws for strong static location hhp-bisimulation]
The monoid laws for strong static location hhp-bisimulation are as follows.

\begin{enumerate}
  \item $P+Q\sim_{hhp}^{sl} Q+P$;
  \item $P+(Q+R)\sim_{hhp}^{sl} (P+Q)+R$;
  \item $P+P\sim_{hhp}^{sl} P$;
  \item $P+\textbf{nil}\sim_{hhp}^{sl} P$.
\end{enumerate}

\end{proposition}

\begin{proof}
\begin{enumerate}
  \item $P+Q\sim_{hhp}^{sl} Q+P$. It is sufficient to prove the relation $R=\{(P+Q, Q+P)\}\cup \textbf{Id}$ is a strong static location hhp-bisimulation for some distributions. It can be proved similarly to the proof of
  Monoid laws for strong hhp-bisimulation in CTC, we omit it;
  \item $P+(Q+R)\sim_{hhp}^{sl} (P+Q)+R$. It is sufficient to prove the relation $R=\{(P+(Q+R), (P+Q)+R)\}\cup \textbf{Id}$ is a strong static location hhp-bisimulation for some distributions. It can be proved similarly to the proof of
  Monoid laws for strong hhp-bisimulation in CTC, we omit it;
  \item $P+P\sim_{hhp}^{sl} P$. It is sufficient to prove the relation $R=\{(P+P, P)\}\cup \textbf{Id}$ is a strong static location hhp-bisimulation for some distributions. It can be proved similarly to the proof of
  Monoid laws for strong hhp-bisimulation in CTC, we omit it;
  \item $P+\textbf{nil}\sim_{hhp}^{sl} P$. It is sufficient to prove the relation $R=\{(P+\textbf{nil}, P)\}\cup \textbf{Id}$ is a strong static location hhp-bisimulation for some distributions. It can be proved similarly to the proof of
  Monoid laws for strong hhp-bisimulation in CTC, we omit it.
\end{enumerate}
\end{proof}

\begin{proposition}[Static laws for strong static location pomset bisimulation]
The static laws for strong static location pomset bisimulation are as follows.

\begin{enumerate}
  \item $P\parallel Q\sim_p^{sl} Q\parallel P$;
  \item $P\parallel(Q\parallel R)\sim_p^{sl} (P\parallel Q)\parallel R$;
  \item $P\parallel \textbf{nil}\sim_p^{sl} P$;
  \item $P\setminus L\sim_p^{sl} P$, if $\mathcal{L}(P)\cap(L\cup\overline{L})=\emptyset$;
  \item $P\setminus K\setminus L\sim_p^{sl} P\setminus(K\cup L)$;
  \item $P[f]\setminus L\sim_p^{sl} P\setminus f^{-1}(L)[f]$;
  \item $(P\parallel Q)\setminus L\sim_p^{sl} P\setminus L\parallel Q\setminus L$, if $\mathcal{L}(P)\cap\overline{\mathcal{L}(Q)}\cap(L\cup\overline{L})=\emptyset$;
  \item $P[Id]\sim_p^{sl} P$;
  \item $P[f]\sim_p^{sl} P[f']$, if $f\upharpoonright\mathcal{L}(P)=f'\upharpoonright\mathcal{L}(P)$;
  \item $P[f][f']\sim_p^{sl} P[f'\circ f]$;
  \item $(P\parallel Q)[f]\sim_p^{sl} P[f]\parallel Q[f]$, if $f\upharpoonright(L\cup\overline{L})$ is one-to-one, where $L=\mathcal{L}(P)\cup\mathcal{L}(Q)$.
\end{enumerate}
\end{proposition}

\begin{proof}
\begin{enumerate}
  \item $P\parallel Q\sim_p^{sl} Q\parallel P$. It is sufficient to prove the relation $R=\{(P\parallel Q, Q\parallel P)\}\cup \textbf{Id}$ is a strong static location pomset bisimulation for some distributions. It can be proved similarly to the proof of
  static laws for strong pomset bisimulation in CTC, we omit it;
  \item $P\parallel(Q\parallel R)\sim_p^{sl} (P\parallel Q)\parallel R$. It is sufficient to prove the relation $R=\{(P\parallel(Q\parallel R), (P\parallel Q)\parallel R)\}\cup \textbf{Id}$ is a strong static location pomset bisimulation for some distributions. It can be proved similarly to the proof of
  static laws for strong pomset bisimulation in CTC, we omit it;
  \item $P\parallel \textbf{nil}\sim_p^{sl} P$. It is sufficient to prove the relation $R=\{(P\parallel \textbf{nil}, P)\}\cup \textbf{Id}$ is a strong static location pomset bisimulation for some distributions. It can be proved similarly to the proof of
  static laws for strong pomset bisimulation in CTC, we omit it;
  \item $P\setminus L\sim_p^{sl} P$, if $\mathcal{L}(P)\cap(L\cup\overline{L})=\emptyset$. It is sufficient to prove the relation $R=\{(P\setminus L, P)\}\cup \textbf{Id}$, if $\mathcal{L}(P)\cap(L\cup\overline{L})=\emptyset$, is a strong static location pomset bisimulation for some distributions. It can be proved similarly to the proof of
  static laws for strong pomset bisimulation in CTC, we omit it;
  \item $P\setminus K\setminus L\sim_p^{sl} P\setminus(K\cup L)$. It is sufficient to prove the relation $R=\{(P\setminus K\setminus L, P\setminus(K\cup L))\}\cup \textbf{Id}$ is a strong static location pomset bisimulation for some distributions. It can be proved similarly to the proof of
  static laws for strong pomset bisimulation in CTC, we omit it;
  \item $P[f]\setminus L\sim_p^{sl} P\setminus f^{-1}(L)[f]$. It is sufficient to prove the relation $R=\{(P[f]\setminus L, P\setminus f^{-1}(L)[f])\}\cup \textbf{Id}$ is a strong static location pomset bisimulation for some distributions. It can be proved similarly to the proof of
  static laws for strong pomset bisimulation in CTC, we omit it;
  \item $(P\parallel Q)\setminus L\sim_p^{sl} P\setminus L\parallel Q\setminus L$, if $\mathcal{L}(P)\cap\overline{\mathcal{L}(Q)}\cap(L\cup\overline{L})=\emptyset$. It is sufficient to prove the relation $R=\{((P\parallel Q)\setminus L, P\setminus L\parallel Q\setminus L)\}\cup \textbf{Id}$, if $\mathcal{L}(P)\cap\overline{\mathcal{L}(Q)}\cap(L\cup\overline{L})=\emptyset$, is a strong static location pomset bisimulation for some distributions. It can be proved similarly to the proof of
  static laws for strong pomset bisimulation in CTC, we omit it;
  \item $P[Id]\sim_p^{sl} P$. It is sufficient to prove the relation $R=\{(P[Id], P)\}\cup \textbf{Id}$ is a strong static location pomset bisimulation for some distributions. It can be proved similarly to the proof of
  static laws for strong pomset bisimulation in CTC, we omit it;
  \item $P[f]\sim_p^{sl} P[f']$, if $f\upharpoonright\mathcal{L}(P)=f'\upharpoonright\mathcal{L}(P)$. It is sufficient to prove the relation $R=\{(P[f], P[f'])\}\cup \textbf{Id}$, if $f\upharpoonright\mathcal{L}(P)=f'\upharpoonright\mathcal{L}(P)$, is a strong static location pomset bisimulation for some distributions. It can be proved similarly to the proof of
  static laws for strong pomset bisimulation in CTC, we omit it;
  \item $P[f][f']\sim_p^{sl} P[f'\circ f]$. It is sufficient to prove the relation $R=\{(P[f][f'], P[f'\circ f])\}\cup \textbf{Id}$ is a strong static location pomset bisimulation for some distributions. It can be proved similarly to the proof of
  static laws for strong pomset bisimulation in CTC, we omit it;
  \item $(P\parallel Q)[f]\sim_p^{sl} P[f]\parallel Q[f]$, if $f\upharpoonright(L\cup\overline{L})$ is one-to-one, where $L=\mathcal{L}(P)\cup\mathcal{L}(Q)$. It is sufficient to prove the relation $R=\{((P\parallel Q)[f], P[f]\parallel Q[f])\}\cup \textbf{Id}$, if $f\upharpoonright(L\cup\overline{L})$ is one-to-one, where $L=\mathcal{L}(P)\cup\mathcal{L}(Q)$, is a strong static location pomset bisimulation for some distributions. It can be proved similarly to the proof of
  static laws for strong pomset bisimulation in CTC, we omit it.
\end{enumerate}
\end{proof}

\begin{proposition}[Static laws for strong static location step bisimulation]
The static laws for strong static location step bisimulation are as follows.

\begin{enumerate}
  \item $P\parallel Q\sim_s^{sl} Q\parallel P$;
  \item $P\parallel(Q\parallel R)\sim_s^{sl} (P\parallel Q)\parallel R$;
  \item $P\parallel \textbf{nil}\sim_s^{sl} P$;
  \item $P\setminus L\sim_s^{sl} P$, if $\mathcal{L}(P)\cap(L\cup\overline{L})=\emptyset$;
  \item $P\setminus K\setminus L\sim_s^{sl} P\setminus(K\cup L)$;
  \item $P[f]\setminus L\sim_s^{sl} P\setminus f^{-1}(L)[f]$;
  \item $(P\parallel Q)\setminus L\sim_s^{sl} P\setminus L\parallel Q\setminus L$, if $\mathcal{L}(P)\cap\overline{\mathcal{L}(Q)}\cap(L\cup\overline{L})=\emptyset$;
  \item $P[Id]\sim_s^{sl} P$;
  \item $P[f]\sim_s^{sl} P[f']$, if $f\upharpoonright\mathcal{L}(P)=f'\upharpoonright\mathcal{L}(P)$;
  \item $P[f][f']\sim_s^{sl} P[f'\circ f]$;
  \item $(P\parallel Q)[f]\sim_s^{sl} P[f]\parallel Q[f]$, if $f\upharpoonright(L\cup\overline{L})$ is one-to-one, where $L=\mathcal{L}(P)\cup\mathcal{L}(Q)$.
\end{enumerate}
\end{proposition}

\begin{proof}
\begin{enumerate}
  \item $P\parallel Q\sim_s^{sl} Q\parallel P$. It is sufficient to prove the relation $R=\{(P\parallel Q, Q\parallel P)\}\cup \textbf{Id}$ is a strong static location step bisimulation for some distributions. It can be proved similarly to the proof of
  static laws for strong step bisimulation in CTC, we omit it;
  \item $P\parallel(Q\parallel R)\sim_s^{sl} (P\parallel Q)\parallel R$. It is sufficient to prove the relation $R=\{(P\parallel(Q\parallel R), (P\parallel Q)\parallel R)\}\cup \textbf{Id}$ is a strong static location step bisimulation for some distributions. It can be proved similarly to the proof of
  static laws for strong step bisimulation in CTC, we omit it;
  \item $P\parallel \textbf{nil}\sim_s^{sl} P$. It is sufficient to prove the relation $R=\{(P\parallel \textbf{nil}, P)\}\cup \textbf{Id}$ is a strong static location step bisimulation for some distributions. It can be proved similarly to the proof of
  static laws for strong step bisimulation in CTC, we omit it;
  \item $P\setminus L\sim_s^{sl} P$, if $\mathcal{L}(P)\cap(L\cup\overline{L})=\emptyset$. It is sufficient to prove the relation $R=\{(P\setminus L, P)\}\cup \textbf{Id}$, if $\mathcal{L}(P)\cap(L\cup\overline{L})=\emptyset$, is a strong static location step bisimulation for some distributions. It can be proved similarly to the proof of
  static laws for strong step bisimulation in CTC, we omit it;
  \item $P\setminus K\setminus L\sim_s^{sl} P\setminus(K\cup L)$. It is sufficient to prove the relation $R=\{(P\setminus K\setminus L, P\setminus(K\cup L))\}\cup \textbf{Id}$ is a strong static location step bisimulation for some distributions. It can be proved similarly to the proof of
  static laws for strong step bisimulation in CTC, we omit it;
  \item $P[f]\setminus L\sim_s^{sl} P\setminus f^{-1}(L)[f]$. It is sufficient to prove the relation $R=\{(P[f]\setminus L, P\setminus f^{-1}(L)[f])\}\cup \textbf{Id}$ is a strong static location step bisimulation for some distributions. It can be proved similarly to the proof of
  static laws for strong step bisimulation in CTC, we omit it;
  \item $(P\parallel Q)\setminus L\sim_s^{sl} P\setminus L\parallel Q\setminus L$, if $\mathcal{L}(P)\cap\overline{\mathcal{L}(Q)}\cap(L\cup\overline{L})=\emptyset$. It is sufficient to prove the relation $R=\{((P\parallel Q)\setminus L, P\setminus L\parallel Q\setminus L)\}\cup \textbf{Id}$, if $\mathcal{L}(P)\cap\overline{\mathcal{L}(Q)}\cap(L\cup\overline{L})=\emptyset$, is a strong static location step bisimulation for some distributions. It can be proved similarly to the proof of
  static laws for strong step bisimulation in CTC, we omit it;
  \item $P[Id]\sim_s^{sl} P$. It is sufficient to prove the relation $R=\{(P[Id], P)\}\cup \textbf{Id}$ is a strong static location step bisimulation for some distributions. It can be proved similarly to the proof of
  static laws for strong step bisimulation in CTC, we omit it;
  \item $P[f]\sim_s^{sl} P[f']$, if $f\upharpoonright\mathcal{L}(P)=f'\upharpoonright\mathcal{L}(P)$. It is sufficient to prove the relation $R=\{(P[f], P[f'])\}\cup \textbf{Id}$, if $f\upharpoonright\mathcal{L}(P)=f'\upharpoonright\mathcal{L}(P)$, is a strong static location step bisimulation for some distributions. It can be proved similarly to the proof of
  static laws for strong step bisimulation in CTC, we omit it;
  \item $P[f][f']\sim_s^{sl} P[f'\circ f]$. It is sufficient to prove the relation $R=\{(P[f][f'], P[f'\circ f])\}\cup \textbf{Id}$ is a strong static location step bisimulation for some distributions. It can be proved similarly to the proof of
  static laws for strong step bisimulation in CTC, we omit it;
  \item $(P\parallel Q)[f]\sim_s^{sl} P[f]\parallel Q[f]$, if $f\upharpoonright(L\cup\overline{L})$ is one-to-one, where $L=\mathcal{L}(P)\cup\mathcal{L}(Q)$. It is sufficient to prove the relation $R=\{((P\parallel Q)[f], P[f]\parallel Q[f])\}\cup \textbf{Id}$, if $f\upharpoonright(L\cup\overline{L})$ is one-to-one, where $L=\mathcal{L}(P)\cup\mathcal{L}(Q)$, is a strong static location step bisimulation for some distributions. It can be proved similarly to the proof of
  static laws for strong step bisimulation in CTC, we omit it.
\end{enumerate}
\end{proof}

\begin{proposition}[Static laws for strong static location hp-bisimulation]
The static laws for strong static location hp-bisimulation are as follows.

\begin{enumerate}
  \item $P\parallel Q\sim_{hp}^{sl} Q\parallel P$;
  \item $P\parallel(Q\parallel R)\sim_{hp}^{sl} (P\parallel Q)\parallel R$;
  \item $P\parallel \textbf{nil}\sim_{hp}^{sl} P$;
  \item $P\setminus L\sim_{hp}^{sl} P$, if $\mathcal{L}(P)\cap(L\cup\overline{L})=\emptyset$;
  \item $P\setminus K\setminus L\sim_{hp}^{sl} P\setminus(K\cup L)$;
  \item $P[f]\setminus L\sim_{hp}^{sl} P\setminus f^{-1}(L)[f]$;
  \item $(P\parallel Q)\setminus L\sim_{hp}^{sl} P\setminus L\parallel Q\setminus L$, if $\mathcal{L}(P)\cap\overline{\mathcal{L}(Q)}\cap(L\cup\overline{L})=\emptyset$;
  \item $P[Id]\sim_{hp}^{sl} P$;
  \item $P[f]\sim_{hp}^{sl} P[f']$, if $f\upharpoonright\mathcal{L}(P)=f'\upharpoonright\mathcal{L}(P)$;
  \item $P[f][f']\sim_{hp}^{sl} P[f'\circ f]$;
  \item $(P\parallel Q)[f]\sim_{hp}^{sl} P[f]\parallel Q[f]$, if $f\upharpoonright(L\cup\overline{L})$ is one-to-one, where $L=\mathcal{L}(P)\cup\mathcal{L}(Q)$.
\end{enumerate}
\end{proposition}

\begin{proof}
\begin{enumerate}
  \item $P\parallel Q\sim_{hp}^{sl} Q\parallel P$. It is sufficient to prove the relation $R=\{(P\parallel Q, Q\parallel P)\}\cup \textbf{Id}$ is a strong static location hp-bisimulation for some distributions. It can be proved similarly to the proof of
  static laws for strong hp-bisimulation in CTC, we omit it;
  \item $P\parallel(Q\parallel R)\sim_{hp}^{sl} (P\parallel Q)\parallel R$. It is sufficient to prove the relation $R=\{(P\parallel(Q\parallel R), (P\parallel Q)\parallel R)\}\cup \textbf{Id}$ is a strong static location hp-bisimulation for some distributions. It can be proved similarly to the proof of
  static laws for strong hp-bisimulation in CTC, we omit it;
  \item $P\parallel \textbf{nil}\sim_{hp}^{sl} P$. It is sufficient to prove the relation $R=\{(P\parallel \textbf{nil}, P)\}\cup \textbf{Id}$ is a strong static location hp-bisimulation for some distributions. It can be proved similarly to the proof of
  static laws for strong hp-bisimulation in CTC, we omit it;
  \item $P\setminus L\sim_{hp}^{sl} P$, if $\mathcal{L}(P)\cap(L\cup\overline{L})=\emptyset$. It is sufficient to prove the relation $R=\{(P\setminus L, P)\}\cup \textbf{Id}$, if $\mathcal{L}(P)\cap(L\cup\overline{L})=\emptyset$, is a strong static location hp-bisimulation for some distributions. It can be proved similarly to the proof of
  static laws for strong hp-bisimulation in CTC, we omit it;
  \item $P\setminus K\setminus L\sim_{hp}^{sl} P\setminus(K\cup L)$. It is sufficient to prove the relation $R=\{(P\setminus K\setminus L, P\setminus(K\cup L))\}\cup \textbf{Id}$ is a strong static location hp-bisimulation for some distributions. It can be proved similarly to the proof of
  static laws for strong hp-bisimulation in CTC, we omit it;
  \item $P[f]\setminus L\sim_{hp}^{sl} P\setminus f^{-1}(L)[f]$. It is sufficient to prove the relation $R=\{(P[f]\setminus L, P\setminus f^{-1}(L)[f])\}\cup \textbf{Id}$ is a strong static location hp-bisimulation for some distributions. It can be proved similarly to the proof of
  static laws for strong hp-bisimulation in CTC, we omit it;
  \item $(P\parallel Q)\setminus L\sim_{hp}^{sl} P\setminus L\parallel Q\setminus L$, if $\mathcal{L}(P)\cap\overline{\mathcal{L}(Q)}\cap(L\cup\overline{L})=\emptyset$. It is sufficient to prove the relation $R=\{((P\parallel Q)\setminus L, P\setminus L\parallel Q\setminus L)\}\cup \textbf{Id}$, if $\mathcal{L}(P)\cap\overline{\mathcal{L}(Q)}\cap(L\cup\overline{L})=\emptyset$, is a strong static location hp-bisimulation for some distributions. It can be proved similarly to the proof of
  static laws for strong hp-bisimulation in CTC, we omit it;
  \item $P[Id]\sim_{hp}^{sl} P$. It is sufficient to prove the relation $R=\{(P[Id], P)\}\cup \textbf{Id}$ is a strong static location hp-bisimulation for some distributions. It can be proved similarly to the proof of
  static laws for strong hp-bisimulation in CTC, we omit it;
  \item $P[f]\sim_{hp}^{sl} P[f']$, if $f\upharpoonright\mathcal{L}(P)=f'\upharpoonright\mathcal{L}(P)$. It is sufficient to prove the relation $R=\{(P[f], P[f'])\}\cup \textbf{Id}$, if $f\upharpoonright\mathcal{L}(P)=f'\upharpoonright\mathcal{L}(P)$, is a strong static location hp-bisimulation for some distributions. It can be proved similarly to the proof of
  static laws for strong hp-bisimulation in CTC, we omit it;
  \item $P[f][f']\sim_{hp}^{sl} P[f'\circ f]$. It is sufficient to prove the relation $R=\{(P[f][f'], P[f'\circ f])\}\cup \textbf{Id}$ is a strong static location hp-bisimulation for some distributions. It can be proved similarly to the proof of
  static laws for strong hp-bisimulation in CTC, we omit it;
  \item $(P\parallel Q)[f]\sim_{hp}^{sl} P[f]\parallel Q[f]$, if $f\upharpoonright(L\cup\overline{L})$ is one-to-one, where $L=\mathcal{L}(P)\cup\mathcal{L}(Q)$. It is sufficient to prove the relation $R=\{((P\parallel Q)[f], P[f]\parallel Q[f])\}\cup \textbf{Id}$, if $f\upharpoonright(L\cup\overline{L})$ is one-to-one, where $L=\mathcal{L}(P)\cup\mathcal{L}(Q)$, is a strong static location hp-bisimulation for some distributions. It can be proved similarly to the proof of
  static laws for strong hp-bisimulation in CTC, we omit it.
\end{enumerate}
\end{proof}

\begin{proposition}[Static laws for strong static location hhp-bisimulation]
The static laws for strong static location hhp-bisimulation are as follows.

\begin{enumerate}
  \item $P\parallel Q\sim_{hhp}^{sl} Q\parallel P$;
  \item $P\parallel(Q\parallel R)\sim_{hhp}^{sl} (P\parallel Q)\parallel R$;
  \item $P\parallel \textbf{nil}\sim_{hhp}^{sl} P$;
  \item $P\setminus L\sim_{hhp}^{sl} P$, if $\mathcal{L}(P)\cap(L\cup\overline{L})=\emptyset$;
  \item $P\setminus K\setminus L\sim_{hhp}^{sl} P\setminus(K\cup L)$;
  \item $P[f]\setminus L\sim_{hhp}^{sl} P\setminus f^{-1}(L)[f]$;
  \item $(P\parallel Q)\setminus L\sim_{hhp}^{sl} P\setminus L\parallel Q\setminus L$, if $\mathcal{L}(P)\cap\overline{\mathcal{L}(Q)}\cap(L\cup\overline{L})=\emptyset$;
  \item $P[Id]\sim_{hhp}^{sl} P$;
  \item $P[f]\sim_{hhp}^{sl} P[f']$, if $f\upharpoonright\mathcal{L}(P)=f'\upharpoonright\mathcal{L}(P)$;
  \item $P[f][f']\sim_{hhp}^{sl} P[f'\circ f]$;
  \item $(P\parallel Q)[f]\sim_{hhp}^{sl} P[f]\parallel Q[f]$, if $f\upharpoonright(L\cup\overline{L})$ is one-to-one, where $L=\mathcal{L}(P)\cup\mathcal{L}(Q)$.
\end{enumerate}
\end{proposition}

\begin{proof}
\begin{enumerate}
  \item $P\parallel Q\sim_{hhp}^{sl} Q\parallel P$. It is sufficient to prove the relation $R=\{(P\parallel Q, Q\parallel P)\}\cup \textbf{Id}$ is a strong static location hhp-bisimulation for some distributions. It can be proved similarly to the proof of
  static laws for strong hhp-bisimulation in CTC, we omit it;
  \item $P\parallel(Q\parallel R)\sim_{hhp}^{sl} (P\parallel Q)\parallel R$. It is sufficient to prove the relation $R=\{(P\parallel(Q\parallel R), (P\parallel Q)\parallel R)\}\cup \textbf{Id}$ is a strong static location hhp-bisimulation for some distributions. It can be proved similarly to the proof of
  static laws for strong hhp-bisimulation in CTC, we omit it;
  \item $P\parallel \textbf{nil}\sim_{hhp}^{sl} P$. It is sufficient to prove the relation $R=\{(P\parallel \textbf{nil}, P)\}\cup \textbf{Id}$ is a strong static location hhp-bisimulation for some distributions. It can be proved similarly to the proof of
  static laws for strong hhp-bisimulation in CTC, we omit it;
  \item $P\setminus L\sim_{hhp}^{sl} P$, if $\mathcal{L}(P)\cap(L\cup\overline{L})=\emptyset$. It is sufficient to prove the relation $R=\{(P\setminus L, P)\}\cup \textbf{Id}$, if $\mathcal{L}(P)\cap(L\cup\overline{L})=\emptyset$, is a strong static location hhp-bisimulation for some distributions. It can be proved similarly to the proof of
  static laws for strong hhp-bisimulation in CTC, we omit it;
  \item $P\setminus K\setminus L\sim_{hhp}^{sl} P\setminus(K\cup L)$. It is sufficient to prove the relation $R=\{(P\setminus K\setminus L, P\setminus(K\cup L))\}\cup \textbf{Id}$ is a strong static location hhp-bisimulation for some distributions. It can be proved similarly to the proof of
  static laws for strong hhp-bisimulation in CTC, we omit it;
  \item $P[f]\setminus L\sim_{hhp}^{sl} P\setminus f^{-1}(L)[f]$. It is sufficient to prove the relation $R=\{(P[f]\setminus L, P\setminus f^{-1}(L)[f])\}\cup \textbf{Id}$ is a strong static location hhp-bisimulation for some distributions. It can be proved similarly to the proof of
  static laws for strong hhp-bisimulation in CTC, we omit it;
  \item $(P\parallel Q)\setminus L\sim_{hhp}^{sl} P\setminus L\parallel Q\setminus L$, if $\mathcal{L}(P)\cap\overline{\mathcal{L}(Q)}\cap(L\cup\overline{L})=\emptyset$. It is sufficient to prove the relation $R=\{((P\parallel Q)\setminus L, P\setminus L\parallel Q\setminus L)\}\cup \textbf{Id}$, if $\mathcal{L}(P)\cap\overline{\mathcal{L}(Q)}\cap(L\cup\overline{L})=\emptyset$, is a strong static location hhp-bisimulation for some distributions. It can be proved similarly to the proof of
  static laws for strong hhp-bisimulation in CTC, we omit it;
  \item $P[Id]\sim_{hhp}^{sl} P$. It is sufficient to prove the relation $R=\{(P[Id], P)\}\cup \textbf{Id}$ is a strong static location hhp-bisimulation for some distributions. It can be proved similarly to the proof of
  static laws for strong hhp-bisimulation in CTC, we omit it;
  \item $P[f]\sim_{hhp}^{sl} P[f']$, if $f\upharpoonright\mathcal{L}(P)=f'\upharpoonright\mathcal{L}(P)$. It is sufficient to prove the relation $R=\{(P[f], P[f'])\}\cup \textbf{Id}$, if $f\upharpoonright\mathcal{L}(P)=f'\upharpoonright\mathcal{L}(P)$, is a strong static location hhp-bisimulation for some distributions. It can be proved similarly to the proof of
  static laws for strong hhp-bisimulation in CTC, we omit it;
  \item $P[f][f']\sim_{hhp}^{sl} P[f'\circ f]$. It is sufficient to prove the relation $R=\{(P[f][f'], P[f'\circ f])\}\cup \textbf{Id}$ is a strong static location hhp-bisimulation for some distributions. It can be proved similarly to the proof of
  static laws for strong hhp-bisimulation in CTC, we omit it;
  \item $(P\parallel Q)[f]\sim_{hhp}^{sl} P[f]\parallel Q[f]$, if $f\upharpoonright(L\cup\overline{L})$ is one-to-one, where $L=\mathcal{L}(P)\cup\mathcal{L}(Q)$. It is sufficient to prove the relation $R=\{((P\parallel Q)[f], P[f]\parallel Q[f])\}\cup \textbf{Id}$, if $f\upharpoonright(L\cup\overline{L})$ is one-to-one, where $L=\mathcal{L}(P)\cup\mathcal{L}(Q)$, is a strong static location hhp-bisimulation for some distributions. It can be proved similarly to the proof of
  static laws for strong hhp-bisimulation in CTC, we omit it.
\end{enumerate}
\end{proof}

\begin{proposition}[Location laws for strong static location pomset bisimulation]
The location laws for strong static location pomset bisimulation are as follows.

\begin{enumerate}
  \item $\epsilon::P\sim_p^{sl} P$;
  \item $u::\textbf{nil}\sim_p^{sl} \textbf{nil}$;
  \item $u::(\alpha.P)\sim_p^{sl} u::\alpha.u::P$;
  \item $u::(P+Q)\sim_p^{sl} u::P+u::Q$;
  \item $u::(P\parallel Q)\sim_p^{sl}u::P\parallel u::Q$;
  \item $u::(P\setminus L)\sim_p^{sl}u::P\setminus L$;
  \item $u::(P[f])\sim_p^{sl}u::P[f]$;
  \item $u::(v::P)\sim_p^{sl}uv::P$.
\end{enumerate}
\end{proposition}

\begin{proof}
\begin{enumerate}
  \item $\epsilon::P\sim_p^{sl} P$. It is sufficient to prove the relation $R=\{(\epsilon::P, P)\}\cup \textbf{Id}$ is a strong static location pomset bisimulation, we omit it;
  \item $u::\textbf{nil}\sim_p^{sl} \textbf{nil}$. It is sufficient to prove the relation $R=\{(u::\textbf{nil}, \textbf{nil})\}\cup \textbf{Id}$ is a strong static location pomset bisimulation, we omit it;
  \item $u::(\alpha.P)\sim_p^{sl} u::\alpha.u::P$. It is sufficient to prove the relation $R=\{(u::(\alpha.P), u::\alpha.u::P)\}\cup \textbf{Id}$ is a strong static location pomset bisimulation, we omit it;
  \item $u::(P+Q)\sim_p^{sl} u::P+u::Q$. It is sufficient to prove the relation $R=\{(u::(P+Q), u::P+u::Q)\}\cup \textbf{Id}$ is a strong static location pomset bisimulation, we omit it;
  \item $u::(P\parallel Q)\sim_p^{sl}u::P\parallel u::Q$. It is sufficient to prove the relation $R=\{(u::(P\parallel Q), u::P\parallel u::Q)\}\cup \textbf{Id}$ is a strong static location pomset bisimulation, we omit it;
  \item $u::(P\setminus L)\sim_p^{sl}u::P\setminus L$. It is sufficient to prove the relation $R=\{(u::(P\setminus L), u::P\setminus L)\}\cup \textbf{Id}$ is a strong static location pomset bisimulation, we omit it;
  \item $u::(P[f])\sim_p^{sl}u::P[f]$. It is sufficient to prove the relation $R=\{(u::(P[f]), u::P[f])\}\cup \textbf{Id}$ is a strong static location pomset bisimulation, we omit it;
  \item $u::(v::P)\sim_p^{sl}uv::P$. It is sufficient to prove the relation $R=\{(u::(v::P), uv::P)\}\cup \textbf{Id}$ is a strong static location pomset bisimulation, we omit it.
\end{enumerate}
\end{proof}

\begin{proposition}[Location laws for strong static location step bisimulation]
The location laws for strong static location step bisimulation are as follows.

\begin{enumerate}
  \item $\epsilon::P\sim_s^{sl} P$;
  \item $u::\textbf{nil}\sim_s^{sl} \textbf{nil}$;
  \item $u::(\alpha.P)\sim_s^{sl} u::\alpha.u::P$;
  \item $u::(P+Q)\sim_s^{sl} u::P+u::Q$;
  \item $u::(P\parallel Q)\sim_s^{sl}u::P\parallel u::Q$;
  \item $u::(P\setminus L)\sim_s^{sl}u::P\setminus L$;
  \item $u::(P[f])\sim_s^{sl}u::P[f]$;
  \item $u::(v::P)\sim_s^{sl}uv::P$.
\end{enumerate}
\end{proposition}

\begin{proof}
\begin{enumerate}
  \item $\epsilon::P\sim_s^{sl} P$. It is sufficient to prove the relation $R=\{(\epsilon::P, P)\}\cup \textbf{Id}$ is a strong static location step bisimulation, we omit it;
  \item $u::\textbf{nil}\sim_s^{sl} \textbf{nil}$. It is sufficient to prove the relation $R=\{(u::\textbf{nil}, \textbf{nil})\}\cup \textbf{Id}$ is a strong static location step bisimulation, we omit it;
  \item $u::(\alpha.P)\sim_s^{sl} u::\alpha.u::P$. It is sufficient to prove the relation $R=\{(u::(\alpha.P), u::\alpha.u::P)\}\cup \textbf{Id}$ is a strong static location step bisimulation, we omit it;
  \item $u::(P+Q)\sim_s^{sl} u::P+u::Q$. It is sufficient to prove the relation $R=\{(u::(P+Q), u::P+u::Q)\}\cup \textbf{Id}$ is a strong static location step bisimulation, we omit it;
  \item $u::(P\parallel Q)\sim_s^{sl}u::P\parallel u::Q$. It is sufficient to prove the relation $R=\{(u::(P\parallel Q), u::P\parallel u::Q)\}\cup \textbf{Id}$ is a strong static location step bisimulation, we omit it;
  \item $u::(P\setminus L)\sim_s^{sl}u::P\setminus L$. It is sufficient to prove the relation $R=\{(u::(P\setminus L), u::P\setminus L)\}\cup \textbf{Id}$ is a strong static location step bisimulation, we omit it;
  \item $u::(P[f])\sim_s^{sl}u::P[f]$. It is sufficient to prove the relation $R=\{(u::(P[f]), u::P[f])\}\cup \textbf{Id}$ is a strong static location step bisimulation, we omit it;
  \item $u::(v::P)\sim_s^{sl}uv::P$. It is sufficient to prove the relation $R=\{(u::(v::P), uv::P)\}\cup \textbf{Id}$ is a strong static location step bisimulation, we omit it.
\end{enumerate}
\end{proof}

\begin{proposition}[Location laws for strong static location hp-bisimulation]
The location laws for strong static location hp-bisimulation are as follows.

\begin{enumerate}
  \item $\epsilon::P\sim_{hp}^{sl} P$;
  \item $u::\textbf{nil}\sim_{hp}^{sl} \textbf{nil}$;
  \item $u::(\alpha.P)\sim_{hp}^{sl} u::\alpha.u::P$;
  \item $u::(P+Q)\sim_{hp}^{sl} u::P+u::Q$;
  \item $u::(P\parallel Q)\sim_{hp}^{sl}u::P\parallel u::Q$;
  \item $u::(P\setminus L)\sim_{hp}^{sl}u::P\setminus L$;
  \item $u::(P[f])\sim_{hp}^{sl}u::P[f]$;
  \item $u::(v::P)\sim_{hp}^{sl}uv::P$.
\end{enumerate}
\end{proposition}

\begin{proof}
\begin{enumerate}
  \item $\epsilon::P\sim_{hp}^{sl} P$. It is sufficient to prove the relation $R=\{(\epsilon::P, P)\}\cup \textbf{Id}$ is a strong static location hp-bisimulation, we omit it;
  \item $u::\textbf{nil}\sim_{hp}^{sl} \textbf{nil}$. It is sufficient to prove the relation $R=\{(u::\textbf{nil}, \textbf{nil})\}\cup \textbf{Id}$ is a strong static location hp-bisimulation, we omit it;
  \item $u::(\alpha.P)\sim_{hp}^{sl} u::\alpha.u::P$. It is sufficient to prove the relation $R=\{(u::(\alpha.P), u::\alpha.u::P)\}\cup \textbf{Id}$ is a strong static location hp-bisimulation, we omit it;
  \item $u::(P+Q)\sim_{hp}^{sl} u::P+u::Q$. It is sufficient to prove the relation $R=\{(u::(P+Q), u::P+u::Q)\}\cup \textbf{Id}$ is a strong static location hp-bisimulation, we omit it;
  \item $u::(P\parallel Q)\sim_{hp}^{sl}u::P\parallel u::Q$. It is sufficient to prove the relation $R=\{(u::(P\parallel Q), u::P\parallel u::Q)\}\cup \textbf{Id}$ is a strong static location hp-bisimulation, we omit it;
  \item $u::(P\setminus L)\sim_{hp}^{sl}u::P\setminus L$. It is sufficient to prove the relation $R=\{(u::(P\setminus L), u::P\setminus L)\}\cup \textbf{Id}$ is a strong static location hp-bisimulation, we omit it;
  \item $u::(P[f])\sim_{hp}^{sl}u::P[f]$. It is sufficient to prove the relation $R=\{(u::(P[f]), u::P[f])\}\cup \textbf{Id}$ is a strong static location hp-bisimulation, we omit it;
  \item $u::(v::P)\sim_{hp}^{sl}uv::P$. It is sufficient to prove the relation $R=\{(u::(v::P), uv::P)\}\cup \textbf{Id}$ is a strong static location hp-bisimulation, we omit it.
\end{enumerate}
\end{proof}

\begin{proposition}[Location laws for strong static location hhp-bisimulation]
The location laws for strong static location hhp-bisimulation are as follows.

\begin{enumerate}
  \item $\epsilon::P\sim_{hhp}^{sl} P$;
  \item $u::\textbf{nil}\sim_{hhp}^{sl} \textbf{nil}$;
  \item $u::(\alpha.P)\sim_{hhp}^{sl} u::\alpha.u::P$;
  \item $u::(P+Q)\sim_{hhp}^{sl} u::P+u::Q$;
  \item $u::(P\parallel Q)\sim_{hhp}^{sl}u::P\parallel u::Q$;
  \item $u::(P\setminus L)\sim_{hhp}^{sl}u::P\setminus L$;
  \item $u::(P[f])\sim_{hhp}^{sl}u::P[f]$;
  \item $u::(v::P)\sim_{hhp}^{sl}uv::P$.
\end{enumerate}
\end{proposition}

\begin{proof}
\begin{enumerate}
  \item $\epsilon::P\sim_{hhp}^{sl} P$. It is sufficient to prove the relation $R=\{(\epsilon::P, P)\}\cup \textbf{Id}$ is a strong static location hhp-bisimulation, we omit it;
  \item $u::\textbf{nil}\sim_{hhp}^{sl} \textbf{nil}$. It is sufficient to prove the relation $R=\{(u::\textbf{nil}, \textbf{nil})\}\cup \textbf{Id}$ is a strong static location hhp-bisimulation, we omit it;
  \item $u::(\alpha.P)\sim_{hhp}^{sl} u::\alpha.u::P$. It is sufficient to prove the relation $R=\{(u::(\alpha.P), u::\alpha.u::P)\}\cup \textbf{Id}$ is a strong static location hhp-bisimulation, we omit it;
  \item $u::(P+Q)\sim_{hhp}^{sl} u::P+u::Q$. It is sufficient to prove the relation $R=\{(u::(P+Q), u::P+u::Q)\}\cup \textbf{Id}$ is a strong static location hhp-bisimulation, we omit it;
  \item $u::(P\parallel Q)\sim_{hhp}^{sl}u::P\parallel u::Q$. It is sufficient to prove the relation $R=\{(u::(P\parallel Q), u::P\parallel u::Q)\}\cup \textbf{Id}$ is a strong static location hhp-bisimulation, we omit it;
  \item $u::(P\setminus L)\sim_{hhp}^{sl}u::P\setminus L$. It is sufficient to prove the relation $R=\{(u::(P\setminus L), u::P\setminus L)\}\cup \textbf{Id}$ is a strong static location hhp-bisimulation, we omit it;
  \item $u::(P[f])\sim_{hhp}^{sl}u::P[f]$. It is sufficient to prove the relation $R=\{(u::(P[f]), u::P[f])\}\cup \textbf{Id}$ is a strong static location hhp-bisimulation, we omit it;
  \item $u::(v::P)\sim_{hhp}^{sl}uv::P$. It is sufficient to prove the relation $R=\{(u::(v::P), uv::P)\}\cup \textbf{Id}$ is a strong static location hhp-bisimulation, we omit it.
\end{enumerate}
\end{proof}

\begin{proposition}[Expansion law for strong static location pomset bisimulation]
Let $P\equiv (P_1[f_1]\parallel\cdots\parallel P_n[f_n])\setminus L$, with $n\geq 1$. Then

\begin{eqnarray}
P\sim_p^{sl} \{(f_1(\alpha_1)\parallel\cdots\parallel f_n(\alpha_n)).(P_1'[f_1]\parallel\cdots\parallel P_n'[f_n])\setminus L: \nonumber\\
P_i\xrightarrow[u_i]{\alpha_i}P_i',i\in\{1,\cdots,n\},f_i(\alpha_i)\notin L\cup\overline{L}\} \nonumber\\
+\sum\{\tau.(P_1[f_1]\parallel\cdots\parallel P_i'[f_i]\parallel\cdots\parallel P_j'[f_j]\parallel\cdots\parallel P_n[f_n])\setminus L: \nonumber\\
P_i\xrightarrow[v_i]{l_1}P_i',P_j\xrightarrow[v_j]{l_2}P_j',f_i(l_1)=\overline{f_j(l_2)},i<j\} \nonumber
\end{eqnarray}
\end{proposition}

\begin{proof}
Firstly, we consider the case without Restriction and Relabeling. That is, we suffice to prove the following case by induction on the size $n$. Note that, we consider the general distribution.

For $P\equiv loc_1::P_1\parallel\cdots\parallel loc_n::P_n$, with $n\geq 1$, we need to prove

\begin{eqnarray}
P\sim_p^{sl} \{(loc_1::\alpha_1\parallel\cdots\parallel loc_n::\alpha_n).(loc_1::P_1'\parallel\cdots\parallel loc_n::P_n'): \nonumber\\
loc_i::P_i\xrightarrow[u_i]{\alpha_i}loc_i::P_i',i\in\{1,\cdots,n\}\nonumber\\
+\sum\{\tau.(loc_1::P_1\parallel\cdots\parallel loc_i::P_i'\parallel\cdots\parallel loc_j::P_j'\parallel\cdots\parallel loc_n::P_n): \nonumber\\
loc_i::P_i\xrightarrow[v_i]{l}loc_i::P_i',loc_j::P_j\xrightarrow[v_j]{\overline{l}}loc_j::P_j',i<j\} \nonumber
\end{eqnarray}

For $n=1$, $loc_1::P_1\sim_p^{sl} loc_1::\alpha_1.loc_1::P_1':loc_1::P_1\xrightarrow[u_1]{\alpha_1}loc_1::P_1'$ is obvious. Then with a hypothesis $n$, we consider $R\equiv P\parallel loc_{n+1}::P_{n+1}$. By the transition rules $\textbf{Com}_{1,2,3,4}$, we can get

\begin{eqnarray}
R\sim_p^{sl} \{(loc::p\parallel loc_{n+1}::\alpha_{n+1}).(loc::P'\parallel loc_{n+1}::P_{n+1}'): \nonumber\\
loc::P\xrightarrow[u]{p}loc::P',loc_{n+1}::P_{n+1}\xrightarrow[u_{n+1}]{\alpha_{n+1}}loc_{n+1}::P_{n+1}',p\subseteq P\}\nonumber\\
+\sum\{\tau.(loc::P'\parallel loc_{n+1}::P_{n+1}'): \nonumber\\
loc::P\xrightarrow[v]{l}loc::P',loc_{n+1}::P_{n+1}\xrightarrow[v_{n+1}]{\overline{l}}loc_{n+1}::P_{n+1}'\} \nonumber
\end{eqnarray}

Now with the induction assumption $P\equiv loc_1::P_1\parallel\cdots\parallel loc_n::P_n$, the right-hand side can be reformulated as follows.

\begin{eqnarray}
\{(loc_1::\alpha_1\parallel\cdots\parallel loc_n::\alpha_n\parallel loc_{n+1}::\alpha_{n+1}).(loc_1::P_1'\parallel\cdots\parallel loc_n::P_n'\parallel loc_{n+1}::P_{n+1}'): \nonumber\\
loc_i::P_i\xrightarrow[u_i]{\alpha_i}loc_i::P_i',i\in\{1,\cdots,n+1\}\nonumber\\
+\sum\{\tau.(loc_1::P_1\parallel\cdots\parallel loc_i::P_i'\parallel\cdots\parallel loc_j::P_j'\parallel\cdots\parallel loc_n::P_n\parallel loc_{n+1}::P_{n+1}): \nonumber\\
loc_i::P_i\xrightarrow[v_i]{l}loc_i::P_i',loc_j::P_j\xrightarrow[v_j]{\overline{l}}loc_j::P_j',i<j\} \nonumber\\
+\sum\{\tau.(loc_1::P_1\parallel\cdots\parallel loc_i::P_i'\parallel\cdots\parallel loc_j::P_j\parallel\cdots\parallel loc_n::P_n\parallel loc_{n+1}::P_{n+1}'): \nonumber\\
loc_i::P_i\xrightarrow[v_i]{l}loc_i::P_i',loc_{n+1}::P_{n+1}\xrightarrow[v_{n+1}]{\overline{l}}loc_{n+1}::P_{n+1}',i\in\{1,\cdots, n\}\} \nonumber
\end{eqnarray}

So,

\begin{eqnarray}
R\sim_p^{sl} \{(loc_1::\alpha_1\parallel\cdots\parallel loc_n::\alpha_n\parallel loc_{n+1}::\alpha_{n+1}).(loc_1::P_1'\parallel\cdots\parallel loc_n::P_n'\parallel loc_{n+1}::P_{n+1}'): \nonumber\\
loc_i::P_i\xrightarrow[u_i]{\alpha_i}loc_i::P_i',i\in\{1,\cdots,n+1\}\nonumber\\
+\sum\{\tau.(loc_1::P_1\parallel\cdots\parallel loc_i::P_i'\parallel\cdots\parallel loc_j::P_j'\parallel\cdots\parallel loc_n::P_n): \nonumber\\
loc_i::P_i\xrightarrow[v_i]{l}loc_i::P_i',loc_j::P_j\xrightarrow[v_j]{\overline{l}}loc_j::P_j',1 \leq i<j\geq n+1\} \nonumber
\end{eqnarray}

Then, we can easily add the full conditions with Restriction and Relabeling.
\end{proof}

\begin{proposition}[Expansion law for strong static location step bisimulation]
Let $P\equiv (P_1[f_1]\parallel\cdots\parallel P_n[f_n])\setminus L$, with $n\geq 1$. Then

\begin{eqnarray}
P\sim_s^{sl} \{(f_1(\alpha_1)\parallel\cdots\parallel f_n(\alpha_n)).(P_1'[f_1]\parallel\cdots\parallel P_n'[f_n])\setminus L: \nonumber\\
P_i\xrightarrow[u_i]{\alpha_i}P_i',i\in\{1,\cdots,n\},f_i(\alpha_i)\notin L\cup\overline{L}\} \nonumber\\
+\sum\{\tau.(P_1[f_1]\parallel\cdots\parallel P_i'[f_i]\parallel\cdots\parallel P_j'[f_j]\parallel\cdots\parallel P_n[f_n])\setminus L: \nonumber\\
P_i\xrightarrow[v_i]{l_1}P_i',P_j\xrightarrow[v_j]{l_2}P_j',f_i(l_1)=\overline{f_j(l_2)},i<j\} \nonumber
\end{eqnarray}
\end{proposition}

\begin{proof}
Firstly, we consider the case without Restriction and Relabeling. That is, we suffice to prove the following case by induction on the size $n$. Note that, we consider the general distribution.

For $P\equiv loc_1::P_1\parallel\cdots\parallel loc_n::P_n$, with $n\geq 1$, we need to prove

\begin{eqnarray}
P\sim_s^{sl} \{(loc_1::\alpha_1\parallel\cdots\parallel loc_n::\alpha_n).(loc_1::P_1'\parallel\cdots\parallel loc_n::P_n'): \nonumber\\
loc_i::P_i\xrightarrow[u_i]{\alpha_i}loc_i::P_i',i\in\{1,\cdots,n\}\nonumber\\
+\sum\{\tau.(loc_1::P_1\parallel\cdots\parallel loc_i::P_i'\parallel\cdots\parallel loc_j::P_j'\parallel\cdots\parallel loc_n::P_n): \nonumber\\
loc_i::P_i\xrightarrow[v_i]{l}loc_i::P_i',loc_j::P_j\xrightarrow[v_j]{\overline{l}}loc_j::P_j',i<j\} \nonumber
\end{eqnarray}

For $n=1$, $loc_1::P_1\sim_s^{sl} loc_1::\alpha_1.loc_1::P_1':loc_1::P_1\xrightarrow[u_1]{\alpha_1}loc_1::P_1'$ is obvious. Then with a hypothesis $n$, we consider $R\equiv P\parallel loc_{n+1}::P_{n+1}$. By the transition rules $\textbf{Com}_{1,2,3,4}$, we can get

\begin{eqnarray}
R\sim_s^{sl} \{(loc::p\parallel loc_{n+1}::\alpha_{n+1}).(loc::P'\parallel loc_{n+1}::P_{n+1}'): \nonumber\\
loc::P\xrightarrow[u]{p}loc::P',loc_{n+1}::P_{n+1}\xrightarrow[u_{n+1}]{\alpha_{n+1}}loc_{n+1}::P_{n+1}',p\subseteq P\}\nonumber\\
+\sum\{\tau.(loc::P'\parallel loc_{n+1}::P_{n+1}'): \nonumber\\
loc::P\xrightarrow[v]{l}loc::P',loc_{n+1}::P_{n+1}\xrightarrow[v_{n+1}]{\overline{l}}loc_{n+1}::P_{n+1}'\} \nonumber
\end{eqnarray}

Now with the induction assumption $P\equiv loc_1::P_1\parallel\cdots\parallel loc_n::P_n$, the right-hand side can be reformulated as follows.

\begin{eqnarray}
\{(loc_1::\alpha_1\parallel\cdots\parallel loc_n::\alpha_n\parallel loc_{n+1}::\alpha_{n+1}).(loc_1::P_1'\parallel\cdots\parallel loc_n::P_n'\parallel loc_{n+1}::P_{n+1}'): \nonumber\\
loc_i::P_i\xrightarrow[u_i]{\alpha_i}loc_i::P_i',i\in\{1,\cdots,n+1\}\nonumber\\
+\sum\{\tau.(loc_1::P_1\parallel\cdots\parallel loc_i::P_i'\parallel\cdots\parallel loc_j::P_j'\parallel\cdots\parallel loc_n::P_n\parallel loc_{n+1}::P_{n+1}): \nonumber\\
loc_i::P_i\xrightarrow[v_i]{l}loc_i::P_i',loc_j::P_j\xrightarrow[v_j]{\overline{l}}loc_j::P_j',i<j\} \nonumber\\
+\sum\{\tau.(loc_1::P_1\parallel\cdots\parallel loc_i::P_i'\parallel\cdots\parallel loc_j::P_j\parallel\cdots\parallel loc_n::P_n\parallel loc_{n+1}::P_{n+1}'): \nonumber\\
loc_i::P_i\xrightarrow[v_i]{l}loc_i::P_i',loc_{n+1}::P_{n+1}\xrightarrow[v_{n+1}]{\overline{l}}loc_{n+1}::P_{n+1}',i\in\{1,\cdots, n\}\} \nonumber
\end{eqnarray}

So,

\begin{eqnarray}
R\sim_s^{sl} \{(loc_1::\alpha_1\parallel\cdots\parallel loc_n::\alpha_n\parallel loc_{n+1}::\alpha_{n+1}).(loc_1::P_1'\parallel\cdots\parallel loc_n::P_n'\parallel loc_{n+1}::P_{n+1}'): \nonumber\\
loc_i::P_i\xrightarrow[u_i]{\alpha_i}loc_i::P_i',i\in\{1,\cdots,n+1\}\nonumber\\
+\sum\{\tau.(loc_1::P_1\parallel\cdots\parallel loc_i::P_i'\parallel\cdots\parallel loc_j::P_j'\parallel\cdots\parallel loc_n::P_n): \nonumber\\
loc_i::P_i\xrightarrow[v_i]{l}loc_i::P_i',loc_j::P_j\xrightarrow[v_j]{\overline{l}}loc_j::P_j',1 \leq i<j\geq n+1\} \nonumber
\end{eqnarray}

Then, we can easily add the full conditions with Restriction and Relabeling.
\end{proof}

\begin{proposition}[Expansion law for strong static location hp-bisimulation]
Let $P\equiv (P_1[f_1]\parallel\cdots\parallel P_n[f_n])\setminus L$, with $n\geq 1$. Then

\begin{eqnarray}
P\sim_{hp}^{sl} \{(f_1(\alpha_1)\parallel\cdots\parallel f_n(\alpha_n)).(P_1'[f_1]\parallel\cdots\parallel P_n'[f_n])\setminus L: \nonumber\\
P_i\xrightarrow[u_i]{\alpha_i}P_i',i\in\{1,\cdots,n\},f_i(\alpha_i)\notin L\cup\overline{L}\} \nonumber\\
+\sum\{\tau.(P_1[f_1]\parallel\cdots\parallel P_i'[f_i]\parallel\cdots\parallel P_j'[f_j]\parallel\cdots\parallel P_n[f_n])\setminus L: \nonumber\\
P_i\xrightarrow[v_i]{l_1}P_i',P_j\xrightarrow[v_j]{l_2}P_j',f_i(l_1)=\overline{f_j(l_2)},i<j\} \nonumber
\end{eqnarray}
\end{proposition}

\begin{proof}
Firstly, we consider the case without Restriction and Relabeling. That is, we suffice to prove the following case by induction on the size $n$. Note that, we consider the general distribution.

For $P\equiv loc_1::P_1\parallel\cdots\parallel loc_n::P_n$, with $n\geq 1$, we need to prove

\begin{eqnarray}
P\sim_{hp}^{sl} \{(loc_1::\alpha_1\parallel\cdots\parallel loc_n::\alpha_n).(loc_1::P_1'\parallel\cdots\parallel loc_n::P_n'): \nonumber\\
loc_i::P_i\xrightarrow[u_i]{\alpha_i}loc_i::P_i',i\in\{1,\cdots,n\}\nonumber\\
+\sum\{\tau.(loc_1::P_1\parallel\cdots\parallel loc_i::P_i'\parallel\cdots\parallel loc_j::P_j'\parallel\cdots\parallel loc_n::P_n): \nonumber\\
loc_i::P_i\xrightarrow[v_i]{l}loc_i::P_i',loc_j::P_j\xrightarrow[v_j]{\overline{l}}loc_j::P_j',i<j\} \nonumber
\end{eqnarray}

For $n=1$, $loc_1::P_1\sim_{hp}^{sl} loc_1::\alpha_1.loc_1::P_1':loc_1::P_1\xrightarrow[u_1]{\alpha_1}loc_1::P_1'$ is obvious. Then with a hypothesis $n$, we consider $R\equiv P\parallel loc_{n+1}::P_{n+1}$. By the transition rules $\textbf{Com}_{1,2,3,4}$, we can get

\begin{eqnarray}
R\sim_{hp}^{sl} \{(loc::p\parallel loc_{n+1}::\alpha_{n+1}).(loc::P'\parallel loc_{n+1}::P_{n+1}'): \nonumber\\
loc::P\xrightarrow[u]{p}loc::P',loc_{n+1}::P_{n+1}\xrightarrow[u_{n+1}]{\alpha_{n+1}}loc_{n+1}::P_{n+1}',p\subseteq P\}\nonumber\\
+\sum\{\tau.(loc::P'\parallel loc_{n+1}::P_{n+1}'): \nonumber\\
loc::P\xrightarrow[v]{l}loc::P',loc_{n+1}::P_{n+1}\xrightarrow[v_{n+1}]{\overline{l}}loc_{n+1}::P_{n+1}'\} \nonumber
\end{eqnarray}

Now with the induction assumption $P\equiv loc_1::P_1\parallel\cdots\parallel loc_n::P_n$, the right-hand side can be reformulated as follows.

\begin{eqnarray}
\{(loc_1::\alpha_1\parallel\cdots\parallel loc_n::\alpha_n\parallel loc_{n+1}::\alpha_{n+1}).(loc_1::P_1'\parallel\cdots\parallel loc_n::P_n'\parallel loc_{n+1}::P_{n+1}'): \nonumber\\
loc_i::P_i\xrightarrow[u_i]{\alpha_i}loc_i::P_i',i\in\{1,\cdots,n+1\}\nonumber\\
+\sum\{\tau.(loc_1::P_1\parallel\cdots\parallel loc_i::P_i'\parallel\cdots\parallel loc_j::P_j'\parallel\cdots\parallel loc_n::P_n\parallel loc_{n+1}::P_{n+1}): \nonumber\\
loc_i::P_i\xrightarrow[v_i]{l}loc_i::P_i',loc_j::P_j\xrightarrow[v_j]{\overline{l}}loc_j::P_j',i<j\} \nonumber\\
+\sum\{\tau.(loc_1::P_1\parallel\cdots\parallel loc_i::P_i'\parallel\cdots\parallel loc_j::P_j\parallel\cdots\parallel loc_n::P_n\parallel loc_{n+1}::P_{n+1}'): \nonumber\\
loc_i::P_i\xrightarrow[v_i]{l}loc_i::P_i',loc_{n+1}::P_{n+1}\xrightarrow[v_{n+1}]{\overline{l}}loc_{n+1}::P_{n+1}',i\in\{1,\cdots, n\}\} \nonumber
\end{eqnarray}

So,

\begin{eqnarray}
R\sim_{hp}^{sl} \{(loc_1::\alpha_1\parallel\cdots\parallel loc_n::\alpha_n\parallel loc_{n+1}::\alpha_{n+1}).(loc_1::P_1'\parallel\cdots\parallel loc_n::P_n'\parallel loc_{n+1}::P_{n+1}'): \nonumber\\
loc_i::P_i\xrightarrow[u_i]{\alpha_i}loc_i::P_i',i\in\{1,\cdots,n+1\}\nonumber\\
+\sum\{\tau.(loc_1::P_1\parallel\cdots\parallel loc_i::P_i'\parallel\cdots\parallel loc_j::P_j'\parallel\cdots\parallel loc_n::P_n): \nonumber\\
loc_i::P_i\xrightarrow[v_i]{l}loc_i::P_i',loc_j::P_j\xrightarrow[v_j]{\overline{l}}loc_j::P_j',1 \leq i<j\geq n+1\} \nonumber
\end{eqnarray}

Then, we can easily add the full conditions with Restriction and Relabeling.
\end{proof}

\begin{proposition}[Expansion law for strong static location hhp-bisimulation]
Let $P\equiv (P_1[f_1]\parallel\cdots\parallel P_n[f_n])\setminus L$, with $n\geq 1$. Then

\begin{eqnarray}
P\sim_{hhp}^{sl} \{(f_1(\alpha_1)\parallel\cdots\parallel f_n(\alpha_n)).(P_1'[f_1]\parallel\cdots\parallel P_n'[f_n])\setminus L: \nonumber\\
P_i\xrightarrow[u_i]{\alpha_i}P_i',i\in\{1,\cdots,n\},f_i(\alpha_i)\notin L\cup\overline{L}\} \nonumber\\
+\sum\{\tau.(P_1[f_1]\parallel\cdots\parallel P_i'[f_i]\parallel\cdots\parallel P_j'[f_j]\parallel\cdots\parallel P_n[f_n])\setminus L: \nonumber\\
P_i\xrightarrow[v_i]{l_1}P_i',P_j\xrightarrow[v_j]{l_2}P_j',f_i(l_1)=\overline{f_j(l_2)},i<j\} \nonumber
\end{eqnarray}
\end{proposition}

\begin{proof}
Firstly, we consider the case without Restriction and Relabeling. That is, we suffice to prove the following case by induction on the size $n$. Note that, we consider the general distribution.

For $P\equiv loc_1::P_1\parallel\cdots\parallel loc_n::P_n$, with $n\geq 1$, we need to prove

\begin{eqnarray}
P\sim_{hhp}^{sl} \{(loc_1::\alpha_1\parallel\cdots\parallel loc_n::\alpha_n).(loc_1::P_1'\parallel\cdots\parallel loc_n::P_n'): \nonumber\\
loc_i::P_i\xrightarrow[u_i]{\alpha_i}loc_i::P_i',i\in\{1,\cdots,n\}\nonumber\\
+\sum\{\tau.(loc_1::P_1\parallel\cdots\parallel loc_i::P_i'\parallel\cdots\parallel loc_j::P_j'\parallel\cdots\parallel loc_n::P_n): \nonumber\\
loc_i::P_i\xrightarrow[v_i]{l}loc_i::P_i',loc_j::P_j\xrightarrow[v_j]{\overline{l}}loc_j::P_j',i<j\} \nonumber
\end{eqnarray}

For $n=1$, $loc_1::P_1\sim_{hhp}^{sl} loc_1::\alpha_1.loc_1::P_1':loc_1::P_1\xrightarrow[u_1]{\alpha_1}loc_1::P_1'$ is obvious. Then with a hypothesis $n$, we consider $R\equiv P\parallel loc_{n+1}::P_{n+1}$. By the transition rules $\textbf{Com}_{1,2,3,4}$, we can get

\begin{eqnarray}
R\sim_{hhp}^{sl} \{(loc::p\parallel loc_{n+1}::\alpha_{n+1}).(loc::P'\parallel loc_{n+1}::P_{n+1}'): \nonumber\\
loc::P\xrightarrow[u]{p}loc::P',loc_{n+1}::P_{n+1}\xrightarrow[u_{n+1}]{\alpha_{n+1}}loc_{n+1}::P_{n+1}',p\subseteq P\}\nonumber\\
+\sum\{\tau.(loc::P'\parallel loc_{n+1}::P_{n+1}'): \nonumber\\
loc::P\xrightarrow[v]{l}loc::P',loc_{n+1}::P_{n+1}\xrightarrow[v_{n+1}]{\overline{l}}loc_{n+1}::P_{n+1}'\} \nonumber
\end{eqnarray}

Now with the induction assumption $P\equiv loc_1::P_1\parallel\cdots\parallel loc_n::P_n$, the right-hand side can be reformulated as follows.

\begin{eqnarray}
\{(loc_1::\alpha_1\parallel\cdots\parallel loc_n::\alpha_n\parallel loc_{n+1}::\alpha_{n+1}).(loc_1::P_1'\parallel\cdots\parallel loc_n::P_n'\parallel loc_{n+1}::P_{n+1}'): \nonumber\\
loc_i::P_i\xrightarrow[u_i]{\alpha_i}loc_i::P_i',i\in\{1,\cdots,n+1\}\nonumber\\
+\sum\{\tau.(loc_1::P_1\parallel\cdots\parallel loc_i::P_i'\parallel\cdots\parallel loc_j::P_j'\parallel\cdots\parallel loc_n::P_n\parallel loc_{n+1}::P_{n+1}): \nonumber\\
loc_i::P_i\xrightarrow[v_i]{l}loc_i::P_i',loc_j::P_j\xrightarrow[v_j]{\overline{l}}loc_j::P_j',i<j\} \nonumber\\
+\sum\{\tau.(loc_1::P_1\parallel\cdots\parallel loc_i::P_i'\parallel\cdots\parallel loc_j::P_j\parallel\cdots\parallel loc_n::P_n\parallel loc_{n+1}::P_{n+1}'): \nonumber\\
loc_i::P_i\xrightarrow[v_i]{l}loc_i::P_i',loc_{n+1}::P_{n+1}\xrightarrow[v_{n+1}]{\overline{l}}loc_{n+1}::P_{n+1}',i\in\{1,\cdots, n\}\} \nonumber
\end{eqnarray}

So,

\begin{eqnarray}
R\sim_{hhp}^{sl} \{(loc_1::\alpha_1\parallel\cdots\parallel loc_n::\alpha_n\parallel loc_{n+1}::\alpha_{n+1}).(loc_1::P_1'\parallel\cdots\parallel loc_n::P_n'\parallel loc_{n+1}::P_{n+1}'): \nonumber\\
loc_i::P_i\xrightarrow[u_i]{\alpha_i}loc_i::P_i',i\in\{1,\cdots,n+1\}\nonumber\\
+\sum\{\tau.(loc_1::P_1\parallel\cdots\parallel loc_i::P_i'\parallel\cdots\parallel loc_j::P_j'\parallel\cdots\parallel loc_n::P_n): \nonumber\\
loc_i::P_i\xrightarrow[v_i]{l}loc_i::P_i',loc_j::P_j\xrightarrow[v_j]{\overline{l}}loc_j::P_j',1 \leq i<j\geq n+1\} \nonumber
\end{eqnarray}

Then, we can easily add the full conditions with Restriction and Relabeling.
\end{proof}

\begin{theorem}[Congruence for strong static location pomset bisimulation] \label{CSSB3}
We can enjoy the full congruence for strong static location pomset bisimulation as follows.
\begin{enumerate}
  \item If $A\overset{\text{def}}{=}P$, then $A\sim_p^{sl} P$;
  \item Let $P_1\sim_p^{sl} P_2$. Then
        \begin{enumerate}
           \item $loc::P_1\sim_p^{sl}loc::P_2$;
           \item $\alpha.P_1\sim_p^{sl} \alpha.P_2$;
           \item $(\alpha_1\parallel\cdots\parallel\alpha_n).P_1\sim_p^{sl} (\alpha_1\parallel\cdots\parallel\alpha_n).P_2$;
           \item $P_1+Q\sim_p^{sl} P_2 +Q$;
           \item $P_1\parallel Q\sim_p^{sl} P_2\parallel Q$;
           \item $P_1\setminus L\sim_p^{sl} P_2\setminus L$;
           \item $P_1[f]\sim_p^{sl} P_2[f]$.
         \end{enumerate}
\end{enumerate}
\end{theorem}

\begin{proof}
\begin{enumerate}
  \item If $A\overset{\text{def}}{=}P$, then $A\sim_p^{sl} P$. It is obvious.
  \item Let $P_1\sim_p^{sl} P_2$. Then
        \begin{enumerate}
           \item $loc::P_1\sim_p^{sl}loc::P_2$. It is sufficient to prove the relation $R=\{(loc::P_1, loc::P_2)\}\cup \textbf{Id}$ is a strong static location pomset bisimulation, we omit it;
           \item $\alpha.P_1\sim_p^{sl} \alpha.P_2$. It is sufficient to prove the relation $R=\{(\alpha.P_1, \alpha.P_2)\}\cup \textbf{Id}$ is a strong static location pomset bisimulation for some distributions. It can be proved similarly to the proof of
           congruence for strong pomset bisimulation in CTC, we omit it;
           \item $(\alpha_1\parallel\cdots\parallel\alpha_n).P_1\sim_p^{sl} (\alpha_1\parallel\cdots\parallel\alpha_n).P_2$. It is sufficient to prove the relation $R=\{((\alpha_1\parallel\cdots\parallel\alpha_n).P_1, (\alpha_1\parallel\cdots\parallel\alpha_n).P_2)\}\cup \textbf{Id}$ is a strong static location pomset bisimulation for some distributions. It can be proved similarly to the proof of
           congruence for strong pomset bisimulation in CTC, we omit it;
           \item $P_1+Q\sim_p^{sl} P_2 +Q$. It is sufficient to prove the relation $R=\{(P_1+Q, P_2+Q)\}\cup \textbf{Id}$ is a strong static location pomset bisimulation for some distributions. It can be proved similarly to the proof of
           congruence for strong pomset bisimulation in CTC, we omit it;
           \item $P_1\parallel Q\sim_p^{sl} P_2\parallel Q$. It is sufficient to prove the relation $R=\{(P_1\parallel Q, P_2\parallel Q)\}\cup \textbf{Id}$ is a strong static location pomset bisimulation for some distributions. It can be proved similarly to the proof of
           congruence for strong pomset bisimulation in CTC, we omit it;
           \item $P_1\setminus L\sim_p^{sl} P_2\setminus L$. It is sufficient to prove the relation $R=\{(P_1\setminus L, P_2\setminus L)\}\cup \textbf{Id}$ is a strong static location pomset bisimulation for some distributions. It can be proved similarly to the proof of
           congruence for strong pomset bisimulation in CTC, we omit it;
           \item $P_1[f]\sim_p^{sl} P_2[f]$. It is sufficient to prove the relation $R=\{(P_1[f], P_2[f])\}\cup \textbf{Id}$ is a strong static location pomset bisimulation for some distributions. It can be proved similarly to the proof of
           congruence for strong pomset bisimulation in CTC, we omit it.
         \end{enumerate}
\end{enumerate}
\end{proof}

\begin{theorem}[Congruence for strong static location step bisimulation] \label{CSSB3}
We can enjoy the full congruence for strong static location step bisimulation as follows.
\begin{enumerate}
  \item If $A\overset{\text{def}}{=}P$, then $A\sim_s^{sl} P$;
  \item Let $P_1\sim_s^{sl} P_2$. Then
        \begin{enumerate}
           \item $loc::P_1\sim_s^{sl}loc::P_2$;
           \item $\alpha.P_1\sim_s^{sl} \alpha.P_2$;
           \item $(\alpha_1\parallel\cdots\parallel\alpha_n).P_1\sim_s^{sl} (\alpha_1\parallel\cdots\parallel\alpha_n).P_2$;
           \item $P_1+Q\sim_s^{sl} P_2 +Q$;
           \item $P_1\parallel Q\sim_s^{sl} P_2\parallel Q$;
           \item $P_1\setminus L\sim_s^{sl} P_2\setminus L$;
           \item $P_1[f]\sim_s^{sl} P_2[f]$.
         \end{enumerate}
\end{enumerate}
\end{theorem}

\begin{proof}
\begin{enumerate}
  \item If $A\overset{\text{def}}{=}P$, then $A\sim_s^{sl} P$. It is obvious.
  \item Let $P_1\sim_s^{sl} P_2$. Then
        \begin{enumerate}
           \item $loc::P_1\sim_s^{sl}loc::P_2$. It is sufficient to prove the relation $R=\{(loc::P_1, loc::P_2)\}\cup \textbf{Id}$ is a strong static location step bisimulation, we omit it;
           \item $\alpha.P_1\sim_s^{sl} \alpha.P_2$. It is sufficient to prove the relation $R=\{(\alpha.P_1, \alpha.P_2)\}\cup \textbf{Id}$ is a strong static location step bisimulation for some distributions. It can be proved similarly to the proof of
           congruence for strong step bisimulation in CTC, we omit it;
           \item $(\alpha_1\parallel\cdots\parallel\alpha_n).P_1\sim_s^{sl} (\alpha_1\parallel\cdots\parallel\alpha_n).P_2$. It is sufficient to prove the relation $R=\{((\alpha_1\parallel\cdots\parallel\alpha_n).P_1, (\alpha_1\parallel\cdots\parallel\alpha_n).P_2)\}\cup \textbf{Id}$ is a strong static location step bisimulation for some distributions. It can be proved similarly to the proof of
           congruence for strong step bisimulation in CTC, we omit it;
           \item $P_1+Q\sim_s^{sl} P_2 +Q$. It is sufficient to prove the relation $R=\{(P_1+Q, P_2+Q)\}\cup \textbf{Id}$ is a strong static location step bisimulation for some distributions. It can be proved similarly to the proof of
           congruence for strong step bisimulation in CTC, we omit it;
           \item $P_1\parallel Q\sim_s^{sl} P_2\parallel Q$. It is sufficient to prove the relation $R=\{(P_1\parallel Q, P_2\parallel Q)\}\cup \textbf{Id}$ is a strong static location step bisimulation for some distributions. It can be proved similarly to the proof of
           congruence for strong step bisimulation in CTC, we omit it;
           \item $P_1\setminus L\sim_s^{sl} P_2\setminus L$. It is sufficient to prove the relation $R=\{(P_1\setminus L, P_2\setminus L)\}\cup \textbf{Id}$ is a strong static location step bisimulation for some distributions. It can be proved similarly to the proof of
           congruence for strong step bisimulation in CTC, we omit it;
           \item $P_1[f]\sim_s^{sl} P_2[f]$. It is sufficient to prove the relation $R=\{(P_1[f], P_2[f])\}\cup \textbf{Id}$ is a strong static location step bisimulation for some distributions. It can be proved similarly to the proof of
           congruence for strong step bisimulation in CTC, we omit it.
         \end{enumerate}
\end{enumerate}
\end{proof}

\begin{theorem}[Congruence for strong static location hp-bisimulation] \label{CSSB3}
We can enjoy the full congruence for strong static location hp-bisimulation as follows.
\begin{enumerate}
  \item If $A\overset{\text{def}}{=}P$, then $A\sim_{hp}^{sl} P$;
  \item Let $P_1\sim_{hp}^{sl} P_2$. Then
        \begin{enumerate}
           \item $loc::P_1\sim_{hp}^{sl}loc::P_2$;
           \item $\alpha.P_1\sim_{hp}^{sl} \alpha.P_2$;
           \item $(\alpha_1\parallel\cdots\parallel\alpha_n).P_1\sim_{hp}^{sl} (\alpha_1\parallel\cdots\parallel\alpha_n).P_2$;
           \item $P_1+Q\sim_{hp}^{sl} P_2 +Q$;
           \item $P_1\parallel Q\sim_{hp}^{sl} P_2\parallel Q$;
           \item $P_1\setminus L\sim_{hp}^{sl} P_2\setminus L$;
           \item $P_1[f]\sim_{hp}^{sl} P_2[f]$.
         \end{enumerate}
\end{enumerate}
\end{theorem}

\begin{proof}
\begin{enumerate}
  \item If $A\overset{\text{def}}{=}P$, then $A\sim_{hp}^{sl} P$. It is obvious.
  \item Let $P_1\sim_{hp}^{sl} P_2$. Then
        \begin{enumerate}
           \item $loc::P_1\sim_{hp}^{sl}loc::P_2$. It is sufficient to prove the relation $R=\{(loc::P_1, loc::P_2)\}\cup \textbf{Id}$ is a strong static location hp-bisimulation, we omit it;
           \item $\alpha.P_1\sim_{hp}^{sl} \alpha.P_2$. It is sufficient to prove the relation $R=\{(\alpha.P_1, \alpha.P_2)\}\cup \textbf{Id}$ is a strong static location hp-bisimulation for some distributions. It can be proved similarly to the proof of
           congruence for strong hp-bisimulation in CTC, we omit it;
           \item $(\alpha_1\parallel\cdots\parallel\alpha_n).P_1\sim_{hp}^{sl} (\alpha_1\parallel\cdots\parallel\alpha_n).P_2$. It is sufficient to prove the relation $R=\{((\alpha_1\parallel\cdots\parallel\alpha_n).P_1, (\alpha_1\parallel\cdots\parallel\alpha_n).P_2)\}\cup \textbf{Id}$ is a strong static location hp-bisimulation for some distributions. It can be proved similarly to the proof of
           congruence for strong hp-bisimulation in CTC, we omit it;
           \item $P_1+Q\sim_{hp}^{sl} P_2 +Q$. It is sufficient to prove the relation $R=\{(P_1+Q, P_2+Q)\}\cup \textbf{Id}$ is a strong static location hp-bisimulation for some distributions. It can be proved similarly to the proof of
           congruence for strong hp-bisimulation in CTC, we omit it;
           \item $P_1\parallel Q\sim_{hp}^{sl} P_2\parallel Q$. It is sufficient to prove the relation $R=\{(P_1\parallel Q, P_2\parallel Q)\}\cup \textbf{Id}$ is a strong static location hp-bisimulation for some distributions. It can be proved similarly to the proof of
           congruence for strong hp-bisimulation in CTC, we omit it;
           \item $P_1\setminus L\sim_{hp}^{sl} P_2\setminus L$. It is sufficient to prove the relation $R=\{(P_1\setminus L, P_2\setminus L)\}\cup \textbf{Id}$ is a strong static location hp-bisimulation for some distributions. It can be proved similarly to the proof of
           congruence for strong hp-bisimulation in CTC, we omit it;
           \item $P_1[f]\sim_{hp}^{sl} P_2[f]$. It is sufficient to prove the relation $R=\{(P_1[f], P_2[f])\}\cup \textbf{Id}$ is a strong static location hp-bisimulation for some distributions. It can be proved similarly to the proof of
           congruence for strong hp-bisimulation in CTC, we omit it.
         \end{enumerate}
\end{enumerate}
\end{proof}

\begin{theorem}[Congruence for strong static location hhp-bisimulation] \label{CSSB3}
We can enjoy the full congruence for strong static location hhp-bisimulation as follows.
\begin{enumerate}
  \item If $A\overset{\text{def}}{=}P$, then $A\sim_{hhp}^{sl} P$;
  \item Let $P_1\sim_{hhp}^{sl} P_2$. Then
        \begin{enumerate}
           \item $loc::P_1\sim_{hhp}^{sl}loc::P_2$;
           \item $\alpha.P_1\sim_{hhp}^{sl} \alpha.P_2$;
           \item $(\alpha_1\parallel\cdots\parallel\alpha_n).P_1\sim_{hhp}^{sl} (\alpha_1\parallel\cdots\parallel\alpha_n).P_2$;
           \item $P_1+Q\sim_{hhp}^{sl} P_2 +Q$;
           \item $P_1\parallel Q\sim_{hhp}^{sl} P_2\parallel Q$;
           \item $P_1\setminus L\sim_{hhp}^{sl} P_2\setminus L$;
           \item $P_1[f]\sim_{hhp}^{sl} P_2[f]$.
         \end{enumerate}
\end{enumerate}
\end{theorem}

\begin{proof}
\begin{enumerate}
  \item If $A\overset{\text{def}}{=}P$, then $A\sim_{hhp}^{sl} P$. It is obvious.
  \item Let $P_1\sim_{hhp}^{sl} P_2$. Then
        \begin{enumerate}
           \item $loc::P_1\sim_{hhp}^{sl}loc::P_2$. It is sufficient to prove the relation $R=\{(loc::P_1, loc::P_2)\}\cup \textbf{Id}$ is a strong static location hhp-bisimulation, we omit it;
           \item $\alpha.P_1\sim_{hhp}^{sl} \alpha.P_2$. It is sufficient to prove the relation $R=\{(\alpha.P_1, \alpha.P_2)\}\cup \textbf{Id}$ is a strong static location hhp-bisimulation for some distributions. It can be proved similarly to the proof of
           congruence for strong hhp-bisimulation in CTC, we omit it;
           \item $(\alpha_1\parallel\cdots\parallel\alpha_n).P_1\sim_{hhp}^{sl} (\alpha_1\parallel\cdots\parallel\alpha_n).P_2$. It is sufficient to prove the relation $R=\{((\alpha_1\parallel\cdots\parallel\alpha_n).P_1, (\alpha_1\parallel\cdots\parallel\alpha_n).P_2)\}\cup \textbf{Id}$ is a strong static location hhp-bisimulation for some distributions. It can be proved similarly to the proof of
           congruence for strong hhp-bisimulation in CTC, we omit it;
           \item $P_1+Q\sim_{hhp}^{sl} P_2 +Q$. It is sufficient to prove the relation $R=\{(P_1+Q, P_2+Q)\}\cup \textbf{Id}$ is a strong static location hhp-bisimulation for some distributions. It can be proved similarly to the proof of
           congruence for strong hhp-bisimulation in CTC, we omit it;
           \item $P_1\parallel Q\sim_{hhp}^{sl} P_2\parallel Q$. It is sufficient to prove the relation $R=\{(P_1\parallel Q, P_2\parallel Q)\}\cup \textbf{Id}$ is a strong static location hhp-bisimulation for some distributions. It can be proved similarly to the proof of
           congruence for strong hhp-bisimulation in CTC, we omit it;
           \item $P_1\setminus L\sim_{hhp}^{sl} P_2\setminus L$. It is sufficient to prove the relation $R=\{(P_1\setminus L, P_2\setminus L)\}\cup \textbf{Id}$ is a strong static location hhp-bisimulation for some distributions. It can be proved similarly to the proof of
           congruence for strong hhp-bisimulation in CTC, we omit it;
           \item $P_1[f]\sim_{hhp}^{sl} P_2[f]$. It is sufficient to prove the relation $R=\{(P_1[f], P_2[f])\}\cup \textbf{Id}$ is a strong static location hhp-bisimulation for some distributions. It can be proved similarly to the proof of
           congruence for strong hhp-bisimulation in CTC, we omit it.
         \end{enumerate}
\end{enumerate}
\end{proof}


\begin{definition}[Weakly guarded recursive expression]
$X$ is weakly guarded in $E$ if each occurrence of $X$ is with some subexpression $loc::\alpha.F$ or $(loc_1::\alpha_1\parallel\cdots\parallel loc_n::\alpha_n).F$ of $E$.
\end{definition}

\begin{lemma}\label{LUS3}
If the variables $\widetilde{X}$ are weakly guarded in $E$, and $E\{\widetilde{P}/\widetilde{X}\}\xrightarrow[u]{\{\alpha_1,\cdots,\alpha_n\}}P'$, then $P'$ takes the form
$E'\{\widetilde{P}/\widetilde{X}\}$ for some expression $E'$, and moreover, for any $\widetilde{Q}$,
$E\{\widetilde{Q}/\widetilde{X}\}\xrightarrow[u]{\{\alpha_1,\cdots,\alpha_n\}}E'\{\widetilde{Q}/\widetilde{X}\}$.
\end{lemma}

\begin{proof}
It needs to induct on the depth of the inference of $E\{\widetilde{P}/\widetilde{X}\}\xrightarrow[u]{\{\alpha_1,\cdots,\alpha_n\}}P'$. Note that, we consider the general distribution.

\begin{enumerate}
  \item Case $E\equiv Y$, a variable. Then $Y\notin \widetilde{X}$. Since $\widetilde{X}$ are weakly guarded, $Y\{\widetilde{P}/\widetilde{X}\equiv Y\}\nrightarrow$, this case is impossible.
  \item Case $E\equiv\beta.F$. Then we must have $\alpha=\beta$, and $P'\equiv F\{\widetilde{P}/\widetilde{X}\}$, and $E\{\widetilde{Q}/\widetilde{X}\}\equiv \beta.F\{\widetilde{Q}/\widetilde{X}\} \xrightarrow[v]{\beta}F\{\widetilde{Q}/\widetilde{X}\}$, then, let $E'$ be $F$, as desired.
  \item Case $E\equiv(\beta_1\parallel\cdots\parallel\beta_n).F$. Then we must have $\alpha_i=\beta_i$ for $1\leq i\leq n$, and $P'\equiv F\{\widetilde{P}/\widetilde{X}\}$, and $E\{\widetilde{Q}/\widetilde{X}\}\equiv (\beta_1\parallel\cdots\parallel\beta_n).F\{\widetilde{Q}/\widetilde{X}\} \xrightarrow[v]{\{\beta_1,\cdots,\beta_n\}}F\{\widetilde{Q}/\widetilde{X}\}$, then, let $E'$ be $F$, as desired.
  \item Case $E\equiv E_1+E_2$. Then either $E_1\{\widetilde{P}/\widetilde{X}\} \xrightarrow[u]{\{\alpha_1,\cdots,\alpha_n\}}P'$ or $E_2\{\widetilde{P}/\widetilde{X}\} \xrightarrow[u]{\{\alpha_1,\cdots,\alpha_n\}}P'$, then, we can apply this lemma in either case, as desired.
  \item Case $E\equiv E_1\parallel E_2$. There are four possibilities.
  \begin{enumerate}
    \item We may have $E_1\{\widetilde{P}/\widetilde{X}\} \xrightarrow[u]{\alpha}P_1'$ and $E_2\{\widetilde{P}/\widetilde{X}\}\nrightarrow$ with $P'\equiv P_1'\parallel (E_2\{\widetilde{P}/\widetilde{X}\})$, then by applying this lemma, $P_1'$ is of the form $E_1'\{\widetilde{P}/\widetilde{X}\}$, and for any $Q$, $E_1\{\widetilde{Q}/\widetilde{X}\}\xrightarrow[u]{\alpha} E_1'\{\widetilde{Q}/\widetilde{X}\}$. So, $P'$ is of the form $E_1'\parallel E_2\{\widetilde{P}/\widetilde{X}\}$, and for any $Q$, $E\{\widetilde{Q}/\widetilde{X}\}\equiv E_1\{\widetilde{Q}/\widetilde{X}\}\parallel E_2\{\widetilde{Q}/\widetilde{X}\}\xrightarrow[u]{\alpha} (E_1'\parallel E_2)\{\widetilde{Q}/\widetilde{X}\}$, then, let $E'$ be $E_1'\parallel E_2$, as desired.
    \item We may have $E_2\{\widetilde{P}/\widetilde{X}\} \xrightarrow[u]{\alpha}P_2'$ and $E_1\{\widetilde{P}/\widetilde{X}\}\nrightarrow$ with $P'\equiv P_2'\parallel (E_1\{\widetilde{P}/\widetilde{X}\})$, this case can be prove similarly to the above subcase, as desired.
    \item We may have $E_1\{\widetilde{P}/\widetilde{X}\} \xrightarrow[u]{\alpha}P_1'$ and $E_2\{\widetilde{P}/\widetilde{X}\}\xrightarrow[v]{\beta}P_2'$ with $\alpha\neq\overline{\beta}$ and $P'\equiv P_1'\parallel P_2'$, then by applying this lemma, $P_1'$ is of the form $E_1'\{\widetilde{P}/\widetilde{X}\}$, and for any $Q$, $E_1\{\widetilde{Q}/\widetilde{X}\}\xrightarrow[u]{\alpha} E_1'\{\widetilde{Q}/\widetilde{X}\}$; $P_2'$ is of the form $E_2'\{\widetilde{P}/\widetilde{X}\}$, and for any $Q$, $E_2\{\widetilde{Q}/\widetilde{X}\}\xrightarrow[u]{\alpha} E_2'\{\widetilde{Q}/\widetilde{X}\}$. So, $P'$ is of the form $E_1'\parallel E_2'\{\widetilde{P}/\widetilde{X}\}$, and for any $Q$, $E\{\widetilde{Q}/\widetilde{X}\}\equiv E_1\{\widetilde{Q}/\widetilde{X}\}\parallel E_2\{\widetilde{Q}/\widetilde{X}\}\xrightarrow[u\diamond v]{\{\alpha,\beta\}} (E_1'\parallel E_2')\{\widetilde{Q}/\widetilde{X}\}$, then, let $E'$ be $E_1'\parallel E_2'$, as desired.
    \item We may have $E_1\{\widetilde{P}/\widetilde{X}\} \xrightarrow[u]{l}P_1'$ and $E_2\{\widetilde{P}/\widetilde{X}\}\xrightarrow[v]{\overline{l}}P_2'$ with $P'\equiv P_1'\parallel P_2'$, then by applying this lemma, $P_1'$ is of the form $E_1'\{\widetilde{P}/\widetilde{X}\}$, and for any $Q$, $E_1\{\widetilde{Q}/\widetilde{X}\}\xrightarrow[u]{l} E_1'\{\widetilde{Q}/\widetilde{X}\}$; $P_2'$ is of the form $E_2'\{\widetilde{P}/\widetilde{X}\}$, and for any $Q$, $E_2\{\widetilde{Q}/\widetilde{X}\}\xrightarrow[v]{\overline{l}} E_2'\{\widetilde{Q}/\widetilde{X}\}$. So, $P'$ is of the form $E_1'\parallel E_2'\{\widetilde{P}/\widetilde{X}\}$, and for any $Q$, $E\{\widetilde{Q}/\widetilde{X}\}\equiv E_1\{\widetilde{Q}/\widetilde{X}\}\parallel E_2\{\widetilde{Q}/\widetilde{X}\}\xrightarrow[u\diamond v]{\tau} (E_1'\parallel E_2')\{\widetilde{Q}/\widetilde{X}\}$, then, let $E'$ be $E_1'\parallel E_2'$, as desired.
  \end{enumerate}
  \item Case $E\equiv F[R]$ and $E\equiv F\setminus L$. These cases can be prove similarly to the above case.
  \item Case $E\equiv C$, an agent constant defined by $C\overset{\text{def}}{=}R$. Then there is no $X\in\widetilde{X}$ occurring in $E$, so $C\xrightarrow[u]{\{\alpha_1,\cdots,\alpha_n\}}P'$, let $E'$ be $P'$, as desired.
\end{enumerate}
\end{proof}

\begin{theorem}[Unique solution of equations for strong static location pomset bisimulation]\label{USSSB3}
Let the recursive expressions $E_i(i\in I)$ contain at most the variables $X_i(i\in I)$, and let each $X_j(j\in I)$ be weakly guarded in each $E_i$. Then,

If $\widetilde{P}\sim_p^{sl} \widetilde{E}\{\widetilde{P}/\widetilde{X}\}$ and $\widetilde{Q}\sim_p^{sl} \widetilde{E}\{\widetilde{Q}/\widetilde{X}\}$, then $\widetilde{P}\sim_p^{sl} \widetilde{Q}$.
\end{theorem}

\begin{proof}
It is sufficient to induct on the depth of the inference of $E\{\widetilde{P}/\widetilde{X}\}\xrightarrow[u]{\{\alpha_1,\cdots,\alpha_n\}}P'$. Note that, we consider the general distribution.

\begin{enumerate}
  \item Case $E\equiv X_i$. Then we have $E\{\widetilde{P}/\widetilde{X}\}\equiv P_i\xrightarrow[u]{\{\alpha_1,\cdots,\alpha_n\}}P'$, since $P_i\sim_p^{sl} E_i\{\widetilde{P}/\widetilde{X}\}$, we have $E_i\{\widetilde{P}/\widetilde{X}\}\xrightarrow[u]{\{\alpha_1,\cdots,\alpha_n\}}P''\sim_p^{sl} P'$. Since $\widetilde{X}$ are weakly guarded in $E_i$, by Lemma \ref{LUS3}, $P''\equiv E'\{\widetilde{P}/\widetilde{X}\}$ and $E_i\{\widetilde{P}/\widetilde{X}\}\xrightarrow[u]{\{\alpha_1,\cdots,\alpha_n\}} E'\{\widetilde{P}/\widetilde{X}\}$. Since $E\{\widetilde{Q}/\widetilde{X}\}\equiv X_i\{\widetilde{Q}/\widetilde{X}\} \equiv Q_i\sim_p^{sl} E_i\{\widetilde{Q}/\widetilde{X}\}$, $E\{\widetilde{Q}/\widetilde{X}\}\xrightarrow[u]{\{\alpha_1,\cdots,\alpha_n\}}Q'\sim_p^{sl} E'\{\widetilde{Q}/\widetilde{X}\}$. So, $P'\sim_p^{sl} Q'$, as desired.
  \item Case $E\equiv\alpha.F$. This case can be proven similarly.
  \item Case $E\equiv(\alpha_1\parallel\cdots\parallel\alpha_n).F$. This case can be proven similarly.
  \item Case $E\equiv E_1+E_2$. We have $E_i\{\widetilde{P}/\widetilde{X}\} \xrightarrow[u]{\{\alpha_1,\cdots,\alpha_n\}}P'$, $E_i\{\widetilde{Q}/\widetilde{X}\} \xrightarrow[u]{\{\alpha_1,\cdots,\alpha_n\}}Q'$, then, $P'\sim_p^{sl} Q'$, as desired.
  \item Case $E\equiv E_1\parallel E_2$, $E\equiv F[R]$ and $E\equiv F\setminus L$, $E\equiv C$. These cases can be prove similarly to the above case.
\end{enumerate}
\end{proof}

\begin{theorem}[Unique solution of equations for strong static location step bisimulation]\label{USSSB3}
Let the recursive expressions $E_i(i\in I)$ contain at most the variables $X_i(i\in I)$, and let each $X_j(j\in I)$ be weakly guarded in each $E_i$. Then,

If $\widetilde{P}\sim_s^{sl} \widetilde{E}\{\widetilde{P}/\widetilde{X}\}$ and $\widetilde{Q}\sim_s^{sl} \widetilde{E}\{\widetilde{Q}/\widetilde{X}\}$, then $\widetilde{P}\sim_s^{sl} \widetilde{Q}$.
\end{theorem}

\begin{proof}
It is sufficient to induct on the depth of the inference of $E\{\widetilde{P}/\widetilde{X}\}\xrightarrow[u]{\{\alpha_1,\cdots,\alpha_n\}}P'$. Note that, we consider the general distribution.

\begin{enumerate}
  \item Case $E\equiv X_i$. Then we have $E\{\widetilde{P}/\widetilde{X}\}\equiv P_i\xrightarrow[u]{\{\alpha_1,\cdots,\alpha_n\}}P'$, since $P_i\sim_s^{sl} E_i\{\widetilde{P}/\widetilde{X}\}$, we have $E_i\{\widetilde{P}/\widetilde{X}\}\xrightarrow[u]{\{\alpha_1,\cdots,\alpha_n\}}P''\sim_s^{sl} P'$. Since $\widetilde{X}$ are weakly guarded in $E_i$, by Lemma \ref{LUS3}, $P''\equiv E'\{\widetilde{P}/\widetilde{X}\}$ and $E_i\{\widetilde{P}/\widetilde{X}\}\xrightarrow[u]{\{\alpha_1,\cdots,\alpha_n\}} E'\{\widetilde{P}/\widetilde{X}\}$. Since $E\{\widetilde{Q}/\widetilde{X}\}\equiv X_i\{\widetilde{Q}/\widetilde{X}\} \equiv Q_i\sim_s^{sl} E_i\{\widetilde{Q}/\widetilde{X}\}$, $E\{\widetilde{Q}/\widetilde{X}\}\xrightarrow[u]{\{\alpha_1,\cdots,\alpha_n\}}Q'\sim_s^{sl} E'\{\widetilde{Q}/\widetilde{X}\}$. So, $P'\sim_s^{sl} Q'$, as desired.
  \item Case $E\equiv\alpha.F$. This case can be proven similarly.
  \item Case $E\equiv(\alpha_1\parallel\cdots\parallel\alpha_n).F$. This case can be proven similarly.
  \item Case $E\equiv E_1+E_2$. We have $E_i\{\widetilde{P}/\widetilde{X}\} \xrightarrow[u]{\{\alpha_1,\cdots,\alpha_n\}}P'$, $E_i\{\widetilde{Q}/\widetilde{X}\} \xrightarrow[u]{\{\alpha_1,\cdots,\alpha_n\}}Q'$, then, $P'\sim_s^{sl} Q'$, as desired.
  \item Case $E\equiv E_1\parallel E_2$, $E\equiv F[R]$ and $E\equiv F\setminus L$, $E\equiv C$. These cases can be prove similarly to the above case.
\end{enumerate}
\end{proof}

\begin{theorem}[Unique solution of equations for strong static location hp-bisimulation]\label{USSSB3}
Let the recursive expressions $E_i(i\in I)$ contain at most the variables $X_i(i\in I)$, and let each $X_j(j\in I)$ be weakly guarded in each $E_i$. Then,

If $\widetilde{P}\sim_{hp}^{sl} \widetilde{E}\{\widetilde{P}/\widetilde{X}\}$ and $\widetilde{Q}\sim_{hp}^{sl} \widetilde{E}\{\widetilde{Q}/\widetilde{X}\}$, then $\widetilde{P}\sim_{hp}^{sl} \widetilde{Q}$.
\end{theorem}

\begin{proof}
It is sufficient to induct on the depth of the inference of $E\{\widetilde{P}/\widetilde{X}\}\xrightarrow[u]{\{\alpha_1,\cdots,\alpha_n\}}P'$. Note that, we consider the general distribution.

\begin{enumerate}
  \item Case $E\equiv X_i$. Then we have $E\{\widetilde{P}/\widetilde{X}\}\equiv P_i\xrightarrow[u]{\{\alpha_1,\cdots,\alpha_n\}}P'$, since $P_i\sim_{hp}^{sl} E_i\{\widetilde{P}/\widetilde{X}\}$, we have $E_i\{\widetilde{P}/\widetilde{X}\}\xrightarrow[u]{\{\alpha_1,\cdots,\alpha_n\}}P''\sim_{hp}^{sl} P'$. Since $\widetilde{X}$ are weakly guarded in $E_i$, by Lemma \ref{LUS3}, $P''\equiv E'\{\widetilde{P}/\widetilde{X}\}$ and $E_i\{\widetilde{P}/\widetilde{X}\}\xrightarrow[u]{\{\alpha_1,\cdots,\alpha_n\}} E'\{\widetilde{P}/\widetilde{X}\}$. Since $E\{\widetilde{Q}/\widetilde{X}\}\equiv X_i\{\widetilde{Q}/\widetilde{X}\} \equiv Q_i\sim_{hp}^{sl} E_i\{\widetilde{Q}/\widetilde{X}\}$, $E\{\widetilde{Q}/\widetilde{X}\}\xrightarrow[u]{\{\alpha_1,\cdots,\alpha_n\}}Q'\sim_{hp}^{sl} E'\{\widetilde{Q}/\widetilde{X}\}$. So, $P'\sim_{hp}^{sl} Q'$, as desired.
  \item Case $E\equiv\alpha.F$. This case can be proven similarly.
  \item Case $E\equiv(\alpha_1\parallel\cdots\parallel\alpha_n).F$. This case can be proven similarly.
  \item Case $E\equiv E_1+E_2$. We have $E_i\{\widetilde{P}/\widetilde{X}\} \xrightarrow[u]{\{\alpha_1,\cdots,\alpha_n\}}P'$, $E_i\{\widetilde{Q}/\widetilde{X}\} \xrightarrow[u]{\{\alpha_1,\cdots,\alpha_n\}}Q'$, then, $P'\sim_{hp}^{sl} Q'$, as desired.
  \item Case $E\equiv E_1\parallel E_2$, $E\equiv F[R]$ and $E\equiv F\setminus L$, $E\equiv C$. These cases can be prove similarly to the above case.
\end{enumerate}
\end{proof}

\begin{theorem}[Unique solution of equations for strong static location hhp-bisimulation]\label{USSSB3}
Let the recursive expressions $E_i(i\in I)$ contain at most the variables $X_i(i\in I)$, and let each $X_j(j\in I)$ be weakly guarded in each $E_i$. Then,

If $\widetilde{P}\sim_{hhp}^{sl} \widetilde{E}\{\widetilde{P}/\widetilde{X}\}$ and $\widetilde{Q}\sim_{hhp}^{sl} \widetilde{E}\{\widetilde{Q}/\widetilde{X}\}$, then $\widetilde{P}\sim_{hhp}^{sl} \widetilde{Q}$.
\end{theorem}

\begin{proof}
It is sufficient to induct on the depth of the inference of $E\{\widetilde{P}/\widetilde{X}\}\xrightarrow[u]{\{\alpha_1,\cdots,\alpha_n\}}P'$. Note that, we consider the general distribution.

\begin{enumerate}
  \item Case $E\equiv X_i$. Then we have $E\{\widetilde{P}/\widetilde{X}\}\equiv P_i\xrightarrow[u]{\{\alpha_1,\cdots,\alpha_n\}}P'$, since
  $P_i\sim_{hhp}^{sl} E_i\{\widetilde{P}/\widetilde{X}\}$, we have $E_i\{\widetilde{P}/\widetilde{X}\}\xrightarrow[u]{\{\alpha_1,\cdots,\alpha_n\}}P''\sim_{hhp}^{sl} P'$. Since
  $\widetilde{X}$ are weakly guarded in $E_i$, by Lemma \ref{LUS3}, $P''\equiv E'\{\widetilde{P}/\widetilde{X}\}$ and
  $E_i\{\widetilde{P}/\widetilde{X}\}\xrightarrow[u]{\{\alpha_1,\cdots,\alpha_n\}} E'\{\widetilde{P}/\widetilde{X}\}$. Since
  $E\{\widetilde{Q}/\widetilde{X}\}\equiv X_i\{\widetilde{Q}/\widetilde{X}\} \equiv Q_i\sim_{hhp}^{sl} E_i\{\widetilde{Q}/\widetilde{X}\}$,
  $E\{\widetilde{Q}/\widetilde{X}\}\xrightarrow[u]{\{\alpha_1,\cdots,\alpha_n\}}Q'\sim_{hhp}^{sl} E'\{\widetilde{Q}/\widetilde{X}\}$. So, $P'\sim_{hhp}^{sl} Q'$, as desired.
  \item Case $E\equiv\alpha.F$. This case can be proven similarly.
  \item Case $E\equiv(\alpha_1\parallel\cdots\parallel\alpha_n).F$. This case can be proven similarly.
  \item Case $E\equiv E_1+E_2$. We have $E_i\{\widetilde{P}/\widetilde{X}\} \xrightarrow[u]{\{\alpha_1,\cdots,\alpha_n\}}P'$,
  $E_i\{\widetilde{Q}/\widetilde{X}\} \xrightarrow[u]{\{\alpha_1,\cdots,\alpha_n\}}Q'$, then, $P'\sim_{hhp}^{sl} Q'$, as desired.
  \item Case $E\equiv E_1\parallel E_2$, $E\equiv F[R]$ and $E\equiv F\setminus L$, $E\equiv C$. These cases can be prove similarly to the above case.
\end{enumerate}
\end{proof}

\subsubsection{Weak Bisimulations}

The weak transition rules for CTC with static localities are listed in Table \ref{WTRForCTC3}.

\begin{center}
    \begin{table}
        \[\textbf{Act}_1\quad \frac{}{\alpha.P\xRightarrow[\epsilon]{\alpha}P}\]

        \[\textbf{Loc}\quad \frac{P\xRightarrow[u]{\alpha}P'}{loc::P\xRightarrow[loc\ll u]{\alpha}loc::P'}\]

        \[\textbf{Sum}_1\quad \frac{P\xRightarrow[u]{\alpha}P'}{P+Q\xRightarrow[u]{\alpha}P'}\]

        \[\textbf{Com}_1\quad \frac{P\xRightarrow[u]{\alpha}P'\quad Q\nrightarrow}{P\parallel Q\xRightarrow[u]{\alpha}P'\parallel Q}\]

        \[\textbf{Com}_2\quad \frac{Q\xRightarrow[u]{\alpha}Q'\quad P\nrightarrow}{P\parallel Q\xRightarrow[u]{\alpha}P\parallel Q'}\]

        \[\textbf{Com}_3\quad \frac{P\xRightarrow[u]{\alpha}P'\quad Q\xRightarrow[v]{\beta}Q'}{P\parallel Q\xRightarrow[u\diamond v]{\{\alpha,\beta\}}P'\parallel Q'}\quad (\beta\neq\overline{\alpha})\]

        \[\textbf{Com}_4\quad \frac{P\xRightarrow[u]{l}P'\quad Q\xRightarrow[v]{\overline{l}}Q'}{P\parallel Q\xRightarrow{\tau}P'\parallel Q'}\]

        \[\textbf{Act}_2\quad \frac{}{(\alpha_1\parallel\cdots\parallel\alpha_n).P\xRightarrow[\epsilon]{\{\alpha_1,\cdots,\alpha_n\}}P}\quad (\alpha_i\neq\overline{\alpha_j}\quad i,j\in\{1,\cdots,n\})\]

        \[\textbf{Sum}_2\quad \frac{P\xRightarrow[u]{\{\alpha_1,\cdots,\alpha_n\}}P'}{P+Q\xRightarrow[u]{\{\alpha_1,\cdots,\alpha_n\}}P'}\]
%
%
%
%
%
%
        \caption{Weak transition rules of CTC with static localities}
        \label{WTRForCTC3}
    \end{table}
\end{center}

\begin{center}
    \begin{table}
%
%
%
%
%
%
%
%
%
        \[\textbf{Res}_1\quad \frac{P\xRightarrow[u]{\alpha}P'}{P\setminus L\xRightarrow[u]{\alpha}P'\setminus L}\quad (\alpha,\overline{\alpha}\notin L)\]

        \[\textbf{Res}_2\quad \frac{P\xRightarrow[u]{\{\alpha_1,\cdots,\alpha_n\}}P'}{P\setminus L\xRightarrow[u]{\{\alpha_1,\cdots,\alpha_n\}}P'\setminus L}\quad (\alpha_1,\overline{\alpha_1},\cdots,\alpha_n,\overline{\alpha_n}\notin L)\]

        \[\textbf{Rel}_1\quad \frac{P\xRightarrow[u]{\alpha}P'}{P[f]\xRightarrow[u]{f(\alpha)}P'[f]}\]

        \[\textbf{Rel}_2\quad \frac{P\xRightarrow[u]{\{\alpha_1,\cdots,\alpha_n\}}P'}{P[f]\xRightarrow[u]{\{f(\alpha_1),\cdots,f(\alpha_n)\}}P'[f]}\]

        \[\textbf{Con}_1\quad\frac{P\xRightarrow[u]{\alpha}P'}{A\xRightarrow[u]{\alpha}P'}\quad (A\overset{\text{def}}{=}P)\]

        \[\textbf{Con}_2\quad\frac{P\xRightarrow[u]{\{\alpha_1,\cdots,\alpha_n\}}P'}{A\xRightarrow[u]{\{\alpha_1,\cdots,\alpha_n\}}P'}\quad (A\overset{\text{def}}{=}P)\]
        \caption{Weak transition rules of CTC with static localities (continuing)}
        \label{WTRForCTC32}
    \end{table}
\end{center}


\begin{proposition}[$\tau$ laws for weak static location pomset bisimulation]\label{TAUWSB3}
The $\tau$ laws for weak static location pomset bisimulation is as follows.
\begin{enumerate}
  \item $P\approx_p^{sl} \tau.P$;
  \item $\alpha.\tau.P\approx_p^{sl} \alpha.P$;
  \item $(\alpha_1\parallel\cdots\parallel\alpha_n).\tau.P\approx_p^{sl} (\alpha_1\parallel\cdots\parallel\alpha_n).P$;
  \item $P+\tau.P\approx_p^{sl} \tau.P$;
  \item $\alpha.(P+\tau.Q)+\alpha.Q\approx_p^{sl}\alpha.(P+\tau.Q)$;
  \item $(\alpha_1\parallel\cdots\parallel\alpha_n).(P+\tau.Q)+ (\alpha_1\parallel\cdots\parallel\alpha_n).Q\approx_p^{sl} (\alpha_1\parallel\cdots\parallel\alpha_n).(P+\tau.Q)$;
  \item $P\approx_p^{sl} \tau\parallel P$.
\end{enumerate}
\end{proposition}

\begin{proof}
\begin{enumerate}
  \item $P\approx_p^{sl} \tau.P$. It is sufficient to prove the relation $R=\{(P, \tau.P)\}\cup \textbf{Id}$ is a weak static location pomset bisimulation for some distributions. It can be proved similarly to the proof of
  $\tau$-laws for weak pomset bisimulation in CTC, we omit it;
  \item $\alpha.\tau.P\approx_p^{sl} \alpha.P$. It is sufficient to prove the relation $R=\{(\alpha.\tau.P, \alpha.P)\}\cup \textbf{Id}$ is a weak static location pomset bisimulation for some distributions. It can be proved similarly to the proof of
  $\tau$-laws for weak pomset bisimulation in CTC, we omit it;
  \item $(\alpha_1\parallel\cdots\parallel\alpha_n).\tau.P\approx_p^{sl} (\alpha_1\parallel\cdots\parallel\alpha_n).P$. It is sufficient to prove the relation $R=\{((\alpha_1\parallel\cdots\parallel\alpha_n).\tau.P, (\alpha_1\parallel\cdots\parallel\alpha_n).P)\}\cup \textbf{Id}$ is a weak static location pomset bisimulation for some distributions. It can be proved similarly to the proof of
  $\tau$-laws for weak pomset bisimulation in CTC, we omit it;
  \item $P+\tau.P\approx_p^{sl} \tau.P$. It is sufficient to prove the relation $R=\{(P+\tau.P, \tau.P)\}\cup \textbf{Id}$ is a weak static location pomset bisimulation for some distributions. It can be proved similarly to the proof of
  $\tau$-laws for weak pomset bisimulation in CTC, we omit it;
  \item $\alpha.(P+\tau.Q)+\alpha.Q\approx_p^{sl}\alpha.(P+\tau.Q)$. It is sufficient to prove the relation $R=\{(\alpha.(P+\tau.Q)+\alpha.Q, \alpha.(P+\tau.Q))\}\cup \textbf{Id}$ is a weak static location pomset bisimulation for some distributions. It can be proved similarly to the proof of
  $\tau$-laws for weak pomset bisimulation in CTC, we omit it;
  \item $(\alpha_1\parallel\cdots\parallel\alpha_n).(P+\tau.Q)+ (\alpha_1\parallel\cdots\parallel\alpha_n).Q\approx_p^{sl} (\alpha_1\parallel\cdots\parallel\alpha_n).(P+\tau.Q)$. It is sufficient to prove the relation $R=\{((\alpha_1\parallel\cdots\parallel\alpha_n).(P+\tau.Q)+ (\alpha_1\parallel\cdots\parallel\alpha_n).Q, (\alpha_1\parallel\cdots\parallel\alpha_n).(P+\tau.Q))\}\cup \textbf{Id}$ is a weak static location pomset bisimulation for some distributions. It can be proved similarly to the proof of
  $\tau$-laws for weak pomset bisimulation in CTC, we omit it;
  \item $P\approx_p^{sl} \tau\parallel P$. It is sufficient to prove the relation $R=\{(P, \tau\parallel P)\}\cup \textbf{Id}$ is a weak static location pomset bisimulation for some distributions. It can be proved similarly to the proof of
  $\tau$-laws for weak pomset bisimulation in CTC, we omit it.
\end{enumerate}
\end{proof}

\begin{proposition}[$\tau$ laws for weak static location step bisimulation]\label{TAUWSB3}
The $\tau$ laws for weak static location step bisimulation is as follows.
\begin{enumerate}
  \item $P\approx_s^{sl} \tau.P$;
  \item $\alpha.\tau.P\approx_s^{sl} \alpha.P$;
  \item $(\alpha_1\parallel\cdots\parallel\alpha_n).\tau.P\approx_s^{sl} (\alpha_1\parallel\cdots\parallel\alpha_n).P$;
  \item $P+\tau.P\approx_s^{sl} \tau.P$;
  \item $\alpha.(P+\tau.Q)+\alpha.Q\approx_s^{sl}\alpha.(P+\tau.Q)$;
  \item $(\alpha_1\parallel\cdots\parallel\alpha_n).(P+\tau.Q)+ (\alpha_1\parallel\cdots\parallel\alpha_n).Q\approx_s^{sl} (\alpha_1\parallel\cdots\parallel\alpha_n).(P+\tau.Q)$;
  \item $P\approx_s^{sl} \tau\parallel P$.
\end{enumerate}
\end{proposition}

\begin{proof}
\begin{enumerate}
  \item $P\approx_s^{sl} \tau.P$. It is sufficient to prove the relation $R=\{(P, \tau.P)\}\cup \textbf{Id}$ is a weak static location step bisimulation for some distributions. It can be proved similarly to the proof of
  $\tau$-laws for weak step bisimulation in CTC, we omit it;
  \item $\alpha.\tau.P\approx_s^{sl} \alpha.P$. It is sufficient to prove the relation $R=\{(\alpha.\tau.P, \alpha.P)\}\cup \textbf{Id}$ is a weak static location step bisimulation for some distributions. It can be proved similarly to the proof of
  $\tau$-laws for weak step bisimulation in CTC, we omit it;
  \item $(\alpha_1\parallel\cdots\parallel\alpha_n).\tau.P\approx_s^{sl} (\alpha_1\parallel\cdots\parallel\alpha_n).P$. It is sufficient to prove the relation $R=\{((\alpha_1\parallel\cdots\parallel\alpha_n).\tau.P, (\alpha_1\parallel\cdots\parallel\alpha_n).P)\}\cup \textbf{Id}$ is a weak static location step bisimulation for some distributions. It can be proved similarly to the proof of
  $\tau$-laws for weak step bisimulation in CTC, we omit it;
  \item $P+\tau.P\approx_s^{sl} \tau.P$. It is sufficient to prove the relation $R=\{(P+\tau.P, \tau.P)\}\cup \textbf{Id}$ is a weak static location step bisimulation for some distributions. It can be proved similarly to the proof of
  $\tau$-laws for weak step bisimulation in CTC, we omit it;
  \item $\alpha.(P+\tau.Q)+\alpha.Q\approx_s^{sl}\alpha.(P+\tau.Q)$. It is sufficient to prove the relation $R=\{(\alpha.(P+\tau.Q)+\alpha.Q, \alpha.(P+\tau.Q))\}\cup \textbf{Id}$ is a weak static location step bisimulation for some distributions. It can be proved similarly to the proof of
  $\tau$-laws for weak step bisimulation in CTC, we omit it;
  \item $(\alpha_1\parallel\cdots\parallel\alpha_n).(P+\tau.Q)+ (\alpha_1\parallel\cdots\parallel\alpha_n).Q\approx_s^{sl} (\alpha_1\parallel\cdots\parallel\alpha_n).(P+\tau.Q)$. It is sufficient to prove the relation $R=\{((\alpha_1\parallel\cdots\parallel\alpha_n).(P+\tau.Q)+ (\alpha_1\parallel\cdots\parallel\alpha_n).Q, (\alpha_1\parallel\cdots\parallel\alpha_n).(P+\tau.Q))\}\cup \textbf{Id}$ is a weak static location step bisimulation for some distributions. It can be proved similarly to the proof of
  $\tau$-laws for weak step bisimulation in CTC, we omit it;
  \item $P\approx_s^{sl} \tau\parallel P$. It is sufficient to prove the relation $R=\{(P, \tau\parallel P)\}\cup \textbf{Id}$ is a weak static location step bisimulation for some distributions. It can be proved similarly to the proof of
  $\tau$-laws for weak step bisimulation in CTC, we omit it.
\end{enumerate}
\end{proof}

\begin{proposition}[$\tau$ laws for weak static location hp-bisimulation]\label{TAUWSB3}
The $\tau$ laws for weak static location hp-bisimulation is as follows.
\begin{enumerate}
  \item $P\approx_{hp}^{sl} \tau.P$;
  \item $\alpha.\tau.P\approx_{hp}^{sl} \alpha.P$;
  \item $(\alpha_1\parallel\cdots\parallel\alpha_n).\tau.P\approx_{hp}^{sl} (\alpha_1\parallel\cdots\parallel\alpha_n).P$;
  \item $P+\tau.P\approx_{hp}^{sl} \tau.P$;
  \item $\alpha.(P+\tau.Q)+\alpha.Q\approx_{hp}^{sl}\alpha.(P+\tau.Q)$;
  \item $(\alpha_1\parallel\cdots\parallel\alpha_n).(P+\tau.Q)+ (\alpha_1\parallel\cdots\parallel\alpha_n).Q\approx_{hp}^{sl} (\alpha_1\parallel\cdots\parallel\alpha_n).(P+\tau.Q)$;
  \item $P\approx_{hp}^{sl} \tau\parallel P$.
\end{enumerate}
\end{proposition}

\begin{proof}
\begin{enumerate}
  \item $P\approx_{hp}^{sl} \tau.P$. It is sufficient to prove the relation $R=\{(P, \tau.P)\}\cup \textbf{Id}$ is a weak static location hp-bisimulation for some distributions. It can be proved similarly to the proof of
  $\tau$-laws for weak hp-bisimulation in CTC, we omit it;
  \item $\alpha.\tau.P\approx_{hp}^{sl} \alpha.P$. It is sufficient to prove the relation $R=\{(\alpha.\tau.P, \alpha.P)\}\cup \textbf{Id}$ is a weak static location hp-bisimulation for some distributions. It can be proved similarly to the proof of
  $\tau$-laws for weak hp-bisimulation in CTC, we omit it;
  \item $(\alpha_1\parallel\cdots\parallel\alpha_n).\tau.P\approx_{hp}^{sl} (\alpha_1\parallel\cdots\parallel\alpha_n).P$. It is sufficient to prove the relation $R=\{((\alpha_1\parallel\cdots\parallel\alpha_n).\tau.P, (\alpha_1\parallel\cdots\parallel\alpha_n).P)\}\cup \textbf{Id}$ is a weak static location hp-bisimulation for some distributions. It can be proved similarly to the proof of
  $\tau$-laws for weak hp-bisimulation in CTC, we omit it;
  \item $P+\tau.P\approx_{hp}^{sl} \tau.P$. It is sufficient to prove the relation $R=\{(P+\tau.P, \tau.P)\}\cup \textbf{Id}$ is a weak static location hp-bisimulation for some distributions. It can be proved similarly to the proof of
  $\tau$-laws for weak hp-bisimulation in CTC, we omit it;
  \item $\alpha.(P+\tau.Q)+\alpha.Q\approx_{hp}^{sl}\alpha.(P+\tau.Q)$. It is sufficient to prove the relation $R=\{(\alpha.(P+\tau.Q)+\alpha.Q, \alpha.(P+\tau.Q))\}\cup \textbf{Id}$ is a weak static location hp-bisimulation for some distributions. It can be proved similarly to the proof of
  $\tau$-laws for weak hp-bisimulation in CTC, we omit it;
  \item $(\alpha_1\parallel\cdots\parallel\alpha_n).(P+\tau.Q)+ (\alpha_1\parallel\cdots\parallel\alpha_n).Q\approx_{hp}^{sl} (\alpha_1\parallel\cdots\parallel\alpha_n).(P+\tau.Q)$. It is sufficient to prove the relation $R=\{((\alpha_1\parallel\cdots\parallel\alpha_n).(P+\tau.Q)+ (\alpha_1\parallel\cdots\parallel\alpha_n).Q, (\alpha_1\parallel\cdots\parallel\alpha_n).(P+\tau.Q))\}\cup \textbf{Id}$ is a weak static location hp-bisimulation for some distributions. It can be proved similarly to the proof of
  $\tau$-laws for weak hp-bisimulation in CTC, we omit it;
  \item $P\approx_{hp}^{sl} \tau\parallel P$. It is sufficient to prove the relation $R=\{(P, \tau\parallel P)\}\cup \textbf{Id}$ is a weak static location hp-bisimulation for some distributions. It can be proved similarly to the proof of
  $\tau$-laws for weak hp-bisimulation in CTC, we omit it.
\end{enumerate}
\end{proof}

\begin{proposition}[$\tau$ laws for weak static location hhp-bisimulation]\label{TAUWSB3}
The $\tau$ laws for weak static location hhp-bisimulation is as follows.
\begin{enumerate}
  \item $P\approx_{hhp}^{sl} \tau.P$;
  \item $\alpha.\tau.P\approx_{hhp}^{sl} \alpha.P$;
  \item $(\alpha_1\parallel\cdots\parallel\alpha_n).\tau.P\approx_{hhp}^{sl} (\alpha_1\parallel\cdots\parallel\alpha_n).P$;
  \item $P+\tau.P\approx_{hhp}^{sl} \tau.P$;
  \item $\alpha.(P+\tau.Q)+\alpha.Q\approx_{hhp}^{sl}\alpha.(P+\tau.Q)$;
  \item $(\alpha_1\parallel\cdots\parallel\alpha_n).(P+\tau.Q)+ (\alpha_1\parallel\cdots\parallel\alpha_n).Q\approx_{hhp}^{sl} (\alpha_1\parallel\cdots\parallel\alpha_n).(P+\tau.Q)$;
  \item $P\approx_{hhp}^{sl} \tau\parallel P$.
\end{enumerate}
\end{proposition}

\begin{proof}
\begin{enumerate}
  \item $P\approx_{hhp}^{sl} \tau.P$. It is sufficient to prove the relation $R=\{(P, \tau.P)\}\cup \textbf{Id}$ is a weak static location hhp-bisimulation for some distributions. It can be proved similarly to the proof of
  $\tau$-laws for weak hhp-bisimulation in CTC, we omit it;
  \item $\alpha.\tau.P\approx_{hhp}^{sl} \alpha.P$. It is sufficient to prove the relation $R=\{(\alpha.\tau.P, \alpha.P)\}\cup \textbf{Id}$ is a weak static location hhp-bisimulation for some distributions. It can be proved similarly to the proof of
  $\tau$-laws for weak hhp-bisimulation in CTC, we omit it;
  \item $(\alpha_1\parallel\cdots\parallel\alpha_n).\tau.P\approx_{hhp}^{sl} (\alpha_1\parallel\cdots\parallel\alpha_n).P$. It is sufficient to prove the relation $R=\{((\alpha_1\parallel\cdots\parallel\alpha_n).\tau.P, (\alpha_1\parallel\cdots\parallel\alpha_n).P)\}\cup \textbf{Id}$ is a weak static location hhp-bisimulation for some distributions. It can be proved similarly to the proof of
  $\tau$-laws for weak hhp-bisimulation in CTC, we omit it;
  \item $P+\tau.P\approx_{hhp}^{sl} \tau.P$. It is sufficient to prove the relation $R=\{(P+\tau.P, \tau.P)\}\cup \textbf{Id}$ is a weak static location hhp-bisimulation for some distributions. It can be proved similarly to the proof of
  $\tau$-laws for weak hhp-bisimulation in CTC, we omit it;
  \item $\alpha.(P+\tau.Q)+\alpha.Q\approx_{hhp}^{sl}\alpha.(P+\tau.Q)$. It is sufficient to prove the relation $R=\{(\alpha.(P+\tau.Q)+\alpha.Q, \alpha.(P+\tau.Q))\}\cup \textbf{Id}$ is a weak static location hhp-bisimulation for some distributions. It can be proved similarly to the proof of
  $\tau$-laws for weak hhp-bisimulation in CTC, we omit it;
  \item $(\alpha_1\parallel\cdots\parallel\alpha_n).(P+\tau.Q)+ (\alpha_1\parallel\cdots\parallel\alpha_n).Q\approx_{hhp}^{sl} (\alpha_1\parallel\cdots\parallel\alpha_n).(P+\tau.Q)$. It is sufficient to prove the relation $R=\{((\alpha_1\parallel\cdots\parallel\alpha_n).(P+\tau.Q)+ (\alpha_1\parallel\cdots\parallel\alpha_n).Q, (\alpha_1\parallel\cdots\parallel\alpha_n).(P+\tau.Q))\}\cup \textbf{Id}$ is a weak static location hhp-bisimulation for some distributions. It can be proved similarly to the proof of
  $\tau$-laws for weak hhp-bisimulation in CTC, we omit it;
  \item $P\approx_{hhp}^{sl} \tau\parallel P$. It is sufficient to prove the relation $R=\{(P, \tau\parallel P)\}\cup \textbf{Id}$ is a weak static location hhp-bisimulation for some distributions. It can be proved similarly to the proof of
  $\tau$-laws for weak hhp-bisimulation in CTC, we omit it.
\end{enumerate}
\end{proof}

\begin{definition}[Sequential]
$X$ is sequential in $E$ if every subexpression of $E$ which contains $X$, apart from $X$ itself, is of the form $loc::\alpha.F$, or $(loc_1::\alpha_1\parallel\cdots\parallel loc_n::\alpha_n).F$, or $\sum\widetilde{F}$.
\end{definition}

\begin{definition}[Guarded recursive expression]
$X$ is guarded in $E$ if each occurrence of $X$ is with some subexpression $loc::l.F$ or $(loc_1::l_1\parallel\cdots\parallel loc_n::l_n).F$ of $E$.
\end{definition}

\begin{lemma}\label{LUSWW3}
Let $G$ be guarded and sequential, $Vars(G)\subseteq\widetilde{X}$, and let $G\{\widetilde{P}/\widetilde{X}\}\xrightarrow[u]{\{\alpha_1,\cdots,\alpha_n\}}P'$. Then there is an expression $H$ such that $G\xrightarrow[u]{\{\alpha_1,\cdots,\alpha_n\}}H$, $P'\equiv H\{\widetilde{P}/\widetilde{X}\}$, and for any $\widetilde{Q}$, $G\{\widetilde{Q}/\widetilde{X}\}\xrightarrow[u]{\{\alpha_1,\cdots,\alpha_n\}} H\{\widetilde{Q}/\widetilde{X}\}$. Moreover $H$ is sequential, $Vars(H)\subseteq\widetilde{X}$, and if $\alpha_1=\cdots=\alpha_n=\tau$, then $H$ is also guarded.
\end{lemma}

\begin{proof}
We need to induct on the structure of $G$. Note that, we consider the general distribution.

If $G$ is a Constant, a Composition, a Restriction or a Relabeling then it contains no variables, since $G$ is sequential and guarded, then $G\xrightarrow[u]{\{\alpha_1,\cdots,\alpha_n\}}P'$, then let $H\equiv P'$, as desired.

$G$ cannot be a variable, since it is guarded.

If $G\equiv G_1+G_2$. Then either $G_1\{\widetilde{P}/\widetilde{X}\} \xrightarrow[u]{\{\alpha_1,\cdots,\alpha_n\}}P'$ or $G_2\{\widetilde{P}/\widetilde{X}\} \xrightarrow[u]{\{\alpha_1,\cdots,\alpha_n\}}P'$, then, we can apply this lemma in either case, as desired.

If $G\equiv\beta.H$. Then we must have $\alpha=\beta$, and $P'\equiv H\{\widetilde{P}/\widetilde{X}\}$, and $G\{\widetilde{Q}/\widetilde{X}\}\equiv \beta.H\{\widetilde{Q}/\widetilde{X}\} \xrightarrow[v]{\beta}H\{\widetilde{Q}/\widetilde{X}\}$, then, let $G'$ be $H$, as desired.

If $G\equiv(\beta_1\parallel\cdots\parallel\beta_n).H$. Then we must have $\alpha_i=\beta_i$ for $1\leq i\leq n$, and $P'\equiv H\{\widetilde{P}/\widetilde{X}\}$, and $G\{\widetilde{Q}/\widetilde{X}\}\equiv (\beta_1\parallel\cdots\parallel\beta_n).H\{\widetilde{Q}/\widetilde{X}\} \xrightarrow[v]{\{\beta_1,\cdots,\beta_n\}}H\{\widetilde{Q}/\widetilde{X}\}$, then, let $G'$ be $H$, as desired.

If $G\equiv\tau.H$. Then we must have $\tau=\tau$, and $P'\equiv H\{\widetilde{P}/\widetilde{X}\}$, and $G\{\widetilde{Q}/\widetilde{X}\}\equiv \tau.H\{\widetilde{Q}/\widetilde{X}\} \xrightarrow{\tau}H\{\widetilde{Q}/\widetilde{X}\}$, then, let $G'$ be $H$, as desired.
\end{proof}

\begin{theorem}[Unique solution of equations for weak static location pomset bisimulation]
Let the guarded and sequential expressions $\widetilde{E}$ contain free variables $\subseteq \widetilde{X}$, then,

If $\widetilde{P}\approx_p^{sl} \widetilde{E}\{\widetilde{P}/\widetilde{X}\}$ and $\widetilde{Q}\approx_p^{sl} \widetilde{E}\{\widetilde{Q}/\widetilde{X}\}$, then $\widetilde{P}\approx_p^{sl} \widetilde{Q}$.
\end{theorem}

\begin{proof}
Like the corresponding theorem in CCS, without loss of generality, we only consider a single equation $X=E$. So we assume $P\approx_p^{sl} E(P)$, $Q\approx_p^{sl} E(Q)$, then $P\approx_p^{sl} Q$. Note that, we consider the general distribution.

We will prove $\{(H(P),H(Q)): H\}$ sequential, if $H(P)\xrightarrow[u]{\{\alpha_1,\cdots,\alpha_n\}}P'$, then, for some $Q'$, $H(Q)\xRightarrow[u]{\{\alpha_1.\cdots,\alpha_n\}}Q'$ and $P'\approx_p^{sl} Q'$.

Let $H(P)\xrightarrow[u]{\{\alpha_1,\cdots,\alpha_n\}}P'$, then $H(E(P))\xRightarrow[u]{\{\alpha_1,\cdots,\alpha_n\}}P''$ and $P'\approx_p^{sl} P''$.

By Lemma \ref{LUSWW3}, we know there is a sequential $H'$ such that $H(E(P))\xRightarrow[u]{\{\alpha_1,\cdots,\alpha_n\}}H'(P)\Rightarrow P''\approx_p^{sl} P'$.

And, $H(E(Q))\xRightarrow[u]{\{\alpha_1,\cdots,\alpha_n\}}H'(Q)\Rightarrow Q''$ and $P''\approx_p^{sl} Q''$. And $H(Q)\xrightarrow[u]{\{\alpha_1,\cdots,\alpha_n\}}Q'\approx_p^{sl} Q''$. Hence, $P'\approx_p^{sl} Q'$, as desired.
\end{proof}

\begin{theorem}[Unique solution of equations for weak static location step bisimulation]\label{USWSB3}
Let the guarded and sequential expressions $\widetilde{E}$ contain free variables $\subseteq \widetilde{X}$, then,

If $\widetilde{P}\approx_s^{sl} \widetilde{E}\{\widetilde{P}/\widetilde{X}\}$ and $\widetilde{Q}\approx_s^{sl} \widetilde{E}\{\widetilde{Q}/\widetilde{X}\}$, then $\widetilde{P}\approx_s^{sl} \widetilde{Q}$.
\end{theorem}

\begin{proof}
Like the corresponding theorem in CCS, without loss of generality, we only consider a single equation $X=E$. So we assume $P\approx_s^{sl} E(P)$, $Q\approx_s^{sl} E(Q)$, then $P\approx_s^{sl} Q$. Note that, we consider the general distribution.

We will prove $\{(H(P),H(Q)): H\}$ sequential, if $H(P)\xrightarrow[u]{\{\alpha_1,\cdots,\alpha_n\}}P'$, then, for some $Q'$, $H(Q)\xRightarrow[u]{\{\alpha_1.\cdots,\alpha_n\}}Q'$ and $P'\approx_s^{sl} Q'$.

Let $H(P)\xrightarrow[u]{\{\alpha_1,\cdots,\alpha_n\}}P'$, then $H(E(P))\xRightarrow[u]{\{\alpha_1,\cdots,\alpha_n\}}P''$ and $P'\approx_s^{sl} P''$.

By Lemma \ref{LUSWW3}, we know there is a sequential $H'$ such that $H(E(P))\xRightarrow[u]{\{\alpha_1,\cdots,\alpha_n\}}H'(P)\Rightarrow P''\approx_s^{sl} P'$.

And, $H(E(Q))\xRightarrow[u]{\{\alpha_1,\cdots,\alpha_n\}}H'(Q)\Rightarrow Q''$ and $P''\approx_s^{sl} Q''$. And $H(Q)\xrightarrow[u]{\{\alpha_1,\cdots,\alpha_n\}}Q'\approx_s^{sl} Q''$. Hence, $P'\approx_s^{sl} Q'$, as desired.
\end{proof}

\begin{theorem}[Unique solution of equations for weak static location hp-bisimulation]\label{USWSB3}
Let the guarded and sequential expressions $\widetilde{E}$ contain free variables $\subseteq \widetilde{X}$, then,

If $\widetilde{P}\approx_{hp}^{sl} \widetilde{E}\{\widetilde{P}/\widetilde{X}\}$ and $\widetilde{Q}\approx_{hp}^{sl} \widetilde{E}\{\widetilde{Q}/\widetilde{X}\}$, then $\widetilde{P}\approx_{hp}^{sl} \widetilde{Q}$.
\end{theorem}

\begin{proof}
Like the corresponding theorem in CCS, without loss of generality, we only consider a single equation $X=E$. So we assume $P\approx_{hp}^{sl} E(P)$, $Q\approx_{hp}^{sl} E(Q)$, then $P\approx_{hp}^{sl} Q$. Note that, we consider the general distribution.

We will prove $\{(H(P),H(Q)): H\}$ sequential, if $H(P)\xrightarrow[u]{\{\alpha_1,\cdots,\alpha_n\}}P'$, then, for some $Q'$, $H(Q)\xRightarrow[u]{\{\alpha_1.\cdots,\alpha_n\}}Q'$ and $P'\approx_{hp}^{sl} Q'$.

Let $H(P)\xrightarrow[u]{\{\alpha_1,\cdots,\alpha_n\}}P'$, then $H(E(P))\xRightarrow[u]{\{\alpha_1,\cdots,\alpha_n\}}P''$ and $P'\approx_{hp}^{sl} P''$.

By Lemma \ref{LUSWW3}, we know there is a sequential $H'$ such that $H(E(P))\xRightarrow[u]{\{\alpha_1,\cdots,\alpha_n\}}H'(P)\Rightarrow P''\approx_{hp}^{sl} P'$.

And, $H(E(Q))\xRightarrow[u]{\{\alpha_1,\cdots,\alpha_n\}}H'(Q)\Rightarrow Q''$ and $P''\approx_{hp}^{sl} Q''$. And $H(Q)\xrightarrow[u]{\{\alpha_1,\cdots,\alpha_n\}}Q'\approx_{hp}^{sl} Q''$. Hence, $P'\approx_{hp}^{sl} Q'$, as desired.
\end{proof}

\begin{theorem}[Unique solution of equations for weak static location hhp-bisimulation]\label{USWSB3}
Let the guarded and sequential expressions $\widetilde{E}$ contain free variables $\subseteq \widetilde{X}$, then,

If $\widetilde{P}\approx_{hhp}^{sl} \widetilde{E}\{\widetilde{P}/\widetilde{X}\}$ and $\widetilde{Q}\approx_{hhp}^{sl} \widetilde{E}\{\widetilde{Q}/\widetilde{X}\}$, then $\widetilde{P}\approx_{hhp}^{sl} \widetilde{Q}$.
\end{theorem}

\begin{proof}
Like the corresponding theorem in CCS, without loss of generality, we only consider a single equation $X=E$. So we assume $P\approx_{hhp}^{sl} E(P)$, $Q\approx_{hhp}^{sl} E(Q)$, then $P\approx_{hhp}^{sl} Q$. Note that, we consider the general distribution.

We will prove $\{(H(P),H(Q)): H\}$ sequential, if $H(P)\xrightarrow[u]{\{\alpha_1,\cdots,\alpha_n\}}P'$, then, for some $Q'$, $H(Q)\xRightarrow[u]{\{\alpha_1.\cdots,\alpha_n\}}Q'$ and $P'\approx_{hhp}^{sl} Q'$.

Let $H(P)\xrightarrow[u]{\{\alpha_1,\cdots,\alpha_n\}}P'$, then $H(E(P))\xRightarrow[u]{\{\alpha_1,\cdots,\alpha_n\}}P''$ and $P'\approx_{hhp}^{sl} P''$.

By Lemma \ref{LUSWW3}, we know there is a sequential $H'$ such that $H(E(P))\xRightarrow[u]{\{\alpha_1,\cdots,\alpha_n\}}H'(P)\Rightarrow P''\approx_{hhp}^{sl} P'$.

And, $H(E(Q))\xRightarrow[u]{\{\alpha_1,\cdots,\alpha_n\}}H'(Q)\Rightarrow Q''$ and $P''\approx_{hhp}^{sl} Q''$. And $H(Q)\xrightarrow[u]{\{\alpha_1,\cdots,\alpha_n\}}Q'\approx_{hhp}^{sl} Q''$. Hence, $P'\approx_{hhp}^{sl} Q'$, as desired.
\end{proof}

%% file: section4/section4.2.tex
\subsection{CTC with Dynamic Localities}{\label{ctcdl}}

CTC with dynamic localities is almost the same as CTC with static localities in section \ref{ctcsl}, as the locations are dynamically generated but not allocated statically. The LTSs-based
operational semantics and the laws are almost the same, except for the transition rules of $\textbf{Act}$ as follows.

\[\textbf{Act}_1\quad \frac{}{\alpha.P\xrightarrow[loc]{\alpha}loc::P}\]

\[\textbf{Act}_2\quad \frac{}{(\alpha_1\parallel\cdots\parallel\alpha_n).P\xrightarrow[loc]{\{\alpha_1,\cdots,\alpha_n\}}loc::P}\quad (\alpha_i\neq\overline{\alpha_j}\quad i,j\in\{1,\cdots,n\})\]

%% file: section5.tex
\section{APTC with Localities}\label{aptcl}

In this chapter, we introduce APTC with localities, including APTC with static localities in section \ref{aptcsl}, APTC with dynamic localities in section \ref{aptcdl}.

\input{section5/section5.1.tex}

\input{section5/section5.2.tex}

%% file: section5/section5.1.tex
\subsection{APTC with Static Localities}{\label{aptcsl}}

\subsubsection{BATC with Static Localities}

Let $Loc$ be the set of locations, and $loc\in Loc$, $u,v\in Loc^*$, $\epsilon$ is the empty location. A distribution allocates a location $u\in Loc*$ to an action $e$ denoted
$u::e$ or a process $x$ denoted $u::x$.


In the following, let $e_1, e_2, e_1', e_2'\in \mathbb{E}$, and let variables $x,y,z$ range over the set of terms for true concurrency, $p,q,s$ range over the set of closed terms.
The set of axioms of BATC with static localities ($BATC^{sl}$) consists of the laws given in Table \ref{AxiomsForBATC}.

\begin{center}
    \begin{table}
        \begin{tabular}{@{}ll@{}}
            \hline No. &Axiom\\
            $A1$ & $x+ y = y+ x$\\
            $A2$ & $(x+ y)+ z = x+ (y+ z)$\\
            $A3$ & $x+ x = x$\\
            $A4$ & $(x+ y)\cdot z = x\cdot z + y\cdot z$\\
            $A5$ & $(x\cdot y)\cdot z = x\cdot(y\cdot z)$\\
            $L1$ & $\epsilon::x=x$\\
            $L2$ & $u::(x\cdot y)=u::x\cdot u::y$\\
            $L3$ & $u::(x+ y)=u::x+ u::y$\\
            $L4$ & $u::(v::x)=uv::x$\\
        \end{tabular}
        \caption{Axioms of BATC with static localities}
        \label{AxiomsForBATC}
    \end{table}
\end{center}


\begin{definition}[Basic terms of BATC with static localities]\label{BTBATC}
The set of basic terms of BATC with static localities, $\mathcal{B}(BATC^{sl})$, is inductively defined as follows:
\begin{enumerate}
  \item $\mathbb{E}\subset\mathcal{B}(BATC^{sl})$;
  \item if $u\in Loc^*, t\in\mathcal{B}(BATC^{sl})$ then $u::t\in\mathcal{B}(BATC^{sl})$;
  \item if $e\in \mathbb{E}, t\in\mathcal{B}(BATC^{sl})$ then $e\cdot t\in\mathcal{B}(BATC^{sl})$;
  \item if $t,s\in\mathcal{B}(BATC^{sl})$ then $t+ s\in\mathcal{B}(BATC^{sl})$.
\end{enumerate}
\end{definition}

\begin{theorem}[Elimination theorem of BATC with static localities]\label{ETBATC}
Let $p$ be a closed BATC with static localities term. Then there is a basic BATC with static localities term $q$ such that $BATC^{sl}\vdash p=q$.
\end{theorem}

\begin{proof}
(1) Firstly, suppose that the following ordering on the signature of BATC with static localities is defined: $::>\cdot > +$ and the symbol $::$ is given the lexicographical status for
the first argument, then for each rewrite rule $p\rightarrow q$ in Table \ref{TRSForBATC} relation $p>_{lpo} q$ can easily be proved. We obtain that the term rewrite system shown in
Table \ref{TRSForBATC} is strongly normalizing, for it has finitely many rewriting rules, and $>$ is a well-founded ordering on the signature of BATC with static localities, and if
$s>_{lpo} t$, for each rewriting rule $s\rightarrow t$ is in Table \ref{TRSForBATC} (see Theorem \ref{SN}).

\begin{center}
    \begin{table}
        \begin{tabular}{@{}ll@{}}
            \hline No. &Rewriting Rule\\
            $RA3$ & $x+ x \rightarrow x$\\
            $RA4$ & $(x+ y)\cdot z \rightarrow x\cdot z + y\cdot z$\\
            $RA5$ & $(x\cdot y)\cdot z \rightarrow x\cdot(y\cdot z)$\\
            $RL1$ & $\epsilon::x\rightarrow x$\\
            $RL2$ & $u::(x\cdot y)\rightarrow u::x\cdot u::y$\\
            $RL3$ & $u::(x+ y)\rightarrow u::x+ u::y$\\
            $RL4$ & $u::(v::x)\rightarrow uv::x$\\
        \end{tabular}
        \caption{Term rewrite system of BATC with static localities}
        \label{TRSForBATC}
    \end{table}
\end{center}

(2) Then we prove that the normal forms of closed BATC with static localities terms are basic BATC with static localities terms.

Suppose that $p$ is a normal form of some closed BATC with static localities term and suppose that $p$ is not a basic term. Let $p'$ denote the smallest sub-term of $p$ which is not a
basic term. It implies that each sub-term of $p'$ is a basic term. Then we prove that $p$ is not a term in normal form. It is sufficient to induct on the structure of $p'$:

\begin{itemize}
  \item Case $p'\equiv u::e, e\in \mathbb{E}$. $p'$ is a basic term, which contradicts the assumption that $p'$ is not a basic term, so this case should not occur.
  \item Case $p'\equiv p_1\cdot p_2$. By induction on the structure of the basic term $p_1$:
      \begin{itemize}
        \item Subcase $p_1\in \mathbb{E}$. $p'$ would be a basic term, which contradicts the assumption that $p'$ is not a basic term;
        \item Subcase $p_1\equiv u:: p_1'$. $p'$ would be a basic term, which contradicts the assumption that $p'$ is not a basic term;
        \item Subcase $p_1\equiv e\cdot p_1'$. $RA5$ rewriting rule can be applied. So $p$ is not a normal form;
        \item Subcase $p_1\equiv p_1'+ p_1''$. $RA4$ rewriting rule can be applied. So $p$ is not a normal form.
      \end{itemize}
  \item Case $p'\equiv p_1+ p_2$. By induction on the structure of the basic terms both $p_1$ and $p_2$, all subcases will lead to that $p'$ would be a basic term, which contradicts
  the assumption that $p'$ is not a basic term.
\end{itemize}
\end{proof}


In this subsection, we will define a term-deduction system which gives the operational semantics of BATC with static localities. We give the operational transition rules for operators
$\cdot$ and $+$ as Table \ref{SETRForBATC} shows. And the predicate $\xrightarrow[u]{e}\surd$ represents successful termination after execution of the event $e$ at the location $u$.

\begin{center}
    \begin{table}
        $$\frac{}{e\xrightarrow[\epsilon]{e}\surd}\quad \frac{}{loc::e\xrightarrow[loc]{e}\surd}$$
        $$\frac{x\xrightarrow[u]{e}x'}{loc::x\xrightarrow[loc\ll u]{e}loc::x'}$$
        $$\frac{x\xrightarrow[u]{e}\surd}{x+ y\xrightarrow[u]{e}\surd} \quad\frac{x\xrightarrow[u]{e}x'}{x+ y\xrightarrow[u]{e}x'} \quad\frac{y\xrightarrow[u]{e}\surd}{x+ y\xrightarrow[u]{e}\surd} \quad\frac{y\xrightarrow[u]{e}y'}{x+ y\xrightarrow[u]{e}y'}$$
        $$\frac{x\xrightarrow[u]{e}\surd}{x\cdot y\xrightarrow[u]{e} y} \quad\frac{x\xrightarrow[u]{e}x'}{x\cdot y\xrightarrow[u]{e}x'\cdot y}$$
        \caption{Single event transition rules of BATC with static localities}
        \label{SETRForBATC}
    \end{table}
\end{center}

\begin{theorem}[Congruence of BATC with static localities with respect to static location pomset bisimulation equivalence]
Static location pomset bisimulation equivalence $\sim_p^{sl}$ is a congruence with respect to BATC with static localities.
\end{theorem}

\begin{proof}
It is easy to see that static location pomset bisimulation is an equivalent relation on BATC with static localities terms, we only need to prove that $\sim_p^{sl}$ is preserved by the operators $::$,
$\cdot$ and $+$, we omit the proof.
\end{proof}

\begin{theorem}[Soundness of BATC with static localities modulo static location pomset bisimulation equivalence]\label{SBATCPBE}
Let $x$ and $y$ be BATC with static localities terms. If $BATC^{sl}\vdash x=y$, then $x\sim_p^{sl} y$.
\end{theorem}

\begin{proof}
Since static location pomset bisimulation $\sim_p^{sl}$ is both an equivalent and a congruent relation, we only need to check if each axiom in Table \ref{AxiomsForBATC} is sound modulo
static location pomset bisimulation equivalence, the proof is trivial and we omit it.
\end{proof}

\begin{theorem}[Completeness of BATC with static localities modulo static location pomset bisimulation equivalence]\label{CBATCPBE}
Let $p$ and $q$ be closed BATC with static localities terms, if $p\sim_p^{sl} q$ then $p=q$.
\end{theorem}

\begin{proof}
Firstly, by the elimination theorem of BATC with static localities, we know that for each closed BATC with static localities term $p$, there exists a closed basic BATC with static
localities term $p'$, such that $BATC^{sl}\vdash p=p'$, so, we only need to consider closed basic BATC with static localities terms.

The basic terms (see Definition \ref{BTBATC}) modulo associativity and commutativity (AC) of conflict $+$ (defined by axioms $A1$ and $A2$ in Table \ref{AxiomsForBATC}), and this
equivalence is denoted by $=_{AC}$. Then, each equivalence class $s$ modulo AC of $+$ has the following normal form

$$s_1+\cdots+ s_k$$

with each $s_i$ either an atomic event or of the form $t_1\cdot t_2$, and each $s_i$ is called the summand of $s$.

Now, we prove that for normal forms $n$ and $n'$, if $n\sim_p^{sl} n'$ then $n=_{AC}n'$. It is sufficient to induct on the sizes of $n$ and $n'$.

\begin{itemize}
  \item Consider a summand $u::e$ of $n$. Then $n\xrightarrow[u]{e}\surd$, so $n\sim_p^{sl} n'$ implies $n'\xrightarrow[u]{e}\surd$, meaning that $n'$ also contains the summand $u::e$.
  \item Consider a summand $t_1\cdot t_2$ of $n$. Then $n\xrightarrow[u]{t_1}t_2$, so $n\sim_p^{sl} n'$ implies $n'\xrightarrow[u]{t_1}t_2'$ with $t_2\sim_p^{sl} t_2'$, meaning that
  $n'$ contains a summand $t_1\cdot t_2'$. Since $t_2$ and $t_2'$ are normal forms and have sizes smaller than $n$ and $n'$, by the induction hypotheses $t_2\sim_p^{sl} t_2'$ implies
  $t_2=_{AC} t_2'$.
\end{itemize}

So, we get $n=_{AC} n'$.

Finally, let $s$ and $t$ be basic terms, and $s\sim_p^{sl} t$, there are normal forms $n$ and $n'$, such that $s=n$ and $t=n'$. The soundness theorem of BATC with static localities
modulo static location pomset bisimulation equivalence (see Theorem \ref{SBATCPBE}) yields $s\sim_p^{sl} n$ and $t\sim_p^{sl} n'$, so $n\sim_p^{sl} s\sim_p^{sl} t\sim_p^{sl} n'$. Since
if $n\sim_p^{sl} n'$ then $n=_{AC}n'$, $s=n=_{AC}n'=t$, as desired.
\end{proof}

\begin{theorem}[Congruence of BATC with static localities with respect to static location step bisimulation equivalence]
Static location step bisimulation equivalence $\sim_s^{sl}$ is a congruence with respect to BATC with static localities.
\end{theorem}

\begin{proof}
It is easy to see that static location step bisimulation is an equivalent relation on BATC with static localities terms, we only need to prove that $\sim_s^{sl}$ is preserved by the
operators $::$, $\cdot$ and $+$, the proof is trivial and we omit it.
\end{proof}

\begin{theorem}[Soundness of BATC with static localities modulo static location step bisimulation equivalence]\label{SBATCSBE}
Let $x$ and $y$ be BATC with static localities terms. If $BATC^{sl}\vdash x=y$, then $x\sim_s^{sl} y$.
\end{theorem}

\begin{proof}
Since static location step bisimulation $\sim_s^{sl}$ is both an equivalent and a congruent relation, we only need to check if each axiom in Table \ref{AxiomsForBATC} is sound modulo
static location step bisimulation equivalence, the proof is trivial and we omit it.
\end{proof}

\begin{theorem}[Completeness of BATC with static localities modulo static location step bisimulation equivalence]\label{CBATCSBE}
Let $p$ and $q$ be closed BATC with static localities terms, if $p\sim_s^{sl} q$ then $p=q$.
\end{theorem}

\begin{proof}
Firstly, by the elimination theorem of BATC with static localities, we know that for each closed BATC with static localities term $p$, there exists a closed basic BATC with static
localities term $p'$, such that $BATC^{sl}\vdash p=p'$, so, we only need to consider closed basic BATC with static localities terms.

The basic terms (see Definition \ref{BTBATC}) modulo associativity and commutativity (AC) of conflict $+$ (defined by axioms $A1$ and $A2$ in Table \ref{AxiomsForBATC}), and this
equivalence is denoted by $=_{AC}$. Then, each equivalence class $s$ modulo AC of $+$ has the following normal form

$$s_1+\cdots+ s_k$$

with each $s_i$ either an atomic event or of the form $t_1\cdot t_2$, and each $s_i$ is called the summand of $s$.

Now, we prove that for normal forms $n$ and $n'$, if $n\sim_s^{sl} n'$ then $n=_{AC}n'$. It is sufficient to induct on the sizes of $n$ and $n'$.

\begin{itemize}
  \item Consider a summand $u::e$ of $n$. Then $n\xrightarrow[u]{e}\surd$, so $n\sim_s^{sl} n'$ implies $n'\xrightarrow[u]{e}\surd$, meaning that $n'$ also contains the summand $u::e$.
  \item Consider a summand $t_1\cdot t_2$ of $n$. Then $n\xrightarrow[u]{t_1}t_2$($\forall e_1, e_2\in t_1$ are pairwise concurrent), so $n\sim_s^{sl} n'$ implies
  $n'\xrightarrow[u]{t_1}t_2'$($\forall e_1, e_2\in t_1$ are pairwise concurrent) with $t_2\sim_s^{sl} t_2'$, meaning that $n'$ contains a summand $t_1\cdot t_2'$. Since $t_2$ and
  $t_2'$ are normal forms and have sizes smaller than $n$ and $n'$, by the induction hypotheses if $t_2\sim_s^{sl} t_2'$ then $t_2=_{AC} t_2'$.
\end{itemize}

So, we get $n=_{AC} n'$.

Finally, let $s$ and $t$ be basic terms, and $s\sim_s^{sl} t$, there are normal forms $n$ and $n'$, such that $s=n$ and $t=n'$. The soundness theorem of BATC with static localities
modulo static location step bisimulation equivalence (see Theorem \ref{SBATCSBE}) yields $s\sim_s^{sl} n$ and $t\sim_s^{sl} n'$, so $n\sim_s^{sl} s\sim_s^{sl} t\sim_s^{sl} n'$. Since
if $n\sim_s^{sl} n'$ then $n=_{AC}n'$, $s=n=_{AC}n'=t$, as desired.
\end{proof}

\begin{theorem}[Congruence of BATC with static localities with respect to static location hp-bisimulation equivalence]
Static location hp-bisimulation equivalence $\sim_{hp}^{sl}$ is a congruence with respect to BATC with static localities.
\end{theorem}

\begin{proof}
It is easy to see that static location history-preserving bisimulation is an equivalent relation on BATC with static localities terms, we only need to prove that $\sim_{hp}^{sl}$ is
preserved by the operators $::$, $\cdot$ and $+$, the proof is trivial and we omit it.
\end{proof}

\begin{theorem}[Soundness of BATC with static localities modulo static location hp-bisimulation equivalence]\label{SBATCHPBE}
Let $x$ and $y$ be BATC with static localities terms. If $BATC\vdash x=y$, then $x\sim_{hp}^{sl} y$.
\end{theorem}

\begin{proof}
Since static location hp-bisimulation $\sim_{hp}^{sl}$ is both an equivalent and a congruent relation, we only need to check if each axiom in Table \ref{AxiomsForBATC} is sound modulo
static location hp-bisimulation equivalence, the proof is trivial and we omit it.
\end{proof}

\begin{theorem}[Completeness of BATC with static localities modulo static location hp-bisimulation equivalence]\label{CBATCHPBE}
Let $p$ and $q$ be closed BATC with static localities terms, if $p\sim_{hp}^{sl} q$ then $p=q$.
\end{theorem}

\begin{proof}
Firstly, by the elimination theorem of BATC with static localities, we know that for each closed BATC with static localities term $p$, there exists a closed basic BATC with static
localities term $p'$, such that $BATC^{sl}\vdash p=p'$, so, we only need to consider closed basic BATC with static localities terms.

The basic terms (see Definition \ref{BTBATC}) modulo associativity and commutativity (AC) of conflict $+$ (defined by axioms $A1$ and $A2$ in Table \ref{AxiomsForBATC}), and this
equivalence is denoted by $=_{AC}$. Then, each equivalence class $s$ modulo AC of $+$ has the following normal form

$$s_1+\cdots+ s_k$$

with each $s_i$ either an atomic event or of the form $t_1\cdot t_2$, and each $s_i$ is called the summand of $s$.

Now, we prove that for normal forms $n$ and $n'$, if $n\sim_{hp}^{sl} n'$ then $n=_{AC}n'$. It is sufficient to induct on the sizes of $n$ and $n'$.

\begin{itemize}
  \item Consider a summand $u::e$ of $n$. Then $n\xrightarrow[u]{e}\surd$, so $n\sim_{hp}^{sl} n'$ implies $n'\xrightarrow[u]{e}\surd$, meaning that $n'$ also contains the summand $u::e$.
  \item Consider a summand $u::e\cdot s$ of $n$. Then $n\xrightarrow[u]{e}s$, so $n\sim_{hp}^{sl} n'$ implies $n'\xrightarrow[u]{e}t$ with $s\sim_{hp}^{sl} t$, meaning that $n'$
  contains a summand $u::e\cdot t$. Since $s$ and $t$ are normal forms and have sizes smaller than $n$ and $n'$, by the induction hypotheses $s\sim_{hp}^{sl} t$ implies $s=_{AC} t$.
\end{itemize}

So, we get $n=_{AC} n'$.

Finally, let $s$ and $t$ be basic terms, and $s\sim_{hp}^{sl} t$, there are normal forms $n$ and $n'$, such that $s=n$ and $t=n'$. The soundness theorem of BATC with static localities
modulo static location hp-bisimulation equivalence (see Theorem \ref{SBATCHPBE}) yields $s\sim_{hp}^{sl} n$ and $t\sim_{hp}^{sl} n'$, so $n\sim_{hp}^{sl} s\sim_{hp}^{sl} t\sim_{hp}^{sl} n'$.
Since if $n\sim_{hp}^{sl} n'$ then $n=_{AC}n'$, $s=n=_{AC}n'=t$, as desired.
\end{proof}

\begin{theorem}[Congruence of BATC with static localities with respect to static location hhp-bisimulation equivalence]
Static location hhp-bisimulation equivalence $\sim_{hhp}^{sl}$ is a congruence with respect to BATC with static localities.
\end{theorem}

\begin{proof}
It is easy to see that static location hhp-bisimulation is an equivalent relation on BATC with static localities terms, we only need to prove that $\sim_{hhp}^{sl}$ is preserved by the
operators $::$, $\cdot$ and $+$, the proof is trivial and we omit it.
\end{proof}

\begin{theorem}[Soundness of BATC with static localities modulo static location hhp-bisimulation equivalence]\label{SBATCHHPBE}
Let $x$ and $y$ be BATC with static localities terms. If $BATC\vdash x=y$, then $x\sim_{hhp}^{sl} y$.
\end{theorem}

\begin{proof}
Since static location hhp-bisimulation $\sim_{hhp}^{sl}$ is both an equivalent and a congruent relation, we only need to check if each axiom in Table \ref{AxiomsForBATC} is sound
modulo static location hhp-bisimulation equivalence, the proof is trivial and we omit it.
\end{proof}

\begin{theorem}[Completeness of BATC with static localities modulo static location hhp-bisimulation equivalence]\label{CBATCHHPBE}
Let $p$ and $q$ be closed BATC with static localities terms, if $p\sim_{hhp}^{sl} q$ then $p=q$.
\end{theorem}

\begin{proof}
Firstly, by the elimination theorem of BATC with static localities, we know that for each closed BATC with static localities term $p$, there exists a closed basic BATC with static
localities term $p'$, such that $BATC^{sl}\vdash p=p'$, so, we only need to consider closed basic BATC with static localities terms.

The basic terms (see Definition \ref{BTBATC}) modulo associativity and commutativity (AC) of conflict $+$ (defined by axioms $A1$ and $A2$ in Table \ref{AxiomsForBATC}), and this
equivalence is denoted by $=_{AC}$. Then, each equivalence class $s$ modulo AC of $+$ has the following normal form

$$s_1+\cdots+ s_k$$

with each $s_i$ either an atomic event or of the form $t_1\cdot t_2$, and each $s_i$ is called the summand of $s$.

Now, we prove that for normal forms $n$ and $n'$, if $n\sim_{hhp}^{sl} n'$ then $n=_{AC}n'$. It is sufficient to induct on the sizes of $n$ and $n'$.

\begin{itemize}
  \item Consider a summand $u::e$ of $n$. Then $n\xrightarrow[u]{e}\surd$, so $n\sim_{hhp}^{sl} n'$ implies $n'\xrightarrow[u]{e}\surd$, meaning that $n'$ also contains the summand $u::e$.
  \item Consider a summand $u::e\cdot s$ of $n$. Then $n\xrightarrow[u]{e}s$, so $n\sim_{hhp}^{sl} n'$ implies $n'\xrightarrow[u]{e}t$ with $s\sim_{hhp}^{sl} t$, meaning that $n'$
  contains a summand $u::e\cdot t$. Since $s$ and $t$ are normal forms and have sizes smaller than $n$ and $n'$, by the induction hypotheses $s\sim_{hhp}^{sl} t$ implies $s=_{AC} t$.
\end{itemize}

So, we get $n=_{AC} n'$.

Finally, let $s$ and $t$ be basic terms, and $s\sim_{hhp}^{sl} t$, there are normal forms $n$ and $n'$, such that $s=n$ and $t=n'$. The soundness theorem of BATC with static localities
modulo history-preserving bisimulation equivalence (see Theorem \ref{SBATCHHPBE}) yields $s\sim_{hhp}^{sl} n$ and $t\sim_{hhp}^{sl} n'$, so
$n\sim_{hhp}^{sl} s\sim_{hhp}^{sl} t\sim_{hhp}^{sl} n'$. Since if $n\sim_{hhp}^{sl} n'$ then $n=_{AC}n'$, $s=n=_{AC}n'=t$, as desired.
\end{proof}

\subsubsection{APTC with Static Localities}

We give the transition rules of APTC with static localities as Table \ref{TRForAPTC} shows.

\begin{center}
    \begin{table}
        $$\frac{x\xrightarrow[u]{e_1}\surd\quad y\xrightarrow[v]{e_2}\surd}{x\parallel y\xrightarrow[u\diamond v]{\{e_1,e_2\}}\surd} \quad\frac{x\xrightarrow[u]{e_1}x'\quad y\xrightarrow[v]{e_2}\surd}{x\parallel y\xrightarrow[u\diamond v]{\{e_1,e_2\}}x'}$$
        $$\frac{x\xrightarrow[u]{e_1}\surd\quad y\xrightarrow[v]{e_2}y'}{x\parallel y\xrightarrow[u\diamond v]{\{e_1,e_2\}}y'} \quad\frac{x\xrightarrow[u]{e_1}x'\quad y\xrightarrow[v]{e_2}y'}{x\parallel y\xrightarrow[u\diamond v]{\{e_1,e_2\}}x'\between y'}$$
        $$\frac{x\xrightarrow[u]{e_1}\surd\quad y\xrightarrow[v]{e_2}\surd \quad(e_1\leq e_2)}{x\leftmerge y\xrightarrow[u\diamond v]{\{e_1,e_2\}}\surd} \quad\frac{x\xrightarrow[u]{e_1}x'\quad y\xrightarrow[v]{e_2}\surd \quad(e_1\leq e_2)}{x\leftmerge y\xrightarrow[u\diamond v]{\{e_1,e_2\}}x'}$$
        $$\frac{x\xrightarrow[u]{e_1}\surd\quad y\xrightarrow[v]{e_2}y' \quad(e_1\leq e_2)}{x\leftmerge y\xrightarrow[u\diamond v]{\{e_1,e_2\}}y'} \quad\frac{x\xrightarrow[u]{e_1}x'\quad y\xrightarrow[v]{e_2}y' \quad(e_1\leq e_2)}{x\leftmerge y\xrightarrow[u\diamond v]{\{e_1,e_2\}}x'\between y'}$$
        $$\frac{x\xrightarrow[u]{e_1}\surd\quad y\xrightarrow[v]{e_2}\surd}{x\mid y\xrightarrow[u\diamond v]{\gamma(e_1,e_2)}\surd} \quad\frac{x\xrightarrow[u]{e_1}x'\quad y\xrightarrow[v]{e_2}\surd}{x\mid y\xrightarrow[u\diamond v]{\gamma(e_1,e_2)}x'}$$
        $$\frac{x\xrightarrow[u]{e_1}\surd\quad y\xrightarrow[v]{e_2}y'}{x\mid y\xrightarrow[u\diamond v]{\gamma(e_1,e_2)}y'} \quad\frac{x\xrightarrow[u]{e_1}x'\quad y\xrightarrow[v]{e_2}y'}{x\mid y\xrightarrow[u\diamond v]{\gamma(e_1,e_2)}x'\between y'}$$
        $$\frac{x\xrightarrow[u]{e_1}\surd\quad (\sharp(e_1,e_2))}{\Theta(x)\xrightarrow[u]{e_1}\surd} \quad\frac{x\xrightarrow[u]{e_2}\surd\quad (\sharp(e_1,e_2))}{\Theta(x)\xrightarrow[u]{e_2}\surd}$$
        $$\frac{x\xrightarrow[u]{e_1}x'\quad (\sharp(e_1,e_2))}{\Theta(x)\xrightarrow[u]{e_1}\Theta(x')} \quad\frac{x\xrightarrow[u]{e_2}x'\quad (\sharp(e_1,e_2))}{\Theta(x)\xrightarrow[u]{e_2}\Theta(x')}$$
        $$\frac{x\xrightarrow[u]{e_1}\surd \quad y\nrightarrow^{e_2}\quad (\sharp(e_1,e_2))}{x\triangleleft y\xrightarrow[u]{\tau}\surd}
        \quad\frac{x\xrightarrow[u]{e_1}x' \quad y\nrightarrow^{e_2}\quad (\sharp(e_1,e_2))}{x\triangleleft y\xrightarrow[u]{\tau}x'}$$
        $$\frac{x\xrightarrow[u]{e_1}\surd \quad y\nrightarrow^{e_3}\quad (\sharp(e_1,e_2),e_2\leq e_3)}{x\triangleleft y\xrightarrow[u]{e_1}\surd}
        \quad\frac{x\xrightarrow[u]{e_1}x' \quad y\nrightarrow^{e_3}\quad (\sharp(e_1,e_2),e_2\leq e_3)}{x\triangleleft y\xrightarrow[u]{e_1}x'}$$
        $$\frac{x\xrightarrow[u]{e_3}\surd \quad y\nrightarrow^{e_2}\quad (\sharp(e_1,e_2),e_1\leq e_3)}{x\triangleleft y\xrightarrow[u]{\tau}\surd}
        \quad\frac{x\xrightarrow[u]{e_3}x' \quad y\nrightarrow^{e_2}\quad (\sharp(e_1,e_2),e_1\leq e_3)}{x\triangleleft y\xrightarrow[u]{\tau}x'}$$
        \caption{Transition rules of APTC with static localities}
        \label{TRForAPTC}
    \end{table}
\end{center}

In the following, we show that the elimination theorem does not hold for truly concurrent processes combined the operators $\cdot$, $+$ and $\leftmerge$. Firstly, we define the basic terms for APTC with static localities.

\begin{definition}[Basic terms of APTC with static localities]\label{BTAPTC}
The set of basic terms of APTC with static localities, $\mathcal{B}(APTC^{sl})$, is inductively defined as follows:
\begin{enumerate}
  \item $\mathbb{E}\subset\mathcal{B}(APTC^{sl})$;
  \item if $u\in Loc^*, t\in\mathcal{B}(APTC^{sl})$ then $u::t\in\mathcal{B}(APTC^{sl})$;
  \item if $e\in \mathbb{E}, t\in\mathcal{B}(APTC^{sl})$ then $e\cdot t\in\mathcal{B}(APTC^{sl})$;
  \item if $t,s\in\mathcal{B}(APTC^{sl})$ then $t+ s\in\mathcal{B}(APTC^{sl})$;
  \item if $t,s\in\mathcal{B}(APTC^{sl})$ then $t\leftmerge s\in\mathcal{B}(APTC^{sl})$.
\end{enumerate}
\end{definition}

\begin{theorem}[Congruence theorem of APTC with static localities]
Static location truly concurrent bisimulation equivalences $\sim_p^{sl}$, $\sim_s^{sl}$, $\sim_{hp}^{sl}$ and $\sim_{hhp}^{sl}$ are all congruences with respect to APTC with static
localities.
\end{theorem}

\begin{proof}
It is easy to see that static location pomset, step, hp-, hhp- bisimulation are all equivalent relations on APTC with static localities terms, we only need to prove that
$\sim_p^{sl}$, $\sim_s^{sl}$, $\sim_{hp}^{sl}$ and $\sim_{hhp}^{sl}$ are all preserved by the operators $\between$, $\parallel$, $\leftmerge$, $\mid$, $\Theta$ and $\triangleleft$,
the proof is trivial and we omit it.
\end{proof}

So, we design the axioms of parallelism in Table \ref{AxiomsForParallelism}, including algebraic laws for parallel operator $\parallel$, communication operator $\mid$, conflict elimination operator $\Theta$ and unless operator $\triangleleft$, and also the whole parallel operator $\between$. Since the communication between two communicating events in different parallel branches may cause deadlock (a state of inactivity), which is caused by mismatch of two communicating events or the imperfectness of the communication channel. We introduce a new constant $\delta$ to denote the deadlock, and let the atomic event $e\in \mathbb{E}\cup\{\delta\}$.

\begin{center}
    \begin{table}
        \begin{tabular}{@{}ll@{}}
            \hline No. &Axiom\\
            $A6$ & $x+ \delta = x$\\
            $A7$ & $\delta\cdot x =\delta$\\
            $P1$ & $x\between y = x\parallel y + x\mid y$\\
            $P2$ & $x\parallel y = y \parallel x$\\
            $P3$ & $(x\parallel y)\parallel z = x\parallel (y\parallel z)$\\
            $P4$ & $x\parallel y = x\leftmerge y + y\leftmerge x$\\
            $P5$ & $(e_1\leq e_2)\quad e_1\leftmerge (e_2\cdot y) = (e_1\leftmerge e_2)\cdot y$\\
            $P6$ & $(e_1\leq e_2)\quad (e_1\cdot x)\leftmerge e_2 = (e_1\leftmerge e_2)\cdot x$\\
            $P7$ & $(e_1\leq e_2)\quad (e_1\cdot x)\leftmerge (e_2\cdot y) = (e_1\leftmerge e_2)\cdot (x\between y)$\\
            $P8$ & $(x+ y)\leftmerge z = (x\leftmerge z)+ (y\leftmerge z)$\\
            $P9$ & $\delta\leftmerge x = \delta$\\
            $C1$ & $e_1\mid e_2 = \gamma(e_1,e_2)$\\
            $C2$ & $e_1\mid (e_2\cdot y) = \gamma(e_1,e_2)\cdot y$\\
            $C3$ & $(e_1\cdot x)\mid e_2 = \gamma(e_1,e_2)\cdot x$\\
            $C4$ & $(e_1\cdot x)\mid (e_2\cdot y) = \gamma(e_1,e_2)\cdot (x\between y)$\\
            $C5$ & $(x+ y)\mid z = (x\mid z) + (y\mid z)$\\
            $C6$ & $x\mid (y+ z) = (x\mid y)+ (x\mid z)$\\
            $C7$ & $\delta\mid x = \delta$\\
            $C8$ & $x\mid\delta = \delta$\\
            $CE1$ & $\Theta(e) = e$\\
            $CE2$ & $\Theta(\delta) = \delta$\\
            $CE3$ & $\Theta(x+ y) = \Theta(x)\triangleleft y + \Theta(y)\triangleleft x$\\
            $CE4$ & $\Theta(x\cdot y)=\Theta(x)\cdot\Theta(y)$\\
            $CE5$ & $\Theta(x\parallel y) = ((\Theta(x)\triangleleft y)\parallel y)+ ((\Theta(y)\triangleleft x)\parallel x)$\\
            $CE6$ & $\Theta(x\mid y) = ((\Theta(x)\triangleleft y)\mid y)+ ((\Theta(y)\triangleleft x)\mid x)$\\
            $U1$ & $(\sharp(e_1,e_2))\quad e_1\triangleleft e_2 = \tau$\\
            $U2$ & $(\sharp(e_1,e_2),e_2\leq e_3)\quad e_1\triangleleft e_3 = e_1$\\
            $U3$ & $(\sharp(e_1,e_2),e_2\leq e_3)\quad e3\triangleleft e_1 = \tau$\\
            $U4$ & $e\triangleleft \delta = e$\\
            $U5$ & $\delta \triangleleft e = \delta$\\
            $U6$ & $(x+ y)\triangleleft z = (x\triangleleft z)+ (y\triangleleft z)$\\
            $U7$ & $(x\cdot y)\triangleleft z = (x\triangleleft z)\cdot (y\triangleleft z)$\\
            $U8$ & $(x\leftmerge y)\triangleleft z = (x\triangleleft z)\leftmerge (y\triangleleft z)$\\
            $U9$ & $(x\mid y)\triangleleft z = (x\triangleleft z)\mid (y\triangleleft z)$\\
            $U10$ & $x\triangleleft (y+ z) = (x\triangleleft y)\triangleleft z$\\
            $U11$ & $x\triangleleft (y\cdot z)=(x\triangleleft y)\triangleleft z$\\
            $U12$ & $x\triangleleft (y\leftmerge z) = (x\triangleleft y)\triangleleft z$\\
            $U13$ & $x\triangleleft (y\mid z) = (x\triangleleft y)\triangleleft z$\\
            $L5$ & $u::(x\between y) = u::x\between u:: y$\\
            $L6$ & $u::(x\parallel y) = u::x\parallel u:: y$\\
            $L7$ & $u::(x\mid y) = u::x\mid u:: y$\\
            $L8$ & $u::(\Theta(x)) = \Theta(u::x)$\\
            $L9$ & $u::(x\triangleleft y) = u::x\triangleleft u:: y$\\
            $L10$ & $u::\delta=\delta$\\
        \end{tabular}
        \caption{Axioms of parallelism}
        \label{AxiomsForParallelism}
    \end{table}
\end{center}


Based on the definition of basic terms for APTC with static localities (see Definition \ref{BTAPTC}) and axioms of parallelism (see Table \ref{AxiomsForParallelism}), we can prove the
elimination theorem of parallelism.

\begin{theorem}[Elimination theorem of parallelism]\label{ETParallelism}
Let $p$ be a closed APTC with static localities term. Then there is a basic APTC with static localities term $q$ such that $APTC^{sl}\vdash p=q$.
\end{theorem}

\begin{proof}
(1) Firstly, suppose that the following ordering on the signature of APTC with static localities is defined: $::>\leftmerge > \cdot > +$ and the symbol $::$ is given the lexicographical
status for the first argument, then for each rewrite rule $p\rightarrow q$ in Table \ref{TRSForAPTC} relation $p>_{lpo} q$ can easily be proved. We obtain that the term rewrite system
shown in Table \ref{TRSForAPTC} is strongly normalizing, for it has finitely many rewriting rules, and $>$ is a well-founded ordering on the signature of APTC with static localities,
and if $s>_{lpo} t$, for each rewriting rule $s\rightarrow t$ is in Table \ref{TRSForAPTC} (see Theorem \ref{SN}).

\begin{center}
    \begin{table}
        \begin{tabular}{@{}ll@{}}
            \hline No. &Rewriting Rule\\
            $RA6$ & $x+ \delta \rightarrow x$\\
            $RA7$ & $\delta\cdot x \rightarrow\delta$\\
            $RP1$ & $x\between y \rightarrow x\parallel y + x\mid y$\\
            $RP2$ & $x\parallel y \rightarrow y \parallel x$\\
            $RP3$ & $(x\parallel y)\parallel z \rightarrow x\parallel (y\parallel z)$\\
            $RP4$ & $x\parallel y \rightarrow x\leftmerge y + y\leftmerge x$\\
            $RP5$ & $(e_1\leq e_2)\quad e_1\leftmerge (e_2\cdot y) \rightarrow (e_1\leftmerge e_2)\cdot y$\\
            $RP6$ & $(e_1\leq e_2)\quad (e_1\cdot x)\leftmerge e_2 \rightarrow (e_1\leftmerge e_2)\cdot x$\\
            $RP7$ & $(e_1\leq e_2)\quad (e_1\cdot x)\leftmerge (e_2\cdot y) \rightarrow (e_1\leftmerge e_2)\cdot (x\between y)$\\
            $RP8$ & $(x+ y)\leftmerge z \rightarrow (x\leftmerge z)+ (y\leftmerge z)$\\
            $RP9$ & $\delta\leftmerge x \rightarrow \delta$\\
            $RC1$ & $e_1\mid e_2 \rightarrow \gamma(e_1,e_2)$\\
            $RC2$ & $e_1\mid (e_2\cdot y) \rightarrow \gamma(e_1,e_2)\cdot y$\\
            $RC3$ & $(e_1\cdot x)\mid e_2 \rightarrow \gamma(e_1,e_2)\cdot x$\\
            $RC4$ & $(e_1\cdot x)\mid (e_2\cdot y) \rightarrow \gamma(e_1,e_2)\cdot (x\between y)$\\
            $RC5$ & $(x+ y)\mid z \rightarrow (x\mid z) + (y\mid z)$\\
            $RC6$ & $x\mid (y+ z) \rightarrow (x\mid y)+ (x\mid z)$\\
            $RC7$ & $\delta\mid x \rightarrow \delta$\\
            $RC8$ & $x\mid\delta \rightarrow \delta$\\
            $RCE1$ & $\Theta(e) \rightarrow e$\\
            $RCE2$ & $\Theta(\delta) \rightarrow \delta$\\
            $RCE3$ & $\Theta(x+ y) \rightarrow \Theta(x)\triangleleft y + \Theta(y)\triangleleft x$\\
            $RCE4$ & $\Theta(x\cdot y)\rightarrow\Theta(x)\cdot\Theta(y)$\\
            $RCE5$ & $\Theta(x\parallel y) \rightarrow ((\Theta(x)\triangleleft y)\parallel y)+ ((\Theta(y)\triangleleft x)\parallel x)$\\
            $RCE6$ & $\Theta(x\mid y) \rightarrow ((\Theta(x)\triangleleft y)\mid y)+ ((\Theta(y)\triangleleft x)\mid x)$\\
            $RU1$ & $(\sharp(e_1,e_2))\quad e_1\triangleleft e_2 \rightarrow \tau$\\
            $RU2$ & $(\sharp(e_1,e_2),e_2\leq e_3)\quad e_1\triangleleft e_3 \rightarrow e_1$\\
            $RU3$ & $(\sharp(e_1,e_2),e_2\leq e_3)\quad e3\triangleleft e_1 \rightarrow \tau$\\
            $RU4$ & $e\triangleleft \delta \rightarrow e$\\
            $RU5$ & $\delta \triangleleft e \rightarrow \delta$\\
            $RU6$ & $(x+ y)\triangleleft z \rightarrow (x\triangleleft z)+ (y\triangleleft z)$\\
            $RU7$ & $(x\cdot y)\triangleleft z \rightarrow (x\triangleleft z)\cdot (y\triangleleft z)$\\
            $RU8$ & $(x\leftmerge y)\triangleleft z \rightarrow (x\triangleleft z)\leftmerge (y\triangleleft z)$\\
            $RU9$ & $(x\mid y)\triangleleft z \rightarrow (x\triangleleft z)\mid (y\triangleleft z)$\\
            $RU10$ & $x\triangleleft (y+ z) \rightarrow (x\triangleleft y)\triangleleft z$\\
            $RU11$ & $x\triangleleft (y\cdot z)\rightarrow(x\triangleleft y)\triangleleft z$\\
            $RU12$ & $x\triangleleft (y\leftmerge z) \rightarrow (x\triangleleft y)\triangleleft z$\\
            $RU13$ & $x\triangleleft (y\mid z) \rightarrow (x\triangleleft y)\triangleleft z$\\
            $RL5$ & $u::(x\between y) \rightarrow u::x\between u:: y$\\
            $RL6$ & $u::(x\parallel y) \rightarrow u::x\parallel u:: y$\\
            $RL7$ & $u::(x\mid y) \rightarrow u::x\mid u:: y$\\
            $RL8$ & $u::(\Theta(x)) \rightarrow \Theta(u::x)$\\
            $RL9$ & $u::(x\triangleleft y) \rightarrow u::x\triangleleft u:: y$\\
            $RL10$ & $u::\delta\rightarrow\delta$\\
        \end{tabular}
        \caption{Term rewrite system of APTC with static localities}
        \label{TRSForAPTC}
    \end{table}
\end{center}

(2) Then we prove that the normal forms of closed APTC with static localities terms are basic APTC with static localities terms.

Suppose that $p$ is a normal form of some closed APTC with static localities term and suppose that $p$ is not a basic APTC with static localities term. Let $p'$ denote the smallest
sub-term of $p$ which is not a basic APTC with static localities term. It implies that each sub-term of $p'$ is a basic APTC with static localities term. Then we prove that $p$ is not
a term in normal form. It is sufficient to induct on the structure of $p'$:

\begin{itemize}
  \item Case $p'\equiv u::e, e\in \mathbb{E}$. $p'$ is a basic APTC with static localities term, which contradicts the assumption that $p'$ is not a basic APTC with static localities
  term, so this case should not occur.
  \item Case $p'\equiv p_1\cdot p_2$. By induction on the structure of the basic APTC with static localities term $p_1$:
      \begin{itemize}
        \item Subcase $p_1\in \mathbb{E}$. $p'$ would be a basic APTC with static localities term, which contradicts the assumption that $p'$ is not a basic APTC with static localities term;
        \item Subcase $p_1\equiv u::e$. $p'$ would be a basic APTC with static localities term, which contradicts the assumption that $p'$ is not a basic APTC with static localities term;
        \item Subcase $p_1\equiv e\cdot p_1'$. $RA5$ rewriting rule in Table \ref{TRSForBATC} can be applied. So $p$ is not a normal form;
        \item Subcase $p_1\equiv p_1'+ p_1''$. $RA4$ rewriting rule in Table \ref{TRSForBATC} can be applied. So $p$ is not a normal form;
        \item Subcase $p_1\equiv p_1'\parallel p_1''$. $p'$ would be a basic APTC with static localities term, which contradicts the assumption that $p'$ is not a basic APTC with static localities term;
        \item Subcase $p_1\equiv p_1'\mid p_1''$. $RC1$ rewrite rule in Table \ref{TRSForAPTC} can be applied. So $p$ is not a normal form;
        \item Subcase $p_1\equiv \Theta(p_1')$. $RCE1$-$RCE6$ rewrite rules in Table \ref{TRSForAPTC} can be applied. So $p$ is not a normal form.
      \end{itemize}
  \item Case $p'\equiv p_1+ p_2$. By induction on the structure of the basic APTC with static localities terms both $p_1$ and $p_2$, all subcases will lead to that $p'$ would be a basic APTC with static localities term, which contradicts the assumption that $p'$ is not a basic APTC with static localities term.
  \item Case $p'\equiv p_1\leftmerge p_2$. By induction on the structure of the basic APTC with static localities terms both $p_1$ and $p_2$, all subcases will lead to that $p'$ would be a basic APTC with static localities term, which contradicts the assumption that $p'$ is not a basic APTC with static localities term.
  \item Case $p'\equiv p_1\mid p_2$. By induction on the structure of the basic APTC with static localities terms both $p_1$ and $p_2$, all subcases will lead to that $p'$ would be a basic APTC with static localities term, which contradicts the assumption that $p'$ is not a basic APTC with static localities term.
  \item Case $p'\equiv \Theta(p_1)$. By induction on the structure of the basic APTC with static localities term $p_1$, $RCE1-RCE6$ rewrite rules in Table \ref{TRSForAPTC} can be applied. So $p$ is not a normal form.
  \item Case $p'\equiv p_1\triangleleft p_2$. By induction on the structure of the basic APTC with static localities terms both $p_1$ and $p_2$, all subcases will lead to that $p'$ would be a basic APTC with static localities term, which contradicts the assumption that $p'$ is not a basic APTC with static localities term.
\end{itemize}
\end{proof}


\begin{theorem}[Generalization of APTC with static localities with respect to BATC with static localities]
APTC with static localities is a generalization of BATC with static localities.
\end{theorem}

\begin{proof}
It follows from the following three facts.

\begin{enumerate}
  \item The transition rules of BATC with static localities are all source-dependent;
  \item The sources of the transition rules for APTC with static localities contain an occurrence of $\between$, or $\parallel$, or $\leftmerge$, or $\mid$, or $\Theta$, or $\triangleleft$;
  \item The transition rules of APTC with static localities are all source-dependent.
\end{enumerate}

So, APTC with static localities is a generalization of BATC with static localities, that is, BATC with static localities is an embedding of APTC with static localities, as desired.
\end{proof}

\begin{theorem}[Soundness of APTC with static localities modulo static location pomset bisimulation equivalence]\label{SPPBE}
Let $x$ and $y$ be APTC with static localities terms. If $APTC^{sl}\vdash x=y$, then $x\sim_p^{sl} y$.
\end{theorem}

\begin{proof}
Since static location pomset bisimulation $\sim_p^{sl}$ is both an equivalent and a congruent relation with respect to the operators $\between$, $\parallel$, $\leftmerge$, $\mid$,
$\Theta$ and $\triangleleft$, we only need to check if each axiom in Table \ref{AxiomsForParallelism} is sound modulo static location pomset bisimulation equivalence, the proof is trivial
and we omit it.
\end{proof}

\begin{theorem}[Completeness of APTC with static localities modulo static location pomset bisimulation equivalence]\label{CPPBE}
Let $p$ and $q$ be closed APTC with static localities terms, if $p\sim_p^{sl} q$ then $p=q$.
\end{theorem}

\begin{proof}
Firstly, by the elimination theorem of APTC with static localities (see Theorem \ref{ETParallelism}), we know that for each closed APTC with static localities term $p$, there exists a
closed basic APTC with static localities term $p'$, such that $APTC^{sl}\vdash p=p'$, so, we only need to consider closed basic APTC with static localities terms.

The basic terms (see Definition \ref{BTAPTC}) modulo associativity and commutativity (AC) of conflict $+$ (defined by axioms $A1$ and $A2$ in Table \ref{AxiomsForBATC}), and these
equivalences is denoted by $=_{AC}$. Then, each equivalence class $s$ modulo AC of $+$ has the following normal form

$$s_1+\cdots+ s_k$$

with each $s_i$ either an atomic event or of the form

$$t_1\cdot\cdots\cdot t_m$$

with each $t_j$ either an atomic event or of the form

$$u_1\parallel\cdots\parallel u_l$$

with each $u_l$ an atomic event with a location, and each $s_i$ is called the summand of $s$.

Now, we prove that for normal forms $n$ and $n'$, if $n\sim_p^{sl} n'$ then $n=_{AC}n'$. It is sufficient to induct on the sizes of $n$ and $n'$.

\begin{itemize}
  \item Consider a summand $u::e$ of $n$. Then $n\xrightarrow[u]{e}\surd$, so $n\sim_p^{sl} n'$ implies $n'\xrightarrow[u]{e}\surd$, meaning that $n'$ also contains the summand $u::e$.
  \item Consider a summand $t_1\cdot t_2$ of $n$,
  \begin{itemize}
    \item if $t_1\equiv u'::e'$, then $n\xrightarrow[u']{e'}t_2$, so $n\sim_p^{sl} n'$ implies $n'\xrightarrow[u']{e'}t_2'$ with $t_2\sim_p^{sl} t_2'$, meaning that $n'$ contains a
    summand $u'::e'\cdot t_2'$. Since $t_2$ and $t_2'$ are normal forms and have sizes smaller than $n$ and $n'$, by the induction hypotheses if $t_2\sim_p^{sl} t_2'$ then
    $t_2=_{AC} t_2'$;
    \item if $t_1\equiv u_1::e_1\leftmerge\cdots\leftmerge u_l::e_l$, then $n\xrightarrow[u]{\{e_1,\cdots,e_l\}}t_2$, so $n\sim_p^{sl} n'$ implies
    $n'\xrightarrow[u]{\{e_1,\cdots,e_l\}}t_2'$ with $t_2\sim_p^{sl} t_2'$, meaning that $n'$ contains a summand $(u_1::e_1\leftmerge\cdots\leftmerge u_l::e_l)\cdot t_2'$. Since $t_2$
    and $t_2'$ are normal forms and have sizes smaller than $n$ and $n'$, by the induction hypotheses if $t_2\sim_p^{sl} t_2'$ then $t_2=_{AC} t_2'$.
  \end{itemize}
\end{itemize}

So, we get $n=_{AC} n'$.

Finally, let $s$ and $t$ be basic APTC with static localities terms, and $s\sim_p^{sl} t$, there are normal forms $n$ and $n'$, such that $s=n$ and $t=n'$. The soundness theorem of
parallelism modulo static location pomset bisimulation equivalence (see Theorem \ref{SPPBE}) yields $s\sim_p^{sl} n$ and $t\sim_p^{sl} n'$, so
$n\sim_p^{sl} s\sim_p^{sl} t\sim_p^{sl} n'$. Since if $n\sim_p^{sl} n'$ then $n=_{AC}n'$, $s=n=_{AC}n'=t$, as desired.
\end{proof}

\begin{theorem}[Soundness of APTC with static localities modulo static location step bisimulation equivalence]\label{SPSBE}
Let $x$ and $y$ be APTC with static localities terms. If $APTC^{sl}\vdash x=y$, then $x\sim_s^{sl} y$.
\end{theorem}

\begin{proof}
Since static location step bisimulation $\sim_s^{sl}$ is both an equivalent and a congruent relation with respect to the operators $\between$, $\parallel$, $\leftmerge$, $\mid$,
$\Theta$ and $\triangleleft$, we only need to check if each axiom in Table \ref{AxiomsForParallelism} is sound modulo static location step bisimulation equivalence, the proof is trivial
and we omit it.
\end{proof}

\begin{theorem}[Completeness of APTC with static localities modulo static location step bisimulation equivalence]\label{CPSBE}
Let $p$ and $q$ be closed APTC with static localities terms, if $p\sim_s^{sl} q$ then $p=q$.
\end{theorem}

\begin{proof}
Firstly, by the elimination theorem of APTC with static localities (see Theorem \ref{ETParallelism}), we know that for each closed APTC with static localities term $p$, there exists a
closed basic APTC with static localities term $p'$, such that $APTC^{sl}\vdash p=p'$, so, we only need to consider closed basic APTC with static localities terms.

The basic terms (see Definition \ref{BTAPTC}) modulo associativity and commutativity (AC) of conflict $+$ (defined by axioms $A1$ and $A2$ in Table \ref{AxiomsForBATC}), and these
equivalences is denoted by $=_{AC}$. Then, each equivalence class $s$ modulo AC of $+$ has the following normal form

$$s_1+\cdots+ s_k$$

with each $s_i$ either an atomic event or of the form

$$t_1\cdot\cdots\cdot t_m$$

with each $t_j$ either an atomic event or of the form

$$u_1\parallel\cdots\parallel u_l$$

with each $u_l$ an atomic event with a location, and each $s_i$ is called the summand of $s$.

Now, we prove that for normal forms $n$ and $n'$, if $n\sim_s^{sl} n'$ then $n=_{AC}n'$. It is sufficient to induct on the sizes of $n$ and $n'$.

\begin{itemize}
  \item Consider a summand $u::e$ of $n$. Then $n\xrightarrow[u]{e}\surd$, so $n\sim_s^{sl} n'$ implies $n'\xrightarrow[u]{e}\surd$, meaning that $n'$ also contains the summand $u::e$.
  \item Consider a summand $t_1\cdot t_2$ of $n$,
  \begin{itemize}
    \item if $t_1\equiv u'::e'$, then $n\xrightarrow[u']{e'}t_2$, so $n\sim_s^{sl} n'$ implies $n'\xrightarrow[u']{e'}t_2'$ with $t_2\sim_s^{sl} t_2'$, meaning that $n'$ contains a
    summand $u'::e'\cdot t_2'$. Since $t_2$ and $t_2'$ are normal forms and have sizes smaller than $n$ and $n'$, by the induction hypotheses if $t_2\sim_s^{sl} t_2'$ then
    $t_2=_{AC} t_2'$;
    \item if $t_1\equiv u_1::e_1\leftmerge\cdots\leftmerge u_l::e_l$, then $n\xrightarrow[u]{\{e_1,\cdots,e_l\}}t_2$, so $n\sim_s^{sl} n'$ implies
    $n'\xrightarrow[u]{\{e_1,\cdots,e_l\}}t_2'$ with $t_2\sim_s^{sl} t_2'$, meaning that $n'$ contains a summand $(u_1::e_1\leftmerge\cdots\leftmerge u_l::e_l)\cdot t_2'$. Since $t_2$
    and $t_2'$ are normal forms and have sizes smaller than $n$ and $n'$, by the induction hypotheses if $t_2\sim_s^{sl} t_2'$ then $t_2=_{AC} t_2'$.
  \end{itemize}
\end{itemize}

So, we get $n=_{AC} n'$.

Finally, let $s$ and $t$ be basic APTC with static localities terms, and $s\sim_s^{sl} t$, there are normal forms $n$ and $n'$, such that $s=n$ and $t=n'$. The soundness theorem of
parallelism modulo static location step bisimulation equivalence (see Theorem \ref{SPSBE}) yields $s\sim_s^{sl} n$ and $t\sim_s^{sl} n'$, so $n\sim_s^{sl} s\sim_s^{sl} t\sim_s^{sl} n'$.
Since if $n\sim_s^{sl} n'$ then $n=_{AC}n'$, $s=n=_{AC}n'=t$, as desired.
\end{proof}

\begin{theorem}[Soundness of APTC with static localities modulo static location hp-bisimulation equivalence]\label{SPHPBE}
Let $x$ and $y$ be APTC with static localities terms. If $APTC^{sl}\vdash x=y$, then $x\sim_{hp}^{sl} y$.
\end{theorem}

\begin{proof}
Since static location hp-bisimulation $\sim_{hp}^{sl}$ is both an equivalent and a congruent relation with respect to the operators $\between$, $\parallel$, $\leftmerge$, $\mid$,
$\Theta$ and $\triangleleft$, we only need to check if each axiom in Table \ref{AxiomsForParallelism} is sound modulo static location hp-bisimulation equivalence, the proof is trivial
and we omit it.
\end{proof}

\begin{theorem}[Completeness of APTC with static localities modulo static location hp-bisimulation equivalence]\label{CPHPBE}
Let $p$ and $q$ be closed APTC with static localities terms, if $p\sim_{hp}^{sl} q$ then $p=q$.
\end{theorem}

\begin{proof}
Firstly, by the elimination theorem of APTC with static localities (see Theorem \ref{ETParallelism}), we know that for each closed APTC with static localities term $p$, there exists a
closed basic APTC with static localities term $p'$, such that $APTC^{sl}\vdash p=p'$, so, we only need to consider closed basic APTC with static localities terms.

The basic terms (see Definition \ref{BTAPTC}) modulo associativity and commutativity (AC) of conflict $+$ (defined by axioms $A1$ and $A2$ in Table \ref{AxiomsForBATC}), and these
equivalences is denoted by $=_{AC}$. Then, each equivalence class $s$ modulo AC of $+$ has the following normal form

$$s_1+\cdots+ s_k$$

with each $s_i$ either an atomic event or of the form

$$t_1\cdot\cdots\cdot t_m$$

with each $t_j$ either an atomic event or of the form

$$u_1\parallel\cdots\parallel u_l$$

with each $u_l$ an atomic event with a location, and each $s_i$ is called the summand of $s$.

Now, we prove that for normal forms $n$ and $n'$, if $n\sim_{hp}^{sl} n'$ then $n=_{AC}n'$. It is sufficient to induct on the sizes of $n$ and $n'$.

\begin{itemize}
  \item Consider a summand $u::e$ of $n$. Then $n\xrightarrow[u]{e}\surd$, so $n\sim_{hp}^{sl} n'$ implies $n'\xrightarrow[u]{e}\surd$, meaning that $n'$ also contains the summand $u::e$.
  \item Consider a summand $t_1\cdot t_2$ of $n$,
  \begin{itemize}
    \item if $t_1\equiv u'::e'$, then $n\xrightarrow[u']{e'}t_2$, so $n\sim_{hp}^{sl} n'$ implies $n'\xrightarrow[u']{e'}t_2'$ with $t_2\sim_{hp}^{sl} t_2'$, meaning that $n'$ contains
    a summand $u'::e'\cdot t_2'$. Since $t_2$ and $t_2'$ are normal forms and have sizes smaller than $n$ and $n'$, by the induction hypotheses if $t_2\sim_{hp}^{sl} t_2'$ then
    $t_2=_{AC} t_2'$;
    \item if $t_1\equiv u_1::e_1\leftmerge\cdots\leftmerge u_l::e_l$, then $n\xrightarrow[u]{\{e_1,\cdots,e_l\}}t_2$, so $n\sim_{hp}^{sl} n'$ implies
    $n'\xrightarrow[u]{\{e_1,\cdots,e_l\}}t_2'$ with $t_2\sim_{hp}^{sl} t_2'$, meaning that $n'$ contains a summand $(u_1::e_1\leftmerge\cdots\leftmerge u_l::e_l)\cdot t_2'$. Since
    $t_2$ and $t_2'$ are normal forms and have sizes smaller than $n$ and $n'$, by the induction hypotheses if $t_2\sim_{hp}^{sl} t_2'$ then $t_2=_{AC} t_2'$.
  \end{itemize}
\end{itemize}

So, we get $n=_{AC} n'$.

Finally, let $s$ and $t$ be basic APTC with static localities terms, and $s\sim_{hp}^{sl} t$, there are normal forms $n$ and $n'$, such that $s=n$ and $t=n'$. The soundness theorem of
parallelism modulo static location hp-bisimulation equivalence (see Theorem \ref{SPPBE}) yields $s\sim_{hp}^{sl} n$ and $t\sim_{hp}^{sl} n'$, so
$n\sim_{hp}^{sl} s\sim_{hp}^{sl} t\sim_{hp}^{sl} n'$. Since if $n\sim_{hp}^{sl} n'$ then $n=_{AC}n'$, $s=n=_{AC}n'=t$, as desired.
\end{proof}

\begin{theorem}[Soundness of APTC with static localities modulo static location hhp-bisimulation equivalence]\label{SPHPBE}
Let $x$ and $y$ be APTC with static localities terms. If $APTC^{sl}\vdash x=y$, then $x\sim_{hhp}^{sl} y$.
\end{theorem}

\begin{proof}
Since static location hhp-bisimulation $\sim_{hhp}^{sl}$ is both an equivalent and a congruent relation with respect to the operators $\between$, $\parallel$, $\leftmerge$, $\mid$,
$\Theta$ and $\triangleleft$, we only need to check if each axiom in Table \ref{AxiomsForParallelism} is sound modulo static location hhp-bisimulation equivalence, the proof is trivial
and we omit it.
\end{proof}

\begin{theorem}[Completeness of APTC with static localities modulo static location hhp-bisimulation equivalence]\label{CPHPBE}
Let $p$ and $q$ be closed APTC with static localities terms, if $p\sim_{hhp}^{sl} q$ then $p=q$.
\end{theorem}

\begin{proof}
Firstly, by the elimination theorem of APTC with static localities (see Theorem \ref{ETParallelism}), we know that for each closed APTC with static localities term $p$, there exists a
closed basic APTC with static localities term $p'$, such that $APTC^{sl}\vdash p=p'$, so, we only need to consider closed basic APTC with static localities terms.

The basic terms (see Definition \ref{BTAPTC}) modulo associativity and commutativity (AC) of conflict $+$ (defined by axioms $A1$ and $A2$ in Table \ref{AxiomsForBATC}), and these
equivalences is denoted by $=_{AC}$. Then, each equivalence class $s$ modulo AC of $+$ has the following normal form

$$s_1+\cdots+ s_k$$

with each $s_i$ either an atomic event or of the form

$$t_1\cdot\cdots\cdot t_m$$

with each $t_j$ either an atomic event or of the form

$$u_1\parallel\cdots\parallel u_l$$

with each $u_l$ an atomic event with a location, and each $s_i$ is called the summand of $s$.

Now, we prove that for normal forms $n$ and $n'$, if $n\sim_{hhp}^{sl} n'$ then $n=_{AC}n'$. It is sufficient to induct on the sizes of $n$ and $n'$.

\begin{itemize}
  \item Consider a summand $u::e$ of $n$. Then $n\xrightarrow[u]{e}\surd$, so $n\sim_{hhp}^{sl} n'$ implies $n'\xrightarrow[u]{e}\surd$, meaning that $n'$ also contains the summand $u::e$.
  \item Consider a summand $t_1\cdot t_2$ of $n$,
  \begin{itemize}
    \item if $t_1\equiv u'::e'$, then $n\xrightarrow[u']{e'}t_2$, so $n\sim_{hhp}^{sl} n'$ implies $n'\xrightarrow[u']{e'}t_2'$ with $t_2\sim_{hhp}^{sl} t_2'$, meaning that $n'$ contains
    a summand $u'::e'\cdot t_2'$. Since $t_2$ and $t_2'$ are normal forms and have sizes smaller than $n$ and $n'$, by the induction hypotheses if $t_2\sim_{hhp}^{sl} t_2'$ then
    $t_2=_{AC} t_2'$;
    \item if $t_1\equiv u_1::e_1\leftmerge\cdots\leftmerge u_l::e_l$, then $n\xrightarrow[u]{\{e_1,\cdots,e_l\}}t_2$, so $n\sim_{hhp}^{sl} n'$ implies
    $n'\xrightarrow[u]{\{e_1,\cdots,e_l\}}t_2'$ with $t_2\sim_{hhp}^{sl} t_2'$, meaning that $n'$ contains a summand $(u_1::e_1\leftmerge\cdots\leftmerge u_l::e_l)\cdot t_2'$. Since
    $t_2$ and $t_2'$ are normal forms and have sizes smaller than $n$ and $n'$, by the induction hypotheses if $t_2\sim_{hhp}^{sl} t_2'$ then $t_2=_{AC} t_2'$.
  \end{itemize}
\end{itemize}

So, we get $n=_{AC} n'$.

Finally, let $s$ and $t$ be basic APTC with static localities terms, and $s\sim_{hhp}^{sl} t$, there are normal forms $n$ and $n'$, such that $s=n$ and $t=n'$. The soundness theorem of
parallelism modulo static location hhp-bisimulation equivalence (see Theorem \ref{SPPBE}) yields $s\sim_{hhp}^{sl} n$ and $t\sim_{hhp}^{sl} n'$, so
$n\sim_{hhp}^{sl} s\sim_{hhp}^{sl} t\sim_{hhp}^{sl} n'$. Since if $n\sim_{hhp}^{sl} n'$ then $n=_{AC}n'$, $s=n=_{AC}n'=t$, as desired.
\end{proof}


The transition rules of encapsulation operator $\partial_H$ are shown in Table \ref{TRForEncapsulation}.

\begin{center}
    \begin{table}
        $$\frac{x\xrightarrow[u]{e}\surd}{\partial_H(x)\xrightarrow[u]{e}\surd}\quad (e\notin H)\quad\quad\frac{x\xrightarrow[u]{e}x'}{\partial_H(x)\xrightarrow[u]{e}\partial_H(x')}\quad(e\notin H)$$
        \caption{Transition rules of encapsulation operator $\partial_H$}
        \label{TRForEncapsulation}
    \end{table}
\end{center}

Based on the transition rules for encapsulation operator $\partial_H$ in Table \ref{TRForEncapsulation}, we design the axioms as Table \ref{AxiomsForEncapsulation} shows.

\begin{center}
    \begin{table}
        \begin{tabular}{@{}ll@{}}
            \hline No. &Axiom\\
            $D1$ & $e\notin H\quad\partial_H(e) = e$\\
            $D2$ & $e\in H\quad \partial_H(e) = \delta$\\
            $D3$ & $\partial_H(\delta) = \delta$\\
            $D4$ & $\partial_H(x+ y) = \partial_H(x)+\partial_H(y)$\\
            $D5$ & $\partial_H(x\cdot y) = \partial_H(x)\cdot\partial_H(y)$\\
            $D6$ & $\partial_H(x\leftmerge y) = \partial_H(x)\leftmerge\partial_H(y)$\\
            $L11$ & $u::\partial_H(x) = \partial_H(u::x)$\\
        \end{tabular}
        \caption{Axioms of encapsulation operator}
        \label{AxiomsForEncapsulation}
    \end{table}
\end{center}

\begin{theorem}[Congruence theorem of encapsulation operator $\partial_H$]
Static location truly concurrent bisimulation equivalences $\sim_p^{sl}$, $\sim_s^{sl}$, $\sim_{hp}^{sl}$ and $\sim_{hhp}^{sl}$ are all congruences with respect to encapsulation
operator $\partial_H$.
\end{theorem}

\begin{proof}
It is easy to see that static location pomset, step, hp-, hhp- bisimulation are all equivalent relations on APTC with static localities terms, we only need to prove that
$\sim_p^{sl}$, $\sim_s^{sl}$, $\sim_{hp}^{sl}$ and $\sim_{hhp}^{sl}$ are all preserved by the operator $\partial_H$,
the proof is trivial and we omit it.
\end{proof}

\begin{theorem}[Elimination theorem of APTC with static localities]\label{ETEncapsulation}
Let $p$ be a closed APTC with static localities term including the encapsulation operator $\partial_H$. Then there is a basic APTC with static localities term $q$ such that
$APTC\vdash p=q$.
\end{theorem}

\begin{proof}
(1) Firstly, suppose that the following ordering on the signature of APTC with static localities is defined: $::>\leftmerge > \cdot > +$ and the symbol $::$ is given the lexicographical
status for the first argument, then for each rewrite rule $p\rightarrow q$ in Table \ref{TRSForEncapsulation} relation $p>_{lpo} q$ can easily be proved. We obtain that the term
rewrite system shown in Table \ref{TRSForEncapsulation} is strongly normalizing, for it has finitely many rewriting rules, and $>$ is a well-founded ordering on the signature of APTC
with static localities, and if $s>_{lpo} t$, for each rewriting rule $s\rightarrow t$ is in Table \ref{TRSForEncapsulation} (see Theorem \ref{SN}).

\begin{center}
    \begin{table}
        \begin{tabular}{@{}ll@{}}
            \hline No. &Rewriting Rule\\
            $RD1$ & $e\notin H\quad\partial_H(e) \rightarrow e$\\
            $RD2$ & $e\in H\quad \partial_H(e) \rightarrow \delta$\\
            $RD3$ & $\partial_H(\delta) \rightarrow \delta$\\
            $RD4$ & $\partial_H(x+ y) \rightarrow \partial_H(x)+\partial_H(y)$\\
            $RD5$ & $\partial_H(x\cdot y) \rightarrow \partial_H(x)\cdot\partial_H(y)$\\
            $RD6$ & $\partial_H(x\leftmerge y) \rightarrow \partial_H(x)\leftmerge\partial_H(y)$\\
            $RL11$ & $u::\partial_H(x) \rightarrow \partial_H(u::x)$\\
        \end{tabular}
        \caption{Term rewrite system of encapsulation operator $\partial_H$}
        \label{TRSForEncapsulation}
    \end{table}
\end{center}

(2) Then we prove that the normal forms of closed APTC with static localities terms including encapsulation operator $\partial_H$ are basic APTC with static localities terms.

Suppose that $p$ is a normal form of some closed APTC with static localities term and suppose that $p$ is not a basic APTC with static localities term. Let $p'$ denote the smallest
sub-term of $p$ which is not a basic APTC with static localities term. It implies that each sub-term of $p'$ is a basic APTC with static localities term. Then we prove that $p$ is not
a term in normal form. It is sufficient to induct on the structure of $p'$, we only prove the new case $p'\equiv \partial_H(p_1)$:

\begin{itemize}
  \item Case $p_1\equiv e$. The transition rules $RD1$ or $RD2$ can be applied, so $p$ is not a normal form;
  \item Case $p_1\equiv u::e$. The transition rules $RD1$ or $RD2$ can be applied, so $p$ is not a normal form;
  \item Case $p_1\equiv \delta$. The transition rules $RD3$ can be applied, so $p$ is not a normal form;
  \item Case $p_1\equiv p_1'+ p_1''$. The transition rules $RD4$ can be applied, so $p$ is not a normal form;
  \item Case $p_1\equiv p_1'\cdot p_1''$. The transition rules $RD5$ can be applied, so $p$ is not a normal form;
  \item Case $p_1\equiv p_1'\leftmerge p_1''$. The transition rules $RD6$ can be applied, so $p$ is not a normal form.
\end{itemize}
\end{proof}

\begin{theorem}[Soundness of APTC with static localities modulo static location pomset bisimulation equivalence]\label{SAPTCPBE}
Let $x$ and $y$ be APTC with static localities terms including encapsulation operator $\partial_H$. If $APTC^{sl}\vdash x=y$, then $x\sim_p^{sl} y$.
\end{theorem}

\begin{proof}
Since static location pomset bisimulation $\sim_p^{sl}$ is both an equivalent and a congruent relation with respect to the operator $\partial_H$, we only need to check if each axiom in
Table \ref{AxiomsForEncapsulation} is sound modulo static location pomset bisimulation equivalence, the proof is trivial and we omit it.
\end{proof}

\begin{theorem}[Completeness of APTC with static localities modulo static location pomset bisimulation equivalence]\label{CAPTCPBE}
Let $p$ and $q$ be closed APTC with static localities terms including encapsulation operator $\partial_H$, if $p\sim_p^{sl} q$ then $p=q$.
\end{theorem}

\begin{proof}
Firstly, by the elimination theorem of APTC with static localities (see Theorem \ref{ETEncapsulation}), we know that the normal form of APTC with static localities does not contain
$\partial_H$, and for each closed APTC with static localities term $p$, there exists a closed basic APTC with static localities term $p'$, such that $APTC^{sl}\vdash p=p'$, so, we only need
to consider closed basic APTC with static localities terms.

Similarly to Theorem \ref{CAPTCSBE}, we can prove that for normal forms $n$ and $n'$, if $n\sim_p^{sl} n'$ then $n=_{AC}n'$.

Finally, let $s$ and $t$ be basic APTC with static localities terms, and $s\sim_p^{sl} t$, there are normal forms $n$ and $n'$, such that $s=n$ and $t=n'$. The soundness theorem of
APTC with static localities modulo static location pomset bisimulation equivalence (see Theorem \ref{SAPTCPBE}) yields $s\sim_p^{sl} n$ and $t\sim_p^{sl} n'$, so
$n\sim_p^{sl} s\sim_p^{sl} t\sim_p^{sl} n'$. Since if $n\sim_p^{sl} n'$ then $n=_{AC}n'$, $s=n=_{AC}n'=t$, as desired.
\end{proof}

\begin{theorem}[Soundness of APTC with static localities modulo static location step bisimulation equivalence]\label{SAPTCSBE}
Let $x$ and $y$ be APTC with static localities terms including encapsulation operator $\partial_H$. If $APTC^{sl}\vdash x=y$, then $x\sim_s^{sl} y$.
\end{theorem}

\begin{proof}
Since static location step bisimulation $\sim_s^{sl}$ is both an equivalent and a congruent relation with respect to the operator $\partial_H$, we only need to check if each axiom in
Table \ref{AxiomsForEncapsulation} is sound modulo static location step bisimulation equivalence, the proof is trivial and we omit it.
\end{proof}

\begin{theorem}[Completeness of APTC with static localities modulo static location step bisimulation equivalence]\label{CAPTCSBE}
Let $p$ and $q$ be closed APTC with static localities terms including encapsulation operator $\partial_H$, if $p\sim_s^{sl} q$ then $p=q$.
\end{theorem}

\begin{proof}
Firstly, by the elimination theorem of APTC with static localities (see Theorem \ref{ETEncapsulation}), we know that the normal form of APTC with static localities does not contain
$\partial_H$, and for each closed APTC with static localities term $p$, there exists a closed basic APTC with static localities term $p'$, such that $APTC^{sl}\vdash p=p'$, so, we only
need to consider closed basic APTC with static localities terms.

Similarly to Theorem \ref{CPSBE}, we can prove that for normal forms $n$ and $n'$, if $n\sim_s^{sl} n'$ then $n=_{AC}n'$.

Finally, let $s$ and $t$ be basic APTC with static localities terms, and $s\sim_s^{sl} t$, there are normal forms $n$ and $n'$, such that $s=n$ and $t=n'$. The soundness theorem of
APTC with static localities modulo static location step bisimulation equivalence (see Theorem \ref{SAPTCSBE}) yields $s\sim_s^{sl} n$ and $t\sim_s^{sl} n'$, so
$n\sim_s^{sl} s\sim_s^{sl} t\sim_s^{sl} n'$. Since if $n\sim_s^{sl} n'$ then $n=_{AC}n'$, $s=n=_{AC}n'=t$, as desired.
\end{proof}

\begin{theorem}[Soundness of APTC with static localities modulo static location hp-bisimulation equivalence]\label{SAPTCHPBE}
Let $x$ and $y$ be APTC with static localities terms including encapsulation operator $\partial_H$. If $APTC^{sl}\vdash x=y$, then $x\sim_{hp}^{sl} y$.
\end{theorem}

\begin{proof}
Since static location hp-bisimulation $\sim_{hp}^{sl}$ is both an equivalent and a congruent relation with respect to the operator $\partial_H$, we only need to check if each axiom in
Table \ref{AxiomsForEncapsulation} is sound modulo static location hp-bisimulation equivalence, the proof is trivial and we omit it.
\end{proof}

\begin{theorem}[Completeness of APTC with static localities modulo static location hp-bisimulation equivalence]\label{CAPTCHPBE}
Let $p$ and $q$ be closed APTC with static localities terms including encapsulation operator $\partial_H$, if $p\sim_{hp}^{sl} q$ then $p=q$.
\end{theorem}

\begin{proof}
Firstly, by the elimination theorem of APTC with static localities (see Theorem \ref{ETEncapsulation}), we know that the normal form of APTC with static localities does not contain
$\partial_H$, and for each closed APTC with static localities term $p$, there exists a closed basic APTC with static localities term $p'$, such that $APTC^{sl}\vdash p=p'$, so, we only
need to consider closed basic APTC with static localities terms.

Similarly to Theorem \ref{CAPTCPBE}, we can prove that for normal forms $n$ and $n'$, if $n\sim_{hp}^{sl} n'$ then $n=_{AC}n'$.

Finally, let $s$ and $t$ be basic APTC with static localities terms, and $s\sim_{hp}^{sl} t$, there are normal forms $n$ and $n'$, such that $s=n$ and $t=n'$. The soundness theorem of
APTC with static localities modulo static location hp-bisimulation equivalence (see Theorem \ref{SAPTCHPBE}) yields $s\sim_{hp}^{sl} n$ and $t\sim_{hp}^{sl} n'$, so
$n\sim_{hp}^{sl} s\sim_{hp}^{sl} t\sim_{hp}^{sl} n'$. Since if $n\sim_{hp}^{sl} n'$ then $n=_{AC}n'$, $s=n=_{AC}n'=t$, as desired.
\end{proof}

\begin{theorem}[Soundness of APTC with static localities modulo static location hhp-bisimulation equivalence]\label{SAPTCHPBE}
Let $x$ and $y$ be APTC with static localities terms including encapsulation operator $\partial_H$. If $APTC^{sl}\vdash x=y$, then $x\sim_{hhp}^{sl} y$.
\end{theorem}

\begin{proof}
Since static location hhp-bisimulation $\sim_{hhp}^{sl}$ is both an equivalent and a congruent relation with respect to the operator $\partial_H$, we only need to check if each axiom in
Table \ref{AxiomsForEncapsulation} is sound modulo static location hhp-bisimulation equivalence, the proof is trivial and we omit it.
\end{proof}

\begin{theorem}[Completeness of APTC with static localities modulo static location hhp-bisimulation equivalence]\label{CAPTCHPBE}
Let $p$ and $q$ be closed APTC with static localities terms including encapsulation operator $\partial_H$, if $p\sim_{hhp}^{sl} q$ then $p=q$.
\end{theorem}

\begin{proof}
Firstly, by the elimination theorem of APTC with static localities (see Theorem \ref{ETEncapsulation}), we know that the normal form of APTC with static localities does not contain
$\partial_H$, and for each closed APTC with static localities term $p$, there exists a closed basic APTC with static localities term $p'$, such that $APTC^{sl}\vdash p=p'$, so, we only
need to consider closed basic APTC with static localities terms.

Similarly to Theorem \ref{CAPTCPBE}, we can prove that for normal forms $n$ and $n'$, if $n\sim_{hhp}^{sl} n'$ then $n=_{AC}n'$.

Finally, let $s$ and $t$ be basic APTC with static localities terms, and $s\sim_{hhp}^{sl} t$, there are normal forms $n$ and $n'$, such that $s=n$ and $t=n'$. The soundness theorem of
APTC with static localities modulo static location hhp-bisimulation equivalence (see Theorem \ref{SAPTCHPBE}) yields $s\sim_{hhp}^{sl} n$ and $t\sim_{hhp}^{sl} n'$, so
$n\sim_{hhp}^{sl} s\sim_{hhp}^{sl} t\sim_{hhp}^{sl} n'$. Since if $n\sim_{hhp}^{sl} n'$ then $n=_{AC}n'$, $s=n=_{AC}n'=t$, as desired.
\end{proof}

\subsubsection{Recursion}\label{rec}

In this section, we introduce recursion to capture infinite processes based on APTC with static localities. Since in APTC with static localities, there are four basic operators
$::$, $\cdot$, $+$ and $\leftmerge$, the recursion must be adapted this situation to include $\leftmerge$.

In the following, $E,F,G$ are recursion specifications, $X,Y,Z$ are recursive variables.


\begin{definition}[Recursive specification]
A recursive specification is a finite set of recursive equations

$$X_1=t_1(X_1,\cdots,X_n)$$
$$\cdots$$
$$X_n=t_n(X_1,\cdots,X_n)$$

where the left-hand sides of $X_i$ are called recursion variables, and the right-hand sides $t_i(X_1,\cdots,X_n)$ are process terms in APTC with static localities with possible
occurrences of the recursion variables $X_1,\cdots,X_n$.
\end{definition}

\begin{definition}[Solution]
Processes $p_1,\cdots,p_n$ are a solution for a recursive specification $\{X_i=t_i(X_1,\cdots,X_n)|i\in\{1,\cdots,n\}\}$ (with respect to static location truly concurrent bisimulation equivalences
$\sim_s^{sl}$($\sim_p^{sl}$, $\sim_{hp}^{sl}$, $\sim_{hhp}^{sl}$)) if $p_i\sim_s^{sl} (\sim_p^{sl}, \sim_{hp}^{sl}, \sim_{hhp}^{sl})t_i(p_1,\cdots,p_n)$ for $i\in\{1,\cdots,n\}$.
\end{definition}

\begin{definition}[Guarded recursive specification]
A recursive specification

$$X_1=t_1(X_1,\cdots,X_n)$$
$$...$$
$$X_n=t_n(X_1,\cdots,X_n)$$

is guarded if the right-hand sides of its recursive equations can be adapted to the form by applications of the axioms in APTC with static localities and replacing recursion variables
by the right-hand sides of their recursive equations,

$(u_{11}::a_{11}\leftmerge\cdots\leftmerge u_{1i_1}::a_{1i_1})\cdot s_1(X_1,\cdots,X_n)+\cdots+(u_{k1}::a_{k1}\leftmerge\cdots\leftmerge u_{ki_k}::a_{ki_k})\cdot s_k(X_1,\cdots,X_n)\\
+(v_{11}::b_{11}\leftmerge\cdots\leftmerge v_{1j_1}::b_{1j_1})+\cdots+(v_{1j_1}::b_{1j_1}\leftmerge\cdots\leftmerge v_{1j_l}::b_{lj_l})$

where $a_{11},\cdots,a_{1i_1},a_{k1},\cdots,a_{ki_k},b_{11},\cdots,b_{1j_1},b_{1j_1},\cdots,b_{lj_l}\in \mathbb{E}$, and the sum above is allowed to be empty, in which case it
represents the deadlock $\delta$.
\end{definition}

\begin{definition}[Linear recursive specification]\label{LRS}
A recursive specification is linear if its recursive equations are of the form

$(u_{11}::a_{11}\leftmerge\cdots\leftmerge u_{1i_1}::a_{1i_1})X_1+\cdots+(u_{k1}::a_{k1}\leftmerge\cdots\leftmerge u_{ki_k}::a_{ki_k})X_k\\
+(v_{11}::b_{11}\leftmerge\cdots\leftmerge v_{1j_1}::b_{1j_1})+\cdots+(v_{1j_1}::b_{1j_1}\leftmerge\cdots\leftmerge v_{1j_l}::b_{lj_l})$

where $a_{11},\cdots,a_{1i_1},a_{k1},\cdots,a_{ki_k},b_{11},\cdots,b_{1j_1},b_{1j_1},\cdots,b_{lj_l}\in \mathbb{E}$, and the sum above is allowed to be empty, in which case it
represents the deadlock $\delta$.
\end{definition}

For a guarded recursive specifications $E$ with the form

$$X_1=t_1(X_1,\cdots,X_n)$$
$$\cdots$$
$$X_n=t_n(X_1,\cdots,X_n)$$

the behavior of the solution $\langle X_i|E\rangle$ for the recursion variable $X_i$ in $E$, where $i\in\{1,\cdots,n\}$, is exactly the behavior of their right-hand sides
$t_i(X_1,\cdots,X_n)$, which is captured by the two transition rules in Table \ref{TRForGR}.

\begin{center}
    \begin{table}
        $$\frac{t_i(\langle X_1|E\rangle,\cdots,\langle X_n|E\rangle)\xrightarrow[u]{\{e_1,\cdots,e_k\}}\surd}{\langle X_i|E\rangle\xrightarrow[u]{\{e_1,\cdots,e_k\}}\surd}$$
        $$\frac{t_i(\langle X_1|E\rangle,\cdots,\langle X_n|E\rangle)\xrightarrow[u]{\{e_1,\cdots,e_k\}} y}{\langle X_i|E\rangle\xrightarrow[u]{\{e_1,\cdots,e_k\}} y}$$
        \caption{Transition rules of guarded recursion}
        \label{TRForGR}
    \end{table}
\end{center}

\begin{theorem}[Conservitivity of APTC with static localities and guarded recursion]
APTC with static localities and guarded recursion is a conservative extension of APTC with static localities.
\end{theorem}

\begin{proof}
Since the transition rules of APTC with static localities are source-dependent, and the transition rules for guarded recursion in Table \ref{TRForGR} contain only a fresh constant in
their source, so the transition rules of APTC with static localities and guarded recursion are a conservative extension of those of APTC with static localities.
\end{proof}

\begin{theorem}[Congruence theorem of APTC with static localities and guarded recursion]
Static location truly concurrent bisimulation equivalences $\sim_p^{sl}$, $\sim_s^{sl}$, $\sim_{hp}^{sl}$ and $\sim_{hhp}^{sl}$ are all congruences with respect to APTC with static localities and guarded
recursion.
\end{theorem}

\begin{proof}
It follows the following two facts:
\begin{enumerate}
  \item in a guarded recursive specification, right-hand sides of its recursive equations can be adapted to the form by applications of the axioms in APTC with static localities and
  replacing recursion variables by the right-hand sides of their recursive equations;
  \item static location truly concurrent bisimulation equivalences $\sim_p^{sl}$, $\sim_s^{sl}$, $\sim_{hp}^{sl}$ and $\sim_{hhp}^{sl}$ are all congruences with respect to all operators of APTC with
  static localities.
\end{enumerate}
\end{proof}


The $RDP$ (Recursive Definition Principle) and the $RSP$ (Recursive Specification Principle) are shown in Table \ref{RDPRSP}.

\begin{center}
\begin{table}
  \begin{tabular}{@{}ll@{}}
\hline No. &Axiom\\
  $RDP$ & $\langle X_i|E\rangle = t_i(\langle X_1|E\rangle,\cdots,\langle X_n|E\rangle)\quad (i\in\{1,\cdots,n\})$\\
  $RSP$ & if $y_i=t_i(y_1,\cdots,y_n)$ for $i\in\{1,\cdots,n\}$, then $y_i=\langle X_i|E\rangle \quad(i\in\{1,\cdots,n\})$\\
\end{tabular}
\caption{Recursive definition and specification principle}
\label{RDPRSP}
\end{table}
\end{center}

\begin{theorem}[Elimination theorem of APTC with static localities and linear recursion]\label{ETRecursion}
Each process term in APTC with static localities and linear recursion is equal to a process term $\langle X_1|E\rangle$ with $E$ a linear recursive specification.
\end{theorem}

\begin{proof}
By applying structural induction with respect to term size, each process term $t_1$ in APTC with static localities and linear recursion generates a process can be expressed in the form of equations

$t_i=(u_{i11}::a_{i11}\leftmerge\cdots\leftmerge u_{i1i_1}::a_{i1i_1})t_{i1}+\cdots+(u_{ik_i1}::a_{ik_i1}\leftmerge\cdots\leftmerge u_{ik_ii_k}::a_{ik_ii_k})t_{ik_i}\\
+(v_{i11}::b_{i11}\leftmerge\cdots\leftmerge v_{i1i_1}::b_{i1i_1})+\cdots+(v_{i1_i1}::b_{il_i1}\leftmerge\cdots\leftmerge v_{i1_ii_l}::b_{il_ii_l})$

for $i\in\{1,\cdots,n\}$. Let the linear recursive specification $E$ consist of the recursive equations

$X_i=(u_{i11}::a_{i11}\leftmerge\cdots\leftmerge u_{i1i_1}::a_{i1i_1})X_{i1}+\cdots+(u_{ik_i1}::a_{ik_i1}\leftmerge\cdots\leftmerge u_{ik_ii_k}::a_{ik_ii_k})X_{ik_i}\\
+(v_{i11}::b_{i11}\leftmerge\cdots\leftmerge v_{i1i_1}::b_{i1i_1})+\cdots+(v_{i1_i1}::b_{il_i1}\leftmerge\cdots\leftmerge v_{i1_ii_l}::b_{il_ii_l})$

for $i\in\{1,\cdots,n\}$. Replacing $X_i$ by $t_i$ for $i\in\{1,\cdots,n\}$ is a solution for $E$, $RSP$ yields $t_1=\langle X_1|E\rangle$.
\end{proof}

\begin{theorem}[Soundness of APTC with static localities and guarded recursion]\label{SAPTCR}
Let $x$ and $y$ be APTC with static localities and guarded recursion terms. If $APTC\textrm{ with guarded recursion}\vdash x=y$, then
\begin{enumerate}
  \item $x\sim_s^{sl} y$;
  \item $x\sim_p^{sl} y$;
  \item $x\sim_{hp}^{sl} y$;
  \item $x\sim_{hhp}^{sl} y$.
\end{enumerate}
\end{theorem}

\begin{proof}
Since $\sim_p^{sl}$, $\sim_s^{sl}$, $\sim_{hp}^{sl}$, and $\sim_{hhp}^{sl}$ are all both equivalent and congruent relations, we only need to check if each axiom in
Table \ref{RDPRSP} is sound modulo $\sim_p^{sl}$, $\sim_s^{sl}$, $\sim_{hp}^{sl}$, and $\sim_{hhp}^{sl}$, the proof is trivial and we omit it.
\end{proof}

\begin{theorem}[Completeness of APTC with static localities and linear recursion]\label{CAPTCR}
Let $p$ and $q$ be closed APTC with static localities and linear recursion terms, then,
\begin{enumerate}
  \item if $p\sim_s^{sl} q$ then $p=q$;
  \item if $p\sim_p^{sl} q$ then $p=q$;
  \item if $p\sim_{hp}^{sl} q$ then $p=q$;
  \item if $p\sim_{hhp}^{sl} q$ then $p=q$.
\end{enumerate}
\end{theorem}

\begin{proof}
Firstly, by the elimination theorem of APTC with static localities and guarded recursion (see Theorem \ref{ETRecursion}), we know that each process term in APTC with static localities and linear recursion is equal to a process term $\langle X_1|E\rangle$ with $E$ a linear recursive specification.

It remains to prove the following cases.

(1) If $\langle X_1|E_1\rangle \sim_s^{sl} \langle Y_1|E_2\rangle$ for linear recursive specification $E_1$ and $E_2$, then $\langle X_1|E_1\rangle = \langle Y_1|E_2\rangle$.

Let $E_1$ consist of recursive equations $X=t_X$ for $X\in \mathcal{X}$ and $E_2$
consists of recursion equations $Y=t_Y$ for $Y\in\mathcal{Y}$. Let the linear recursive specification $E$ consist of recursion equations $Z_{XY}=t_{XY}$, and
$\langle X|E_1\rangle\sim_s^{sl}\langle Y|E_2\rangle$, and $t_{XY}$ consists of the following summands:

\begin{enumerate}
  \item $t_{XY}$ contains a summand $(u_1::a_1\leftmerge\cdots\leftmerge u_m::a_m)Z_{X'Y'}$ iff $t_X$ contains the summand $(u_1::a_1\leftmerge\cdots\leftmerge u_m::a_m)X'$ and $t_Y$
  contains the summand $(u_1::a_1\leftmerge\cdots\leftmerge u_m::a_m)Y'$ such that $\langle X'|E_1\rangle\sim_s^{sl}\langle Y'|E_2\rangle$;
  \item $t_{XY}$ contains a summand $v_1::b_1\leftmerge\cdots\leftmerge v_n::b_n$ iff $t_X$ contains the summand $v_1::b_1\leftmerge\cdots\leftmerge v_n::b_n$ and $t_Y$ contains the
  summand $v_1::b_1\leftmerge\cdots\leftmerge v_n::b_n$.
\end{enumerate}

Let $\sigma$ map recursion variable $X$ in $E_1$ to $\langle X|E_1\rangle$, and let $\psi$ map recursion variable $Z_{XY}$ in $E$ to $\langle X|E_1\rangle$. So,
$\sigma((u_1::a_1\leftmerge\cdots\leftmerge u_m::a_m)X')\equiv(u_1::a_1\leftmerge\cdots\leftmerge u_m::a_m)\langle X'|E_1\rangle\equiv\psi((u_1::a_1\leftmerge\cdots\leftmerge u_m::a_m)Z_{X'Y'})$,
so by $RDP$, we get $\langle X|E_1\rangle=\sigma(t_X)=\psi(t_{XY})$. Then by $RSP$, $\langle X|E_1\rangle=\langle Z_{XY}|E\rangle$, particularly,
$\langle X_1|E_1\rangle=\langle Z_{X_1Y_1}|E\rangle$. Similarly, we can obtain $\langle Y_1|E_2\rangle=\langle Z_{X_1Y_1}|E\rangle$. Finally,
$\langle X_1|E_1\rangle=\langle Z_{X_1Y_1}|E\rangle=\langle Y_1|E_2\rangle$, as desired.

(2) If $\langle X_1|E_1\rangle \sim_p^{sl} \langle Y_1|E_2\rangle$ for linear recursive specification $E_1$ and $E_2$, then $\langle X_1|E_1\rangle = \langle Y_1|E_2\rangle$.

It can be proven similarly to (1), we omit it.

(3) If $\langle X_1|E_1\rangle \sim_{hp}^{sl} \langle Y_1|E_2\rangle$ for linear recursive specification $E_1$ and $E_2$, then $\langle X_1|E_1\rangle = \langle Y_1|E_2\rangle$.

It can be proven similarly to (1), we omit it.

(4) If $\langle X_1|E_1\rangle \sim_{hhp}^{sl} \langle Y_1|E_2\rangle$ for linear recursive specification $E_1$ and $E_2$, then $\langle X_1|E_1\rangle = \langle Y_1|E_2\rangle$.

It can be proven similarly to (1), we omit it.
\end{proof}


In this subsection, we introduce approximation induction principle ($AIP$) and try to explain that $AIP$ is still valid. $AIP$ can be used to try and equate static location truly concurrent bisimilar
guarded recursive specifications. $AIP$ says that if two process terms are truly concurrent bisimilar up to any finite depth, then they are static location truly concurrent bisimilar.

Also, we need the auxiliary unary projection operator $\Pi_n$ for $n\in\mathbb{N}$ and $\mathbb{N}\triangleq\{0,1,2,\cdots\}$. The transition rules of $\Pi_n$ are expressed in Table
\ref{TRForProjection}.

\begin{center}
    \begin{table}
        $$\frac{x\xrightarrow[u]{\{e_1,\cdots,e_k\}}\surd}{\Pi_{n+1}(x)\xrightarrow[u]{\{e_1,\cdots,e_k\}}\surd}
        \quad\frac{x\xrightarrow[u]{\{e_1,\cdots,e_k\}}x'}{\Pi_{n+1}(x)\xrightarrow[u]{\{e_1,\cdots,e_k\}}\Pi_n(x')}$$
        \caption{Transition rules of projection operator $\Pi_n$}
        \label{TRForProjection}
    \end{table}
\end{center}

Based on the transition rules for projection operator $\Pi_n$ in Table \ref{TRForProjection}, we design the axioms as Table \ref{AxiomsForProjection} shows.

\begin{center}
    \begin{table}
        \begin{tabular}{@{}ll@{}}
            \hline No. &Axiom\\
            $PR1$ & $\Pi_n(x+y)=\Pi_n(x)+\Pi_n(y)$\\
            $PR2$ & $\Pi_n(x\leftmerge y)=\Pi_n(x)\leftmerge \Pi_n(y)$\\
            $PR3$ & $\Pi_{n+1}(e_1\leftmerge\cdots\leftmerge e_k)=e_1\leftmerge\cdots\leftmerge e_k$\\
            $PR4$ & $\Pi_{n+1}((e_1\leftmerge\cdots\leftmerge e_k)\cdot x)=(e_1\leftmerge\cdots\leftmerge e_k)\cdot\Pi_n(x)$\\
            $PR5$ & $\Pi_0(x)=\delta$\\
            $PR6$ & $\Pi_n(\delta)=\delta$\\
            $L12$ & $u::\Pi_n(x) = \Pi_n(u::x)$\\
        \end{tabular}
        \caption{Axioms of projection operator}
        \label{AxiomsForProjection}
    \end{table}
\end{center}

\begin{theorem}[Conservativity of APTC with static localities and projection operator and guarded recursion]
APTC with static localities and projection operator and guarded recursion is a conservative extension of APTC with static localities and guarded recursion.
\end{theorem}

\begin{proof}
It follows from the following two facts.

\begin{enumerate}
  \item The transition rules of APTC with static localities and guarded recursion are all source-dependent;
  \item The sources of the transition rules for the projection operator contain an occurrence of $\Pi_n$.
\end{enumerate}
\end{proof}

\begin{theorem}[Congruence theorem of projection operator $\Pi_n$]
Static location truly concurrent bisimulation equivalences $\sim_p^{sl}$, $\sim_s^{sl}$, $\sim_{hp}^{sl}$ and $\sim_{hhp}^{sl}$ are all congruences with respect to projection operator
$\Pi_n$.
\end{theorem}

\begin{proof}
It is easy to see that static location pomset, step, hp-, hhp- bisimulation are all equivalent relations on APTC with static localities and guarded recursion terms, we only need to prove that
$\sim_p^{sl}$, $\sim_s^{sl}$, $\sim_{hp}^{sl}$ and $\sim_{hhp}^{sl}$ are all preserved by the operator $\Pi_n$,
the proof is trivial and we omit it.
\end{proof}

\begin{theorem}[Elimination theorem of APTC with static localities and linear recursion and projection operator]\label{ETProjection}
Each process term in APTC with static localities and linear recursion and projection operator is equal to a process term $\langle X_1|E\rangle$ with $E$ a linear recursive specification.
\end{theorem}

\begin{proof}
By applying structural induction with respect to term size, each process term $t_1$ in APTC with static localities and linear recursion and projection operator $\Pi_n$ generates a
process can be expressed in the form of equations

$t_i=(u_{i11}::a_{i11}\leftmerge\cdots\leftmerge u_{i1i_1}::a_{i1i_1})t_{i1}+\cdots+(u_{ik_i1}::a_{ik_i1}\leftmerge\cdots\leftmerge u_{ik_ii_k}::a_{ik_ii_k})t_{ik_i}\\
+(v_{i11}::b_{i11}\leftmerge\cdots\leftmerge v_{i1i_1}::b_{i1i_1})+\cdots+(v_{i1_i1}::b_{il_i1}\leftmerge\cdots\leftmerge v_{i1_ii_l}::b_{il_ii_l})$

for $i\in\{1,\cdots,n\}$. Let the linear recursive specification $E$ consist of the recursive equations

$X_i=(u_{i11}::a_{i11}\leftmerge\cdots\leftmerge u_{i1i_1}::a_{i1i_1})X_{i1}+\cdots+(u_{ik_i1}::a_{ik_i1}\leftmerge\cdots\leftmerge u_{ik_ii_k}::a_{ik_ii_k})X_{ik_i}\\
+(v_{i11}::b_{i11}\leftmerge\cdots\leftmerge v_{i1i_1}::b_{i1i_1})+\cdots+(v_{i1_i1}::b_{il_i1}\leftmerge\cdots\leftmerge v_{i1_ii_l}::b_{il_ii_l})$

for $i\in\{1,\cdots,n\}$. Replacing $X_i$ by $t_i$ for $i\in\{1,\cdots,n\}$ is a solution for $E$, $RSP$ yields $t_1=\langle X_1|E\rangle$.

That is, in $E$, there is not the occurrence of projection operator $\Pi_n$.
\end{proof}

\begin{theorem}[Soundness of APTC with static localities and projection operator and guarded recursion]\label{SAPTCR}
Let $x$ and $y$ be APTC with static localities and projection operator and guarded recursion terms. If APTC with static localities and projection operator and guarded recursion $\vdash x=y$, then
\begin{enumerate}
  \item $x\sim_s^{sl} y$;
  \item $x\sim_p^{sl} y$;
  \item $x\sim_{hp}^{sl} y$;
  \item $x\sim_{hhp}^{sl} y$.
\end{enumerate}
\end{theorem}

\begin{proof}
Since $\sim_p^{sl}$, $\sim_s^{sl}$, $\sim_{hp}^{sl}$, and $\sim_{hhp}^{sl}$ are all both equivalent and congruent relations, we only need to check if each axiom in
Table \ref{AxiomsForProjection} is sound modulo $\sim_p^{sl}$, $\sim_s^{sl}$, $\sim_{hp}^{sl}$, and $\sim_{hhp}^{sl}$, the proof is trivial and we omit it.
\end{proof}

Then $AIP$ is given in Table \ref{AIP}.

\begin{center}
    \begin{table}
        \begin{tabular}{@{}ll@{}}
            \hline No. &Axiom\\
            $AIP$ & if $\Pi_n(x)=\Pi_n(y)$ for $n\in\mathbb{N}$, then $x=y$\\
        \end{tabular}
        \caption{$AIP$}
        \label{AIP}
    \end{table}
\end{center}

\begin{theorem}[Soundness of $AIP$]\label{SAIP}
Let $x$ and $y$ be APTC with static localities and projection operator and guarded recursion terms.

\begin{enumerate}
  \item If $\Pi_n(x)\sim_s^{sl}\Pi_n(y)$ for $n\in\mathbb{N}$, then $x\sim_s^{sl} y$;
  \item If $\Pi_n(x)\sim_p^{sl}\Pi_n(y)$ for $n\in\mathbb{N}$, then $x\sim_p^{sl} y$;
  \item If $\Pi_n(x)\sim_{hp}^{sl}\Pi_n(y)$ for $n\in\mathbb{N}$, then $x\sim_{hp}^{sl} y$;
  \item If $\Pi_n(x)\sim_{hhp}^{sl}\Pi_n(y)$ for $n\in\mathbb{N}$, then $x\sim_{hhp}^{sl} y$.
\end{enumerate}
\end{theorem}

\begin{proof}
(1) If $\Pi_n(x)\sim_s^{sl}\Pi_n(y)$ for $n\in\mathbb{N}$, then $x\sim_s^{sl} y$.

Since static location step bisimulation $\sim_s^{sl}$ is both an equivalent and a congruent relation with respect to APTC with static localities and guarded recursion and projection
operator, we only need to check if $AIP$ in Table \ref{AIP} is sound modulo static location step bisimulation equivalence.

Let $p,p_0$ and $q,q_0$ be closed APTC with static localities and projection operator and guarded recursion terms such that $\Pi_n(p_0)\sim_s^{sl}\Pi_n(q_0)$ for $n\in\mathbb{N}$. We
define a relation $R$ such that $p R q$ iff $\Pi_n(p)\sim_s^{sl} \Pi_n(q)$. Obviously, $p_0 R q_0$, next, we prove that $R\in\sim_s^{sl}$.

Let $p R q$ and $p\xrightarrow[u]{\{e_1,\cdots,e_k\}}\surd$, then $\Pi_1(p)\xrightarrow[u]{\{e_1,\cdots,e_k\}}\surd$, $\Pi_1(p)\sim_s^{sl}\Pi_1(q)$ yields
$\Pi_1(q)\xrightarrow[u]{\{e_1,\cdots,e_k\}}\surd$. Similarly, $q\xrightarrow[u]{\{e_1,\cdots,e_k\}}\surd$ implies $p\xrightarrow[u]{\{e_1,\cdots,e_k\}}\surd$.

Let $p R q$ and $p\xrightarrow[u]{\{e_1,\cdots,e_k\}}p'$. We define the set of process terms

$$S_n\triangleq\{q'|q\xrightarrow[u]{\{e_1,\cdots,e_k\}}q'\textrm{ and }\Pi_n(p')\sim_s^{sl}\Pi_n(q')\}$$

\begin{enumerate}
  \item Since $\Pi_{n+1}(p)\sim_s^{sl}\Pi_{n+1}(q)$ and $\Pi_{n+1}(p)\xrightarrow[u]{\{e_1,\cdots,e_k\}}\Pi_n(p')$, there exist $q'$ such that
  $\Pi_{n+1}(q)\xrightarrow[u]{\{e_1,\cdots,e_k\}}\Pi_n(q')$ and $\Pi_{n}(p')\sim_s^{sl}\Pi_{n}(q')$. So, $S_n$ is not empty.
  \item There are only finitely many $q'$ such that $q\xrightarrow[u]{\{e_1,\cdots,e_k\}}q'$, so, $S_n$ is finite.
  \item $\Pi_{n+1}(p)\sim_s^{sl}\Pi_{n+1}(q)$ implies $\Pi_{n}(p')\sim_s^{sl}\Pi_{n}(q')$, so $S_n\supseteq S_{n+1}$.
\end{enumerate}

So, $S_n$ has a non-empty intersection, and let $q'$ be in this intersection, then $q\xrightarrow[u]{\{e_1,\cdots,e_k\}}q'$ and $\Pi_n(p')\sim_s^{sl}\Pi_n(q')$, so $p' R q'$.
Similarly, let $p\mathcal{q}q$, we can obtain $q\xrightarrow[u]{\{e_1,\cdots,e_k\}}q'$ implies $p\xrightarrow[u]{\{e_1,\cdots,e_k\}}p'$ such that $p' R q'$.

Finally, $R\in\sim_s^{sl}$ and $p_0\sim_s^{sl} q_0$, as desired.

(2) If $\Pi_n(x)\sim_p^{sl}\Pi_n(y)$ for $n\in\mathbb{N}$, then $x\sim_p^{sl} y$.

Similarly to the proof of soundness of $AIP$ modulo static location step bisimulation equivalence (1), we can prove that $AIP$ in Table \ref{AIP} is sound modulo static location pomset
bisimulation equivalence, we omit it.

(3) If $\Pi_n(x)\sim_{hp}^{sl}\Pi_n(y)$ for $n\in\mathbb{N}$, then $x\sim_{hp}^{sl} y$.

Similarly to the proof of soundness of $AIP$ modulo static location pomset bisimulation equivalence (2), we can prove that $AIP$ in Table \ref{AIP} is sound modulo static location
hp-bisimulation equivalence, we omit it.

(4) If $\Pi_n(x)\sim_{hhp}^{sl}\Pi_n(y)$ for $n\in\mathbb{N}$, then $x\sim_{hhp}^{sl} y$.

Similarly to the proof of soundness of $AIP$ modulo static location hp-bisimulation equivalence (3), we can prove that $AIP$ in Table \ref{AIP} is sound modulo static location
hhp-bisimulation equivalence, we omit it.
\end{proof}

\begin{theorem}[Completeness of $AIP$]\label{CAIP}
Let $p$ and $q$ be closed APTC with static localities and linear recursion and projection operator terms, then,
\begin{enumerate}
  \item if $p\sim_s^{sl} q$ then $\Pi_n(p)=\Pi_n(q)$;
  \item if $p\sim_p^{sl} q$ then $\Pi_n(p)=\Pi_n(q)$;
  \item if $p\sim_{hp}^{sl} q$ then $\Pi_n(p)=\Pi_n(q)$;
  \item if $p\sim_{hhp}^{sl} q$ then $\Pi_n(p)=\Pi_n(q)$.
\end{enumerate}
\end{theorem}

\begin{proof}
Firstly, by the elimination theorem of APTC with static localities and guarded recursion and projection operator (see Theorem \ref{ETProjection}), we know that each process term in
APTC with static localities and linear recursion and projection operator is equal to a process term $\langle X_1|E\rangle$ with $E$ a linear recursive specification:

$X_i=(u_{i11}::a_{i11}\leftmerge\cdots\leftmerge u_{i1i_1}::a_{i1i_1})X_{i1}+\cdots+(u_{ik_i1}::a_{ik_i1}\leftmerge\cdots\leftmerge u_{ik_ii_k}::a_{ik_ii_k})X_{ik_i}\\
+(v_{i11}::b_{i11}\leftmerge\cdots\leftmerge v_{i1i_1}::b_{i1i_1})+\cdots+(v_{i1_i1}::b_{il_i1}\leftmerge\cdots\leftmerge v_{i1_ii_l}::b_{il_ii_l})$

for $i\in\{1,\cdots,n\}$.

It remains to prove the following cases.

(1) if $p\sim_s^{sl} q$ then $\Pi_n(p)=\Pi_n(q)$.

Let $p\sim_s^{sl} q$, and fix an $n\in\mathbb{N}$, there are $p',q'$ in basic APTC with static localities terms such that $p'=\Pi_n(p)$ and $q'=\Pi_n(q)$. Since $\sim_s^{sl}$ is a
congruence with respect to APTC with static localities, if $p\sim_s^{sl} q$ then $\Pi_n(p)\sim_s^{sl}\Pi_n(q)$. The soundness theorem yields
$p'\sim_s^{sl}\Pi_n(p)\sim_s^{sl}\Pi_n(q)\sim_s^{sl} q'$. Finally, the completeness of APTC with static localities modulo $\sim_s^{sl}$ ensures $p'=q'$,
and $\Pi_n(p)=p'=q'=\Pi_n(q)$, as desired.

(2) if $p\sim_p^{sl} q$ then $\Pi_n(p)=\Pi_n(q)$.

Let $p\sim_p^{sl} q$, and fix an $n\in\mathbb{N}$, there are $p',q'$ in basic APTC with static localities terms such that $p'=\Pi_n(p)$ and $q'=\Pi_n(q)$. Since $\sim_p^{sl}$ is a
congruence with respect to APTC with static localities, if $p\sim_p^{sl} q$ then $\Pi_n(p)\sim_p^{sl}\Pi_n(q)$. The soundness theorem yields
$p'\sim_p^{sl}\Pi_n(p)\sim_p^{sl}\Pi_n(q)\sim_p^{sl} q'$. Finally, the completeness of APTC with static localities modulo $\sim_p^{sl}$ ensures $p'=q'$,
and $\Pi_n(p)=p'=q'=\Pi_n(q)$, as desired.

(3) if $p\sim_{hp}^{sl} q$ then $\Pi_n(p)=\Pi_n(q)$.

Let $p\sim_{hp}^{sl} q$, and fix an $n\in\mathbb{N}$, there are $p',q'$ in basic APTC with static localities terms such that $p'=\Pi_n(p)$ and $q'=\Pi_n(q)$. Since $\sim_{hp}^{sl}$ is
a congruence with respect to APTC with static localities, if $p\sim_{hp}^{sl} q$ then $\Pi_n(p)\sim_{hp}^{sl}\Pi_n(q)$. The soundness theorem yields
$p'\sim_{hp}^{sl}\Pi_n(p)\sim_{hp}^{sl}\Pi_n(q)\sim_{hp}^{sl} q'$. Finally, the completeness of APTC with static localities modulo $\sim_{hp}^{sl}$
ensures $p'=q'$, and $\Pi_n(p)=p'=q'=\Pi_n(q)$, as desired.

(4) if $p\sim_{hhp}^{sl} q$ then $\Pi_n(p)=\Pi_n(q)$.

Let $p\sim_{hhp}^{sl} q$, and fix an $n\in\mathbb{N}$, there are $p',q'$ in basic APTC with static localities terms such that $p'=\Pi_n(p)$ and $q'=\Pi_n(q)$. Since $\sim_{hhp}^{sl}$ is
a congruence with respect to APTC with static localities, if $p\sim_{hhp}^{sl} q$ then $\Pi_n(p)\sim_{hhp}^{sl}\Pi_n(q)$. The soundness theorem yields
$p'\sim_{hhp}^{sl}\Pi_n(p)\sim_{hhp}^{sl}\Pi_n(q)\sim_{hhp}^{sl} q'$. Finally, the completeness of APTC with static localities modulo $\sim_{hhp}^{sl}$
ensures $p'=q'$, and $\Pi_n(p)=p'=q'=\Pi_n(q)$, as desired.
\end{proof}

\subsubsection{Abstraction}

To abstract away from the internal implementations of a program, and verify that the program exhibits the desired external behaviors, the silent step $\tau$ (and making $\tau$ distinct
by $\tau^e$) and abstraction operator $\tau_I$ are introduced, where $I\subseteq \mathbb{E}$ denotes the internal events. The silent step $\tau$ represents the internal
events, when we consider the external behaviors of a process, $\tau$ events can be removed, that is, $\tau$ events must keep silent. The transition rule of $\tau$ is shown in Table
\ref{TRForTau}. In the following, let the atomic event $e$ range over $\mathbb{E}\cup\{\delta\}\cup\{\tau\}$, and let the communication function
$\gamma:\mathbb{E}\cup\{\tau\}\times \mathbb{E}\cup\{\tau\}\rightarrow \mathbb{E}\cup\{\delta\}$, with each communication involved $\tau$ resulting into $\delta$.

\begin{center}
    \begin{table}
        $$\frac{}{\tau\xrightarrow{\tau}\surd}$$
        \caption{Transition rule of the silent step}
        \label{TRForTau}
    \end{table}
\end{center}


The silent step $\tau$ as an atomic event, is introduced into $E$. Considering the recursive specification $X=\tau X$, $\tau s$, $\tau\tau s$, and $\tau\cdots s$ are all its solutions, that is, the solutions make the existence of $\tau$-loops which cause unfairness. To prevent $\tau$-loops, we extend the definition of linear recursive specification (Definition \ref{LRS}) to the guarded one.

\begin{definition}[Guarded linear recursive specification]\label{GLRS}
A recursive specification is linear if its recursive equations are of the form

$(u_{11}::a_{11}\leftmerge\cdots\leftmerge u_{1i_1}::a_{1i_1})X_1+\cdots+(u_{k1}::a_{k1}\leftmerge\cdots\leftmerge u_{ki_k}::a_{ki_k})X_k\\
+(v_{11}::b_{11}\leftmerge\cdots\leftmerge v_{1j_1}::b_{1j_1})+\cdots+(v_{1j_1}::b_{1j_1}\leftmerge\cdots\leftmerge v_{1j_l}::b_{lj_l})$

where $a_{11},\cdots,a_{1i_1},a_{k1},\cdots,a_{ki_k},b_{11},\cdots,b_{1j_1},b_{1j_1},\cdots,b_{lj_l}\in \mathbb{E}\cup\{\tau\}$, and the sum above is allowed to be empty, in which case
it represents the deadlock $\delta$.

A linear recursive specification $E$ is guarded if there does not exist an infinite sequence of $\tau$-transitions
$\langle X|E\rangle\xrightarrow{\tau}\langle X'|E\rangle\xrightarrow{\tau}\langle X''|E\rangle\xrightarrow{\tau}\cdots$.
\end{definition}

\begin{theorem}[Conservitivity of APTC with static localities and silent step and guarded linear recursion]
APTC with static localities and silent step and guarded linear recursion is a conservative extension of APTC with static localities and linear recursion.
\end{theorem}

\begin{proof}
Since the transition rules of APTC with static localities and linear recursion are source-dependent, and the transition rules for silent step in Table \ref{TRForTau} contain only a
fresh constant $\tau$ in their source, so the transition rules of APTC with static localities and silent step and guarded linear recursion is a conservative extension of those of
APTC with static localities and linear recursion.
\end{proof}

\begin{theorem}[Congruence theorem of APTC with static localities and silent step and guarded linear recursion]
Rooted branching static location truly concurrent bisimulation equivalences $\approx_{rbp}^{sl}$, $\approx_{rbs}^{sl}$ and $\approx_{rbhp}^{sl}$ are all congruences with respect to
APTC with static localities and silent step and guarded linear recursion.
\end{theorem}

\begin{proof}
It follows the following three facts:
\begin{enumerate}
  \item in a guarded linear recursive specification, right-hand sides of its recursive equations can be adapted to the form by applications of the axioms in APTC with static localities
  and replacing recursion variables by the right-hand sides of their recursive equations;
  \item static location truly concurrent bisimulation equivalences $\sim_p^{sl}$, $\sim_s^{sl}$, $\sim_{hp}^{sl}$ and $\sim_{hhp}^{sl}$ are all congruences with respect to all operators of APTC with
  static localities, while static location truly concurrent bisimulation equivalences $\sim_p^{sl}$, $\sim_s^{sl}$, $\sim_{hp}^{sl}$ and $\sim_{hhp}^{sl}$ imply the corresponding rooted branching static
  location truly concurrent bisimulations $\approx_{rbp}^{sl}$, $\approx_{rbs}^{sl}$, $\approx_{rbhp}^{sl}$ and $\approx_{rbhhp}^{sl}$, so rooted branching static location truly concurrent bisimulations
  $\approx_{rbp}^{sl}$, $\approx_{rbs}^{sl}$, $\approx_{rbhp}^{sl}$ and $\approx_{rbhhp}^{sl}$ are all congruences with respect to all operators of APTC with static localities;
  \item While $\mathbb{E}$ is extended to $\mathbb{E}\cup\{\tau\}$, it can be proved that rooted branching static location truly concurrent bisimulations $\approx_{rbp}^{sl}$,
  $\approx_{rbs}^{sl}$, $\approx_{rbhp}^{sl}$ and $\approx_{rbhhp}^{sl}$ are all congruences with respect to all operators of APTC with static localities, we omit it.
\end{enumerate}
\end{proof}


We design the axioms for the silent step $\tau$ in Table \ref{AxiomsForTau}.

\begin{center}
\begin{table}
  \begin{tabular}{@{}ll@{}}
\hline No. &Axiom\\
  $B1$ & $e\cdot\tau=e$\\
  $B2$ & $e\cdot(\tau\cdot(x+y)+x)=e\cdot(x+y)$\\
  $B3$ & $x\leftmerge\tau=x$\\
  $L13$ & $u::\tau=\tau$\
\end{tabular}
\caption{Axioms of silent step}
\label{AxiomsForTau}
\end{table}
\end{center}

\begin{theorem}[Elimination theorem of APTC with static localities and silent step and guarded linear recursion]\label{ETTau}
Each process term in APTC with static localities and silent step and guarded linear recursion is equal to a process term $\langle X_1|E\rangle$ with $E$ a guarded linear recursive
specification.
\end{theorem}

\begin{proof}
By applying structural induction with respect to term size, each process term $t_1$ in APTC with static localities and silent step and guarded linear recursion generates a process can
be expressed in the form of equations

$t_i=(u_{i11}::a_{i11}\leftmerge\cdots\leftmerge u_{i1i_1}::a_{i1i_1})t_{i1}+\cdots+(u_{ik_i1}::a_{ik_i1}\leftmerge\cdots\leftmerge u_{ik_ii_k}::a_{ik_ii_k})t_{ik_i}\\
+(v_{i11}::b_{i11}\leftmerge\cdots\leftmerge v_{i1i_1}::b_{i1i_1})+\cdots+(v_{i1_i1}::b_{il_i1}\leftmerge\cdots\leftmerge v_{i1_ii_l}::b_{il_ii_l})$

for $i\in\{1,\cdots,n\}$. Let the linear recursive specification $E$ consist of the recursive equations

$X_i=(u_{i11}::a_{i11}\leftmerge\cdots\leftmerge u_{i1i_1}::a_{i1i_1})X_{i1}+\cdots+(u_{ik_i1}::a_{ik_i1}\leftmerge\cdots\leftmerge u_{ik_ii_k}::a_{ik_ii_k})X_{ik_i}\\
+(v_{i11}::b_{i11}\leftmerge\cdots\leftmerge v_{i1i_1}::b_{i1i_1})+\cdots+(v_{i1_i1}::b_{il_i1}\leftmerge\cdots\leftmerge v_{i1_ii_l}::b_{il_ii_l})$

for $i\in\{1,\cdots,n\}$. Replacing $X_i$ by $t_i$ for $i\in\{1,\cdots,n\}$ is a solution for $E$, $RSP$ yields $t_1=\langle X_1|E\rangle$.
\end{proof}

\begin{theorem}[Soundness of APTC with static localities and silent step and guarded linear recursion]\label{SAPTCTAU}
Let $x$ and $y$ be APTC with static localities and silent step and guarded linear recursion terms. If APTC with static localities and silent step and guarded linear recursion
$\vdash x=y$, then
\begin{enumerate}
  \item $x\approx_{rbs}^{sl} y$;
  \item $x\approx_{rbp}^{sl} y$;
  \item $x\approx_{rbhp}^{sl} y$;
  \item $x\approx_{rbhhp}^{sl} y$.
\end{enumerate}
\end{theorem}

\begin{proof}
Since $\approx_{rbp}^{sl}$, $\approx_{rbs}^{sl}$, $\approx_{rbhp}^{sl}$ and $\approx_{rbhhp}^{sl}$ are all both equivalent and congruent relations, we only need to check if each axiom in
Table \ref{AxiomsForTau} is sound modulo $\approx_{rbp}^{sl}$, $\approx_{rbs}^{sl}$, $\approx_{rbhp}^{sl}$ and $\approx_{rbhhp}^{sl}$, the proof is trivial and we omit it.
\end{proof}

\begin{theorem}[Completeness of APTC with static localities and silent step and guarded linear recursion]\label{CAPTCTAU}
Let $p$ and $q$ be closed APTC with static localities and silent step and guarded linear recursion terms, then,
\begin{enumerate}
  \item if $p\approx_{rbs}^{sl} q$ then $p=q$;
  \item if $p\approx_{rbp}^{sl} q$ then $p=q$;
  \item if $p\approx_{rbhp}^{sl} q$ then $p=q$;
  \item if $p\approx_{rbhhp}^{sl} q$ then $p=q$.
\end{enumerate}
\end{theorem}

\begin{proof}
Firstly, by the elimination theorem of APTC with static localities and silent step and guarded linear recursion (see Theorem \ref{ETTau}), we know that each process term in APTC with
static localities and silent step and guarded linear recursion is equal to a process term $\langle X_1|E\rangle$ with $E$ a guarded linear recursive specification.

It remains to prove the following cases.

(1) If $\langle X_1|E_1\rangle \approx_{rbs}^{sl} \langle Y_1|E_2\rangle$ for guarded linear recursive specification $E_1$ and $E_2$, then
$\langle X_1|E_1\rangle = \langle Y_1|E_2\rangle$.

Firstly, the recursive equation $W=\tau+\cdots+\tau$ with $W\nequiv X_1$ in $E_1$ and $E_2$, can be removed, and the corresponding summands $aW$ are replaced by $a$, to get $E_1'$ and
$E_2'$, by use of the axioms $RDP$, $A3$ and $B1$, and $\langle X|E_1\rangle = \langle X|E_1'\rangle$, $\langle Y|E_2\rangle = \langle Y|E_2'\rangle$.

Let $E_1$ consists of recursive equations $X=t_X$ for $X\in \mathcal{X}$ and $E_2$
consists of recursion equations $Y=t_Y$ for $Y\in\mathcal{Y}$, and are not the form $\tau+\cdots+\tau$. Let the guarded linear recursive specification $E$ consists of recursion
equations $Z_{XY}=t_{XY}$, and $\langle X|E_1\rangle\approx_{rbs}^{sl}\langle Y|E_2\rangle$, and $t_{XY}$ consists of the following summands:

\begin{enumerate}
  \item $t_{XY}$ contains a summand $(u_1::a_1\leftmerge\cdots\leftmerge u_m::a_m)Z_{X'Y'}$ iff $t_X$ contains the summand $(u_1::a_1\leftmerge\cdots\leftmerge u_m::a_m)X'$ and $t_Y$
  contains the summand $(u_1::a_1\leftmerge\cdots\leftmerge u_m::a_m)Y'$ such that $\langle X'|E_1\rangle\approx_{rbs}^{sl}\langle Y'|E_2\rangle$;
  \item $t_{XY}$ contains a summand $v_1::b_1\leftmerge\cdots\leftmerge v_n::b_n$ iff $t_X$ contains the summand $v_1::b_1\leftmerge\cdots\leftmerge v_n::b_n$ and $t_Y$ contains the
  summand $v_1::b_1\leftmerge\cdots\leftmerge v_n::b_n$;
  \item $t_{XY}$ contains a summand $\tau Z_{X'Y}$ iff $XY\nequiv X_1Y_1$, $t_X$ contains the summand $\tau X'$, and $\langle X'|E_1\rangle\approx_{rbs}^{sl}\langle Y|E_2\rangle$;
  \item $t_{XY}$ contains a summand $\tau Z_{XY'}$ iff $XY\nequiv X_1Y_1$, $t_Y$ contains the summand $\tau Y'$, and $\langle X|E_1\rangle\approx_{rbs}^{sl}\langle Y'|E_2\rangle$.
\end{enumerate}

Since $E_1$ and $E_2$ are guarded, $E$ is guarded. Constructing the process term $u_{XY}$ consist of the following summands:

\begin{enumerate}
  \item $u_{XY}$ contains a summand $(u_1::a_1\leftmerge\cdots\leftmerge u_m::a_m)\langle X'|E_1\rangle$ iff $t_X$ contains the summand $(u_1::a_1\leftmerge\cdots\leftmerge u_m::a_m)X'$
  and $t_Y$ contains the summand $(u_1::a_1\leftmerge\cdots\leftmerge u_m::a_m)Y'$ such that $\langle X'|E_1\rangle\approx_{rbs}^{sl}\langle Y'|E_2\rangle$;
  \item $u_{XY}$ contains a summand $v_1::b_1\leftmerge\cdots\leftmerge v_n::b_n$ iff $t_X$ contains the summand $v_1::b_1\leftmerge\cdots\leftmerge v_n::b_n$ and $t_Y$ contains the
  summand $v_1::b_1\leftmerge\cdots\leftmerge v_n::b_n$;
  \item $u_{XY}$ contains a summand $\tau \langle X'|E_1\rangle$ iff $XY\nequiv X_1Y_1$, $t_X$ contains the summand $\tau X'$, and
  $\langle X'|E_1\rangle\approx_{rbs}^{sl}\langle Y|E_2\rangle$.
\end{enumerate}

Let the process term $s_{XY}$ be defined as follows:

\begin{enumerate}
  \item $s_{XY}\triangleq\tau\langle X|E_1\rangle + u_{XY}$ iff $XY\nequiv X_1Y_1$, $t_Y$ contains the summand $\tau Y'$, and
  $\langle X|E_1\rangle\approx_{rbs}^{sl}\langle Y'|E_2\rangle$;
  \item $s_{XY}\triangleq\langle X|E_1\rangle$, otherwise.
\end{enumerate}

So, $\langle X|E_1\rangle=\langle X|E_1\rangle+u_{XY}$, and $(u_1::a_1\leftmerge\cdots\leftmerge u_m::a_m)(\tau\langle X|E_1\rangle+u_{XY})=(u_1::a_1\leftmerge\cdots\leftmerge u_m::a_m)((\tau\langle X|E_1\rangle+u_{XY})+u_{XY})=(u_1::a_1\leftmerge\cdots\leftmerge u_m::a_m)(\langle X|E_1\rangle+u_{XY})=(u_1::a_1\leftmerge\cdots\leftmerge u_m::a_m)\langle X|E_1\rangle$,
hence, $(u_1::a_1\leftmerge\cdots\leftmerge u_m::a_m)s_{XY}=(u_1::a_1\leftmerge\cdots\leftmerge u_m::a_m)\langle X|E_1\rangle$.

Let $\sigma$ map recursion variable $X$ in $E_1$ to $\langle X|E_1\rangle$, and let $\psi$ map recursion variable $Z_{XY}$ in $E$ to $s_{XY}$. It is sufficient to prove
$s_{XY}=\psi(t_{XY})$ for recursion variables $Z_{XY}$ in $E$. Either $XY\equiv X_1Y_1$ or $XY\nequiv X_1Y_1$, we all can get $s_{XY}=\psi(t_{XY})$. So,
$s_{XY}=\langle Z_{XY}|E\rangle$ for recursive variables $Z_{XY}$ in $E$ is a solution for $E$. Then by $RSP$, particularly, $\langle X_1|E_1\rangle=\langle Z_{X_1Y_1}|E\rangle$.
Similarly, we can obtain $\langle Y_1|E_2\rangle=\langle Z_{X_1Y_1}|E\rangle$. Finally, $\langle X_1|E_1\rangle=\langle Z_{X_1Y_1}|E\rangle=\langle Y_1|E_2\rangle$, as desired.

(2) If $\langle X_1|E_1\rangle \approx_{rbp}^{sl} \langle Y_1|E_2\rangle$ for guarded linear recursive specification $E_1$ and $E_2$, then $\langle X_1|E_1\rangle = \langle Y_1|E_2\rangle$.

It can be proven similarly to (1), we omit it.

(3) If $\langle X_1|E_1\rangle \approx_{rbhb} \langle Y_1|E_2\rangle$ for guarded linear recursive specification $E_1$ and $E_2$, then $\langle X_1|E_1\rangle = \langle Y_1|E_2\rangle$.

It can be proven similarly to (1), we omit it.

(4) If $\langle X_1|E_1\rangle \approx_{rbhhb} \langle Y_1|E_2\rangle$ for guarded linear recursive specification $E_1$ and $E_2$, then $\langle X_1|E_1\rangle = \langle Y_1|E_2\rangle$.

It can be proven similarly to (1), we omit it.
\end{proof}


The unary abstraction operator $\tau_I$ ($I\subseteq \mathbb{E}$) renames all atomic events in $I$ into $\tau$. APTC with static localities and silent step and abstraction operator is
called $APTC_{\tau}$ with static localities. The transition rules of operator $\tau_I$ are shown in Table \ref{TRForAbstraction}.

\begin{center}
    \begin{table}
        $$\frac{x\xrightarrow[u]{e}\surd}{\tau_I(x)\xrightarrow[u]{e}\surd}\quad e\notin I
        \quad\quad\frac{x\xrightarrow[u]{e}x'}{\tau_I(x)\xrightarrow[u]{e}\tau_I(x')}\quad e\notin I$$

        $$\frac{x\xrightarrow[u]{e}\surd}{\tau_I(x)\xrightarrow{\tau}\surd}\quad e\in I
        \quad\quad\frac{x\xrightarrow[u]{e}x'}{\tau_I(x)\xrightarrow{\tau}\tau_I(x')}\quad e\in I$$
        \caption{Transition rule of the abstraction operator}
        \label{TRForAbstraction}
    \end{table}
\end{center}

\begin{theorem}[Conservitivity of $APTC_{\tau}$ with static localities and guarded linear recursion]
$APTC_{\tau}$ with static localities and guarded linear recursion is a conservative extension of APTC with static localities and silent step and guarded linear recursion.
\end{theorem}

\begin{proof}
Since the transition rules of APTC with static localities and silent step and guarded linear recursion are source-dependent, and the transition rules for abstraction operator in Table
\ref{TRForAbstraction} contain only a fresh operator $\tau_I$ in their source, so the transition rules of $APTC_{\tau}$ with static localities and guarded linear recursion is a
conservative extension of those of APTC with static localities and silent step and guarded linear recursion.
\end{proof}

\begin{theorem}[Congruence theorem of $APTC_{\tau}$ with static localities and guarded linear recursion]
Rooted branching static location truly concurrent bisimulation equivalences $\approx_{rbp}^{sl}$, $\approx_{rbs}^{sl}$, $\approx_{rbhp}^{sl}$ and $\approx_{rbhhp}^{sl}$ are all
congruences with respect to $APTC_{\tau}$ with static localities and guarded linear recursion.
\end{theorem}

\begin{proof}
It is easy to see that Rooted branching static location truly concurrent bisimulations $\approx_{rbp}^{sl}$, $\approx_{rbs}^{sl}$, $\approx_{rbhp}^{sl}$ and $\approx_{rbhhp}^{sl}$ are
all equivalent relations, we only need to prove that Rooted branching static location truly concurrent bisimulation equivalences $\approx_{rbp}^{sl}$, $\approx_{rbs}^{sl}$,
$\approx_{rbhp}^{sl}$ and $\approx_{rbhhp}^{sl}$ are all preserved by the operator $\tau_I$, the proof is trivial and we omit it.
\end{proof}

We design the axioms for the abstraction operator $\tau_I$ in Table \ref{AxiomsForAbstraction}.

\begin{center}
\begin{table}
  \begin{tabular}{@{}ll@{}}
\hline No. &Axiom\\
  $TI1$ & $e\notin I\quad \tau_I(e)=e$\\
  $TI2$ & $e\in I\quad \tau_I(e)=\tau$\\
  $TI3$ & $\tau_I(\delta)=\delta$\\
  $TI4$ & $\tau_I(x+y)=\tau_I(x)+\tau_I(y)$\\
  $TI5$ & $\tau_I(x\cdot y)=\tau_I(x)\cdot\tau_I(y)$\\
  $TI6$ & $\tau_I(x\leftmerge y)=\tau_I(x)\leftmerge\tau_I(y)$\\
  $L14$ & $u::\tau_I(x)=\tau_I(u::x)$\\
  $L15$ & $e\notin I\quad \tau_I(u::e)=u::e$\\
  $L16$ & $e\in I\quad \tau_I(u::e)=\tau$\\
\end{tabular}
\caption{Axioms of abstraction operator}
\label{AxiomsForAbstraction}
\end{table}
\end{center}

\begin{theorem}[Soundness of $APTC_{\tau}$ with static localities and guarded linear recursion]\label{SAPTCABS}
Let $x$ and $y$ be $APTC_{\tau}$ with static localities and guarded linear recursion terms. If $APTC_{\tau}$ with static localities and guarded linear recursion $\vdash x=y$, then
\begin{enumerate}
  \item $x\approx_{rbs}^{sl} y$;
  \item $x\approx_{rbp}^{sl} y$;
  \item $x\approx_{rbhp}^{sl} y$;
  \item $x\approx_{rbhhp}^{sl} y$.
\end{enumerate}
\end{theorem}

\begin{proof}
Since $\approx_{rbp}^{sl}$, $\approx_{rbs}^{sl}$, $\approx_{rbhp}^{sl}$ and $\approx_{rbhhp}^{sl}$ are all both equivalent and congruent relations, we only need to check if each axiom in
Table \ref{AxiomsForAbstraction} is sound modulo $\approx_{rbp}^{sl}$, $\approx_{rbs}^{sl}$, $\approx_{rbhp}^{sl}$ and $\approx_{rbhhp}^{sl}$, the proof is trivial and we omit it.
\end{proof}

Though $\tau$-loops are prohibited in guarded linear recursive specifications (see Definition \ref{GLRS}) in a specifiable way, they can be constructed using the abstraction operator,
for example, there exist $\tau$-loops in the process term $\tau_{\{a\}}(\langle X|X=aX\rangle)$. To avoid $\tau$-loops caused by $\tau_I$ and ensure fairness, the concept of cluster
and $CFAR$ (Cluster Fair Abstraction Rule) \cite{CFAR} are still valid in true concurrency, we introduce them below.

\begin{definition}[Cluster]\label{CLUSTER}
Let $E$ be a guarded linear recursive specification, and $I\subseteq \mathbb{E}$. Two recursion variable $X$ and $Y$ in $E$ are in the same cluster for $I$ iff there exist sequences of
transitions $\langle X|E\rangle\xrightarrow[u]{\{b_{11},\cdots, b_{1i}\}}\cdots[u]\xrightarrow{\{b_{m1},\cdots, b_{mi}\}}\langle Y|E\rangle$ and
$\langle Y|E\rangle\xrightarrow[v]{\{c_{11},\cdots, c_{1j}\}}\cdots\xrightarrow[v]{\{c_{n1},\cdots, c_{nj}\}}\langle X|E\rangle$, where
$b_{11},\cdots,b_{mi},c_{11},\cdots,c_{nj}\in I\cup\{\tau\}$.

$u_1::a_1\leftmerge\cdots\leftmerge u_k::a_k$ or $(u_1::a_1\leftmerge\cdots\leftmerge u_k::a_k) X$ is an exit for the cluster $C$ iff: (1) $u_1::a_1\leftmerge\cdots\leftmerge u_k::a_k$
or $(u_1::a_1\leftmerge\cdots\leftmerge u_k::a_k) X$ is a summand at the right-hand side of the recursive equation for a recursion variable in $C$, and (2) in the case of
$(u_1::a_1\leftmerge\cdots\leftmerge u_k::a_k) X$, either $a_l\notin I\cup\{\tau\}(l\in\{1,2,\cdots,k\})$ or $X\notin C$.
\end{definition}

\begin{center}
\begin{table}
  \begin{tabular}{@{}ll@{}}
\hline No. &Axiom\\
  $CFAR$ & If $X$ is in a cluster for $I$ with exits \\
           & $\{(u_{11}::a_{11}\leftmerge\cdots\leftmerge u_{1i}::a_{1i})Y_1,\cdots,(u_{m1}::a_{m1}\leftmerge\cdots\leftmerge u_{mi}::a_{mi})Y_m,$ \\
           & $v_{11}::b_{11}\leftmerge\cdots\leftmerge v_{1j}::b_{1j},\cdots,v_{n1}::b_{n1}\leftmerge\cdots\leftmerge v_{nj}::b_{nj}\}$, \\
           & then $\tau\cdot\tau_I(\langle X|E\rangle)=$\\
           & $\tau\cdot\tau_I((u_{11}::a_{11}\leftmerge\cdots\leftmerge u_{1i}::a_{1i})\langle Y_1|E\rangle+\cdots+(u_{m1}::a_{m1}\leftmerge\cdots\leftmerge u_{mi}::a_{mi})\langle Y_m|E\rangle$\\
           & $+v_{11}::b_{11}\leftmerge\cdots\leftmerge v_{1j}::b_{1j}+\cdots+v_{n1}::b_{n1}\leftmerge\cdots\leftmerge v_{nj}::b_{nj})$\\
\end{tabular}
\caption{Cluster fair abstraction rule}
\label{CFAR}
\end{table}
\end{center}

\begin{theorem}[Soundness of $CFAR$]\label{SCFAR}
$CFAR$ is sound modulo rooted branching truly concurrent bisimulation equivalences $\approx_{rbs}^{sl}$, $\approx_{rbp}^{sl}$, $\approx_{rbhp}^{sl}$ and $\approx_{rbhhp}^{sl}$.
\end{theorem}

\begin{proof}
(1) Soundness of $CFAR$ with respect to rooted branching static location step bisimulation $\approx_{rbs}^{sl}$.

Let $X$ be in a cluster for $I$ with exits $\{(u_{11}::a_{11}\leftmerge\cdots\leftmerge u_{1i}::a_{1i})Y_1,\cdots,(u_{m1}::a_{m1}\leftmerge\cdots\leftmerge u_{mi}::a_{mi})Y_m,\\
v_{11}::b_{11}\leftmerge\cdots\leftmerge v_{1j}::b_{1j},\cdots,v_{n1}::b_{n1}\leftmerge\cdots\leftmerge v_{nj}::b_{nj}\}$. Then $\langle X|E\rangle$ can execute a string of atomic
events from $I\cup\{\tau\}$ inside the cluster of $X$, followed by an exit $(u_{i'1}::a_{i'1}\leftmerge\cdots\leftmerge u_{i'i}::a_{i'i})Y_{i'}$ for $i'\in\{1,\cdots,m\}$ or
$v_{j'1}::b_{j'1}\leftmerge\cdots\leftmerge v_{j'j}::b_{j'j}$ for $j'\in\{1,\cdots,n\}$. Hence, $\tau_I(\langle X|E\rangle)$ can execute a string of $\tau^*$ inside the cluster of
$X$, followed by an exit $\tau_I((u_{i'1}::a_{i'1}\leftmerge\cdots\leftmerge u_{i'i}::a_{i'i})\langle Y_{i'}|E\rangle)$ for $i'\in\{1,\cdots,m\}$ or
$\tau_I(v_{j'1}::b_{j'1}\leftmerge\cdots\leftmerge v_{j'j}::b_{j'j})$ for $j'\in\{1,\cdots,n\}$. And these $\tau^*$ are non-initial in $\tau\tau_I(\langle X|E\rangle)$, so they are
truly silent by the axiom $B1$, we obtain $\tau\tau_I(\langle X|E\rangle)\approx_{rbs}^{sl}\tau\cdot\tau_I((u_{11}::a_{11}\leftmerge\cdots\leftmerge u_{1i}::a_{1i})\langle Y_1|E\rangle+
\cdots+(u_{m1}::a_{m1}\leftmerge\cdots\leftmerge u_{mi}::a_{mi})\langle Y_m|E\rangle+v_{11}::b_{11}\leftmerge\cdots\leftmerge v_{1j}::b_{1j}+\cdots+v_{n1}::b_{n1}\leftmerge\cdots\leftmerge v_{nj}::b_{nj})$,
as desired.

(2) Soundness of $CFAR$ with respect to rooted branching static location pomset bisimulation $\approx_{rbp}^{sl}$.

Similarly to the proof of soundness of $CFAR$ modulo rooted branching static location step bisimulation $\approx_{rbs}^{sl}$ (1), we can prove that $CFAR$ in Table \ref{CFAR} is sound
modulo rooted branching static location pomset bisimulation $\approx_{rbp}^{sl}$, we omit them.

(3) Soundness of $CFAR$ with respect to rooted branching static location hp-bisimulation $\approx_{rbhp}^{sl}$.

Similarly to the proof of soundness of $CFAR$ modulo rooted branching static location step bisimulation equivalence (1), we can prove that $CFAR$ in Table \ref{CFAR} is sound modulo
rooted branching static location hp-bisimulation equivalence, we omit them.

(4) Soundness of $CFAR$ with respect to rooted branching static location hhp-bisimulation $\approx_{rbhhp}^{sl}$.

Similarly to the proof of soundness of $CFAR$ modulo rooted branching static location step bisimulation equivalence (1), we can prove that $CFAR$ in Table \ref{CFAR} is sound modulo
rooted branching static location hhp-bisimulation equivalence, we omit them.
\end{proof}

\begin{theorem}[Completeness of $APTC_{\tau}$ with static localities and guarded linear recursion and $CFAR$]\label{CCFAR}
Let $p$ and $q$ be closed $APTC_{\tau}$ with static localities and guarded linear recursion and $CFAR$ terms, then,
\begin{enumerate}
  \item if $p\approx_{rbs}^{sl} q$ then $p=q$;
  \item if $p\approx_{rbp}^{sl} q$ then $p=q$;
  \item if $p\approx_{rbhp}^{sl} q$ then $p=q$;
  \item if $p\approx_{rbhhp}^{sl} q$ then $p=q$.
\end{enumerate}
\end{theorem}

\begin{proof}
(1) For the case of rooted branching static location step bisimulation, the proof is following.

Firstly, in the proof the Theorem \ref{CAPTCTAU}, we know that each process term $p$ in APTC with static localities and silent step and guarded linear recursion is equal to a process
term $\langle X_1|E\rangle$ with $E$ a guarded linear recursive specification. And we prove if $\langle X_1|E_1\rangle\approx_{rbs}^{sl}\langle Y_1|E_2\rangle$, then
$\langle X_1|E_1\rangle=\langle Y_1|E_2\rangle$.

The only new case is $p\equiv\tau_I(q)$. Let $q=\langle X|E\rangle$ with $E$ a guarded linear recursive specification, so $p=\tau_I(\langle X|E\rangle)$. Then the collection of
recursive variables in $E$ can be divided into its clusters $C_1,\cdots,C_N$ for $I$. Let

$(u_{1i1}::a_{1i1}\leftmerge\cdots\leftmerge u_{k_{i1}i1}::a_{k_{i1}i1}) Y_{i1}+\cdots+(u_{1im_i}::a_{1im_i}\leftmerge\cdots\leftmerge u_{k_{im_i}im_i}::a_{k_{im_i}im_i}) Y_{im_i}\\
+v_{1i1}::b_{1i1}\leftmerge\cdots\leftmerge v_{l_{i1}i1}::b_{l_{i1}i1}+\cdots+v_{1im_i}::b_{1im_i}\leftmerge\cdots\leftmerge v_{l_{im_i}im_i}::b_{l_{im_i}im_i}$

be the conflict composition of exits for the cluster $C_i$, with $i\in\{1,\cdots,N\}$.

For $Z\in C_i$ with $i\in\{1,\cdots,N\}$, we define

$s_Z\triangleq (\hat{u_{1i1}::a_{1i1}}\leftmerge\cdots\leftmerge \hat{u_{k_{i1}i1}::a_{k_{i1}i1}}) \tau_I(\langle Y_{i1}|E\rangle)+\cdots+(\hat{u_{1im_i}::a_{1im_i}}\leftmerge\cdots\leftmerge \hat{u_{k_{im_i}im_i}::a_{k_{im_i}im_i}}) \tau_I(\langle Y_{im_i}|E\rangle)\\
+\hat{v_{1i1}::b_{1i1}}\leftmerge\cdots\leftmerge \hat{v_{l_{i1}i1}::b_{l_{i1}i1}}+\cdots+\hat{v_{1im_i}::b_{1im_i}}\leftmerge\cdots\leftmerge \hat{v_{l_{im_i}im_i}::b_{l_{im_i}im_i}}$

For $Z\in C_i$ and $a_1,\cdots,a_j\in \mathbb{E}\cup\{\tau\}$ with $j\in\mathbb{N}$, we have

$(u_1::a_1\leftmerge\cdots\leftmerge u_j::a_j)\tau_I(\langle Z|E\rangle)$

$=(u_1::a_1\leftmerge\cdots\leftmerge u_j::a_j)\tau_I((u_{1i1}::a_{1i1}\leftmerge\cdots\leftmerge u_{k_{i1}i1}::a_{k_{i1}i1}) \langle Y_{i1}|E\rangle+\cdots+(u_{1im_i}::a_{1im_i}\leftmerge\cdots\leftmerge u_{k_{im_i}im_i}::a_{k_{im_i}im_i}) \langle Y_{im_i}|E\rangle
+v_{1i1}::b_{1i1}\leftmerge\cdots\leftmerge v_{l_{i1}i1}::b_{l_{i1}i1}+\cdots+v_{1im_i}::b_{1im_i}\leftmerge\cdots\leftmerge v_{l_{im_i}im_i}::b_{l_{im_i}im_i})$

$=(u_1::a_1\leftmerge\cdots\leftmerge u_j::a_j)s_Z$

Let the linear recursive specification $F$ contain the same recursive variables as $E$, for $Z\in C_i$, $F$ contains the following recursive equation

$Z=(\hat{u_{1i1}::a_{1i1}}\leftmerge\cdots\leftmerge \hat{u_{k_{i1}i1}::a_{k_{i1}i1}}) Y_{i1}+\cdots+(\hat{u_{1im_i}::a_{1im_i}}\leftmerge\cdots\leftmerge \hat{u_{k_{im_i}im_i}::a_{k_{im_i}im_i}})  Y_{im_i}
+\hat{v_{1i1}::b_{1i1}}\leftmerge\cdots\leftmerge \hat{v_{l_{i1}i1}::b_{l_{i1}i1}}+\cdots+\hat{v_{1im_i}::b_{1im_i}}\leftmerge\cdots\leftmerge \hat{v_{l_{im_i}im_i}::b_{l_{im_i}im_i}}$

It is easy to see that there is no sequence of one or more $\tau$-transitions from $\langle Z|F\rangle$ to itself, so $F$ is guarded.

For

$s_Z=(\hat{u_{1i1}::a_{1i1}}\leftmerge\cdots\leftmerge \hat{u_{k_{i1}i1}::a_{k_{i1}i1}}) Y_{i1}+\cdots+(\hat{u_{1im_i}::a_{1im_i}}\leftmerge\cdots\leftmerge \hat{u_{k_{im_i}im_i}::a_{k_{im_i}im_i}}) Y_{im_i}
+\hat{v_{1i1}::b_{1i1}}\leftmerge\cdots\leftmerge \hat{v_{l_{i1}i1}::b_{l_{i1}i1}}+\cdots+\hat{v_{1im_i}::b_{1im_i}}\leftmerge\cdots\leftmerge \hat{v_{l_{im_i}im_i}::b_{l_{im_i}im_i}}$

is a solution for $F$. So, $(u_1::a_1\leftmerge\cdots\leftmerge u_j::a_j)\tau_I(\langle Z|E\rangle)=(u_1::a_1\leftmerge\cdots\leftmerge u_j::a_j)s_Z=(u_1::a_1\leftmerge\cdots\leftmerge u_j::a_j)\langle Z|F\rangle$.

So,

$\langle Z|F\rangle=(\hat{u_{1i1}::a_{1i1}}\leftmerge\cdots\leftmerge \hat{u_{k_{i1}i1}::a_{k_{i1}i1}}) Y_{i1}+\cdots+(\hat{u_{1im_i}::a_{1im_i}}\leftmerge\cdots\leftmerge \hat{u_{k_{im_i}im_i}::a_{k_{im_i}im_i}}) Y_{im_i}
+\hat{v_{1i1}::b_{1i1}}\leftmerge\cdots\leftmerge \hat{v_{l_{i1}i1}::b_{l_{i1}i1}}+\cdots+\hat{v_{1im_i}::b_{1im_i}}\leftmerge\cdots\leftmerge \hat{v_{l_{im_i}im_i}::b_{l_{im_i}im_i}}$

Hence, $\tau_I(\langle X|E\rangle=\langle Z|F\rangle)$, as desired.

(2) For the case of rooted branching static location pomset bisimulation, it can be proven similarly to (1), we omit it.

(3) For the case of rooted branching static location hp-bisimulation, it can be proven similarly to (1), we omit it.

(4) For the case of rooted branching static location hhp-bisimulation, it can be proven similarly to (1), we omit it.
\end{proof}

%% file: section5/section5.2.tex
\subsection{APTC with Dynamic Localities}{\label{aptcdl}}

APTC with dynamic localities is almost the same as APTC with static localities in section \ref{aptcsl}, as the locations are dynamically generated but not allocated statically. The LTSs-based
operational semantics and the laws are almost the same, except for the transition rules of atomic action and sequential composition as follows.

\[\frac{}{e\xrightarrow[loc]{e}\surd}\]

\[\frac{x\xrightarrow[u]{e}\surd}{x\cdot y\xrightarrow[u]{e} u::y} \quad\frac{x\xrightarrow[u]{e}x'}{x\cdot y\xrightarrow[u]{e}u::(x'\cdot y)}\]

%% file: section6.tex
\section{$\pi_{tc}$ with Localities}\label{pitcl}

In this chapter, we introduce $\pi_{tc}$ with localities, including static and dynamic location semantics in section \ref{pitcos}, $\pi_{tc}$ with static localities in section \ref{pitcsl},
$\pi_{tc}$ with dynamic localities in section \ref{pitcdl}.

\input{section6/section6.1.tex}

\input{section6/section6.2.tex}

\input{section6/section6.3.tex}


%% file: section6/section6.1.tex
\subsection{Operational Semantics}\label{pitcos}

\begin{definition}[Strong static location pomset, step bisimilarity]
Let $\mathcal{E}_1$, $\mathcal{E}_2$ be PESs. A strong static location pomset bisimulation is a relation $R_{\varphi}\subseteq\mathcal{C}(\mathcal{E}_1)\times\mathcal{C}(\mathcal{E}_2)$, such that if 
$(C_1,C_2)\in R_{\varphi}$, and $C_1\xrightarrow[u]{X_1}C_1'$ (with $\mathcal{E}_1\xrightarrow[u]{X_1}\mathcal{E}_1'$) then $C_2\xrightarrow[v]{X_2}C_2'$ (with 
$\mathcal{E}_2\xrightarrow[v]{X_2}\mathcal{E}_2'$), with $X_1\subseteq \mathbb{E}_1$, $X_2\subseteq \mathbb{E}_2$, $X_1\sim X_2$ and $(C_1',C_2')\in R_{\varphi\cup\{(u,v)\}}$:
\begin{enumerate}
  \item for each fresh action $\alpha\in X_1$, if $C_1''\xrightarrow[u']{\alpha}C_1'''$ (with $\mathcal{E}_1''\xrightarrow[u']{\alpha}\mathcal{E}_1'''$), then for some $C_2''$ and $C_2'''$, 
  $C_2''\xrightarrow{\alpha}[v']C_2'''$ (with $\mathcal{E}_2''\xrightarrow[v']{\alpha}\mathcal{E}_2'''$), such that if $(C_1'',C_2'')\in R_{\varphi}$ then $(C_1''',C_2''')\in R_{\varphi\cup\{(u',v')\}}$;
  \item for each $x(y)\in X_1$ with ($y\notin n(\mathcal{E}_1, \mathcal{E}_2)$), if $C_1''\xrightarrow[u']{x(y)}C_1'''$ (with $\mathcal{E}_1''\xrightarrow[u']{x(y)}\mathcal{E}_1'''\{w/y\}$) 
  for all $w$, then for some $C_2''$ and $C_2'''$, $C_2''\xrightarrow[v']{x(y)}C_2'''$ (with $\mathcal{E}_2''\xrightarrow[v']{x(y)}\mathcal{E}_2'''\{w/y\}$) for all $w$, such that if 
  $(C_1'',C_2'')\in R_{\varphi}$ then $(C_1''',C_2''')\in R_{\varphi\cup\{(u',v')\}}$;
  \item for each two $x_1(y),x_2(y)\in X_1$ with ($y\notin n(\mathcal{E}_1, \mathcal{E}_2)$), if $C_1''\xrightarrow[u']{\{x_1(y),x_2(y)\}}C_1'''$ (with 
  $\mathcal{E}_1''\xrightarrow[u']{\{x_1(y),x_2(y)\}}\mathcal{E}_1'''\{w/y\}$) for all $w$, then for some $C_2''$ and $C_2'''$, $C_2''\xrightarrow[v']{\{x_1(y),x_2(y)\}}C_2'''$ 
  (with $\mathcal{E}_2''\xrightarrow[v']{\{x_1(y),x_2(y)\}}\mathcal{E}_2'''\{w/y\}$) for all $w$, such that if $(C_1'',C_2'')\in R_{\varphi}$ then $(C_1''',C_2''')\in R_{\varphi\cup\{(u',v')\}}$;
  \item for each $\overline{x}(y)\in X_1$ with $y\notin n(\mathcal{E}_1, \mathcal{E}_2)$, if $C_1''\xrightarrow[u']{\overline{x}(y)}C_1'''$ (with 
  $\mathcal{E}_1''\xrightarrow[u']{\overline{x}(y)}\mathcal{E}_1'''$), then for some $C_2''$ and $C_2'''$, $C_2''\xrightarrow[v']{\overline{x}(y)}C_2'''$ (with 
  $\mathcal{E}_2''\xrightarrow[v']{\overline{x}(y)}\mathcal{E}_2'''$), such that if $(C_1'',C_2'')\in R_{\varphi}$ then $(C_1''',C_2''')\in R_{\varphi\cup\{(u',v')\}}$.
\end{enumerate}
 and vice-versa.

We say that $\mathcal{E}_1$, $\mathcal{E}_2$ are strong static location pomset bisimilar, written $\mathcal{E}_1\sim_p^{sl}\mathcal{E}_2$, if there exists a strong static location pomset bisimulation $R_{\varphi}$, such that 
$(\emptyset,\emptyset)\in R_{\varphi}$. By replacing pomset transitions with steps, we can get the definition of strong static location step bisimulation. When PESs $\mathcal{E}_1$ and $\mathcal{E}_2$ are 
strong static location step bisimilar, we write $\mathcal{E}_1\sim_s^{sl}\mathcal{E}_2$.
\end{definition}

\begin{definition}[Strong (hereditary) history-preserving bisimilarity]
A strong static location history-preserving (hp-) bisimulation is a posetal relation $R_{\varphi}\subseteq\mathcal{C}(\mathcal{E}_1)\overline{\times}\mathcal{C}(\mathcal{E}_2)$ such that if $(C_1,f,C_2)\in R_{\varphi}$, 
and
\begin{enumerate}
  \item for $e_1=\alpha$ a fresh action, if $C_1\xrightarrow[u]{\alpha}C_1'$ (with $\mathcal{E}_1\xrightarrow[u]{\alpha}\mathcal{E}_1'$), then for some $C_2'$ and $e_2=\alpha$, 
  $C_2\xrightarrow[v]{\alpha}C_2'$ (with $\mathcal{E}_2\xrightarrow[v]{\alpha}\mathcal{E}_2'$), such that $(C_1',f[e_1\mapsto e_2],C_2')\in R_{\varphi\cup\{(u,v)\}}$;
  \item for $e_1=x(y)$ with ($y\notin n(\mathcal{E}_1, \mathcal{E}_2)$), if $C_1\xrightarrow[u]{x(y)}C_1'$ (with $\mathcal{E}_1\xrightarrow[u]{x(y)}\mathcal{E}_1'\{w/y\}$) for all $w$, then 
  for some $C_2'$ and $e_2=x(y)$, $C_2\xrightarrow[v]{x(y)}C_2'$ (with $\mathcal{E}_2\xrightarrow[v]{x(y)}\mathcal{E}_2'\{w/y\}$) for all $w$, such that $(C_1',f[e_1\mapsto e_2],C_2')\in R_{\varphi\cup\{(u,v)\}}$;
  \item for $e_1=\overline{x}(y)$ with $y\notin n(\mathcal{E}_1, \mathcal{E}_2)$, if $C_1\xrightarrow[u]{\overline{x}(y)}C_1'$ (with 
  $\mathcal{E}_1\xrightarrow[u]{\overline{x}(y)}\mathcal{E}_1'$), then for some $C_2'$ and $e_2=\overline{x}(y)$, $C_2\xrightarrow[v]{\overline{x}(y)}C_2'$ (with 
  $\mathcal{E}_2\xrightarrow[v]{\overline{x}(y)}\mathcal{E}_2'$), such that $(C_1',f[e_1\mapsto e_2],C_2')\in R_{\varphi\cup\{(u,v)\}}$.
\end{enumerate}

and vice-versa. $\mathcal{E}_1,\mathcal{E}_2$ are strong static locatoin history-preserving (hp-)bisimilar and are written $\mathcal{E}_1\sim_{hp}^{sl}\mathcal{E}_2$ if there exists a strong 
static location hp-bisimulation $R_{\varphi}$ such that $(\emptyset,\emptyset,\emptyset)\in R_{\varphi}$.

A strong static location hereditary history-preserving (hhp-)bisimulation is a downward closed strong static location hp-bisimulation. $\mathcal{E}_1,\mathcal{E}_2$ are strong static location hereditary history-preserving 
(hhp-)bisimilar and are written $\mathcal{E}_1\sim_{hhp}^{sl}\mathcal{E}_2$.
\end{definition}

\begin{definition}[Strong dynamic location pomset, step bisimilarity]
Let $\mathcal{E}_1$, $\mathcal{E}_2$ be PESs. A strong dynamic location pomset bisimulation is a relation $R\subseteq\mathcal{C}(\mathcal{E}_1)\times\mathcal{C}(\mathcal{E}_2)$, such that if
$(C_1,C_2)\in R$, and $C_1\xrightarrow[u]{X_1}C_1'$ (with $\mathcal{E}_1\xrightarrow[u]{X_1}\mathcal{E}_1'$) then $C_2\xrightarrow[u]{X_2}C_2'$ (with
$\mathcal{E}_2\xrightarrow[u]{X_2}\mathcal{E}_2'$), with $X_1\subseteq \mathbb{E}_1$, $X_2\subseteq \mathbb{E}_2$, $X_1\sim X_2$ and $(C_1',C_2')\in R$:
\begin{enumerate}
  \item for each fresh action $\alpha\in X_1$, if $C_1''\xrightarrow[u']{\alpha}C_1'''$ (with $\mathcal{E}_1''\xrightarrow[u']{\alpha}\mathcal{E}_1'''$), then for some $C_2''$ and $C_2'''$,
  $C_2''\xrightarrow[u']{\alpha}C_2'''$ (with $\mathcal{E}_2''\xrightarrow[u']{\alpha}\mathcal{E}_2'''$), such that if $(C_1'',C_2'')\in R$ then $(C_1''',C_2''')\in R$;
  \item for each $x(y)\in X_1$ with ($y\notin n(\mathcal{E}_1, \mathcal{E}_2)$), if $C_1''\xrightarrow[u']{x(y)}C_1'''$ (with $\mathcal{E}_1''\xrightarrow[u']{x(y)}\mathcal{E}_1'''\{w/y\}$)
  for all $w$, then for some $C_2''$ and $C_2'''$, $C_2''\xrightarrow[u']{x(y)}C_2'''$ (with $\mathcal{E}_2''\xrightarrow[u']{x(y)}\mathcal{E}_2'''\{w/y\}$) for all $w$, such that if
  $(C_1'',C_2'')\in R$ then $(C_1''',C_2''')\in R$;
  \item for each two $x_1(y),x_2(y)\in X_1$ with ($y\notin n(\mathcal{E}_1, \mathcal{E}_2)$), if $C_1''\xrightarrow[u']{\{x_1(y),x_2(y)\}}C_1'''$ (with
  $\mathcal{E}_1''\xrightarrow[u']{\{x_1(y),x_2(y)\}}\mathcal{E}_1'''\{w/y\}$) for all $w$, then for some $C_2''$ and $C_2'''$, $C_2''\xrightarrow[u']{\{x_1(y),x_2(y)\}}C_2'''$
  (with $\mathcal{E}_2''\xrightarrow[u']{\{x_1(y),x_2(y)\}}\mathcal{E}_2'''\{w/y\}$) for all $w$, such that if $(C_1'',C_2'')\in R$ then $(C_1''',C_2''')\in R$;
  \item for each $\overline{x}(y)\in X_1$ with $y\notin n(\mathcal{E}_1, \mathcal{E}_2)$, if $C_1''\xrightarrow[u']{\overline{x}(y)}C_1'''$ (with
  $\mathcal{E}_1''\xrightarrow[u']{\overline{x}(y)}\mathcal{E}_1'''$), then for some $C_2''$ and $C_2'''$, $C_2''\xrightarrow[u']{\overline{x}(y)}C_2'''$ (with
  $\mathcal{E}_2''\xrightarrow[u']{\overline{x}(y)}\mathcal{E}_2'''$), such that if $(C_1'',C_2'')\in R$ then $(C_1''',C_2''')\in R$.
\end{enumerate}
 and vice-versa.

We say that $\mathcal{E}_1$, $\mathcal{E}_2$ are strong dynamic location pomset bisimilar, written $\mathcal{E}_1\sim_p^{dl}\mathcal{E}_2$, if there exists a strong dynamic location pomset bisimulation $R$, such that
$(\emptyset,\emptyset)\in R$. By replacing pomset transitions with steps, we can get the definition of strong dynamic location step bisimulation. When PESs $\mathcal{E}_1$ and $\mathcal{E}_2$ are
strong dynamic location step bisimilar, we write $\mathcal{E}_1\sim_s^{dl}\mathcal{E}_2$.
\end{definition}

\begin{definition}[Strong dynamic location (hereditary) history-preserving bisimilarity]
A strong dynamic location history-preserving (hp-) bisimulation is a posetal relation $R\subseteq\mathcal{C}(\mathcal{E}_1)\overline{\times}\mathcal{C}(\mathcal{E}_2)$ such that if $(C_1,f,C_2)\in R$,
and
\begin{enumerate}
  \item for $e_1=\alpha$ a fresh action, if $C_1\xrightarrow[u]{\alpha}C_1'$ (with $\mathcal{E}_1\xrightarrow[u]{\alpha}\mathcal{E}_1'$), then for some $C_2'$ and $e_2=\alpha$,
  $C_2\xrightarrow[u]{\alpha}C_2'$ (with $\mathcal{E}_2\xrightarrow[u]{\alpha}\mathcal{E}_2'$), such that $(C_1',f[e_1\mapsto e_2],C_2')\in R$;
  \item for $e_1=x(y)$ with ($y\notin n(\mathcal{E}_1, \mathcal{E}_2)$), if $C_1\xrightarrow[u]{x(y)}C_1'$ (with $\mathcal{E}_1\xrightarrow[u]{x(y)}\mathcal{E}_1'\{w/y\}$) for all $w$, then
  for some $C_2'$ and $e_2=x(y)$, $C_2\xrightarrow[u]{x(y)}C_2'$ (with $\mathcal{E}_2\xrightarrow[u]{x(y)}\mathcal{E}_2'\{w/y\}$) for all $w$, such that $(C_1',f[e_1\mapsto e_2],C_2')\in R$;
  \item for $e_1=\overline{x}(y)$ with $y\notin n(\mathcal{E}_1, \mathcal{E}_2)$, if $C_1\xrightarrow[u]{\overline{x}(y)}C_1'$ (with
  $\mathcal{E}_1\xrightarrow[u]{\overline{x}(y)}\mathcal{E}_1'$), then for some $C_2'$ and $e_2=\overline{x}(y)$, $C_2\xrightarrow[u]{\overline{x}(y)}C_2'$ (with
  $\mathcal{E}_2\xrightarrow[u]{\overline{x}(y)}\mathcal{E}_2'$), such that $(C_1',f[e_1\mapsto e_2],C_2')\in R$.
\end{enumerate}

and vice-versa. $\mathcal{E}_1,\mathcal{E}_2$ are strong dynamic location history-preserving (hp-)bisimilar and are written $\mathcal{E}_1\sim_{hp}^{dl}\mathcal{E}_2$ if there exists a strong
dynamic location hp-bisimulation $R$ such that $(\emptyset,\emptyset,\emptyset)\in R$.

A strong dynamic location hereditary history-preserving (hhp-)bisimulation is a downward closed strong dynamic location hp-bisimulation. $\mathcal{E}_1,\mathcal{E}_2$ are strong dynamic location hereditary history-preserving
(hhp-)bisimilar and are written $\mathcal{E}_1\sim_{hhp}^{dl}\mathcal{E}_2$.
\end{definition}

%% file: section6/section6.2.tex
\subsection{$\pi_{tc}$ with Static Localities}\label{pitcsl}

\subsubsection{Syntax and Operational Semantics}\label{sos5}

We assume an infinite set $\mathcal{N}$ of (action or event) names, and use $a,b,c,\cdots$ to range over $\mathcal{N}$, use $x,y,z,w,u,v$ as meta-variables over names. We denote by
$\overline{\mathcal{N}}$ the set of co-names and let $\overline{a},\overline{b},\overline{c},\cdots$ range over $\overline{\mathcal{N}}$. Then we set
$\mathcal{L}=\mathcal{N}\cup\overline{\mathcal{N}}$ as the set of labels, and use $l,\overline{l}$ to range over $\mathcal{L}$. We extend complementation to $\mathcal{L}$ such that
$\overline{\overline{a}}=a$. Let $\tau$ denote the silent step (internal action or event) and define $Act=\mathcal{L}\cup\{\tau\}$ to be the set of actions, $\alpha,\beta$ range over
$Act$. And $K,L$ are used to stand for subsets of $\mathcal{L}$ and $\overline{L}$ is used for the set of complements of labels in $L$.

Further, we introduce a set $\mathcal{X}$ of process variables, and a set $\mathcal{K}$ of process constants, and let $X,Y,\cdots$ range over $\mathcal{X}$, and $A,B,\cdots$ range over
$\mathcal{K}$. For each process constant $A$, a nonnegative arity $ar(A)$ is assigned to it. Let $\widetilde{x}=x_1,\cdots,x_{ar(A)}$ be a tuple of distinct name variables, then
$A(\widetilde{x})$ is called a process constant. $\widetilde{X}$ is a tuple of distinct process variables, and also $E,F,\cdots$ range over the recursive expressions. We write
$\mathcal{P}$ for the set of processes. Sometimes, we use $I,J$ to stand for an indexing set, and we write $E_i:i\in I$ for a family of expressions indexed by $I$. $Id_D$ is the
identity function or relation over set $D$. The symbol $\equiv_{\alpha}$ denotes equality under standard alpha-convertibility, note that the subscript $\alpha$ has no relation to the
action $\alpha$.

Let $Loc$ be the set of locations, and $loc\in Loc$, $u,v\in Loc^*$, $\epsilon$ is the empty location. A distribution allocates a location $u\in Loc*$ to an action $\alpha$ denoted
$u::\alpha$ or a process $P$ denoted $u::P$.


\begin{definition}[Syntax]\label{syntax5}
A truly concurrent process $P$ is defined inductively by the following formation rules:

\begin{enumerate}
  \item $A(\widetilde{x})\in\mathcal{P}$;
  \item $\textbf{nil}\in\mathcal{P}$;
  \item if $P\in\mathcal{P}$ and $loc\in Loc$, the Location $loc::P\in\mathcal{P}$;
  \item if $P\in\mathcal{P}$, then the Prefix $\tau.P\in\mathcal{P}$, for $\tau\in Act$ is the silent action;
  \item if $P\in\mathcal{P}$, then the Output $\overline{x}y.P\in\mathcal{P}$, for $x,y\in Act$;
  \item if $P\in\mathcal{P}$, then the Input $x(y).P\in\mathcal{P}$, for $x,y\in Act$;
  \item if $P\in\mathcal{P}$, then the Restriction $(x)P\in\mathcal{P}$, for $x\in Act$;
  \item if $P,Q\in\mathcal{P}$, then the Summation $P+Q\in\mathcal{P}$;
  \item if $P,Q\in\mathcal{P}$, then the Composition $P\parallel Q\in\mathcal{P}$;
\end{enumerate}

The standard BNF grammar of syntax of $\pi_{tc}$ with static localities can be summarized as follows:

$$P::=A(\widetilde{x})|\textbf{nil}|loc::P|\tau.P|\overline{x}y.P| x(y).P | (x)P  |  P+P | P\parallel P.$$
\end{definition}

In $\overline{x}y$, $x(y)$ and $\overline{x}(y)$, $x$ is called the subject, $y$ is called the object and it may be free or bound.

\begin{definition}[Free variables]
The free names of a process $P$, $fn(P)$, are defined as follows.

\begin{enumerate}
  \item $fn(A(\widetilde{x}))\subseteq\{\widetilde{x}\}$;
  \item $fn(\textbf{nil})=\emptyset$;
  \item $fn(loc::P)=fn(P)$;
  \item $fn(\tau.P)=fn(P)$;
  \item $fn(\overline{x}y.P)=fn(P)\cup\{x\}\cup\{y\}$;
  \item $fn(x(y).P)=fn(P)\cup\{x\}-\{y\}$;
  \item $fn((x)P)=fn(P)-\{x\}$;
  \item $fn(P+Q)=fn(P)\cup fn(Q)$;
  \item $fn(P\parallel Q)=fn(P)\cup fn(Q)$.
\end{enumerate}
\end{definition}

\begin{definition}[Bound variables]
Let $n(P)$ be the names of a process $P$, then the bound names $bn(P)=n(P)-fn(P)$.
\end{definition}

For each process constant schema $A(\widetilde{x})$, a defining equation of the form

$$A(\widetilde{x})\overset{\text{def}}{=}P$$

is assumed, where $P$ is a process with $fn(P)\subseteq \{\widetilde{x}\}$.

\begin{definition}[Substitutions]\label{subs5}
A substitution is a function $\sigma:\mathcal{N}\rightarrow\mathcal{N}$. For $x_i\sigma=y_i$ with $1\leq i\leq n$, we write $\{y_1/x_1,\cdots,y_n/x_n\}$ or
$\{\widetilde{y}/\widetilde{x}\}$ for $\sigma$. For a process $P\in\mathcal{P}$, $P\sigma$ is defined inductively as follows:

\begin{enumerate}
  \item if $P$ is a process constant $A(\widetilde{x})=A(x_1,\cdots,x_n)$, then $P\sigma=A(x_1\sigma,\cdots,x_n\sigma)$;
  \item if $P=\textbf{nil}$, then $P\sigma=\textbf{nil}$;
  \item if $P=loc::P'$, then $P\sigma=loc::P'\sigma$;
  \item if $P=\tau.P'$, then $P\sigma=\tau.P'\sigma$;
  \item if $P=\overline{x}y.P'$, then $P\sigma=\overline{x\sigma}y\sigma.P'\sigma$;
  \item if $P=x(y).P'$, then $P\sigma=x\sigma(y).P'\sigma$;
  \item if $P=(x)P'$, then $P\sigma=(x\sigma)P'\sigma$;
  \item if $P=P_1+P_2$, then $P\sigma=P_1\sigma+P_2\sigma$;
  \item if $P=P_1\parallel P_2$, then $P\sigma=P_1\sigma \parallel P_2\sigma$.
\end{enumerate}
\end{definition}


The operational semantics is defined by LTSs (labelled transition systems), and it is detailed by the following definition.

\begin{definition}[Semantics]\label{semantics5}
The operational semantics of $\pi_{tc}$ with static localities corresponding to the syntax in Definition \ref{syntax5} is defined by a series of transition rules, named $\textbf{ACT}$, $\textbf{SUM}$,
$\textbf{IDE}$, $\textbf{PAR}$, $\textbf{COM}$ and $\textbf{CLOSE}$, $\textbf{RES}$ and $\textbf{OPEN}$ indicate that the rules are associated respectively with Prefix, Summation, Match, Identity, Parallel Composition, Communication, and Restriction in Definition \ref{syntax5}. They are shown in Table \ref{TRForPITC5}.

\begin{center}
    \begin{table}
        \[\textbf{TAU-ACT}\quad \frac{}{\tau.P\xrightarrow{\tau}P} \quad \textbf{OUTPUT-ACT}\quad \frac{}{\overline{x}y.P\xrightarrow[u]{\overline{x}y}P}\]

        \[\textbf{INPUT-ACT}\quad \frac{}{x(z).P\xrightarrow[u]{x(w)}P\{w/z\}}\quad (w\notin fn((z)P))\]

        \[\textbf{Loc}\quad \frac{P\xrightarrow[u]{\alpha}P'}{loc::P\xrightarrow[loc\ll u]{\alpha}loc::P'}\]

        \[\textbf{PAR}_1\quad \frac{P\xrightarrow[u]{\alpha}P'\quad Q\nrightarrow}{P\parallel Q\xrightarrow[u]{\alpha}P'\parallel Q}\quad (bn(\alpha)\cap fn(Q)=\emptyset)\]

        \[\textbf{PAR}_2\quad \frac{Q\xrightarrow[u]{\alpha}Q'\quad P\nrightarrow}{P\parallel Q\xrightarrow[u]{\alpha}P\parallel Q'}\quad (bn(\alpha)\cap fn(P)=\emptyset)\]

        \[\textbf{PAR}_3\quad \frac{P\xrightarrow[u]{\alpha}P'\quad Q\xrightarrow[v]{\beta}Q'}{P\parallel Q\xrightarrow[u\diamond v]{\{\alpha,\beta\}}P'\parallel Q'}\quad (\beta\neq\overline{\alpha}, bn(\alpha)\cap bn(\beta)=\emptyset, bn(\alpha)\cap fn(Q)=\emptyset,bn(\beta)\cap fn(P)=\emptyset)\]

        \[\textbf{PAR}_4\quad \frac{P\xrightarrow[u]{x_1(z)}P'\quad Q\xrightarrow[v]{x_2(z)}Q'}{P\parallel Q\xrightarrow[u\diamond v]{\{x_1(w),x_2(w)\}}P'\{w/z\}\parallel Q'\{w/z\}}\quad (w\notin fn((z)P)\cup fn((z)Q))\]

        \[\textbf{COM}\quad \frac{P\xrightarrow[u]{\overline{x}y}P'\quad Q\xrightarrow[v]{x(z)}Q'}{P\parallel Q\xrightarrow{\tau}P'\parallel Q'\{y/z\}}\]

        \[\textbf{CLOSE}\quad \frac{P\xrightarrow[u]{\overline{x}(w)}P'\quad Q\xrightarrow[v]{x(w)}Q'}{P\parallel Q\xrightarrow{\tau}(w)(P'\parallel Q')}\]

        \[\textbf{SUM}_1\quad \frac{P\xrightarrow[u]{\alpha}P'}{P+Q\xrightarrow[u]{\alpha}P'} \quad \textbf{SUM}_2\quad \frac{P\xrightarrow[u]{\{\alpha_1,\cdots,\alpha_n\}}P'}{P+Q\xrightarrow[u]{\{\alpha_1,\cdots,\alpha_n\}}P'}\]
%
%
%
%
%
        \caption{Transition rules of $\pi_{tc}$ with static localities}
        \label{TRForPITC5}
    \end{table}
\end{center}

\begin{center}
    \begin{table}
        \[\textbf{IDE}_1\quad\frac{P\{\widetilde{y}/\widetilde{x}\}\xrightarrow[u]{\alpha}P'}{A(\widetilde{y})\xrightarrow[u]{\alpha}P'}\quad (A(\widetilde{x})\overset{\text{def}}{=}P) \quad \textbf{IDE}_2\quad\frac{P\{\widetilde{y}/\widetilde{x}\}\xrightarrow[u]{\{\alpha_1,\cdots,\alpha_n\}}P'} {A(\widetilde{y})\xrightarrow[u]{\{\alpha_1,\cdots,\alpha_n\}}P'}\quad (A(\widetilde{x})\overset{\text{def}}{=}P)\]

        \[\textbf{RES}_1\quad \frac{P\xrightarrow[u]{\alpha}P'}{(y)P\xrightarrow[u]{\alpha}(y)P'}\quad (y\notin n(\alpha)) \quad \textbf{RES}_2\quad \frac{P\xrightarrow[u]{\{\alpha_1,\cdots,\alpha_n\}}P'}{(y)P\xrightarrow[u]{\{\alpha_1,\cdots,\alpha_n\}}(y)P'}\quad (y\notin n(\alpha_1)\cup\cdots\cup n(\alpha_n))\]

        \[\textbf{OPEN}_1\quad \frac{P\xrightarrow[u]{\overline{x}y}P'}{(y)P\xrightarrow[u]{\overline{x}(w)}P'\{w/y\}} \quad (y\neq x, w\notin fn((y)P'))\]

        \[\textbf{OPEN}_2\quad \frac{P\xrightarrow[u]{\{\overline{x}_1 y,\cdots,\overline{x}_n y\}}P'}{(y)P\xrightarrow[u]{\{\overline{x}_1(w),\cdots,\overline{x}_n(w)\}}P'\{w/y\}} \quad (y\neq x_1\neq\cdots\neq x_n, w\notin fn((y)P'))\]

        \caption{Transition rules of $\pi_{tc}$ with static localities (continuing)}
        \label{TRForPITC52}
    \end{table}
\end{center}
\end{definition}


\begin{proposition}
\begin{enumerate}
  \item If $P\xrightarrow[u]{\alpha}P'$ then
  \begin{enumerate}
    \item $fn(\alpha)\subseteq fn(P)$;
    \item $fn(P')\subseteq fn(P)\cup bn(\alpha)$;
  \end{enumerate}
  \item If $P\xrightarrow[u]{\{\alpha_1,\cdots,\alpha_n\}}P'$ then
  \begin{enumerate}
    \item $fn(\alpha_1)\cup\cdots\cup fn(\alpha_n)\subseteq fn(P)$;
    \item $fn(P')\subseteq fn(P)\cup bn(\alpha_1)\cup\cdots\cup bn(\alpha_n)$.
  \end{enumerate}
\end{enumerate}
\end{proposition}

\begin{proof}
By induction on the depth of inference.
\end{proof}

\begin{proposition}
Suppose that $P\xrightarrow[u]{\alpha(y)}P'$, where $\alpha=x$ or $\alpha=\overline{x}$, and $x\notin n(P)$, then there exists some $P''\equiv_{\alpha}P'\{z/y\}$,
$P\xrightarrow[u]{\alpha(z)}P''$.
\end{proposition}

\begin{proof}
By induction on the depth of inference.
\end{proof}

\begin{proposition}
If $P\rightarrow P'$, $bn(\alpha)\cap fn(P'\sigma)=\emptyset$, and $\sigma\lceil bn(\alpha)=id$, then there exists some $P''\equiv_{\alpha}P'\sigma$,
$P\sigma\xrightarrow[u]{\alpha\sigma}P''$.
\end{proposition}

\begin{proof}
By the definition of substitution (Definition \ref{subs5}) and induction on the depth of inference.
\end{proof}

\begin{proposition}
\begin{enumerate}
  \item If $P\{w/z\}\xrightarrow[u]{\alpha}P'$, where $w\notin fn(P)$ and $bn(\alpha)\cap fn(P,w)=\emptyset$, then there exist some $Q$ and $\beta$ with $Q\{w/z\}\equiv_{\alpha}P'$ and
  $\beta\sigma=\alpha$, $P\xrightarrow[u]{\beta}Q$;
  \item If $P\{w/z\}\xrightarrow[u]{\{\alpha_1,\cdots,\alpha_n\}}P'$, where $w\notin fn(P)$ and $bn(\alpha_1)\cap\cdots\cap bn(\alpha_n)\cap fn(P,w)=\emptyset$, then there exist some
  $Q$ and $\beta_1,\cdots,\beta_n$ with $Q\{w/z\}\equiv_{\alpha}P'$ and $\beta_1\sigma=\alpha_1,\cdots,\beta_n\sigma=\alpha_n$, $P\xrightarrow[u]{\{\beta_1,\cdots,\beta_n\}}Q$.
\end{enumerate}

\end{proposition}

\begin{proof}
By the definition of substitution (Definition \ref{subs5}) and induction on the depth of inference.
\end{proof}

\subsubsection{Strong Bisimilarities}\label{stcb5}



Similarly to $\pi_{tc}$, we can obtain the following laws with respect to probabilistic static location truly concurrent bisimilarities.

\begin{theorem}[Summation laws for strong static location pomset bisimulation]
The summation laws for strong static location pomset bisimulation are as follows.

\begin{enumerate}
  \item $P+\textbf{nil}\sim_p^{sl} P$;
  \item $P+P\sim_p^{sl} P$;
  \item $P_1+P_2\sim_p^{sl} P_2+P_1$;
  \item $P_1+(P_2+P_3)\sim_p^{sl} (P_1+P_2)+P_3$.
\end{enumerate}
\end{theorem}

\begin{proof}
\begin{enumerate}
  \item It is sufficient to prove the relation $R=\{(P+\textbf{nil}, P)\}\cup \textbf{Id}$ is a strong static location pomset bisimulation for some distributions. It can be proved similarly to the proof of Summation laws for strong pomset bisimulation in $\pi_{tc}$, we omit it;
  \item It is sufficient to prove the relation $R=\{(P+P, P)\}\cup \textbf{Id}$ is a strong static location pomset bisimulation for some distributions. It can be proved similarly to the proof of Summation laws for strong pomset bisimulation in $\pi_{tc}$, we omit it;
  \item It is sufficient to prove the relation $R=\{(P_1+P_2, P_2+P_1)\}\cup \textbf{Id}$ is a strong static location pomset bisimulation for some distributions. It can be proved similarly to the proof of Summation laws for strong pomset bisimulation in $\pi_{tc}$, we omit it;
  \item It is sufficient to prove the relation $R=\{(P_1+(P_2+P_3), (P_1+P_2)+P_3)\}\cup \textbf{Id}$ is a strong static location pomset bisimulation for some distributions. It can be proved similarly to the proof of Summation laws for strong pomset bisimulation in $\pi_{tc}$, we omit it.
\end{enumerate}
\end{proof}

\begin{theorem}[Summation laws for strong static location step bisimulation]
The summation laws for strong static location step bisimulation are as follows.

\begin{enumerate}
  \item $P+\textbf{nil}\sim_s^{sl} P$;
  \item $P+P\sim_s^{sl} P$;
  \item $P_1+P_2\sim_s^{sl} P_2+P_1$;
  \item $P_1+(P_2+P_3)\sim_s^{sl} (P_1+P_2)+P_3$.
\end{enumerate}
\end{theorem}

\begin{proof}
\begin{enumerate}
  \item It is sufficient to prove the relation $R=\{(P+\textbf{nil}, P)\}\cup \textbf{Id}$ is a strong static location step bisimulation for some distributions. It can be proved similarly to the proof of Summation laws for strong step bisimulation in $\pi_{tc}$, we omit it;
  \item It is sufficient to prove the relation $R=\{(P+P, P)\}\cup \textbf{Id}$ is a strong static location step bisimulation for some distributions. It can be proved similarly to the proof of Summation laws for strong step bisimulation in $\pi_{tc}$, we omit it;
  \item It is sufficient to prove the relation $R=\{(P_1+P_2, P_2+P_1)\}\cup \textbf{Id}$ is a strong static location step bisimulation for some distributions. It can be proved similarly to the proof of Summation laws for strong step bisimulation in $\pi_{tc}$, we omit it;
  \item It is sufficient to prove the relation $R=\{(P_1+(P_2+P_3), (P_1+P_2)+P_3)\}\cup \textbf{Id}$ is a strong static location step bisimulation for some distributions. It can be proved similarly to the proof of Summation laws for strong step bisimulation in $\pi_{tc}$, we omit it.
\end{enumerate}
\end{proof}

\begin{theorem}[Summation laws for strong static location hp-bisimulation]
The summation laws for strong static location hp-bisimulation are as follows.

\begin{enumerate}
  \item $P+\textbf{nil}\sim_{hp}^{sl} P$;
  \item $P+P\sim_{hp}^{sl} P$;
  \item $P_1+P_2\sim_{hp}^{sl} P_2+P_1$;
  \item $P_1+(P_2+P_3)\sim_{hp}^{sl} (P_1+P_2)+P_3$.
\end{enumerate}
\end{theorem}

\begin{proof}
\begin{enumerate}
  \item It is sufficient to prove the relation $R=\{(P+\textbf{nil}, P)\}\cup \textbf{Id}$ is a strong static location hp-bisimulation for some distributions. It can be proved similarly to the proof of Summation laws for strong hp-bisimulation in $\pi_{tc}$, we omit it;
  \item It is sufficient to prove the relation $R=\{(P+P, P)\}\cup \textbf{Id}$ is a strong static location hp-bisimulation for some distributions. It can be proved similarly to the proof of Summation laws for strong hp-bisimulation in $\pi_{tc}$, we omit it;
  \item It is sufficient to prove the relation $R=\{(P_1+P_2, P_2+P_1)\}\cup \textbf{Id}$ is a strong static location hp-bisimulation for some distributions. It can be proved similarly to the proof of Summation laws for strong hp-bisimulation in $\pi_{tc}$, we omit it;
  \item It is sufficient to prove the relation $R=\{(P_1+(P_2+P_3), (P_1+P_2)+P_3)\}\cup \textbf{Id}$ is a strong static location hp-bisimulation for some distributions. It can be proved similarly to the proof of Summation laws for strong hp-bisimulation in $\pi_{tc}$, we omit it.
\end{enumerate}
\end{proof}

\begin{theorem}[Summation laws for strong static location hhp-bisimulation]
The summation laws for strong static location hhp-bisimulation are as follows.

\begin{enumerate}
  \item $P+\textbf{nil}\sim_{hhp}^{sl} P$;
  \item $P+P\sim_{hhp}^{sl} P$;
  \item $P_1+P_2\sim_{hhp}^{sl} P_2+P_1$;
  \item $P_1+(P_2+P_3)\sim_{hhp}^{sl} (P_1+P_2)+P_3$.
\end{enumerate}
\end{theorem}

\begin{proof}
\begin{enumerate}
  \item It is sufficient to prove the relation $R=\{(P+\textbf{nil}, P)\}\cup \textbf{Id}$ is a strongly static location hhp-bisimulation for some distributions. It can be proved similarly to the proof of Summation laws for strong hhp-bisimulation in $\pi_{tc}$, we omit it;
  \item It is sufficient to prove the relation $R=\{(P+P, P)\}\cup \textbf{Id}$ is a strongly static location hhp-bisimulation for some distributions. It can be proved similarly to the proof of Summation laws for strong hhp-bisimulation in $\pi_{tc}$, we omit it;
  \item It is sufficient to prove the relation $R=\{(P_1+P_2, P_2+P_1)\}\cup \textbf{Id}$ is a strongly static location hhp-bisimulation for some distributions. It can be proved similarly to the proof of Summation laws for strong hhp-bisimulation in $\pi_{tc}$, we omit it;
  \item It is sufficient to prove the relation $R=\{(P_1+(P_2+P_3), (P_1+P_2)+P_3)\}\cup \textbf{Id}$ is a strongly static location hhp-bisimulation for some distributions. It can be proved similarly to the proof of Summation laws for strong hhp-bisimulation in $\pi_{tc}$, we omit it.
\end{enumerate}
\end{proof}

\begin{theorem}[Identity law for truly concurrent bisimilarities]
If $A(\widetilde{x})\overset{\text{def}}{=}P$, then

\begin{enumerate}
  \item $A(\widetilde{y})\sim_p^{sl} P\{\widetilde{y}/\widetilde{x}\}$;
  \item $A(\widetilde{y})\sim_s^{sl} P\{\widetilde{y}/\widetilde{x}\}$;
  \item $A(\widetilde{y})\sim_{hp}^{sl} P\{\widetilde{y}/\widetilde{x}\}$;
  \item $A(\widetilde{y})\sim_{hhp}^{sl} P\{\widetilde{y}/\widetilde{x}\}$.
\end{enumerate}
\end{theorem}

\begin{proof}
\begin{enumerate}
  \item It is straightforward to see that $R=\{A(\widetilde{y},P\{\widetilde{y}/\widetilde{x}\})\}\cup \textbf{Id}$ is a strong static location pomset bisimulation for some distributions, we omit it;
  \item It is straightforward to see that $R=\{A(\widetilde{y},P\{\widetilde{y}/\widetilde{x}\})\}\cup \textbf{Id}$ is a strong static location step bisimulation for some distributions, we omit it;
  \item It is straightforward to see that $R=\{A(\widetilde{y},P\{\widetilde{y}/\widetilde{x}\})\}\cup \textbf{Id}$ is a strong static location hp-bisimulation for some distributions, we omit it;
  \item It is straightforward to see that $R=\{A(\widetilde{y},P\{\widetilde{y}/\widetilde{x}\})\}\cup \textbf{Id}$ is a strongly static location hhp-bisimulation for some distributions, we omit it.
\end{enumerate}
\end{proof}

\begin{theorem}[Restriction Laws for strong static location pomset bisimulation]
The restriction laws for strong static location pomset bisimulation are as follows.

\begin{enumerate}
  \item $(y)P\sim_p^{sl} P$, if $y\notin fn(P)$;
  \item $(y)(z)P\sim_p^{sl} (z)(y)P$;
  \item $(y)(P+Q)\sim_p^{sl} (y)P+(y)Q$;
  \item $(y)\alpha.P\sim_p^{sl} \alpha.(y)P$ if $y\notin n(\alpha)$;
  \item $(y)\alpha.P\sim_p^{sl} \textbf{nil}$ if $y$ is the subject of $\alpha$.
\end{enumerate}
\end{theorem}

\begin{proof}
\begin{enumerate}
  \item It is sufficient to prove the relation $R=\{((y)P, P)|\textrm{ if }y\notin fn(P)\}\cup \textbf{Id}$ is a strong static location pomset bisimulation for some distributions. It can be proved similarly to the proof of restriction laws for strong pomset bisimulation in $\pi_{tc}$, we omit it;
  \item It is sufficient to prove the relation $R=\{((y)(z)P, (z)(y)P)\}\cup \textbf{Id}$ is a strong static location pomset bisimulation for some distributions. It can be proved similarly to the proof of restriction laws for strong pomset bisimulation in $\pi_{tc}$, we omit it;
  \item It is sufficient to prove the relation $R=\{((y)(P+Q), (y)P+(y)Q)\}\cup \textbf{Id}$ is a strong static location pomset bisimulation for some distributions. It can be proved similarly to the proof of restriction laws for strong pomset bisimulation in $\pi_{tc}$, we omit it;
  \item It is sufficient to prove the relation $R=\{((y)\alpha.P, \alpha.(y)P)|\textrm{ if }y\notin n(\alpha)\}\cup \textbf{Id}$ is a strong static location pomset bisimulation for some distributions. It can be proved similarly to the proof of restriction laws for strong pomset bisimulation in $\pi_{tc}$, we omit it;
  \item It is sufficient to prove the relation $R=\{((y)\alpha.P, \textbf{nil})|\textrm{ if }y\textrm{ is the subject of }\alpha\}\cup \textbf{Id}$ is a strong static location pomset bisimulation for some distributions. It can be proved similarly to the proof of restriction laws for strong pomset bisimulation in $\pi_{tc}$, we omit it.
\end{enumerate}
\end{proof}

\begin{theorem}[Restriction Laws for strong static location step bisimulation]
The restriction laws for strong static location step bisimulation are as follows.

\begin{enumerate}
  \item $(y)P\sim_s^{sl} P$, if $y\notin fn(P)$;
  \item $(y)(z)P\sim_s^{sl} (z)(y)P$;
  \item $(y)(P+Q)\sim_s^{sl} (y)P+(y)Q$;
  \item $(y)\alpha.P\sim_s^{sl} \alpha.(y)P$ if $y\notin n(\alpha)$;
  \item $(y)\alpha.P\sim_s^{sl} \textbf{nil}$ if $y$ is the subject of $\alpha$.
\end{enumerate}
\end{theorem}

\begin{proof}
\begin{enumerate}
  \item It is sufficient to prove the relation $R=\{((y)P, P)|\textrm{ if }y\notin fn(P)\}\cup \textbf{Id}$ is a strong static location step bisimulation for some distributions. It can be proved similarly to the proof of restriction laws for strong step bisimulation in $\pi_{tc}$, we omit it;
  \item It is sufficient to prove the relation $R=\{((y)(z)P, (z)(y)P)\}\cup \textbf{Id}$ is a strong static location step bisimulation for some distributions. It can be proved similarly to the proof of restriction laws for strong step bisimulation in $\pi_{tc}$, we omit it;
  \item It is sufficient to prove the relation $R=\{((y)(P+Q), (y)P+(y)Q)\}\cup \textbf{Id}$ is a strong static location step bisimulation for some distributions. It can be proved similarly to the proof of restriction laws for strong step bisimulation in $\pi_{tc}$, we omit it;
  \item It is sufficient to prove the relation $R=\{((y)\alpha.P, \alpha.(y)P)|\textrm{ if }y\notin n(\alpha)\}\cup \textbf{Id}$ is a strong static location step bisimulation for some distributions. It can be proved similarly to the proof of restriction laws for strong step bisimulation in $\pi_{tc}$, we omit it;
  \item It is sufficient to prove the relation $R=\{((y)\alpha.P, \textbf{nil})|\textrm{ if }y\textrm{ is the subject of }\alpha\}\cup \textbf{Id}$ is a strong static location step bisimulation for some distributions. It can be proved similarly to the proof of restriction laws for strong step bisimulation in $\pi_{tc}$, we omit it.
\end{enumerate}
\end{proof}

\begin{theorem}[Restriction Laws for strong static location hp-bisimulation]
The restriction laws for strong static location hp-bisimulation are as follows.

\begin{enumerate}
  \item $(y)P\sim_{hp}^{sl} P$, if $y\notin fn(P)$;
  \item $(y)(z)P\sim_{hp}^{sl} (z)(y)P$;
  \item $(y)(P+Q)\sim_{hp}^{sl} (y)P+(y)Q$;
  \item $(y)\alpha.P\sim_{hp}^{sl} \alpha.(y)P$ if $y\notin n(\alpha)$;
  \item $(y)\alpha.P\sim_{hp}^{sl} \textbf{nil}$ if $y$ is the subject of $\alpha$.
\end{enumerate}
\end{theorem}

\begin{proof}
\begin{enumerate}
  \item It is sufficient to prove the relation $R=\{((y)P, P)|\textrm{ if }y\notin fn(P)\}\cup \textbf{Id}$ is a strong static location hp-bisimulation for some distributions. It can be proved similarly to the proof of restriction laws for strong hp-bisimulation in $\pi_{tc}$, we omit it;
  \item It is sufficient to prove the relation $R=\{((y)(z)P, (z)(y)P)\}\cup \textbf{Id}$ is a strong static location hp-bisimulation for some distributions. It can be proved similarly to the proof of restriction laws for strong hp-bisimulation in $\pi_{tc}$, we omit it;
  \item It is sufficient to prove the relation $R=\{((y)(P+Q), (y)P+(y)Q)\}\cup \textbf{Id}$ is a strong static location hp-bisimulation for some distributions. It can be proved similarly to the proof of restriction laws for strong hp-bisimulation in $\pi_{tc}$, we omit it;
  \item It is sufficient to prove the relation $R=\{((y)\alpha.P, \alpha.(y)P)|\textrm{ if }y\notin n(\alpha)\}\cup \textbf{Id}$ is a strong static location hp-bisimulation for some distributions. It can be proved similarly to the proof of restriction laws for strong hp-bisimulation in $\pi_{tc}$, we omit it;
  \item It is sufficient to prove the relation $R=\{((y)\alpha.P, \textbf{nil})|\textrm{ if }y\textrm{ is the subject of }\alpha\}\cup \textbf{Id}$ is a strong static location hp-bisimulation for some distributions. It can be proved similarly to the proof of restriction laws for strong hp-bisimulation in $\pi_{tc}$, we omit it.
\end{enumerate}
\end{proof}

\begin{theorem}[Restriction Laws for strong static location hhp-bisimulation]
The restriction laws for strong static location hhp-bisimulation are as follows.

\begin{enumerate}
  \item $(y)P\sim_{hhp}^{sl} P$, if $y\notin fn(P)$;
  \item $(y)(z)P\sim_{hhp}^{sl} (z)(y)P$;
  \item $(y)(P+Q)\sim_{hhp}^{sl} (y)P+(y)Q$;
  \item $(y)\alpha.P\sim_{hhp}^{sl} \alpha.(y)P$ if $y\notin n(\alpha)$;
  \item $(y)\alpha.P\sim_{hhp}^{sl} \textbf{nil}$ if $y$ is the subject of $\alpha$.
\end{enumerate}
\end{theorem}

\begin{proof}
\begin{enumerate}
  \item It is sufficient to prove the relation $R=\{((y)P, P)|\textrm{ if }y\notin fn(P)\}\cup \textbf{Id}$ is a strongly static location hhp-bisimulation for some distributions. It can be proved similarly to the proof of restriction laws for strong hhp-bisimulation in $\pi_{tc}$, we omit it;
  \item It is sufficient to prove the relation $R=\{((y)(z)P, (z)(y)P)\}\cup \textbf{Id}$ is a strongly static location hhp-bisimulation for some distributions. It can be proved similarly to the proof of restriction laws for strong hhp-bisimulation in $\pi_{tc}$, we omit it;
  \item It is sufficient to prove the relation $R=\{((y)(P+Q), (y)P+(y)Q)\}\cup \textbf{Id}$ is a strongly static location hhp-bisimulation for some distributions. It can be proved similarly to the proof of restriction laws for strong hhp-bisimulation in $\pi_{tc}$, we omit it;
  \item It is sufficient to prove the relation $R=\{((y)\alpha.P, \alpha.(y)P)|\textrm{ if }y\notin n(\alpha)\}\cup \textbf{Id}$ is a strongly static location hhp-bisimulation for some distributions. It can be proved similarly to the proof of restriction laws for strong hhp-bisimulation in $\pi_{tc}$, we omit it;
  \item It is sufficient to prove the relation $R=\{((y)\alpha.P, \textbf{nil})|\textrm{ if }y\textrm{ is the subject of }\alpha\}\cup \textbf{Id}$ is a strongly static location hhp-bisimulation for some distributions. It can be proved similarly to the proof of restriction laws for strong hhp-bisimulation in $\pi_{tc}$, we omit it.
\end{enumerate}
\end{proof}

\begin{theorem}[Parallel laws for strong static location pomset bisimulation]
The parallel laws for strong static location pomset bisimulation are as follows.
\begin{enumerate}
  \item $P\parallel \textbf{nil}\sim_p^{sl} P$;
  \item $P_1\parallel P_2\sim_p^{sl} P_2\parallel P_1$;
  \item $(y)P_1\parallel P_2\sim_p^{sl} (y)(P_1\parallel P_2)$
  \item $(P_1\parallel P_2)\parallel P_3\sim_p^{sl} P_1\parallel (P_2\parallel P_3)$;
  \item $(y)(P_1\parallel P_2)\sim_p^{sl} (y)P_1\parallel (y)P_2$, if $y\notin fn(P_1)\cap fn(P_2)$.
\end{enumerate}
\end{theorem}

\begin{proof}
\begin{enumerate}
  \item It is sufficient to prove the relation $R=\{(P\parallel \textbf{nil}, P)\}\cup \textbf{Id}$ is a strong static location pomset bisimulation for some distributions. It can be proved similarly to the proof of parallel laws for strong pomset bisimulation in $\pi_{tc}$, we omit it;
  \item It is sufficient to prove the relation $R=\{(P_1\parallel P_2, P_2\parallel P_1)\}\cup \textbf{Id}$ is a strong static location pomset bisimulation for some distributions. It can be proved similarly to the proof of parallel laws for strong pomset bisimulation in $\pi_{tc}$, we omit it;
  \item It is sufficient to prove the relation $R=\{((y)P_1\parallel P_2, (y)(P_1\parallel P_2))\}\cup \textbf{Id}$ is a strong static location pomset bisimulation for some distributions. It can be proved similarly to the proof of parallel laws for strong pomset bisimulation in $\pi_{tc}$, we omit it;
  \item It is sufficient to prove the relation $R=\{((P_1\parallel P_2)\parallel P_3, P_1\parallel (P_2\parallel P_3))\}\cup \textbf{Id}$ is a strong static location pomset bisimulation for some distributions. It can be proved similarly to the proof of parallel laws for strong pomset bisimulation in $\pi_{tc}$, we omit it;
  \item It is sufficient to prove the relation $R=\{(y)(P_1\parallel P_2), (y)P_1\parallel (y)P_2)|\textrm{ if }y\notin fn(P_1)\cap fn(P_2)\}\cup \textbf{Id}$ is a strong static location pomset bisimulation for some distributions. It can be proved similarly to the proof of parallel laws for strong pomset bisimulation in $\pi_{tc}$, we omit it.
\end{enumerate}
\end{proof}

\begin{theorem}[Parallel laws for strong static location step bisimulation]
The parallel laws for strong static location step bisimulation are as follows.
\begin{enumerate}
  \item $P\parallel \textbf{nil}\sim_s^{sl} P$;
  \item $P_1\parallel P_2\sim_s^{sl} P_2\parallel P_1$;
  \item $(y)P_1\parallel P_2\sim_s^{sl} (y)(P_1\parallel P_2)$
  \item $(P_1\parallel P_2)\parallel P_3\sim_s^{sl} P_1\parallel (P_2\parallel P_3)$;
  \item $(y)(P_1\parallel P_2)\sim_s^{sl} (y)P_1\parallel (y)P_2$, if $y\notin fn(P_1)\cap fn(P_2)$.
\end{enumerate}
\end{theorem}

\begin{proof}
\begin{enumerate}
  \item It is sufficient to prove the relation $R=\{(P\parallel \textbf{nil}, P)\}\cup \textbf{Id}$ is a strong static location step bisimulation for some distributions. It can be proved similarly to the proof of parallel laws for strong step bisimulation in $\pi_{tc}$, we omit it;
  \item It is sufficient to prove the relation $R=\{(P_1\parallel P_2, P_2\parallel P_1)\}\cup \textbf{Id}$ is a strong static location step bisimulation for some distributions. It can be proved similarly to the proof of parallel laws for strong step bisimulation in $\pi_{tc}$, we omit it;
  \item It is sufficient to prove the relation $R=\{((y)P_1\parallel P_2, (y)(P_1\parallel P_2))\}\cup \textbf{Id}$ is a strong static location step bisimulation for some distributions. It can be proved similarly to the proof of parallel laws for strong step bisimulation in $\pi_{tc}$, we omit it;
  \item It is sufficient to prove the relation $R=\{((P_1\parallel P_2)\parallel P_3, P_1\parallel (P_2\parallel P_3))\}\cup \textbf{Id}$ is a strong static location step bisimulation for some distributions. It can be proved similarly to the proof of parallel laws for strong step bisimulation in $\pi_{tc}$, we omit it;
  \item It is sufficient to prove the relation $R=\{(y)(P_1\parallel P_2), (y)P_1\parallel (y)P_2)|\textrm{ if }y\notin fn(P_1)\cap fn(P_2)\}\cup \textbf{Id}$ is a strong static location step bisimulation for some distributions. It can be proved similarly to the proof of parallel laws for strong step bisimulation in $\pi_{tc}$, we omit it.
\end{enumerate}
\end{proof}

\begin{theorem}[Parallel laws for strong static location hp-bisimulation]
The parallel laws for strong static location hp-bisimulation are as follows.
\begin{enumerate}
  \item $P\parallel \textbf{nil}\sim_{hp}^{sl} P$;
  \item $P_1\parallel P_2\sim_{hp}^{sl} P_2\parallel P_1$;
  \item $(y)P_1\parallel P_2\sim_{hp}^{sl} (y)(P_1\parallel P_2)$
  \item $(P_1\parallel P_2)\parallel P_3\sim_{hp}^{sl} P_1\parallel (P_2\parallel P_3)$;
  \item $(y)(P_1\parallel P_2)\sim_{hp}^{sl} (y)P_1\parallel (y)P_2$, if $y\notin fn(P_1)\cap fn(P_2)$.
\end{enumerate}
\end{theorem}

\begin{proof}
\begin{enumerate}
  \item It is sufficient to prove the relation $R=\{(P\parallel \textbf{nil}, P)\}\cup \textbf{Id}$ is a strong static location hp-bisimulation for some distributions. It can be proved similarly to the proof of parallel laws for strong hp-bisimulation in $\pi_{tc}$, we omit it;
  \item It is sufficient to prove the relation $R=\{(P_1\parallel P_2, P_2\parallel P_1)\}\cup \textbf{Id}$ is a strong static location hp-bisimulation for some distributions. It can be proved similarly to the proof of parallel laws for strong hp-bisimulation in $\pi_{tc}$, we omit it;
  \item It is sufficient to prove the relation $R=\{((y)P_1\parallel P_2, (y)(P_1\parallel P_2))\}\cup \textbf{Id}$ is a strong static location hp-bisimulation for some distributions. It can be proved similarly to the proof of parallel laws for strong hp-bisimulation in $\pi_{tc}$, we omit it;
  \item It is sufficient to prove the relation $R=\{((P_1\parallel P_2)\parallel P_3, P_1\parallel (P_2\parallel P_3))\}\cup \textbf{Id}$ is a strong static location hp-bisimulation for some distributions. It can be proved similarly to the proof of parallel laws for strong hp-bisimulation in $\pi_{tc}$, we omit it;
  \item It is sufficient to prove the relation $R=\{(y)(P_1\parallel P_2), (y)P_1\parallel (y)P_2)|\textrm{ if }y\notin fn(P_1)\cap fn(P_2)\}\cup \textbf{Id}$ is a strong static location hp-bisimulation for some distributions. It can be proved similarly to the proof of parallel laws for strong hp-bisimulation in $\pi_{tc}$, we omit it.
\end{enumerate}
\end{proof}

\begin{theorem}[Parallel laws for strong static location hhp-bisimulation]
The parallel laws for strong static location hhp-bisimulation are as follows.
\begin{enumerate}
  \item $P\parallel \textbf{nil}\sim_{hhp}^{sl} P$;
  \item $P_1\parallel P_2\sim_{hhp}^{sl} P_2\parallel P_1$;
  \item $(y)P_1\parallel P_2\sim_{hhp}^{sl} (y)(P_1\parallel P_2)$
  \item $(P_1\parallel P_2)\parallel P_3\sim_{hhp}^{sl} P_1\parallel (P_2\parallel P_3)$;
  \item $(y)(P_1\parallel P_2)\sim_{hhp}^{sl} (y)P_1\parallel (y)P_2$, if $y\notin fn(P_1)\cap fn(P_2)$.
\end{enumerate}
\end{theorem}

\begin{proof}
\begin{enumerate}
  \item It is sufficient to prove the relation $R=\{(P\parallel \textbf{nil}, P)\}\cup \textbf{Id}$ is a strong static location hhp-bisimulation for some distributions. It can be proved similarly to the proof of parallel laws for strong hhp-bisimulation in $\pi_{tc}$, we omit it;
  \item It is sufficient to prove the relation $R=\{(P_1\parallel P_2, P_2\parallel P_1)\}\cup \textbf{Id}$ is a strong static location hhp-bisimulation for some distributions. It can be proved similarly to the proof of parallel laws for strong hhp-bisimulation in $\pi_{tc}$, we omit it;
  \item It is sufficient to prove the relation $R=\{((y)P_1\parallel P_2, (y)(P_1\parallel P_2))\}\cup \textbf{Id}$ is a strong static location hhp-bisimulation for some distributions. It can be proved similarly to the proof of parallel laws for strong hhp-bisimulation in $\pi_{tc}$, we omit it;
  \item It is sufficient to prove the relation $R=\{((P_1\parallel P_2)\parallel P_3, P_1\parallel (P_2\parallel P_3))\}\cup \textbf{Id}$ is a strong static location hhp-bisimulation for some distributions. It can be proved similarly to the proof of parallel laws for strong hhp-bisimulation in $\pi_{tc}$, we omit it;
  \item It is sufficient to prove the relation $R=\{(y)(P_1\parallel P_2), (y)P_1\parallel (y)P_2)|\textrm{ if }y\notin fn(P_1)\cap fn(P_2)\}\cup \textbf{Id}$ is a strong static location hhp-bisimulation for some distributions. It can be proved similarly to the proof of parallel laws for strong hhp-bisimulation in $\pi_{tc}$, we omit it.
\end{enumerate}
\end{proof}

\begin{proposition}[Location laws for strong static location pomset bisimulation]
The location laws for strong static location pomset bisimulation are as follows.

\begin{enumerate}
  \item $\epsilon::P\sim_p^{sl} P$;
  \item $u::\textbf{nil}\sim_p^{sl} \textbf{nil}$;
  \item $u::(\alpha.P)\sim_p^{sl} u::\alpha.u::P$;
  \item $u::(P+Q)\sim_p^{sl} u::P+u::Q$;
  \item $u::(P\parallel Q)\sim_p^{sl}u::P\parallel u::Q$;
  \item $u::(P\setminus L)\sim_p^{sl}u::P\setminus L$;
  \item $u::(P[f])\sim_p^{sl}u::P[f]$;
  \item $u::(v::P)\sim_p^{sl}uv::P$.
\end{enumerate}
\end{proposition}

\begin{proof}
\begin{enumerate}
  \item $\epsilon::P\sim_p^{sl} P$. It is sufficient to prove the relation $R=\{(\epsilon::P, P)\}\cup \textbf{Id}$ is a strong static location pomset bisimulation, we omit it;
  \item $u::\textbf{nil}\sim_p^{sl} \textbf{nil}$. It is sufficient to prove the relation $R=\{(u::\textbf{nil}, \textbf{nil})\}\cup \textbf{Id}$ is a strong static location pomset bisimulation, we omit it;
  \item $u::(\alpha.P)\sim_p^{sl} u::\alpha.u::P$. It is sufficient to prove the relation $R=\{(u::(\alpha.P), u::\alpha.u::P)\}\cup \textbf{Id}$ is a strong static location pomset bisimulation, we omit it;
  \item $u::(P+Q)\sim_p^{sl} u::P+u::Q$. It is sufficient to prove the relation $R=\{(u::(P+Q), u::P+u::Q)\}\cup \textbf{Id}$ is a strong static location pomset bisimulation, we omit it;
  \item $u::(P\parallel Q)\sim_p^{sl}u::P\parallel u::Q$. It is sufficient to prove the relation $R=\{(u::(P\parallel Q), u::P\parallel u::Q)\}\cup \textbf{Id}$ is a strong static location pomset bisimulation, we omit it;
  \item $u::(P\setminus L)\sim_p^{sl}u::P\setminus L$. It is sufficient to prove the relation $R=\{(u::(P\setminus L), u::P\setminus L)\}\cup \textbf{Id}$ is a strong static location pomset bisimulation, we omit it;
  \item $u::(P[f])\sim_p^{sl}u::P[f]$. It is sufficient to prove the relation $R=\{(u::(P[f]), u::P[f])\}\cup \textbf{Id}$ is a strong static location pomset bisimulation, we omit it;
  \item $u::(v::P)\sim_p^{sl}uv::P$. It is sufficient to prove the relation $R=\{(u::(v::P), uv::P)\}\cup \textbf{Id}$ is a strong static location pomset bisimulation, we omit it.
\end{enumerate}
\end{proof}

\begin{proposition}[Location laws for strong static location step bisimulation]
The location laws for strong static location step bisimulation are as follows.

\begin{enumerate}
  \item $\epsilon::P\sim_s^{sl} P$;
  \item $u::\textbf{nil}\sim_s^{sl} \textbf{nil}$;
  \item $u::(\alpha.P)\sim_s^{sl} u::\alpha.u::P$;
  \item $u::(P+Q)\sim_s^{sl} u::P+u::Q$;
  \item $u::(P\parallel Q)\sim_s^{sl}u::P\parallel u::Q$;
  \item $u::(P\setminus L)\sim_s^{sl}u::P\setminus L$;
  \item $u::(P[f])\sim_s^{sl}u::P[f]$;
  \item $u::(v::P)\sim_s^{sl}uv::P$.
\end{enumerate}
\end{proposition}

\begin{proof}
\begin{enumerate}
  \item $\epsilon::P\sim_s^{sl} P$. It is sufficient to prove the relation $R=\{(\epsilon::P, P)\}\cup \textbf{Id}$ is a strong static location step bisimulation, we omit it;
  \item $u::\textbf{nil}\sim_s^{sl} \textbf{nil}$. It is sufficient to prove the relation $R=\{(u::\textbf{nil}, \textbf{nil})\}\cup \textbf{Id}$ is a strong static location step bisimulation, we omit it;
  \item $u::(\alpha.P)\sim_s^{sl} u::\alpha.u::P$. It is sufficient to prove the relation $R=\{(u::(\alpha.P), u::\alpha.u::P)\}\cup \textbf{Id}$ is a strong static location step bisimulation, we omit it;
  \item $u::(P+Q)\sim_s^{sl} u::P+u::Q$. It is sufficient to prove the relation $R=\{(u::(P+Q), u::P+u::Q)\}\cup \textbf{Id}$ is a strong static location step bisimulation, we omit it;
  \item $u::(P\parallel Q)\sim_s^{sl}u::P\parallel u::Q$. It is sufficient to prove the relation $R=\{(u::(P\parallel Q), u::P\parallel u::Q)\}\cup \textbf{Id}$ is a strong static location step bisimulation, we omit it;
  \item $u::(P\setminus L)\sim_s^{sl}u::P\setminus L$. It is sufficient to prove the relation $R=\{(u::(P\setminus L), u::P\setminus L)\}\cup \textbf{Id}$ is a strong static location step bisimulation, we omit it;
  \item $u::(P[f])\sim_s^{sl}u::P[f]$. It is sufficient to prove the relation $R=\{(u::(P[f]), u::P[f])\}\cup \textbf{Id}$ is a strong static location step bisimulation, we omit it;
  \item $u::(v::P)\sim_s^{sl}uv::P$. It is sufficient to prove the relation $R=\{(u::(v::P), uv::P)\}\cup \textbf{Id}$ is a strong static location step bisimulation, we omit it.
\end{enumerate}
\end{proof}

\begin{proposition}[Location laws for strong static location hp-bisimulation]
The location laws for strong static location hp-bisimulation are as follows.

\begin{enumerate}
  \item $\epsilon::P\sim_{hp}^{sl} P$;
  \item $u::\textbf{nil}\sim_{hp}^{sl} \textbf{nil}$;
  \item $u::(\alpha.P)\sim_{hp}^{sl} u::\alpha.u::P$;
  \item $u::(P+Q)\sim_{hp}^{sl} u::P+u::Q$;
  \item $u::(P\parallel Q)\sim_{hp}^{sl}u::P\parallel u::Q$;
  \item $u::(P\setminus L)\sim_{hp}^{sl}u::P\setminus L$;
  \item $u::(P[f])\sim_{hp}^{sl}u::P[f]$;
  \item $u::(v::P)\sim_{hp}^{sl}uv::P$.
\end{enumerate}
\end{proposition}

\begin{proof}
\begin{enumerate}
  \item $\epsilon::P\sim_{hp}^{sl} P$. It is sufficient to prove the relation $R=\{(\epsilon::P, P)\}\cup \textbf{Id}$ is a strong static location hp-bisimulation, we omit it;
  \item $u::\textbf{nil}\sim_{hp}^{sl} \textbf{nil}$. It is sufficient to prove the relation $R=\{(u::\textbf{nil}, \textbf{nil})\}\cup \textbf{Id}$ is a strong static location hp-bisimulation, we omit it;
  \item $u::(\alpha.P)\sim_{hp}^{sl} u::\alpha.u::P$. It is sufficient to prove the relation $R=\{(u::(\alpha.P), u::\alpha.u::P)\}\cup \textbf{Id}$ is a strong static location hp-bisimulation, we omit it;
  \item $u::(P+Q)\sim_{hp}^{sl} u::P+u::Q$. It is sufficient to prove the relation $R=\{(u::(P+Q), u::P+u::Q)\}\cup \textbf{Id}$ is a strong static location hp-bisimulation, we omit it;
  \item $u::(P\parallel Q)\sim_{hp}^{sl}u::P\parallel u::Q$. It is sufficient to prove the relation $R=\{(u::(P\parallel Q), u::P\parallel u::Q)\}\cup \textbf{Id}$ is a strong static location hp-bisimulation, we omit it;
  \item $u::(P\setminus L)\sim_{hp}^{sl}u::P\setminus L$. It is sufficient to prove the relation $R=\{(u::(P\setminus L), u::P\setminus L)\}\cup \textbf{Id}$ is a strong static location hp-bisimulation, we omit it;
  \item $u::(P[f])\sim_{hp}^{sl}u::P[f]$. It is sufficient to prove the relation $R=\{(u::(P[f]), u::P[f])\}\cup \textbf{Id}$ is a strong static location hp-bisimulation, we omit it;
  \item $u::(v::P)\sim_{hp}^{sl}uv::P$. It is sufficient to prove the relation $R=\{(u::(v::P), uv::P)\}\cup \textbf{Id}$ is a strong static location hp-bisimulation, we omit it.
\end{enumerate}
\end{proof}

\begin{proposition}[Location laws for strong static location hhp-bisimulation]
The location laws for strong static location hhp-bisimulation are as follows.

\begin{enumerate}
  \item $\epsilon::P\sim_{hhp}^{sl} P$;
  \item $u::\textbf{nil}\sim_{hhp}^{sl} \textbf{nil}$;
  \item $u::(\alpha.P)\sim_{hhp}^{sl} u::\alpha.u::P$;
  \item $u::(P+Q)\sim_{hhp}^{sl} u::P+u::Q$;
  \item $u::(P\parallel Q)\sim_{hhp}^{sl}u::P\parallel u::Q$;
  \item $u::(P\setminus L)\sim_{hhp}^{sl}u::P\setminus L$;
  \item $u::(P[f])\sim_{hhp}^{sl}u::P[f]$;
  \item $u::(v::P)\sim_{hhp}^{sl}uv::P$.
\end{enumerate}
\end{proposition}

\begin{proof}
\begin{enumerate}
  \item $\epsilon::P\sim_{hhp}^{sl} P$. It is sufficient to prove the relation $R=\{(\epsilon::P, P)\}\cup \textbf{Id}$ is a strong static location hhp-bisimulation, we omit it;
  \item $u::\textbf{nil}\sim_{hhp}^{sl} \textbf{nil}$. It is sufficient to prove the relation $R=\{(u::\textbf{nil}, \textbf{nil})\}\cup \textbf{Id}$ is a strong static location hhp-bisimulation, we omit it;
  \item $u::(\alpha.P)\sim_{hhp}^{sl} u::\alpha.u::P$. It is sufficient to prove the relation $R=\{(u::(\alpha.P), u::\alpha.u::P)\}\cup \textbf{Id}$ is a strong static location hhp-bisimulation, we omit it;
  \item $u::(P+Q)\sim_{hhp}^{sl} u::P+u::Q$. It is sufficient to prove the relation $R=\{(u::(P+Q), u::P+u::Q)\}\cup \textbf{Id}$ is a strong static location hhp-bisimulation, we omit it;
  \item $u::(P\parallel Q)\sim_{hhp}^{sl}u::P\parallel u::Q$. It is sufficient to prove the relation $R=\{(u::(P\parallel Q), u::P\parallel u::Q)\}\cup \textbf{Id}$ is a strong static location hhp-bisimulation, we omit it;
  \item $u::(P\setminus L)\sim_{hhp}^{sl}u::P\setminus L$. It is sufficient to prove the relation $R=\{(u::(P\setminus L), u::P\setminus L)\}\cup \textbf{Id}$ is a strong static location hhp-bisimulation, we omit it;
  \item $u::(P[f])\sim_{hhp}^{sl}u::P[f]$. It is sufficient to prove the relation $R=\{(u::(P[f]), u::P[f])\}\cup \textbf{Id}$ is a strong static location hhp-bisimulation, we omit it;
  \item $u::(v::P)\sim_{hhp}^{sl}uv::P$. It is sufficient to prove the relation $R=\{(u::(v::P), uv::P)\}\cup \textbf{Id}$ is a strong static location hhp-bisimulation, we omit it.
\end{enumerate}
\end{proof}

\begin{theorem}[Expansion law for static location truly concurrent bisimilarities]
Let $P\equiv\sum_i \alpha_i.P_i$ and $Q\equiv\sum_j \beta_j.Q_j$, where $bn(\alpha_i)\cap fn(Q)=\emptyset$ for all $i$, and $bn(\beta_j)\cap fn(P)=\emptyset$ for all $j$.
Then

\begin{enumerate}
  \item $P\parallel Q\sim_p^{sl} \sum_i\sum_j (\alpha_i\parallel \beta_j).(P_i\parallel Q_j)+\sum_{\alpha_i \textrm{ comp }\beta_j}\tau.R_{ij}$;
  \item $P\parallel Q\sim_s^{sl} \sum_i\sum_j (\alpha_i\parallel \beta_j).(P_i\parallel Q_j)+\sum_{\alpha_i \textrm{ comp }\beta_j}\tau.R_{ij}$;
  \item $P\parallel Q\sim_{hp}^{sl} \sum_i\sum_j (\alpha_i\parallel \beta_j).(P_i\parallel Q_j)+\sum_{\alpha_i \textrm{ comp }\beta_j}\tau.R_{ij}$;
  \item $P\parallel Q\nsim_{hhp}^{sl} \sum_i\sum_j (\alpha_i\parallel \beta_j).(P_i\parallel Q_j)+\sum_{\alpha_i \textrm{ comp }\beta_j}\tau.R_{ij}$.
\end{enumerate}

Where $\alpha_i$ comp $\beta_j$ and $R_{ij}$ are defined as follows:

\begin{enumerate}
  \item $\alpha_i$ is $\overline{x}u$ and $\beta_j$ is $x(v)$, then $R_{ij}=P_i\parallel Q_j\{u/v\}$;
  \item $\alpha_i$ is $\overline{x}(u)$ and $\beta_j$ is $x(v)$, then $R_{ij}=(w)(P_i\{w/u\}\parallel lQ_j\{w/v\})$, if $w\notin fn((u)P_i)\cup fn((v)Q_j)$;
  \item $\alpha_i$ is $x(v)$ and $\beta_j$ is $\overline{x}u$, then $R_{ij}=P_i\{u/v\}\parallel Q_j$;
  \item $\alpha_i$ is $x(v)$ and $\beta_j$ is $\overline{x}(u)$, then $R_{ij}=(w)(P_i\{w/v\}\parallel Q_j\{w/u\})$, if $w\notin fn((v)P_i)\cup fn((u)Q_j)$.
\end{enumerate}
\end{theorem}

\begin{proof}
\begin{enumerate}
  \item It is sufficient to prove the relation $R=\{(P\parallel Q, \sum_i\sum_j (\alpha_i\parallel \beta_j).(P_i\parallel Q_j)+\sum_{\alpha_i \textrm{ comp }\beta_j}\tau.R_{ij})|\textrm{ if }y\notin fn(P)\}\cup \textbf{Id}$ is a strong static location pomset bisimulation for some distributions. It can be proved similarly to the proof of Expansion law for strong pomset bisimulation in $\pi_{tc}$, we omit it;
  \item It is sufficient to prove the relation $R=\{(P\parallel Q, \sum_i\sum_j (\alpha_i\parallel \beta_j).(P_i\parallel Q_j)+\sum_{\alpha_i \textrm{ comp }\beta_j}\tau.R_{ij})|\textrm{ if }y\notin fn(P)\}\cup \textbf{Id}$ is a strong static location step bisimulation for some distributions. It can be proved similarly to the proof of Expansion law for strong step bisimulation in $\pi_{tc}$, we omit it;
  \item It is sufficient to prove the relation $R=\{(P\parallel Q, \sum_i\sum_j (\alpha_i\parallel \beta_j).(P_i\parallel Q_j)+\sum_{\alpha_i \textrm{ comp }\beta_j}\tau.R_{ij})|\textrm{ if }y\notin fn(P)\}\cup \textbf{Id}$ is a strong static location hp-bisimulation for some distributions. It can be proved similarly to the proof of Expansion law for strong hp-bisimulation in $\pi_{tc}$, we omit it;
  \item We just prove that for free actions $a,b,c$, let $s_1=(a+ b)\parallel c$, $t_1=(a\parallel c)+ (b\parallel c)$, and $s_2=a\parallel (b+ c)$, $t_2=(a\parallel b)+ (a\parallel c)$.
  We know that $s_1\sim_{hp}^{sl} t_1$ and $s_2\sim_{hp}^{sl} t_2$, we prove that $s_1\nsim_{hhp} t_1$ and $s_2\nsim_{hhp} t_2$. Let $(C(s_1),f_1,C(t_1))$ and $(C(s_2),f_2,C(t_2))$ are
  the corresponding posetal products.
    \begin{itemize}
        \item $s_1\nsim_{hhp} t_1$. $s_1\xrightarrow{\{a,c\}}\surd(s_1')$ ($C(s_1)\xrightarrow{\{a,c\}}C(s_1')$), then $t_1\xrightarrow{\{a,c\}}\surd(t_1')$
        ($C(t_1)\xrightarrow{\{a,c\}}C(t_1')$), we define $f_1'=f_1[a\mapsto a, c\mapsto c]$, obviously, $(C(s_1),f_1,C(t_1))\in \sim_{hp}^{sl}$ and
        $(C(s_1'),f_1',C(t_1'))\in \sim_{hp}^{sl}$. But, $(C(s_1),f_1,C(t_1))\in \sim_{hhp}^{sl}$ and $(C(s_1'),f_1',C(t_1'))\in \nsim_{hhp}$, just because they are not downward
        closed. Let $(C(s_1''),f_1'',C(t_1''))$, and $f_1''=f_1[c\mapsto c]$, $s_1\xrightarrow{c}s_1''$ ($C(s_1)\xrightarrow{c}C(s_1'')$), $t_1\xrightarrow{c}t_1''$
        ($C(t_1)\xrightarrow{c}C(t_1'')$), it is easy to see that $(C(s_1''),f_1'',C(t_1''))\subseteq (C(s_1'),f_1',C(t_1'))$ pointwise, while
        $(C(s_1''),f_1'',C(t_1''))\notin \sim_{hp}^{sl}$, because $s_1''$ and $C(s_1'')$ exist, but $t_1''$ and $C(t_1'')$ do not exist.
        \item $s_2\nsim_{hhp} t_2$. $s_2\xrightarrow{\{a,c\}}\surd(s_2')$ ($C(s_2)\xrightarrow{\{a,c\}}C(s_2')$), then $t_2\xrightarrow{\{a,c\}}\surd(t_2')$
        ($C(t_2)\xrightarrow{\{a,c\}}C(t_2')$), we define $f_2'=f_2[a\mapsto a, c\mapsto c]$, obviously, $(C(s_2),f_2,C(t_2))\in \sim_{hp}^{sl}$ and
        $(C(s_2'),f_2',C(t_2'))\in \sim_{hp}^{sl}$. But, $(C(s_2),f_2,C(t_2))\in \sim_{hhp}^{sl}$ and $(C(s_2'),f_2',C(t_2'))\in \nsim_{hhp}$, just because they are not downward
        closed. Let $(C(s_2''),f_2'',C(t_2''))$, and $f_2''=f_2[a\mapsto a]$, $s_2\xrightarrow{a}s_2''$ ($C(s_2)\xrightarrow{a}C(s_2'')$), $t_2\xrightarrow{a}t_2''$
        ($C(t_2)\xrightarrow{a}C(t_2'')$), it is easy to see that $(C(s_2''),f_2'',C(t_2''))\subseteq (C(s_2'),f_2',C(t_2'))$ pointwise, while
        $(C(s_2''),f_2'',C(t_2''))\notin \sim_{hp}^{sl}$, because $s_2''$ and $C(s_2'')$ exist, but $t_2''$ and $C(t_2'')$ do not exist.
    \end{itemize}
\end{enumerate}
\end{proof}

\begin{theorem}[Equivalence and congruence for strong static location pomset bisimulation]
\begin{enumerate}
  \item $\sim_p^{sl}$ is an equivalence relation;
  \item If $P\sim_p^{sl} Q$ then
  \begin{enumerate}
    \item $loc::P\sim_p^{sl}loc::Q$;
    \item $\alpha.P\sim_p^{sl} \alpha.Q$, $\alpha$ is a free action;
    \item $P+R\sim_p^{sl} Q+R$;
    \item $P\parallel R\sim_p^{sl} Q\parallel R$;
    \item $(w)P\sim_p^{sl} (w)Q$;
    \item $x(y).P\sim_p^{sl} x(y).Q$.
  \end{enumerate}
\end{enumerate}
\end{theorem}

\begin{proof}
\begin{enumerate}
  \item It is sufficient to prove that $\sim_p^{sl}$ is reflexivity, symmetry, and transitivity, we omit it.
  \item If $P\sim_p^{sl} Q$, then
  \begin{enumerate}
    \item it is sufficient to prove the relation $R=\{(loc::P, loc::.Q)\}\cup \textbf{Id}$ is a strong static location pomset bisimulation, we omit it;
    \item it is sufficient to prove the relation $R=\{(\alpha.P, \alpha.Q)|\alpha \textrm{ is a free action}\}\cup \textbf{Id}$ is a strong static location pomset bisimulation for some distributions. It can be proved similarly to the proof of congruence for strong  pomset bisimulation in $\pi_{tc}$, we omit it;
    \item it is sufficient to prove the relation $R=\{(P+R, Q+R)\}\cup \textbf{Id}$ is a strong static location pomset bisimulation for some distributions. It can be proved similarly to the proof of congruence for strong pomset bisimulation in $\pi_{tc}$, we omit it;
    \item it is sufficient to prove the relation $R=\{(P\parallel R, Q\parallel R)\}\cup \textbf{Id}$ is a strong static location pomset bisimulation for some distributions. It can be proved similarly to the proof of congruence for strong pomset bisimulation in $\pi_{tc}$, we omit it;
    \item it is sufficient to prove the relation $R=\{((w)P, (w).Q)\}\cup \textbf{Id}$ is a strong static location pomset bisimulation for some distributions. It can be proved similarly to the proof of congruence for strong pomset bisimulation in $\pi_{tc}$, we omit it;
    \item it is sufficient to prove the relation $R=\{(x(y).P, x(y).Q)\}\cup \textbf{Id}$ is a strong static location pomset bisimulation for some distributions. It can be proved similarly to the proof of congruence for strong pomset bisimulation in $\pi_{tc}$, we omit it.
  \end{enumerate}
\end{enumerate}
\end{proof}

\begin{theorem}[Equivalence and congruence for strong static location step bisimulation]
\begin{enumerate}
  \item $\sim_s^{sl}$ is an equivalence relation;
  \item If $P\sim_s^{sl} Q$ then
  \begin{enumerate}
    \item $loc::P\sim_s^{sl}loc::Q$;
    \item $\alpha.P\sim_s^{sl} \alpha.Q$, $\alpha$ is a free action;
    \item $P+R\sim_s^{sl} Q+R$;
    \item $P\parallel R\sim_s^{sl} Q\parallel R$;
    \item $(w)P\sim_s^{sl} (w)Q$;
    \item $x(y).P\sim_s^{sl} x(y).Q$.
  \end{enumerate}
\end{enumerate}
\end{theorem}

\begin{proof}
\begin{enumerate}
  \item It is sufficient to prove that $\sim_s^{sl}$ is reflexivity, symmetry, and transitivity, we omit it.
  \item If $P\sim_s^{sl} Q$, then
  \begin{enumerate}
    \item it is sufficient to prove the relation $R=\{(loc::P, loc::.Q)\}\cup \textbf{Id}$ is a strong static location step bisimulation, we omit it;
    \item it is sufficient to prove the relation $R=\{(\alpha.P, \alpha.Q)|\alpha \textrm{ is a free action}\}\cup \textbf{Id}$ is a strong static location step bisimulation for some distributions. It can be proved similarly to the proof of congruence for strong step bisimulation in $\pi_{tc}$, we omit it;
    \item it is sufficient to prove the relation $R=\{(P+R, Q+R)\}\cup \textbf{Id}$ is a strong static location step bisimulation for some distributions. It can be proved similarly to the proof of congruence for strong step bisimulation in $\pi_{tc}$, we omit it;
    \item it is sufficient to prove the relation $R=\{(P\parallel R, Q\parallel R)\}\cup \textbf{Id}$ is a strong static location step bisimulation for some distributions. It can be proved similarly to the proof of congruence for strong step bisimulation in $\pi_{tc}$, we omit it;
    \item it is sufficient to prove the relation $R=\{((w)P, (w).Q)\}\cup \textbf{Id}$ is a strong static location step bisimulation for some distributions. It can be proved similarly to the proof of congruence for strong step bisimulation in $\pi_{tc}$, we omit it;
    \item it is sufficient to prove the relation $R=\{(x(y).P, x(y).Q)\}\cup \textbf{Id}$ is a strong static location step bisimulation for some distributions. It can be proved similarly to the proof of congruence for strong step bisimulation in $\pi_{tc}$, we omit it.
  \end{enumerate}
\end{enumerate}
\end{proof}

\begin{theorem}[Equivalence and congruence for strong static location hp-bisimulation]
\begin{enumerate}
  \item $\sim_{hp}^{sl}$ is an equivalence relation;
  \item If $P\sim_{hp}^{sl} Q$ then
  \begin{enumerate}
    \item $loc::P\sim_{hp}^{sl}loc::Q$;
    \item $\alpha.P\sim_{hp}^{sl} \alpha.Q$, $\alpha$ is a free action;
    \item $P+R\sim_{hp}^{sl} Q+R$;
    \item $P\parallel R\sim_{hp}^{sl} Q\parallel R$;
    \item $(w)P\sim_{hp}^{sl} (w)Q$;
    \item $x(y).P\sim_{hp}^{sl} x(y).Q$.
  \end{enumerate}
\end{enumerate}
\end{theorem}

\begin{proof}
\begin{enumerate}
  \item It is sufficient to prove that $\sim_{hp}^{sl}$ is reflexivity, symmetry, and transitivity, we omit it.
  \item If $P\sim_{hp}^{sl} Q$, then
  \begin{enumerate}
    \item it is sufficient to prove the relation $R=\{(loc::P, loc::.Q)\}\cup \textbf{Id}$ is a strong static location hp-bisimulation, we omit it;
    \item it is sufficient to prove the relation $R=\{(\alpha.P, \alpha.Q)|\alpha \textrm{ is a free action}\}\cup \textbf{Id}$ is a strong static location hp-bisimulation for some distributions. It can be proved similarly to the proof of congruence for strong hp-bisimulation in $\pi_{tc}$, we omit it;
    \item it is sufficient to prove the relation $R=\{(P+R, Q+R)\}\cup \textbf{Id}$ is a strong static location hp-bisimulation for some distributions. It can be proved similarly to the proof of congruence for strong hp-bisimulation in $\pi_{tc}$, we omit it;
    \item it is sufficient to prove the relation $R=\{(P\parallel R, Q\parallel R)\}\cup \textbf{Id}$ is a strong static location hp-bisimulation for some distributions. It can be proved similarly to the proof of congruence for strong hp-bisimulation in $\pi_{tc}$, we omit it;
    \item it is sufficient to prove the relation $R=\{((w)P, (w).Q)\}\cup \textbf{Id}$ is a strong static location hp-bisimulation for some distributions. It can be proved similarly to the proof of congruence for strong hp-bisimulation in $\pi_{tc}$, we omit it;
    \item it is sufficient to prove the relation $R=\{(x(y).P, x(y).Q)\}\cup \textbf{Id}$ is a strong static location hp-bisimulation for some distributions. It can be proved similarly to the proof of congruence for strong hp-bisimulation in $\pi_{tc}$, we omit it.
  \end{enumerate}
\end{enumerate}
\end{proof}

\begin{theorem}[Equivalence and congruence for strong static location hhp-bisimulation]
\begin{enumerate}
  \item $\sim_{hhp}^{sl}$ is an equivalence relation;
  \item If $P\sim_{hhp}^{sl} Q$ then
  \begin{enumerate}
    \item $loc::P\sim_{hhp}^{sl}loc::Q$;
    \item $\alpha.P\sim_{hhp}^{sl} \alpha.Q$, $\alpha$ is a free action;
    \item $P+R\sim_{hhp}^{sl} Q+R$;
    \item $P\parallel R\sim_{hhp}^{sl} Q\parallel R$;
    \item $(w)P\sim_{hhp}^{sl} (w)Q$;
    \item $x(y).P\sim_{hhp}^{sl} x(y).Q$.
  \end{enumerate}
\end{enumerate}
\end{theorem}

\begin{proof}
\begin{enumerate}
  \item It is sufficient to prove that $\sim_{hhp}^{sl}$ is reflexivity, symmetry, and transitivity, we omit it.
  \item If $P\sim_p^{sl} Q$, then
  \begin{enumerate}
    \item it is sufficient to prove the relation $R=\{(loc::P, loc::.Q)\}\cup \textbf{Id}$ is a strong static location hhp-bisimulation, we omit it;
    \item it is sufficient to prove the relation $R=\{(\alpha.P, \alpha.Q)|\alpha \textrm{ is a free action}\}\cup \textbf{Id}$ is a strongly static location hhp-bisimulation for some distributions. It can be proved similarly to the proof of congruence for strong hhp-bisimulation in $\pi_{tc}$, we omit it;
    \item it is sufficient to prove the relation $R=\{(P+R, Q+R)\}\cup \textbf{Id}$ is a strongly static location hhp-bisimulation for some distributions. It can be proved similarly to the proof of congruence for strong hhp-bisimulation in $\pi_{tc}$, we omit it;
    \item it is sufficient to prove the relation $R=\{(P\parallel R, Q\parallel R)\}\cup \textbf{Id}$ is a strongly static location hhp-bisimulation for some distributions. It can be proved similarly to the proof of congruence for strong hhp-bisimulation in $\pi_{tc}$, we omit it;
    \item it is sufficient to prove the relation $R=\{((w)P, (w).Q)\}\cup \textbf{Id}$ is a strongly static location hhp-bisimulation for some distributions. It can be proved similarly to the proof of congruence for strong hhp-bisimulation in $\pi_{tc}$, we omit it;
    \item it is sufficient to prove the relation $R=\{(x(y).P, x(y).Q)\}\cup \textbf{Id}$ is a strongly static location hhp-bisimulation for some distributions. It can be proved similarly to the proof of congruence for strong hhp-bisimulation in $\pi_{tc}$, we omit it.
  \end{enumerate}
\end{enumerate}
\end{proof}


\begin{definition}
Let $X$ have arity $n$, and let $\widetilde{x}=x_1,\cdots,x_n$ be distinct names, and $fn(P)\subseteq\{x_1,\cdots,x_n\}$. The replacement of $X(\widetilde{x})$ by $P$ in $E$,
written $E\{X(\widetilde{x}):=P\}$, means the result of replacing each subterm $X(\widetilde{y})$ in $E$ by $P\{\widetilde{y}/\widetilde{x}\}$.
\end{definition}

\begin{definition}
Let $E$ and $F$ be two process expressions containing only $X_1,\cdots,X_m$ with associated name sequences $\widetilde{x}_1,\cdots,\widetilde{x}_m$. Then,
\begin{enumerate}
  \item $E\sim_p^{sl} F$ means $E(\widetilde{P})\sim_p^{sl} F(\widetilde{P})$;
  \item $E\sim_s^{sl} F$ means $E(\widetilde{P})\sim_s^{sl} F(\widetilde{P})$;
  \item $E\sim_{hp}^{sl} F$ means $E(\widetilde{P})\sim_{hp}^{sl} F(\widetilde{P})$;
  \item $E\sim_{hhp}^{sl} F$ means $E(\widetilde{P})\sim_{hhp}^{sl} F(\widetilde{P})$;
\end{enumerate}

for all $\widetilde{P}$ such that $fn(P_i)\subseteq \widetilde{x}_i$ for each $i$.
\end{definition}

\begin{definition}
A term or identifier is weakly guarded in $P$ if it lies within some subterm $loc::\alpha.Q$ or $(loc_1::\alpha_1\parallel\cdots\parallel loc_n::\alpha_n).Q$ of $P$.
\end{definition}

\begin{theorem}
Assume that $\widetilde{E}$ and $\widetilde{F}$ are expressions containing only $X_i$ with $\widetilde{x}_i$, and $\widetilde{A}$ and $\widetilde{B}$ are identifiers with $A_i$, $B_i$.
Then, for all $i$,
\begin{enumerate}
  \item $E_i\sim_s^{sl} F_i$, $A_i(\widetilde{x}_i)\overset{\text{def}}{=}E_i(\widetilde{A})$, $B_i(\widetilde{x}_i)\overset{\text{def}}{=}F_i(\widetilde{B})$, then
  $A_i(\widetilde{x}_i)\sim_s^{sl} B_i(\widetilde{x}_i)$;
  \item $E_i\sim_p^{sl} F_i$, $A_i(\widetilde{x}_i)\overset{\text{def}}{=}E_i(\widetilde{A})$, $B_i(\widetilde{x}_i)\overset{\text{def}}{=}F_i(\widetilde{B})$, then
  $A_i(\widetilde{x}_i)\sim_p^{sl} B_i(\widetilde{x}_i)$;
  \item $E_i\sim_{hp}^{sl} F_i$, $A_i(\widetilde{x}_i)\overset{\text{def}}{=}E_i(\widetilde{A})$, $B_i(\widetilde{x}_i)\overset{\text{def}}{=}F_i(\widetilde{B})$, then
  $A_i(\widetilde{x}_i)\sim_{hp}^{sl} B_i(\widetilde{x}_i)$;
  \item $E_i\sim_{hhp}^{sl} F_i$, $A_i(\widetilde{x}_i)\overset{\text{def}}{=}E_i(\widetilde{A})$, $B_i(\widetilde{x}_i)\overset{\text{def}}{=}F_i(\widetilde{B})$, then
  $A_i(\widetilde{x}_i)\sim_{hhp}^{sl} B_i(\widetilde{x}_i)$.
\end{enumerate}
\end{theorem}

\begin{proof}
\begin{enumerate}
  \item $E_i\sim_s^{sl} F_i$, $A_i(\widetilde{x}_i)\overset{\text{def}}{=}E_i(\widetilde{A})$, $B_i(\widetilde{x}_i)\overset{\text{def}}{=}F_i(\widetilde{B})$, then $A_i(\widetilde{x}_i)\sim_s^{sl} B_i(\widetilde{x}_i)$.

      We will consider the case $I=\{1\}$ with loss of generality, and show the following relation $R$ is a strong static location step bisimulation for some distributions.

      $$R=\{(G(A),G(B)):G\textrm{ has only identifier }X\}.$$

      By choosing $G\equiv X(\widetilde{y})$, it follows that $A(\widetilde{y})\sim_s^{sl} B(\widetilde{y})$. It is sufficient to prove the following:
      \begin{enumerate}
        \item If $G(A)\xrightarrow[u]{\{\alpha_1,\cdots,\alpha_n\}}P'$, where $\alpha_i(1\leq i\leq n)$ is a free action or bound output action with $bn(\alpha_1)\cap\cdots\cap bn(\alpha_n)\cap n(G(A),G(B))=\emptyset$, then
        $G(B)\xrightarrow[u]{\{\alpha_1,\cdots,\alpha_n\}}Q''$ such that $P'\sim_s^{sl} Q''$;
        \item If $G(A)[u]\xrightarrow{x(y)}P'$ with $x\notin n(G(A),G(B))$, then $G(B)\xrightarrow[u]{x(y)}Q''$, such that for all $u$, $P'\{u/y\}\sim_s^{sl} Q''\{u/y\}$.
      \end{enumerate}

      To prove the above properties, it is sufficient to induct on the depth of inference and quite routine, we omit it.
  \item $E_i\sim_p^{sl} F_i$, $A_i(\widetilde{x}_i)\overset{\text{def}}{=}E_i(\widetilde{A})$, $B_i(\widetilde{x}_i)\overset{\text{def}}{=}F_i(\widetilde{B})$, then
  $A_i(\widetilde{x}_i)\sim_p^{sl} B_i(\widetilde{x}_i)$. It can be proven similarly to the above case.
  \item $E_i\sim_{hp}^{sl} F_i$, $A_i(\widetilde{x}_i)\overset{\text{def}}{=}E_i(\widetilde{A})$, $B_i(\widetilde{x}_i)\overset{\text{def}}{=}F_i(\widetilde{B})$, then
  $A_i(\widetilde{x}_i)\sim_{hp}^{sl} B_i(\widetilde{x}_i)$. It can be proven similarly to the above case.
  \item $E_i\sim_{hhp}^{sl} F_i$, $A_i(\widetilde{x}_i)\overset{\text{def}}{=}E_i(\widetilde{A})$, $B_i(\widetilde{x}_i)\overset{\text{def}}{=}F_i(\widetilde{B})$, then
  $A_i(\widetilde{x}_i)\sim_{hhp}^{sl} B_i(\widetilde{x}_i)$. It can be proven similarly to the above case.
\end{enumerate}
\end{proof}

\begin{theorem}[Unique solution of equations]
Assume $\widetilde{E}$ are expressions containing only $X_i$ with $\widetilde{x}_i$, and each $X_i$ is weakly guarded in each $E_j$. Assume that $\widetilde{P}$ and $\widetilde{Q}$ are processes such that $fn(P_i)\subseteq \widetilde{x}_i$ and $fn(Q_i)\subseteq \widetilde{x}_i$. Then, for all $i$,
\begin{enumerate}
  \item if $P_i\sim_p^{sl} E_i(\widetilde{P})$, $Q_i\sim_p^{sl} E_i(\widetilde{Q})$, then $P_i\sim_p^{sl} Q_i$;
  \item if $P_i\sim_s^{sl} E_i(\widetilde{P})$, $Q_i\sim_s^{sl} E_i(\widetilde{Q})$, then $P_i\sim_s^{sl} Q_i$;
  \item if $P_i\sim_{hp}^{sl} E_i(\widetilde{P})$, $Q_i\sim_{hp}^{sl} E_i(\widetilde{Q})$, then $P_i\sim_{hp}^{sl} Q_i$;
  \item if $P_i\sim_{hhp}^{sl} E_i(\widetilde{P})$, $Q_i\sim_{hhp}^{sl} E_i(\widetilde{Q})$, then $P_i\sim_{hhp}^{sl} Q_i$.
\end{enumerate}
\end{theorem}

\begin{proof}
\begin{enumerate}
  \item It is similar to the proof of unique solution of equations for strong pomset bisimulation in $\pi_{tc}$, we omit it;
  \item It is similar to the proof of unique solution of equations for strong step bisimulation in $\pi_{tc}$, we omit it;
  \item It is similar to the proof of unique solution of equations for strong hp-bisimulation in $\pi_{tc}$, we omit it;
  \item It is similar to the proof of unique solution of equations for strong hhp-bisimulation in $\pi_{tc}$, we omit it.
\end{enumerate}
\end{proof}

\subsubsection{Algebraic Theory}\label{at5}

In this section, we will try to axiomatize $\pi_{tc}$ with static localities, the theory is \textbf{STC}.

\begin{definition}[STC]
The theory \textbf{STC} is consisted of the following axioms and inference rules:

\begin{enumerate}
  \item Alpha-conversion $\textbf{A}$.
  \[\textrm{if } P\equiv Q, \textrm{ then } P=Q\]
  \item Congruence $\textbf{C}$. If $P=Q$, then,
  \[loc::P=loc::Q\]
  \[\tau.P=\tau.Q\quad \overline{x}y.P=\overline{x}y.Q\]
  \[P+R=Q+R\quad P\parallel R=Q\parallel R\]
  \[(x)P=(x)Q\quad x(y).P=x(y).Q\]
  \item Summation $\textbf{S}$.
  \[\textbf{S0}\quad P+\textbf{nil}=P\]
  \[\textbf{S1}\quad P+P=P\]
  \[\textbf{S2}\quad P+Q=Q+P\]
  \[\textbf{S3}\quad P+(Q+R)=(P+Q)+R\]
  \item Restriction $\textbf{R}$.
  \[\textbf{R0}\quad (x)P=P\quad \textrm{ if }x\notin fn(P)\]
  \[\textbf{R1}\quad (x)(y)P=(y)(x)P\]
  \[\textbf{R2}\quad (x)(P+Q)=(x)P+(x)Q\]
  \[\textbf{R3}\quad (x)\alpha.P=\alpha.(x)P\quad \textrm{ if }x\notin n(\alpha)\]
  \[\textbf{R4}\quad (x)\alpha.P=\textbf{nil}\quad \textrm{ if }x\textrm{is the subject of }\alpha\]
  \item Location $\textbf{L}$.
  \[\textbf{L0}\quad \epsilon::P= P\]
  \[\textbf{L1}\quad u::\textbf{nil}=\textbf{nil}\]
  \[\textbf{L2}\quad u::(\alpha.P)=u::\alpha.u::P\]
  \[\textbf{L3}\quad u::(P+Q)=u::P+u::Q\]
  \[\textbf{L4}\quad u::(P\parallel Q)=u::P\parallel u::Q\]
  \[\textbf{L5}\quad u::(P\setminus L)=u::P\setminus L\]
  \[\textbf{L6}\quad u::(P[f])=u::P[f]\]
  \[\textbf{L7}\quad u::(v::P)=uv::P\]
  \item Expansion $\textbf{E}$.
  Let $P\equiv\sum_i \alpha_i.P_i$ and $Q\equiv\sum_j \beta_j.Q_j$, where $bn(\alpha_i)\cap fn(Q)=\emptyset$ for all $i$, and $bn(\beta_j)\cap fn(P)=\emptyset$ for all $j$.
Then

\begin{enumerate}
  \item $P\parallel Q\sim_p^{sl} \sum_i\sum_j (\alpha_i\parallel \beta_j).(P_i\parallel Q_j)+\sum_{\alpha_i \textrm{ comp }\beta_j}\tau.R_{ij}$;
  \item $P\parallel Q\sim_s^{sl} \sum_i\sum_j (\alpha_i\parallel \beta_j).(P_i\parallel Q_j)+\sum_{\alpha_i \textrm{ comp }\beta_j}\tau.R_{ij}$;
  \item $P\parallel Q\sim_{hp}^{sl} \sum_i\sum_j (\alpha_i\parallel \beta_j).(P_i\parallel Q_j)+\sum_{\alpha_i \textrm{ comp }\beta_j}\tau.R_{ij}$;
  \item $P\parallel Q\nsim_{hhp}^{sl} \sum_i\sum_j (\alpha_i\parallel \beta_j).(P_i\parallel Q_j)+\sum_{\alpha_i \textrm{ comp }\beta_j}\tau.R_{ij}$.
\end{enumerate}

Where $\alpha_i$ comp $\beta_j$ and $R_{ij}$ are defined as follows:

\begin{enumerate}
  \item $\alpha_i$ is $\overline{x}u$ and $\beta_j$ is $x(v)$, then $R_{ij}=P_i\parallel Q_j\{u/v\}$;
  \item $\alpha_i$ is $\overline{x}(u)$ and $\beta_j$ is $x(v)$, then $R_{ij}=(w)(P_i\{w/u\}\parallel lQ_j\{w/v\})$, if $w\notin fn((u)P_i)\cup fn((v)Q_j)$;
  \item $\alpha_i$ is $x(v)$ and $\beta_j$ is $\overline{x}u$, then $R_{ij}=P_i\{u/v\}\parallel Q_j$;
  \item $\alpha_i$ is $x(v)$ and $\beta_j$ is $\overline{x}(u)$, then $R_{ij}=(w)(P_i\{w/v\}\parallel Q_j\{w/u\})$, if $w\notin fn((v)P_i)\cup fn((u)Q_j)$.
\end{enumerate}
  \item Identifier $\textbf{I}$.
  \[\textrm{If }A(\widetilde{x})\overset{\text{def}}{=}P,\textrm{ then }A(\widetilde{y})= P\{\widetilde{y}/\widetilde{x}\}.\]
\end{enumerate}
\end{definition}

\begin{theorem}[Soundness]
If $\textbf{STC}\vdash P=Q$ then
\begin{enumerate}
  \item $P\sim_p^{sl} Q$;
  \item $P\sim_s^{sl} Q$;
  \item $P\sim_{hp}^{sl} Q$.
\end{enumerate}
\end{theorem}

\begin{proof}
The soundness of these laws modulo strongly truly concurrent bisimilarities is already proven in Section \ref{stcb5}.
\end{proof}

\begin{definition}
The agent identifier $A$ is weakly guardedly defined if every agent identifier is weakly guarded in the right-hand side of the definition of $A$.
\end{definition}

\begin{definition}[Head normal form]
A Process $P$ is in head normal form if it is a sum of the prefixes:

$$P\equiv \sum_i(loc_{i1}::\alpha_{i1}\parallel\cdots\parallel loc_{in}::\alpha_{in}).P_i.$$
\end{definition}

\begin{proposition}
If every agent identifier is weakly guardedly defined, then for any process $P$, there is a head normal form $H$ such that

$$\textbf{STC}\vdash P=H.$$
\end{proposition}

\begin{proof}
It is sufficient to induct on the structure of $P$ and quite obvious.
\end{proof}

\begin{theorem}[Completeness]
For all processes $P$ and $Q$,
\begin{enumerate}
  \item if $P\sim_p^{sl} Q$, then $\textbf{STC}\vdash P=Q$;
  \item if $P\sim_s^{sl} Q$, then $\textbf{STC}\vdash P=Q$;
  \item if $P\sim_{hp}^{sl} Q$, then $\textbf{STC}\vdash P=Q$.
\end{enumerate}
\end{theorem}

\begin{proof}
\begin{enumerate}
  \item if $P\sim_s^{sl} Q$, then $\textbf{STC}\vdash P=Q$.
Since $P$ and $Q$ all have head normal forms, let $P\equiv\sum_{i=1}^k\alpha_i.P_i$ and $Q\equiv\sum_{i=1}^k\beta_i.Q_i$. Then the depth of $P$, denoted as $d(P)=0$, if $k=0$;
$d(P)=1+max\{d(P_i)\}$ for $1\leq i\leq k$. The depth $d(Q)$ can be defined similarly.

It is sufficient to induct on $d=d(P)+d(Q)$. When $d=0$, $P\equiv\textbf{nil}$ and $Q\equiv\textbf{nil}$, $P=Q$, as desired. Note that, we consider the general distribution.

Suppose $d>0$.

\begin{itemize}
  \item If $(\alpha_1\parallel\cdots\parallel\alpha_n).M$ with $\alpha_i(1\leq i\leq n)$ free actions is a summand of $P$, then $P\xrightarrow[u]{\{\alpha_1,\cdots,\alpha_n\}}M$. Since
  $Q$ is in head normal form and has a summand $(\alpha_1\parallel\cdots\parallel\alpha_n).N$ such that $M\sim_s^{sl} N$, by the induction hypothesis $\textbf{STC}\vdash M=N$,
  $\textbf{STC}\vdash (\alpha_1\parallel\cdots\parallel\alpha_n).M= (\alpha_1\parallel\cdots\parallel\alpha_n).N$;
  \item If $x(y).M$ is a summand of $P$, then for $z\notin n(P, Q)$, $P\xrightarrow[u]{x(z)}M'\equiv M\{z/y\}$. Since $Q$ is in head normal form and has a summand $x(w).N$ such that for
  all $v$, $M'\{v/z\}\sim_s^{sl} N'\{v/z\}$ where $N'\equiv N\{z/w\}$, by the induction hypothesis $\textbf{STC}\vdash M'\{v/z\}=N'\{v/z\}$, by the axioms $\textbf{C}$ and
  $\textbf{A}$, $\textbf{STC}\vdash x(y).M=x(w).N$;
  \item If $\overline{x}(y).M$ is a summand of $P$, then for $z\notin n(P,Q)$, $P\xrightarrow[u]{\overline{x}(z)}M'\equiv M\{z/y\}$. Since $Q$ is in head normal form and has a summand
  $\overline{x}(w).N$ such that $M'\sim_s^{sl} N'$ where $N'\equiv N\{z/w\}$, by the induction hypothesis $\textbf{STC}\vdash M'=N'$, by the axioms $\textbf{A}$ and $\textbf{C}$,
  $\textbf{STC}\vdash \overline{x}(y).M= \overline{x}(w).N$;
\end{itemize}

  \item if $P\sim_p^{sl} Q$, then $\textbf{STC}\vdash P=Q$. It can be proven similarly to the above case.
  \item if $P\sim_{hp}^{sl} Q$, then $\textbf{STC}\vdash P=Q$. It can be proven similarly to the above case.
\end{enumerate}
\end{proof}

%% file: section6/section6.3.tex
\subsection{$\pi_{tc}$ with Dynamic Localities}\label{pitcdl}

$\pi_{tc}$ with dynamic localities is almost the same as $\pi_{tc}$ with static localities in section \ref{pitcsl}, as the locations are dynamically generated but not allocated statically. The LTSs-based
operational semantics and the laws are almost the same, except for the transition rules of $\textbf{Act}$ as follows.

\[\textbf{OUTPUT-ACT}\quad \frac{}{\overline{x}y.P\xrightarrow[u]{\overline{x}y}u::P}\]

\[\textbf{INPUT-ACT}\quad \frac{}{x(z).P\xrightarrow[u]{x(w)}u::P\{w/z\}}\quad (w\notin fn((z)P))\]